%% file: gibbs.tex
\documentclass[11pt,titlepage]{aart}

\usepackage[letterpaper, hmargin=1in, top=1.1in, bottom=1.3in, footskip=0.7in]{geometry} 

\usepackage{titlesec}
\titleformat{\section}[block]{\filcenter\normalfont\bfseries\large}{\thesection.}{.5em}{}\titlespacing*{\section}{0pt}{2\baselineskip}{1\baselineskip}
\titleformat{\subsection}[runin]{\normalfont\bfseries}{\thesubsection.}{.4em}{}[.]\titlespacing{\subsection}{0pt}{2ex plus .1ex minus .2ex}{.8em}
\titleformat{\subsubsection}[runin]{\normalfont\itshape}{\thesubsubsection.}{.3em}{}[.]\titlespacing{\subsubsection}{0pt}{1ex plus .1ex minus .2ex}{.5em}
\titleformat{\paragraph}[runin]{\normalfont\itshape}{\theparagraph.}{.3em}{}[.]\titlespacing{\paragraph}{0pt}{1ex plus .1ex minus .2ex}{.5em}

\usepackage[labelfont=sc,font=small,labelsep=period]{caption}
\usepackage{amsmath}
\usepackage{amssymb}
\usepackage{amsfonts}
\usepackage{latexsym}
\usepackage{amsthm}
\usepackage{amsxtra}
\usepackage{amscd}
\usepackage{bbm}
\usepackage{mathrsfs}
\usepackage{bm}
\usepackage{mathtools}

\usepackage{graphicx, color}

\definecolor{darkred}{rgb}{0.9,0,0.3}
\definecolor{darkblue}{rgb}{0,0.3,0.9}
\definecolor{darkgreen}{rgb}{0,0.8,0.2}

\definecolor{vdarkred}{rgb}{0.6,0,0.2}
\definecolor{vdarkblue}{rgb}{0,0.2,0.6}
\usepackage[pdftex, colorlinks, linkcolor=vdarkblue,citecolor=vdarkred]{hyperref}

\usepackage[nottoc,notlof,notlot]{tocbibind}
\usepackage{cite} 

\numberwithin{equation}{section}
\numberwithin{figure}{section}

\theoremstyle{plain} 
\newtheorem{theorem}{Theorem}[section]
\newtheorem*{theorem*}{Theorem}
\newtheorem{lemma}[theorem]{Lemma}
\newtheorem*{lemma*}{Lemma}
\newtheorem{corollary}[theorem]{Corollary}
\newtheorem*{corollary*}{Corollary}
\newtheorem{proposition}[theorem]{Proposition}
\newtheorem*{proposition*}{Proposition}

\newtheorem*{conjecture*}{Conjecture}

\theoremstyle{definition} 
\newtheorem{definition}[theorem]{Definition}
\newtheorem*{definition*}{Definition}

\newtheorem*{example*}{Example}
\newtheorem{remark}[theorem]{Remark}
\newtheorem*{remark*}{Remark}

\newtheorem*{assumption*}{Assumption}

\newcommand{\f}[1]{\boldsymbol{\mathrm{#1}}} 
\renewcommand{\r}{\mathrm}  
\newcommand{\bb}{\mathbb} 
\renewcommand{\cal}{\mathcal} 
 
\newcommand{\fra}{\mathfrak} 

\newcommand{\wt}{\widetilde}
\newcommand{\txt}[1]{\text{\rm{#1}}}

\renewcommand{\P}{\mathbb{P}}
\newcommand{\E}{\mathbb{E}}
\newcommand{\R}{\mathbb{R}}
\newcommand{\C}{\mathbb{C}}
\newcommand{\N}{\mathbb{N}}
\newcommand{\Z}{\mathbb{Z}}

\newcommand{\ee}{\mathrm{e}}
\newcommand{\ii}{\mathrm{i}}
\newcommand{\dd}{\mathrm{d}}
\newcommand{\col}{\vcentcolon}
\newcommand*{\deq}{\mathrel{\vcenter{\baselineskip0.65ex \lineskiplimit0pt \hbox{.}\hbox{.}}}=}
\newcommand*{\eqd}{=\mathrel{\vcenter{\baselineskip0.65ex \lineskiplimit0pt \hbox{.}\hbox{.}}}}

\renewcommand{\leq}{\leqslant}
\renewcommand{\geq}{\geqslant}
\renewcommand{\epsilon}{\varepsilon}
\newcommand{\Res}{\mathrm{Res}}

\newcommand{\ind}[1]{\f 1 (#1)}

\newcommand{\indbb}[1]{\f 1 \pbb{#1}}

\newcommand{\p}[1]{({#1})}
\newcommand{\pb}[1]{\bigl({#1}\bigr)}
\newcommand{\pB}[1]{\Bigl({#1}\Bigr)}
\newcommand{\pbb}[1]{\biggl({#1}\biggr)}
\newcommand{\pBB}[1]{\Biggl({#1}\Biggr)}

\newcommand{\qb}[1]{\bigl[{#1}\bigr]}

\newcommand{\qbb}[1]{\biggl[{#1}\biggr]}
\newcommand{\qBB}[1]{\Biggl[{#1}\Biggr]}

\newcommand{\hb}[1]{\bigl\{{#1}\bigr\}}

\newcommand{\abs}[1]{\lvert #1 \rvert}
\newcommand{\absb}[1]{\bigl\lvert #1 \bigr\rvert}

\newcommand{\absBB}[1]{\Biggl\lvert #1 \Biggr\rvert}

\newcommand{\norm}[1]{\lVert #1 \rVert}
\newcommand{\normb}[1]{\bigl\lVert #1 \bigr\rVert}
\newcommand{\normB}[1]{\Bigl\lVert #1 \Bigr\rVert}
\newcommand{\normbb}[1]{\biggl\lVert #1 \biggr\rVert}

\newcommand{\scalar}[2]{\langle{#1} \mspace{2mu}, {#2}\rangle}
\newcommand{\scalarb}[2]{\bigl\langle{#1} \mspace{2mu}, {#2}\bigr\rangle}

\newcommand{\scalarbb}[2]{\biggl\langle{#1} \,\mspace{2mu},\, {#2}\biggr\rangle}

\DeclareMathOperator{\tr}{Tr}

\DeclareMathOperator{\re}{Re}
\DeclareMathOperator{\im}{Im}

\DeclareMathOperator{\dist}{dist}

\DeclareMathOperator{\spec}{spec}
\DeclareMathOperator{\conn}{conn}

\newcommand{\va}{a}
\newcommand{\vb}{b}

\newcommand{\sign}{\delta}
\newcommand{\rad}{r}

\newcommand{\CP}{C_0}
\newcommand{\CA}{C_1}
\newcommand{\CB}{C_2}

\title{Gibbs measures of nonlinear Schr\"odinger equations as limits of many-body quantum states in dimensions $d \leq 3$}

\author{	
J\"urg Fr\"ohlich\footnote{ETH Z\"urich, Institute for Theoretical Physics, {\tt juerg@phys.ethz.ch}.}
\and Antti Knowles\footnote{University of Geneva, Section of Mathematics, {\tt antti.knowles@unige.ch}.}
\and Benjamin Schlein\footnote{University of Z\"urich, Institute of Mathematics, {\tt benjamin.schlein@math.uzh.ch}.}
\and Vedran Sohinger\footnote{University of Warwick, Mathematics Institute, {\tt V.Sohinger@warwick.ac.uk}.}
}

\begin{document}

\maketitle

\begin{abstract}

We prove that Gibbs measures of nonlinear Schr\"odinger equations arise as high-temperature limits of thermal states in many-body quantum mechanics. Our results hold for defocusing interactions in dimensions $d =1,2,3$. The many-body quantum thermal states that we consider are the grand canonical ensemble for $d = 1$ and an appropriate modification of the grand canonical ensemble for $d =2,3$. In dimensions $d =2,3$, the Gibbs measures are supported on singular distributions, and a renormalization of the chemical potential is necessary. On the many-body quantum side, the need for renormalization is manifested by a rapid growth of the number of particles. We relate the original many-body quantum problem to a renormalized version obtained by solving a counterterm problem. Our proof is based on ideas from field theory, using a perturbative expansion in the interaction, organized by using a diagrammatic representation, and on Borel resummation of the resulting series.

\end{abstract}

\newpage
\tableofcontents
\pagebreak

\section{Introduction}

\subsection{Overview}

An invariant Gibbs measure $\P$ of a nonlinear Schr\"odinger equation (NLS) is, at least formally, defined as a probability measure on the space of fields $\phi$ that takes the form
\begin{equation} \label{def_P}
\P(\dd \phi) \;=\; \frac{1}{Z} \ee^{- H(\phi)} \, \dd \phi\,,
\end{equation}
where $Z$ is a normalization constant, $H$ is the Hamilton function, and $\dd \phi$ is the (nonexistent) Lebesgue measure on the space of fields.
We consider fields $\phi \col \Lambda \to \C$ defined on physical space $\Lambda$, which we take to be either the Euclidean space $\R^d$ or the torus $\bb T^d = \R^d / \Z^d \simeq [0,1)^d$ with $d = 1,2,3$, each endowed with their natural operations of addition and subtraction. We are interested in Hamilton functions of the form
\begin{equation} \label{def_classical_H}
H(\phi) \;\deq\; \int \dd x \, \pb{\abs{\nabla \phi(x)}^2 + V(x) \abs{\phi(x)}^2} + \frac{1}{2} \int \dd x \, \dd y \, \abs{\phi(x)}^2 \, w(x - y) \, \abs{\phi(y)}^2\,,
\end{equation}
where $V \geq 0$ is a one-body potential and $w$ is an interaction potential. We always assume $w$ to be repulsive or defocusing, meaning that $w$ or its Fourier transform is nonnegative. Typical examples for $w$ include bounded continuous functions and the delta function. 

Formally, the space of fields $\phi \col \Lambda \to \C$ generates a Poisson algebra with Poisson bracket defined by
\begin{equation*}
\{\phi(x),\bar \phi(y)\} \;=\; \ii \delta(x - y) \,, \qquad \{\phi(x), \phi(y)\} \;=\; \{\bar \phi(x),\bar \phi(y)\} \;=\; 0\,.
\end{equation*}
The Hamiltonian equation of motion associated with the Hamilton function \eqref{def_classical_H} is then given by the time-dependent NLS\footnote{If $w$ is a bounded interaction potential, this equation is often called \emph{Hartree equation}.}
\begin{equation} \label{nonlinear_Schr}
\ii \partial_t \phi(x) \;=\; -\Delta \phi(x) + V(x) \phi(x) + \int \dd y \, \abs{\phi(y)}^2 \, w(x - y)\, \phi(x)\,.
\end{equation}
At least formally, we find that the Gibbs measure \eqref{def_P} is invariant under the flow generated by the NLS \eqref{nonlinear_Schr}. Such invariant Gibbs measures have been extensively studied as tools to construct global solutions of time-dependent NLS with rough initial data
\cite{B,B1,B2,B3,BourgainBulut,BourgainBulut2,BourgainBulut4,BurqThomannTzvetkov,Cacciafesta_deSuzzoni1,Cacciafesta_deSuzzoni2,Deng,LRS,Tz1,Tz,OQ,Zhidkov}.

More precisely, the existence of global solutions is shown for almost all initial data belonging to the support of \eqref{def_P}. The measure \eqref{def_P} is typically supported on a set of distributions of low regularity. In this framework, the invariance of \eqref{def_P} serves as a substitute for a conservation law at low regularities.

In this paper, we are interested in deriving invariant Gibbs measures as high-temperature limits of grand canonical thermal states of many-body quantum mechanics. 
In the grand canonical high-temperature limit, the number of particles grows with the temperature. In order to obtain a nontrivial limit, we need to rescale the strength of the many-body interaction potential with the temperature. This gives rise to a \emph{mean-field} limiting regime for the many-body system, which may be also regarded as a 
classical limit. Here, the classical system is the (formal) Hamiltonian system defined by the Hamilton function $H(\phi)$ from \eqref{def_classical_H}.

The study of the classical limit of quantum mechanics is almost as old as quantum mechanics itself, going back at least to the works of Schr\"{o}dinger \cite{Schr1,Schr2} and Ehrenfest \cite{E}. The first rigorous treatment of the classical limit for systems with infinitely many degrees of freedom is due to Hepp, who recognized in \cite{H} that the time-dependent NLS \eqref{nonlinear_Schr} arises as the Hamiltonian equation of motion approximating the many-body quantum time evolution of coherent states. In his work, Hepp also proved that the time evolution of the fluctuations around the classical dynamics is governed by a time-dependent quadratic Hamiltonian on the Fock space (quadratic in the creation and annihilation operators). The results of Hepp were later extended by Ginibre and Velo \cite{GV} to singular interactions.

A many-body quantum system of $n$ particles has a Hamiltonian of the form 
\begin{equation} \label{eq:ham-first}
H^{(n)} \;\deq\; \sum_{i = 1}^n \pb{-\Delta_{x_i} + V(x_i)} + \lambda \sum_{1 \leq i < j \leq n} w(x_i - x_j)\,,
\end{equation}
acting on the bosonic Hilbert space consisting of wave functions in $L^2(\Lambda^n)$ that are symmetric in their arguments $x_1, \dots, x_n \in \Lambda$. In \eqref{eq:ham-first}, $\lambda > 0$ is the interaction strength. In order to obtain a nontrivial limit as $n \to \infty$, we require both terms of \eqref{eq:ham-first} to be of comparable size, which leads to the mean-field scaling $\lambda = n^{-1}$. The dynamical problem for factorized initial data, where one analyses the convergence of the many-body dynamics $\ee^{-\ii t H^{(n)}}$ generated by \eqref{eq:ham-first} to that generated by \eqref{nonlinear_Schr}, has been extensively studied since the work of Hepp \cite{H} mentioned above; see 
\cite{AFP,AN,CLS,CP,CH, CH2, ESY1,ESY2,ESY3,ESY4,ESY5,EESY,ES,EY,FKP,FKS,HS,KSS,KP,RS,S2,S}. Moreover, quantum fluctuations around the classical dynamics have been considered in \cite{BKS,BSS,XC,GMM1,GMM2,LNS}. 
More recently, fluctuations around the dynamics generated by the NLS with a local interaction (arising from an interaction potential converging to the delta function) have been analysed in \cite{BCS,GM,NN}.

In this paper we focus on the equilibrium state of the many-body quantum system at some given temperature. At zero temperature, the many-body quantum system is in the ground state of the Hamiltonian \eqref{eq:ham-first}. The convergence of the ground state energy of (\ref{eq:ham-first}) towards the minimum of (\ref{def_classical_H}) has been proved, for different choices of the interaction $w$, in  \cite{Bach,BL,K,S,LY,FSV,RW,LSY} and in a more general setting in \cite{LNR0} and in \cite{S,GS,LNSS}, where also the excitation spectrum of (\ref{eq:ham-first}) has been analysed. Under general assumptions on $V,w$ it is also possible to prove (see \cite{LNR0} and also \cite{LS}, for the more subtle Gross-Pitaevskii regime) that the ground state of \eqref{eq:ham-first} exhibits complete condensation, meaning that all particles, up to a fraction vanishing in the limit of large $n$, occupy the orbital $\phi_0 \in L^2 (\R^d)$ minimizing (\ref{def_classical_H}) (provided the minimizer is unique). 
Results on Bose-Einstein condensation of the canonical ensemble for fixed temperature $\tau>0$ were proved in \cite{LNR0,LNSS}.

To obtain a nontrivial limiting measure $\P (\dd \phi)$, one has to increase the temperature $\tau$ of the system in tandem with the particle number $n$.
In \cite{Lewin_Nam_Rougerie}, Lewin, Nam, and Rougerie considered the thermal states associated with (\ref{eq:ham-first}) in the grand canonical ensemble at temperature $\tau$ and with a chemical potential fixing the expected number of particles to be $\tau$. They compared the correlation functions of the grand canonical ensemble at temperature $\tau$ with the correlation functions of the Gibbs measure \eqref{def_P}, in the limit $\tau \to \infty$.
In dimension $d=1$, they proved the convergence of the (relative) partition function (i.e.\ the ratio between the free and the interacting partition functions)
and of all correlation functions. In \cite{Lewin_Nam_Rougerie2}, they further extended their results to sub-harmonic trapping. Moreover, in dimension $d=2,3$ they considered a many-body quantum model with a smooth, non-translation-invariant interaction $w: L^2 (\R^d) \otimes L^2 (\R^d) \to L^2 (\R^d) \otimes L^2 (\R^d)$ satisfying $0 \leq w \leq h^{1-p} \otimes h^{1-p}$ (for $p > 1$ if $d = 2$ and for $p > 3/2$ if $d = 3$). Here $h \deq -\Delta + V$ is the one-body Hamiltonian. Finite-rank operators are a typical example of such smooth interactions. For these models, they established convergence of the (relative) partition function and of the one-point correlation function towards the corresponding classical limits, where the interaction term of \eqref{def_classical_H} is replaced by the smooth interaction $\frac{1}{2} \scalar{\phi^{\otimes 2}}{w \phi^{\otimes 2}}$. 

The main result of our work is the derivation in dimensions $d = 2,3$ of the Gibbs measure \eqref{def_P}, \eqref{def_classical_H} as the high-temperature limit of a modified grand canonical ensemble for a many-body quantum system. In contrast to the work of Lewin, Nam, and Rougerie \cite{Lewin_Nam_Rougerie}, we consider interactions defined by a translation-invariant two-body potential $w(x - y)$ that we choose to be bounded and of positive type. For technical reasons, the starting point of our derivation is an appropriate modification of the standard grand canonical many-body quantum thermal states. This modification enables us to control the remainder term in the perturbative expansion of the many-body quantum state.
It is well known that, unlike in dimension $d = 1$, in dimensions $d = 2,3$ the free Gibbs measure \eqref{def_P} corresponding to $w = 0$ is supported on distributions of negative regularity, i.e.\ $\phi$ is $\P$-almost surely not a function but a singular distribution. Thus, some care is already needed to define the classical probability measure (\ref{def_P}). Indeed, in order to define the Gibbs measure \eqref{def_P} with a nonzero interaction $w$, one has to renormalize the interaction term $\frac{1}{2} \int \dd x \, \dd y \, \abs{\phi(x)}^2 \, w(x - y) \, \abs{\phi(y)}^2$. This renormalization may be performed by a \emph{Wick ordering} of the interaction, whereby it is replaced with a formal expression of the form $\frac{1}{2} \int \dd x \, \dd y \, (\abs{\phi(x)}^2 - \infty)\, w(x - y) \, (\abs{\phi(y)}^2 - \infty)$. Here the infinities are carefully defined by introducing an ultraviolet truncation parameter in the classical field $\phi$, and subtracing from $\abs{\phi}^2$ a function that diverges as the truncation parameter is sent to infinity.

Our methods also apply to the simpler case $d = 1$ where no renormalization is necessary, and provide an alternative approach to the one developed in \cite{Lewin_Nam_Rougerie}. Unlike in dimensions $d = 2,3$, in dimension $d = 1$ we do not need to exploit delicate cancellations arising from the renormalization, and a simple argument using the Feynman-Kac formula may be used to estimate the remainder term of the perturbative expansion of the quantum state. As a consequence, for $d = 1$ we do not need the modification mentioned above and we can consider the usual grand canonical thermal states.

In the PDE literature, the existence and the invariance under \eqref{nonlinear_Schr} of the Gibbs measures (\ref{def_P}) was first shown by Bourgain in \cite{B,B1,B2} (their existence was previously essentially shown by Lebowitz, Rose, and Speer in \cite{LRS}, for one-dimensional systems with local and focusing interaction, whereas the defocusing problem was previously considered in the constructive quantum field theory literature \cite{Glimm_Jaffe,Simon74}). 
As noted above, the invariance of the measure \eqref{def_P} provides a method to construct global solutions of \eqref{nonlinear_Schr} for almost all initial data in the support of \eqref{def_P}.
An alternative method to study local solutions of nonlinear wave equations with rough random initial data was developed in \cite{BT} and was applied to construct global solutions  of the NLS with rough random initial data in \cite{CO}, based on the low-high decomposition from \cite{Bourgain1998}.
Related ideas have been applied in various other dispersive models. We refer the reader to \cite{Bourgain_ZS,BT2,deS,DengTzvetkovVisciglia,GLV,NORBS,NRBSS,TTz} and the references therein for further results in this direction.

In \cite{BrydgesSlade} it was observed that, in dimensions $d > 1$, for a focusing local interaction (where $w = -\delta$ in \eqref{def_classical_H}) one cannot construct measures \eqref{def_P} even if one adds additional truncation assumptions. For an appropriate focusing nonlocal interaction, it was shown in \cite{B2} that for $d=2,3$ it is possible to construct the measure \eqref{def_P} provided that one truncates the Wick-ordered (square) $L^2$-norm.
In addition, in \cite{B2} invariance of this measure was shown.

Unlike the classical system \eqref{def_P}, \eqref{def_classical_H}, the many-body quantum system carries an intrinsic ultraviolet (i.e.\ high-frequency) cutoff, which is proportional to the temperature $\tau$. Indeed, our results may be interpreted as a construction of the Gibbs state of the NLS as a limit of regularized states, with the temperature $\tau$ playing the role of the regularization parameter. This construction is very physical, starting from many-body quantum states, unlike the much simpler but less physical construction by truncation given in Sections \ref{sec:classical_state} and \ref{sec:Wick} below.

In the limit $\tau \to \infty$, the expected number of particles in the many-body grand canonical ensemble grows much faster for $d = 2,3$ than for $d = 1$; this is the manifestation, on the many-body quantum level, of the singularity of the limiting classical field. As for the classical system, we also have to renormalize the chemical potential of the quantum many-body problem. Our main result, Theorem \ref{thm:main} below, is that a modification of the renormalized many-body quantum grand canonical ensemble converges, as $\tau \to \infty$, to the (renormalized) Gibbs measure \eqref{def_P}. Here the convergence is in the sense of the relative partition function and all correlation functions. We also establish the relationship between the original and renormalized many-body quantum problems, which is governed by the so-called counterterm problem.

Our proof uses a different approach from that Lewin, Nam, and Rougerie \cite{Lewin_Nam_Rougerie}. While \cite{Lewin_Nam_Rougerie} makes use of Gibbs' variational principle and the quantum de Finetti theorem, our approach is based on ideas from field theory, using a perturbative expansion in the interaction (for both classical and quantum problems) which is organized using a diagrammatic representation, and on Borel resummation of the resulting series; we use a version of Borel resummation going back to Sokal \cite{Sokal}. We refer to Section \ref{sec:proof_outline} below for a more detailed overview of our proof.

\subsubsection*{Conventions}
We use $C$ to denote a constant that may depend on fixed quantities (such as $w$). If a constant depends on some parameter $\alpha$ then we write it as $C_\alpha$. We use the notation $\N = \{0,1,2,\dots\}$. For a separable Hilbert space $\cal H$ and $q \in [1,\infty]$, the Schatten space $\fra S^q(\cal H)$ is the set of bounded operators $\cal A$ on $\cal H$ satisfying $\|\cal A\|_{\fra S^q(\cal H)} < \infty$, where
\begin{equation*}
\|\cal A\|_{\fra S^q(\cal H)} \;\deq\;
\begin{cases}
( \tr \, \abs{\cal A}^q)^{1/q}  &\mbox{if }q<\infty\\
\sup \spec \, \abs{\cal A} &\mbox{if } q=\infty\,,
\end{cases}
\end{equation*}
and $\abs{\cal A} \deq \sqrt{\cal A^* \cal A}$. If there is no risk of confusion, we sometimes omit the argument $\cal H$ in these norms. We denote by $\dd x$ the Lebesgue measure on $\Lambda$, and we often abbreviate $\int_\Lambda \dd x \equiv \int \dd x$.

\subsubsection*{Acknowledgements}
We are grateful to Nicolas Rougerie for helpful discussions. We gratefully acknowledge support from the NCCR SwissMAP of the Swiss National Foundation of Science (SNF). A.\ Knowles acknowledges SNF support through the grant ``Spectral and eigenvector statistics of large random matrices''. B.\ Schlein also acknowledges SNF support through the Grant ``Dynamical and energetic properties of Bose-Einstein condensates''. V.\ Sohinger acknowledges support of the National Science Foundation Grant No.\ DMS-1440140 while he was in residence at the Mathematical Sciences Research Institute in Berkeley, California, during part of the Fall 2015 semester.

\subsection{The one-body Hamiltonian} \label{sec:h}
We define the \emph{one-particle space} $\fra H \deq L^2(\Lambda; \C)$,
whose scalar product and norm we denote by $\scalar{\cdot}{\cdot}$ and $\norm{\cdot}$ respectively.
We always use the convention that scalar products are linear in the second argument. For $p \in \N$, we define the \emph{$p$-particle space} $\fra H^{(p)}$ as the symmetric subspace of the tensor product $\fra H^{\otimes p}$, i.e.\ the space of functions in $L^2(\Lambda^p;\C)$ that are symmetric under permutation of their arguments.

It is often convenient to identify a closed operator $\xi$ on $\fra H^{(p)}$ with its Schwartz integral kernel, which we denote by $\xi(x_1, \dots, x_p;y_1, \dots, y_p)$. The latter is in general a tempered distribution (see e.g.\ \cite[Corollary V.4.4]{RS1}), and we shall always integrate over the variables $x_1, \dots, x_p, y_1, \dots, y_p$ with respect to a sufficiently regular test function.
Similarly, for a Schwartz distribution $T \col f \mapsto T(f)$, we sometimes use the notation $T(x) \equiv T(\delta_x)$ for the integral kernel of $T$ in expressions of the form $T(f)  = \int \dd x \, T(x) f(x)$.

Let $\kappa > 0$ be a \emph{chemical potential} and $v \col \Lambda \to [0,\infty)$ be a \emph{one-body potential}.
We define the \emph{one-body Hamiltonian}
\begin{equation} \label{def_h}
h \;\deq\; -\Delta + \kappa + v\,,
\end{equation}
a densely defined positive operator on $\fra H$. We assume that $h$ has a compact resolvent and that
\begin{equation} \label{tr_h_assump}
\tr h^{s - 1} \;<\; \infty
\end{equation}
for some $s < 1$. We shall mainly focus on the case $s = -1$, which is relevant in dimensions $d = 2,3$, but the case $s = 0$ is also of interest for $d = 1$. 

We note that the assumption \eqref{tr_h_assump} with $s=-1$ is satisfied when 
\begin{equation*}
\begin{cases}
\Lambda\;=\;\mathbb{T}^d \quad &\text{and} \quad v\;=\;0 \\
\Lambda\;=\;\mathbb{R}^d \quad &\text{and} \quad v\;=\; |x|^r \quad \text{for} \quad r\;>\;\frac{2d}{4-d}\,.
\end{cases}
\end{equation*}
The first claim follows immediately since $d \leq 3$. The second claim is a consequence of the Lieb-Thirring inequality in \cite[Theorem 1]{Dolbeault_Felmer_Loss_Paturel}; see also \cite[Example 3.2]{Lewin_Nam_Rougerie}. In particular, on $\mathbb{R}^2$ the potential is infinitesimally more confining than the harmonic oscillator. 

More generally, we have the following result.

\begin{lemma}
\label{v_assump}
When $\Lambda=\mathbb{T}^d$ and $v=0$ \eqref{tr_h_assump} holds whenever $s<1-\frac{d}{2}$. Furthermore, when $\Lambda = \mathbb{R}^d$ \eqref{tr_h_assump} holds whenever $s<1-\frac{d}{2}$ and the potential $v \in C^{\infty}(\mathbb{R}^d)$ is chosen such that 
\begin{equation*}
\kappa+v \;\geq \; 0\,, \qquad
(\kappa + v)^{\frac{d}{2}-1+s} \in L^1(\mathbb{R}^d)\,. 
\end{equation*}
\end{lemma}
\begin{proof}
The claim for $\Lambda=\mathbb{T}^d$ follows immediately since $\frac{1}{(\kappa+|n|^2)^{d/2+}} \in \ell^2 (\mathbb{Z}^d)$ and the claim for $\Lambda=\mathbb{R}^d$ follows from \cite[Theorem 1]{Dolbeault_Felmer_Loss_Paturel}. \qedhere
\end{proof}
Throughout the following, we regard $\kappa$ and $v$ as fixed, and do not track the dependence of our estimates on them.

\subsection{The classical system and Gibbs measures} \label{sec:classical_state}

For $r \in \R$ denote by $\fra H_r$ the Hilbert space of complex-valued Schwartz distributions on $\Lambda$ with inner product $\scalar{f}{g}_{\fra H_r} \deq \scalar{f}{h^r g}$. In particular, $\fra H_0 = \fra H$.
We define the \emph{classical free field} as the abstract Gaussian process on the Hilbert space $\fra H_{-1}$. For completeness and later use, we give an explicit construction.

Consider an infinite sequence of independent standard complex Gaussians. More precisely, we introduce the probability space $(\C^\N, \cal G, \mu)$, where $\cal G$ is the product sigma-algebra and $\mu \deq \bigotimes_{k \in \N} \mu_k$ with $\mu_k(\dd z) \deq \pi^{-1} \ee^{- \abs{z}^2} \dd z$, where $\dd z$ denotes Lebesgue measure on $\C$. We denote points of the probability space $\C^\N$ by $\omega = (\omega_k)_{k \in \N}$. We use the notation
\begin{equation} \label{h_spectrum}
h \;=\; \sum_{k \in \N} \lambda_k u_k u_k^*
\end{equation}
for the eigenvalues $\lambda_k > 0$ and associated normalized eigenfunctions $u_k \in \fra H$ of $h$. For any $K \in \N$ we define the \emph{truncated classical free field}
\begin{equation} \label{truncated classical free field}
\phi_{[K]} \;\deq\; \sum_{k = 0}^K \frac{\omega_k}{\sqrt{\lambda_k}} u_k\,,
\end{equation}
which is a random element in $\fra H$. We immediately find, for every $f \in \fra H_{-1}$, that $\scalar{f}{\phi_{[K]}}$ converges in $L^2(\mu)$ to a random variable denoted by $\phi(f)$, which is antilinear in $f$. Under $\mu$, the process $(\phi(f))_{f \in \fra H_{-1}}$ is the Gaussian free field with covariance $h^{-1}$:
\begin{equation} \label{classical_covariance}
\int \dd \mu \, \bar \phi(g) \, \phi(f) \;=\; \scalar{f}{h^{-1} g} \,, \qquad \int \dd \mu \, \phi(g) \, \phi(f) \;=\; \int \dd \mu \, \bar \phi(g) \, \bar \phi(f) \;=\; 0\,.
\end{equation}
Moreover, from \eqref{tr_h_assump} we easily find that $\phi_{[K]}$ converges in $L^2(\mu; \fra H_s)$ to
\begin{equation} \label{phi_series}
\phi \;=\; \sum_{k \in \N} \frac{\omega_k}{\sqrt{\lambda_k}} u_k\,,
\end{equation}
so that for $f \in \fra H_{-s} \subset \fra H_{-1}$ we may interpret $\phi(f)$ as the dual pairing of $f$ with a random element $\phi \in \fra H_s$, the \emph{classical free field}:
\begin{equation} \label{classical free field}
\phi(f) \;=\; \scalar{f}{\phi} \,, \qquad \bar \phi(f) \;=\; \scalar{\phi}{f}\,.
\end{equation}
(In fact, an application of Wick's theorem shows that all of the above convergences in $L^2(\mu)$ hold in $L^m(\mu)$ for any $m < \infty$.) 

Next, we define the interaction. Let $w$ be an even function on $\Lambda$ and define the
\emph{classical interaction}
\begin{equation} \label{def_W}
W \;\deq\; \frac{1}{2} \int \dd x \, \dd y \, \abs{\phi(x)}^2 \, w(x - y) \, \abs{\phi(y)}^2\,.
\end{equation}
Note that $W \geq 0$ if $w$ is pointwise nonnegative or if $w$ is of \emph{positive type}, meaning its Fourier transform is a positive measure. We shall always make one of these two assumptions. Physically, they correspond to a repulsive or defocusing interaction. Moreover, a sufficient condition that $W < \infty$ $\mu$-almost surely is that $w \in L^\infty(\Lambda)$ and \eqref{tr_h_assump} holds with $s = 0$, since in that case $\abs{\phi}^2 \in L^1(\Lambda)$ $\mu$-almost surely. We make these assumptions for now, later relaxing them in Sections \ref{sec:Wick}--\ref{sec:results}.

Next, we define the \emph{classical state} $\rho(\cdot)$ associated with the one-body Hamiltonian $h$ and interaction potential $w$ as the expectation with respect to the normalized probability measure $\frac{1}{Z} \ee^{-W} \dd \mu$. Explicitly, for a random variable $X$ we set
\begin{equation} \label{classical state}
\rho(X) \;\deq\; \frac{\int X \, \ee^{-W} \, \dd \mu}{\int \ee^{-W} \, \dd \mu}\,.
\end{equation}

We characterize the classical state $\rho(\cdot)$ through its moments. To that end, for $p \in \N$, we define the \emph{classical $p$-particle correlation function} $\gamma_p$, an operator on $\fra H^{(p)}$, through its kernel
\begin{equation} \label{classical p-particle dm}
\gamma_p(x_1, \dots, x_p;y_1, \dots, y_p) \;\deq\; \rho \pb{ \bar \phi(y_1) \cdots \bar \phi(y_p) \phi(x_1) \cdots \phi(x_p)}\,.
\end{equation}

\begin{remark} \label{moments_classical}
The family of correlation functions $(\gamma_p)_{p \in \N}$ determines the moments of the classical state $\rho(\cdot)$. Indeed, for $f_1, \dots, f_p, g_1, \dots, g_q \in \fra H_{-1}$, the joint moment is
\begin{equation*}
\rho \pb{\bar \phi(g_1) \cdots \bar \phi(g_q) \phi(f_1) \cdots \phi(f_p)} \;=\;
\begin{cases}
\scalarb{f_1 \otimes \cdots \otimes f_p}{\gamma_p \, g_1 \otimes \cdots \otimes g_q} & \text{if }p = q
\\
0 & \text{if } p \neq q\,.
\end{cases}
\end{equation*}
That the left-hand side vanishes for $p \neq q$ is a consequence of the gauge invariance of the expectation $\rho(\cdot)$: the measure $\ee^{-W} \dd \mu$ on $(\C^N, \cal G)$ is invariant under the rotation $\omega \mapsto \ee^{\ii t} \omega$ for any $t \in \R$, as follows immediately from the definitions of $\mu$ and $W$.
\end{remark}

\begin{remark}
Instead of the randomized eigenfunction expansion representation \eqref{phi_series}, we can also characterize the classical field $\phi$ directly by a measure $\P$ on $\fra H_s$. Indeed, as shown above there is an event $\Omega \in \cal G$ such that $\mu(\C^\N \setminus \Omega) = 0$ and $\phi(\Omega) \subset \fra H_s$. We define $\P \deq (\phi \vert_\Omega)_*\pb{\frac{1}{Z} \ee^{-W} \dd \mu}$ as the pushforward of the measure $\frac{1}{Z} \ee^{-W} \dd \mu$ on $\Omega$ under $\phi$. Thus, we may for instance rewrite \eqref{classical p-particle dm} as
\begin{equation*}
\gamma_p(x_1, \dots, x_p;y_1, \dots, y_p) \;=\; \int_{\fra H_s} \P(\dd \phi) \, \bar \phi(y_1) \cdots \bar \phi(y_p) \phi(x_1) \cdots \phi(x_p)\,.
\end{equation*}
This provides a rigorous construction of the formal measure \eqref{def_P}.
For our purposes, however, it is more convenient to use the representation as a randomized eigenfunction expansion from \eqref{phi_series}.
\end{remark}

\subsection{The quantum system} \label{sec:intro_quantum}
We now define the many-body quantum Hamiltonian. On the $n$-particle space $\fra H^{(n)}$, the many-body Hamiltonian reads
\begin{equation} \label{def_H^n}
H^{(n)} \;\deq\; \sum_{i = 1}^n h_i + \lambda \sum_{1 \leq i < j \leq n} w(x_i - x_j)\,,
\end{equation}
where $h_i$ denotes the operator $h$ acting in the variable $x_i$ and $\lambda > 0$ is the interaction strength. Under the assumptions on $h$ and $w$ from Sections \ref{sec:h}--\ref{sec:classical_state}, $H^{(n)}$ is a densely defined positive self-adjoint operator on $\fra H^{(n)}$. We shall show that the classical state \eqref{classical state} arises as the high-temperature limit of the thermal state associated with the Hamiltonian $H^{(n)}$. Thus, we introduce the fundamental large parameter of our work, $\tau$, which has the interpretation of the temperature. We shall always be interested in the limit $\tau \to \infty$. The canonical ensemble associated with the Hamiltonian \eqref{def_H^n} is defined by the density operator $P_\tau^{(n)} \deq \ee^{-H^{(n)} / \tau}$.

As we shall see below (see Section \ref{sec:high_t}), in the high-temperature limit $\tau \to \infty$, the rescaled number of particles $n$ corresponds to
the \emph{square $L^2$-norm}
\begin{equation} \label{mass}
\cal N \;\deq\; \int \dd x \, \abs{\phi(x)}^2
\end{equation}
of the classical field $\phi$; see \eqref{exp_N_1d} below. Since the latter is not fixed in $\rho(\cdot)$, we need to replace the \emph{canonical} ensemble, defined for a single $n$, with a \emph{grand canonical} ensemble, which admits a fluctuating particle number. To that end, we introduce the bosonic Fock space
\begin{equation*}
\cal F \;\equiv\; \cal F(\fra H) \;\deq\; \bigoplus_{n \in \N} \fra H^{(n)}\,.
\end{equation*}
We then consider the grand canonical density operator $P_\tau \deq \bigoplus_{n \in \N} P_\tau^{(n)}$. This gives rise to the \emph{quantum state} $\rho_\tau(\cdot)$, defined by
\begin{equation} \label{def_rho_tau}
\rho_\tau(\cal A) \;\deq\; \frac{\tr (\cal A P_\tau)}{\tr (P_\tau)}\,
\end{equation}
where $\cal A$ is a closed operator on $\cal F$. We remark that the growth of the number of particles proportionally to the temperature is closely related to the symmetry of $\fra H^{(n)}$ imposed by the bosonic statistics; this relationship is discussed in more detail in \cite{LNR16}, where the very different behaviour for distinguishable Boltzmann statistics is also analysed. 

On Fock space we introduce annihilation and creation operators, whose definitions we now review.
We denote vectors of $\cal F$ by $\Psi = (\Psi^{(n)})_{n \in \N}$. For $f \in \fra H$ we define the bosonic annihilation and creation operators $b(f)$ and $b^*(f)$ on $\cal F$ through
\begin{align}
\label{def_b2}
\pb{b(f) \Psi}^{(n)}(x_1, \dots, x_n) &\;=\; \sqrt{n+1} \int \dd x \, \bar f(x) \, \Psi^{(n+1)} (x,x_1, \dots, x_n)\,,
\\
\label{def_b1}
\pb{b^*(f) \Psi}^{(n)}(x_1, \dots, x_n) &\;=\; \frac{1}{\sqrt{n}} \sum_{i = 1}^n f(x_i) \Psi^{(n - 1)}(x_1, \dots, x_{i - 1}, x_{i+1}, \dots, x_n)
\,.
\end{align}
The operators $b(f)$ and $b^*(f)$ are unbounded closed operators on $\cal F$, and are each other's adjoints. They satisfy the canonical commutation relations
\begin{equation} \label{CCR_b}
[b(f), b^*(g)] \;=\; \scalar{f}{g} \,, \qquad [b(f), b(g)] \;=\; [b^*(f), b^*(g)] \;=\;0\,,
\end{equation}
where $[A,B] \deq AB - BA$ denotes the commutator.

For $f \in \fra H$, we define the rescaled annihilation and creation operators
\begin{equation} \label{def_phitau}
\phi_\tau(f) \;\deq\; \tau^{-1/2} b(f) \,, \qquad \phi_\tau^*(f) \;\deq\; \tau^{-1/2} b^*(f)\,.
\end{equation}
In analogy to \eqref{classical free field}, we call $\phi_\tau$ the \emph{quantum field}.
We regard $\phi_\tau$ and $\phi_\tau^*$ as operator-valued distributions and use the notations
\begin{equation} \label{phi_tau f}
\phi_\tau(f) \;=\; \scalar{f}{\phi_\tau} \;=\; \int \dd x \, \bar f(x) \, \phi_\tau(x)\,, \qquad
\phi^*_\tau(f) \;=\; \scalar{\phi_\tau}{f} \;=\; \int \dd x \, f(x) \, \phi_\tau^*(x)\,.
\end{equation}
The distribution kernels $\phi^*_\tau(x)$ and $\phi_\tau(x)$ satisfy the canonical commutation relations
\begin{equation} \label{CCR}
[\phi_\tau(x),\phi_\tau^*(y)] \;=\; \frac{1}{\tau} \delta(x - y) \,, \qquad
[\phi_\tau(x),\phi_\tau(y)] \;=\; [\phi_\tau^*(x),\phi_\tau^*(y)] \;=\; 0\,.
\end{equation}

Next, we define the rescaled particle number operator
\begin{equation} \label{N_tau_sc}
\cal N_\tau \;\deq\; \frac{1}{\tau}\bigoplus_{n \in \N} \cal N^{(n)} \;=\; \int \dd x \, \phi_\tau^*(x) \, \phi_\tau(x)\,.
\end{equation}
Moreover, we can write the Hamiltonian \eqref{def_H^n}, rescaled by $\frac{1}{\tau}$ and extended to Fock space, as
\begin{equation} \label{H_tau_1d}
H_\tau \;\deq\; \frac{1}{\tau} \bigoplus_{n \in \N} H^{(n)} \;=\; \int \dd x \, \dd y \, \phi^*_\tau(x) \, h(x;y) \, \phi_\tau(y)
+ 
\frac{\lambda \tau}{2} \int \dd x \, \dd y \, \phi^*_\tau(x) \phi^*_\tau(y)  \, w(x - y) \, \phi_\tau(x) \phi_\tau(y)\,.
\end{equation}
With these notations, we may write the grand canonical density operator simply as
\begin{equation*}
P_\tau \;=\; \bigoplus_{n \in \N} P_\tau^{(n)} \;=\; \bigoplus_{n \in \N} \ee^{-H^{(n)} / \tau} \;=\; \ee^{-H_\tau}\,.
\end{equation*}
We remark that integrals of operator-valued distribution such as \eqref{H_tau_1d} are carefully defined in the weak sense as densely defined quadratic forms, with domain consisting of $\Psi \in \cal F$ such that there exists $n_0 \in \N$ such that $\Psi^{(n)}$ is a Schwartz function for $n \leq n_0$ and $\Psi^{(n)} = 0$ for $n > n_0$.

Analogously to the classical case, we characterize the quantum state $\rho_\tau(\cdot)$ using correlation functions. For $p \in \N$, we define the \emph{quantum $p$-particle correlation function} $\gamma_{\tau,p}$, an operator on $\fra H^{(p)}$, through its kernel
\begin{equation} \label{def_gamma_p0}
\gamma_{\tau,p}(x_1, \dots, x_p;y_1, \dots, y_p) \;\deq\; \rho_\tau \pb{\phi_\tau^*(y_1) \cdots \phi_\tau^*(y_p) \phi_\tau(x_1) \cdots \phi_\tau(x_p)}\,.
\end{equation}

\begin{remark} \label{moments_quantum}
As in the classical case (see Remark \ref{moments_classical}) the quantum state $\rho_\tau(\cdot)$ is determined by its correlation functions $(\gamma_{\tau,p})_{p \in \N}$. Indeed, for $f_1, \dots, f_p, g_1, \dots, g_q \in \fra H$, we have
\begin{equation} \label{quantum_moment}
\rho_\tau \pb{\phi_\tau^*(g_1) \cdots \phi_\tau^*(g_q) \phi_\tau(f_1) \cdots \phi_\tau(f_p)} \;=\;
\begin{cases}
\scalarb{f_1 \otimes \cdots \otimes f_p}{\gamma_{\tau,p} \, g_1 \otimes \cdots \otimes g_q} & \text{if }p = q
\\
0 & \text{if } p \neq q\,.
\end{cases}
\end{equation}
That the left-hand side vanishes for $p \neq q$ is a consequence of the gauge invariance of the state $\rho_\tau(\cdot)$. To see this, we use the rescaled number of particles operator $\cal N_\tau$ from \eqref{N_tau_sc}, and find $\ee^{\ii t \cal N_\tau} \phi_\tau(x) \ee^{-\ii t \cal N_\tau} = \ee^{-\ii t/\tau} \phi_\tau(x)$ for all $t \in \R$, by differentiation and using \eqref{CCR}. Introducing the identity $I = \ee^{-\ii t \cal N_\tau} \ee^{\ii t \cal N_\tau}$ into $\rho_\tau(\cdot)$ on the left-hand side of \eqref{quantum_moment} and using that $\rho_\tau(\ee^{-\ii t \cal N_\tau} \cal A) = \rho_\tau(\cal A \ee^{-\ii t \cal N_\tau})$ by definition of $H_\tau$, we find that the left-hand side of \eqref{quantum_moment} vanishes unless $p = q$.
\end{remark}

\subsection{The high-temperature limit} \label{sec:high_t}

In order to understand the limit $\tau \to \infty$ of the quantum state $\rho_\tau(\cdot)$, it is instructive to consider first the simple free case. For $f,g \in \fra H$ and $w = 0$ we find from \eqref{h_spectrum}
\begin{equation} \label{rho_tau_1}
\rho_\tau \pb{\phi^*_\tau(g) \phi_\tau(f)} \;=\; \scalarbb{f}{\frac{1}{\tau(\ee^{h/\tau} - 1)}\, g} \;=\; \sum_{k \in \N} \frac{\scalar{f}{u_k} \scalar{u_k}{g}}{\tau(\ee^{\lambda_k / \tau} - 1)}\,.
\end{equation}
By dominated convergence, we find
\begin{equation} \label{rho_tau_2}
\lim_{\tau \to \infty} \rho_\tau \pb{\phi^*_\tau(g) \phi_\tau(f)} \;=\; \sum_{k \in \N} \frac{\scalar{f}{u_k} \scalar{u_k}{g}}{\lambda_k} \;=\; \scalar{f}{h^{-1} g}\,,
\end{equation}
so that, in the high-temperature limit $\tau \to \infty$, the quantum two-point function converges to the classical two-point function from \eqref{classical_covariance}. A similar statement holds for arbitrary polynomials in the rescaled annihilation and creation operators $\phi_\tau$ and $\phi^*_\tau$. In particular, we find that for the second term of \eqref{H_tau_1d} to have a nontrivial limit, we require $\lambda \tau$ to be of order one. Hence, from now on we set
\begin{equation} \label{lambda_tau}
\lambda \;\deq\; \frac{1}{\tau}\,.
\end{equation}

The equations \eqref{rho_tau_1}--\eqref{rho_tau_2} hold for general $h$ and in particular explain the choice \eqref{lambda_tau} of $\lambda$ for $d = 1,2,3$. We may also use them to analyse the expected number of particles.
Indeed, writing $\cal N_\tau = \sum_{k \in \N} \phi^*_\tau(u_k) \phi_\tau(u_k)$, we find, again for $w = 0$, that
\begin{equation} \label{rho_tau_N_tau}
\rho_\tau(\cal N_\tau)
\;=\; \sum_{k \in \N} \frac{1}{\tau(\ee^{\lambda_k / \tau} - 1)}\,.
\end{equation}
If $\tr h^{-1} < \infty$, which corresponds to $d = 1$, we find
\begin{equation} \label{exp_N_1d}
\lim_{\tau \to \infty} \rho_\tau(\cal N_\tau) \;=\; \tr h^{-1} \;=\; \rho(\cal N)\,,
\end{equation}
where $\cal N$ was defined in \eqref{mass}.
A slightly more involved computation using Lemma \ref{Quantum Wick theorem} shows that all moments converge: $\lim_{\tau \to \infty} \rho_\tau(\cal N_\tau^\ell) = \rho(\cal N^\ell)$ for all $\ell \in \N$.

Thus, that the particle number grows like $\tau$ is closely linked to the fact that the classical field has a finite $L^2$-norm. This is true if and only if \eqref{tr_h_assump} holds for $s = 0$, which is generally true for $d = 1$.

In the higher-dimensional case, where \eqref{tr_h_assump} only holds for some $s < 0$, we have $\rho(\cal N) = \infty$ and the left-hand side of \eqref{exp_N_1d} diverges. On the classical side, the field $\phi$ has $\mu$-almost surely negative regularity (i.e.\ it is not a function). This singularity of the classical field is manifested on the quantum side by a diverging rescaled particle number $\cal N_\tau$. Suppose that $\lambda_k \sim k^{1/\alpha}$ for some $\alpha > 0$. For example, on the $d$-dimensional torus $\Lambda = \bb T^d$ with $v = 0$, we have $\alpha = d/2$. Then an elementary analysis of the right-hand side of \eqref{rho_tau_N_tau} shows that as $\tau \to \infty$
\begin{equation} \label{N_tau_div}
\rho_\tau(\cal N_\tau) \;\sim\;
\begin{cases}
1 & \txt{if } \alpha < 1
\\
\log \tau & \txt{if } \alpha = 1
\\
\tau^{\alpha - 1} & \txt{if } \alpha > 1\,,
\end{cases}
\end{equation}
which corresponds to the dimensions $d = 1,2,3$ respectively. Similarly, all moments of $\cal N_\tau$ are finite in $\rho_\tau(\cdot)$, but they diverge as $\tau \to \infty$ unless $\tr h^{-1}$ is finite. This may be interpreted as the quantum state $\rho_\tau(\cdot)$ having a natural subexponential cutoff for eigenvalues $\lambda_k$ larger than $\tau$, as is apparent from the right-hand side of \eqref{rho_tau_N_tau}.

In the one-dimensional case $\tr h^{-1} < \infty$, the high-temperature limit of $\rho_\tau(\cdot)$ was understood in \cite{Lewin_Nam_Rougerie}, where the authors prove that
\begin{equation*}
\lim_{\tau \to \infty} \norm{\gamma_{\tau,p} - \gamma_p}_{\fra S^1(\fra H^{(p)})} \;=\; 0\,,
\end{equation*}
under very general assumptions on $w$, which in particular admit $w \in L^\infty(\Lambda)$ as well as the delta function $w = \delta$ yielding the local quartic interaction in \eqref{def_W}.

\subsection{Higher dimensions and renormalization} \label{sec:Wick}
We now move on to the case where $\tr h^{-1} = \infty$ and $\tr h^{-2} < \infty$, which arises when $d = 2,3$ and is the main focus of our work. Throughout this subsection we assume that $w \in L^\infty(\Lambda)$ is a bounded even function of positive type, and that \eqref{tr_h_assump} holds with $s = -1$.

In this case the classical field $\phi$ from \eqref{phi_series} is not in $\fra H$, so that the interaction \eqref{def_W} is ill-defined. Such divergences are well known in quantum field theory and the theory of stochastic partial differential equations. It is also well known that they can be dealt with by a \emph{Wick ordering} procedure, which involves a truncation and a subtraction of appropriately chosen terms that diverge as the truncation parameter is removed. Thus, recalling the truncated classical field $\phi_{[K]} \in \fra H$ with $K \in \N$ from \eqref{truncated classical free field}, we define the \emph{truncated Wick-ordered classical interaction} as
\begin{equation} \label{Classical interaction}
W_{[K]} \;\deq\; \frac{1}{2} \int \dd x \, \dd y \, \pb{\abs{\phi_{[K]}(x)}^2 - \varrho_{[K]}(x)} \, w(x-y) \, \pb{\abs{\phi_{[K]}(y)}^2 - \varrho_{[K]}(y)}\,,
\end{equation}
where we defined the \emph{classical density at $x$} as
\begin{equation*}
\varrho_{[K]}(x) \;\deq\; \int \dd \mu \, \abs{\phi_{[K]}(x)}^2\,.
\end{equation*}
Note that $\int \dd x \, \varrho_{[K]}(x) = \sum_{k = 0}^K \lambda_k^{-1}$ diverges as $K \to \infty$.
Moreover, by assumption on $w$ we have $W_{[K]} \geq 0$.

The \emph{Wick-ordered classical interaction} is then defined as the limit of the truncated $W_{[K]}$ as $K \to \infty$. Its construction is the content of the following result,
whose proof is an application of Wick's theorem and is given in Section \ref{sec:classical} below.

\begin{lemma}[Definition of $W$] \label{lem:def_W}
Suppose that \eqref{tr_h_assump} holds with $s \geq -1$, and that $w \in L^\infty(\Lambda)$ is an even function of positive type. Then the sequence $(W_{[K]})_{K \in \N}$ is a Cauchy sequence in $\bigcap_{m \geq 1} L^m(\mu)$. We denote its limit by $W$.
\end{lemma}

We now define the classical state $\rho(\cdot)$ as in \eqref{classical state}, with $W$ defined in Lemma \ref{lem:def_W}.

Next, we discuss the quantum problem. Let $V \col \Lambda \to [0,\infty)$ be a one-body potential and $\nu \in \R$ a chemical potential. Our starting point is the many-body Hamiltonian analogous to \eqref{H_tau_1d}, defined as
\begin{equation}\label{eq:standard} 
\wt{H}_\tau \;\deq\; \int \dd x \, \dd y \,  \phi^*_\tau(x) \pb{-\Delta + \nu + V}(x;y) \phi_\tau(y)
+ \frac{1}{2} \int \dd x\, \dd y \, \phi_\tau^* (x) \phi^*_\tau (y)\, w(x-y) \, \phi_\tau (x) \phi_\tau (y) \,.
\end{equation}
As it turns out, in order to obtain a nontrivial high-temperature limit as $\tau \to \infty$, the chemical potential $\nu = \nu(\tau)$ will have to tend to $-\infty$. Heuristically, this can be understood from the fact that the rescaled quantum density diverges as $\tau \to \infty$, as explained in \eqref{N_tau_div}. To understand how to choose the chemical potential $\nu$ as a function of $\tau$, it is convenient to rewrite $\wt{H}_\tau$ in terms of a \emph{renormalized} Hamiltonian, called $H_\tau$, where the interaction only depends on the fluctuations of the density around its mean.
Let $\kappa > 0$ and $v_\tau \col \Lambda \to [0,\infty)$ be a possibly $\tau$-dependent \emph{bare one-body potential}, and define, in analogy to \eqref{def_h}, the corresponding one-body Hamiltonian
\begin{equation} \label{def_h_tau}
h_\tau \;\deq\; -\Delta + \kappa + v_\tau\,.
\end{equation}
Then we define the free Hamiltonian as
\begin{equation} \label{H_{tau,0}}
H_{\tau,0} \;\deq\; \int \dd x \, \dd y \, \phi^*_\tau(x) \, h_\tau(x;y) \, \phi_\tau(y)\,.
\end{equation}
The associated \emph{free quantum state} $\rho_{\tau,0}(\cdot)$ is defined as
\begin{equation}
\label{rho_{tau,0}}
\rho_{\tau, 0}(\cal A) \;\deq\; \frac{\tr(\cal A \, \ee^{- H_{\tau,0}})}{\tr(\ee^{- H_{\tau, 0}})}\,,
\end{equation}
and the \emph{quantum density at $x$} is defined as
\begin{equation} \label{def_varrho_tau}
\varrho_\tau(x) \;\deq\; \rho_{\tau,0}\pb{\phi^*_\tau(x) \phi_\tau(x)}\,.
\end{equation}
By definition, the \emph{renormalized many-body Hamiltonian} is
\begin{equation}
\label{eq:counter-hm} 
H_\tau \;\deq\; H_{\tau,0} + W_\tau\,,
\end{equation}
where we defined the \emph{renormalized quantum interaction}
\begin{equation} \label{Quantum W}
W_\tau \;\deq\; \frac{1}{2} \int \dd x \, \dd y \, \pb{\phi_\tau^* (x) \phi_\tau (x) - \varrho_{\tau} (x)} \, w(x-y) \, \pb{ \phi^*_\tau (y) \phi_\tau (y) - \varrho_{\tau} (y) } \,.
\end{equation}

Having defined the physical many-body Hamiltonian $\tilde H_\tau$ and its renormalized version $H_\tau$, we now explain how they are related. It turns out that this relationship is in general nontrivial, and is related to the so-called \emph{counterterm problem}. Before discussing the counterterm problem, however, we focus on the simple case where $\Lambda = \bb T^d$, $V = 0$, and $\nu > 0$, in which case the counterterm problem can be solved by elementary means.

\begin{itemize}
\item[(i)]
\emph{The case $\Lambda = \bb T^d$, $V = 0$.}
By translation invariance, we find that $\varrho_\tau(x) = \varrho_\tau(0)$ for all $x \in \Lambda$. Using \eqref{CCR}, we therefore find
\begin{equation} \label{H_tau_H_tau}
\tilde H_\tau \;=\; H_\tau + \qbb{\varrho_\tau(0) \hat w(0) - \frac{1}{2 \tau} w(0) + \nu - \kappa} \cal N_\tau - \frac{\varrho_\tau(0)^2}{2} \hat w(0)\,,
\end{equation}
with $v_\tau = 0$ and $\hat w(0) \deq \int \dd x \, w(x)$. Hence, up to an irrelevant (diverging) additive constant, $\tilde H_\tau$ and $H_\tau$ differ by a chemical potential multiplying $\cal N_\tau$.
In order to obtain the renormalized Hamiltonian \eqref{eq:counter-hm} with some fixed $\kappa$, we therefore have to ensure that the expression in square brackets in \eqref{H_tau_H_tau} vanishes, so that we have to choose the original chemical potential $\nu \equiv \nu(\tau)$ in \eqref{eq:standard} to be sufficiently negative, tending to $-\infty$ as $\tau \to \infty$. Indeed, $\hat w(0) \geq 0$ because $w$ is of positive type, and $\varrho_\tau(0) = \int \dd x \varrho_\tau(x) = \rho_{\tau,0}(\cal N_\tau)$, which diverges as $\tau \to \infty$ according to \eqref{N_tau_div}. The physical interpretation is that, in order to obtain a well-defined high-temperature limit, we have to let the chemical potential $\nu \equiv \nu(\tau)$ in \eqref{eq:standard} tend to $-\infty$ to compensate the large repulsive interaction energy arising from the large number of particles. 
\item[(ii)]
\emph{The general case.} In general, $\tilde H_\tau$ and $H_\tau$ are not related by a simple shift in the chemical potential: we also have to change the bare one-body potential $v_\tau$ in a $\tau$-dependent fashion to recover the original one-body potential $V$ in \eqref{eq:standard}. This is because $\varrho_\tau(x)$ is no longer independent of $x$, although it still diverges as $\tau \to \infty$. However, it turns out that there exists a constant $\bar \varrho_\tau$ such that $\bar \varrho_\tau$ diverges as $\tau \to \infty$ and $\varrho_\tau(x) - \bar \varrho_\tau$ converges for all $x \in \Lambda$. We rewrite
\begin{equation} \label{tilde_H_H_tau_gen}
\wt{H}_\tau \;=\; H_\tau  + \left[\bar{\varrho}_{\tau} \hat w(0) - \frac{1}{2 \tau} w(0) + \nu - \kappa \right] \cal N_\tau - \frac{1}{2} \int \dd x \, \dd y \, \varrho_\tau(x) w(x - y) \varrho_\tau(y)\,,
\end{equation}
where we imposed the condition
\begin{equation} \label{counter_intro}
v_{\tau}  \;=\; V + w * (\varrho_{\tau} - \bar{\varrho}_{\tau})\,,
\end{equation}
and $*$ denotes convolution.
The condition \eqref{counter_intro} is a self-consistent equation for $v_\tau$, since $\varrho_\tau$ by definition also depends on $v_\tau$ (see \eqref{def_h_tau}--\eqref{def_varrho_tau}).
Hence, as in (i) above, up to an irrelevant (diverging) constant, $\tilde H_\tau$ and $H_\tau$ differ by a chemical potential, equal to the expression in square brackets in \eqref{tilde_H_H_tau_gen}, multiplying $\cal N_\tau$. As in (i), we shall see that in order for $\varrho_\tau - \bar \varrho_\tau$ to remain bounded, we have to choose $\bar \varrho_\tau$ to diverge as $\tau \to \infty$, which means that the original chemical potential $\nu \equiv \nu(\tau)$ has to be chosen to tend to $-\infty$ as $\tau \to \infty$ so that the expression in square brackets in \eqref{tilde_H_H_tau_gen} vanishes.

Note that now the original one-body potential $V$ and the bare one-body potential $v_\tau$ are different. Finding the matching bare potential $v_\tau$ associated with a given $V$, such that \eqref{counter_intro} holds, is referred to as the \emph{counterterm problem}. In general, the counterterm problem is a nonlinear integral equation for $v_\tau$, and as such its solution requires a nontrivial analysis. The general solution of the counterterm problem is given in Section \ref{sec:counterterm} below, where we prove, under some technical assumptions on $V$, that there exists a $\bar \varrho_\tau$ (explicitly given in \eqref{def_rho_bar} below) such that \eqref{counter_intro} has a unique solution $v_\tau$, which depends on $\tau$, and that $v_\tau$ converges to some limiting potential $v$ in the sense that
\begin{equation} \label{conv_h_tau_intro}
\lim_{\tau \to \infty} \norm{h_\tau^{-1} - h^{-1}}_{\fra S^2(\fra H)} \;=\; 0\,.
\end{equation}
(Recall \eqref{def_h}.) See Theorem \ref{thm:counter} below for the precise statement.
The one-body potential $v$ thus constructed is the limiting bare one-body potential, and choosing it in the definition \eqref{def_h} of the classical one-body potential $h$ yields a classical state \eqref{classical state}, with $W$ defined in Lemma \ref{lem:def_W}, that is the correct rigorous version of the Gibbs measure formally defined in \eqref{def_P}.
\end{itemize}

Next, we define the thermal states associated with the one-body Hamiltonian $h_\tau$ and interaction potential $w$. Instead of considering the usual grand canonical density operator $P_\tau = \ee^{-H_{\tau,0} - W_\tau}$, as in Section \ref{sec:intro_quantum}, we introduce a family of modified grand canonical density operators
\begin{equation} \label{reg par 1}
P_{\tau}^\eta \;\deq\;  \ee^{-\eta H_{\tau,0}} \ee^{-(1 - 2\eta) H_{\tau,0}  - W_\tau} \ee^{-\eta H_{\tau,0}}
\end{equation}
parametrized by $\eta \in [0,1/4]$. Note that in the free case $w = 0$, $P_\tau^\eta$ does not depend on $\eta$. For technical reasons explained in Section \ref{sec:proof_outline} below, it is sometimes more convenient to consider $P_\tau^\eta$ with $\eta > 0$ rather than $P_\tau = P_\tau^0$.

We define the quantum state $\rho_{\tau}^\eta(\cdot)$ associated with $P_\tau^\eta$ through
\begin{equation} \label{reg par 2}
\rho_{\tau}^\eta(\cal A) \;\deq\; \frac{\tr (\cal A P_{\tau}^\eta)}{\tr (P_{\tau}^\eta)}\,.
\end{equation}
Analogously to \eqref{def_gamma_p0}, for $p \in \N$ we define the quantum $p$-particle correlation function $\gamma_{\tau,p}^\eta$ by
\begin{equation}
\label{quantum p-particle dm}
\gamma_{\tau,p}^\eta(x_1, \dots, x_p;y_1, \dots, y_p) \;\deq\; \rho_\tau^\eta \pb{\phi_\tau^*(y_1) \cdots \phi_\tau^*(y_p) \phi_\tau(x_1) \cdots \phi_\tau(x_p)}\,.
\end{equation}
The classical $p$-particle correlation function $\gamma_p$ is defined as in \eqref{classical p-particle dm}.

Unlike in the one-dimensional case $\tr h^{-1} < \infty$, in the higher-dimensional case $\tr h^{-1} = \infty$ the $p$-point correlation function $\gamma_{p}$ does not lie in $\fra S^1(\fra H^{(p)})$. For simplicity, for the following discussion we set $p = 1$ and consider the free case $w = 0$, where $\gamma_1 = h^{-1}$. Thus, $\tr \gamma_1 = \infty$, and the convergence of correlation functions cannot hold in the trace class $\norm{\cdot}_{\fra S^1(\fra H^{(p)})}$. Instead, under the assumption $\tr h^{-2} < \infty$ we find $\norm{\gamma_1}_{\fra S^2(\fra H)}^2 = \tr \gamma_1^2 = \tr h^{-2} < \infty$. We conclude that for higher dimensions we need to replace the notion of convergence in trace norm $\norm{\cdot}_{\fra S^1(\fra H^{(p)})}$ with convergence in the Hilbert-Schmidt norm $\norm{\cdot}_{\fra S^2(\fra H^{(p)})}$.

\subsection{Results} \label{sec:results}

We may now state our main result in the higher-dimensional setting.

\begin{theorem}[Convergence for $d = 2,3$] \label{thm:main}
Let $\kappa > 0$ and $v \col \Lambda \to [0,\infty)$, and define $h$ as in \eqref{def_h}. Suppose that $\tr h^{-2} < \infty$. Moreover, suppose that $h_\tau$ satisfies
\begin{equation} \label{h_tau_conv}
\lim_{\tau \to \infty} \norm{h_\tau^{-1} - h^{-1}}_{\fra S^2(\fra H)} \;=\; 0\,.
\end{equation}
Let $w \in L^\infty(\Lambda)$ be an even function of positive type. Let the classical interaction $W$ be defined as in Lemma \ref{lem:def_W} and the classical $p$-particle correlation function $\gamma_p$ be defined as in \eqref{classical p-particle dm} and \eqref{classical state}. Moreover, let the quantum $p$-particle correlation function $\gamma_{\tau,p}^\eta$ be defined as in \eqref{quantum p-particle dm}, \eqref{reg par 2}, \eqref{reg par 1}, \eqref{H_{tau,0}}, and \eqref{Quantum W}.

Then for every $\eta \in (0,1/4]$ and $p \in \N$ we have
\begin{equation}
\lim_{\tau \to \infty} \norm{\gamma_{\tau,p}^\eta - \gamma_p}_{\fra S^2(\fra H^{(p)})} \;=\; 0\,.
\end{equation}
\end{theorem}

\begin{remark}
By a diagonal sequence argument, we find that under the assumptions of Theorem \ref{thm:main} there exists a family $(\eta_\tau)_{\tau \geq 1}$ such that $\lim_{\tau \to \infty} \eta_\tau = 0$ and 
$
\lim_{\tau \to \infty} \norm{\gamma_{\tau,p}^{\eta_\tau} - \gamma_p}_{\fra S^2(\fra H^{(p)})} = 0
$
for all $p \in \N$.
\end{remark}

In Theorem \ref{thm:main}, one can of course take $h_\tau = h$, in which case the assumption \eqref{h_tau_conv} is trivial. More generally, starting from a physical (non-renormalized) Hamiltonian of the form \eqref{eq:standard}, the assumption \eqref{h_tau_conv} is satisfied by solving the counterterm problem, as explained in \eqref{conv_h_tau_intro}; see Theorem \ref{thm:counter} below.

Our methods can also be applied to the easier case $d = 1$, where no renormalization is necessary. We illustrate this in the following two theorems. The following result was previously proved in \cite{Lewin_Nam_Rougerie}.

\begin{theorem}[Convergence for $d = 1$] \label{thm:1D}
Let $\kappa > 0$ and $v \col \Lambda \to [0,\infty)$, and define $h$ as in \eqref{def_h}. Suppose that $\tr h^{-1} < \infty$. Let $w \in L^{\infty}(\Lambda)$ be an even nonnegative function.
Let the classical interaction $W$ be defined as in \eqref{def_W} and the classical $p$-particle correlation function $\gamma_p$ be defined as in \eqref{classical p-particle dm} and \eqref{classical state}. Moreover, let the quantum $p$-particle correlation function $\gamma_{\tau,p}$ be defined as in \eqref{def_gamma_p0}, \eqref{def_rho_tau}, and \eqref{H_tau_1d}.

Then for all $p \in \N$ we have
\begin{equation} 
\label{1D convergence}
\lim_{\tau \to \infty} \norm{\gamma_{\tau,p} - \gamma_p}_{\fra S^1(\fra H^{(p)})} \;=\; 0\,.
\end{equation}
\end{theorem}

We remark that an analogous result also holds if the density operator $P_\tau = P_\tau^0$ is replaced with $P_\tau^\eta$ from \eqref{reg par 1}. Moreover, we can also admit the quantum one-body Hamiltonian to depend on $\tau$, replacing $h$ on the right-hand side of \eqref{H_tau_1d} with $h_\tau$ satisfying $\lim_{\tau \to \infty} \norm{h_\tau^{-1} - h^{-1}}_{\fra S^1(\fra H)} = 0$. However, for the one-dimensional case there is no natural reason to do so, unlike in the higher-dimensional case where $h_\tau$ was obtained by solving the counterterm problem.

Up to now, we focused on a nonlocal interaction of the form \eqref{def_W}. For $d = 1$, our methods may also be used to derive Gibbs measures with a \emph{local} interaction. We illustrate this in the case $\Lambda = \bb T^1$ and $w = \alpha \delta$ for some constant $\alpha > 0$, in which case the classical interaction \eqref{def_W} becomes
\begin{equation} \label{def_W_local}
W \;=\; \frac{\alpha}{2} \int \dd x \, \abs{\phi(x)}^4\,.
\end{equation}
A physically natural way to obtain a local interaction is to modify the range of the two-body interaction potential. Recall from Section \ref{sec:high_t} that the typical number of particles in the quantum system is of order $\tau$, and the scaling $\lambda = 1/\tau$ in front of the interaction \eqref{def_H^n} ensures the interaction potential per particle is of order one. This corresponds to a mean-field scenario, where each particle interacts with an order $\tau$ other particles and the interaction strength is of order $\tau^{-1}$.

Instead, we may consider short-range interactions, where $w$ in \eqref{H_tau_1d} is replaced with $w_\tau(x) = \tau w(\tau x)$,
where $w$ is an even nonnegative function\footnote{More precisely, $w$ is an even nonnegative function in $L^1(\R)$ and $w_\tau \col \Lambda \to \R$ is defined by $w_\tau(x) \deq \tau w(\tau [x])$, where $[x]$ is the unique representative in the set $(x + \Z) \cap [-1/2,1/2)$.} with integral $\alpha$. Physically, using the interaction $w_\tau$ means that each particle interacts with an order $1$ other particles, and the interaction strength is of order $1$. Indeed, because we have an order $\tau$ particles in $\bb T^1$, the typical separation between neighbouring particles is of order $\tau^{-1}$.

\begin{theorem}[Convergence for $d = 1$ with local interaction]\label{thm:1D local}

Let $\Lambda = \bb T^1$, $\kappa > 0$, $v  = 0$, and define $h$ as in \eqref{def_h}. Let $w_\tau$ be an even nonnegative function on $\Lambda$ satisfying
\begin{equation}
\label{w_tau bound}
\|w_{\tau}\|_{L^1}\;\leq\; C \,,\qquad
\|w_{\tau}\|_{L^{\infty}} \;\leq\; C\tau\,,
\end{equation}
and suppose that $w_\tau$ converges weakly (with respect to bounded continuous functions) to $\alpha > 0$ times the delta function at $0$.
Let the classical interaction $W$ be defined as in \eqref{def_W_local} and the classical $p$-particle correlation function $\gamma_p$ be defined as in \eqref{classical p-particle dm} and \eqref{classical state}. Moreover, let the quantum $p$-particle correlation function $\gamma_{\tau,p}$ be defined as in \eqref{def_gamma_p0} and \eqref{def_rho_tau}, where the many-body quantum Hamiltonian \eqref{H_tau_1d} has interaction potential $w_\tau$ instead of $w$.

Then for all $p \in \N$ we have \eqref{1D convergence}.
\end{theorem}

Note that the assumptions on $w_\tau$ in Theorem \ref{thm:1D local} are satisfied if we take $w_\tau(x) = \tau^\beta w(\tau^\beta x)$ for any $\beta \in (0,1]$ and even nonnegative $w$ with integral $\alpha$.

\subsection{Strategy of proof} \label{sec:proof_outline}
Our basic approach is a perturbative expansion in the interaction, motivated by ideas from field theory. Let us first focus on the classical problem. For simplicity, we explain our strategy for the \emph{partition function} $\int \ee^{-W} \dd \mu$.

Since $W \geq 0$, it is easy to see that the function $z \mapsto \int \ee^{-z W} \dd \mu$ is analytic in the right half-plane $\re z > 0$. However, it is ill-defined for $\re z < 0$, and therefore its radius of convergence at $z = 0$ is zero. Hence, a power series representation in terms of moments $\int W^m \, \dd \mu$ of $W$ with respect to $\mu$ is a hopeless task. This well-known phenomenon is best understood in the toy example $A(z) = (2 \pi)^{-1/2} \int \ee^{- z x^4} \, \ee^{-x^2/2} \, \dd x$, which is obviously analytic for $\re z > 0$. Moreover, using Taylor's theorem and Wick's theorem for the moments of a Gaussian measure, we find the asymptotic expansion
\begin{equation} \label{A_intro}
A(z) \;=\; \sum_{m = 0}^{M - 1} a_m z^m + R_M(z)\,, \qquad a_m \;\deq\; \frac{(-1)^m (4m - 1)!!}{m!}\,, \qquad \abs{R_M(z)} \;\leq\; \frac{(4M - 1)!!}{M!} \abs{z}^M\,,
\end{equation}
for any $M \in \N$ and $\re z \geq 0$. The estimate on the remainder term behaves like $(C M \abs{z})^M$ for large $M$, and in particular diverges as $M \to \infty$ unless $z = 0$. Although the power series for $A(z)$ is not convergent around $z = 0$, one can neverthless reconstruct $A$ effectively from its coefficients $a_m$ by \emph{Borel summation}. This was understood in \cite{Watson} and \cite{Sokal}, building on the previous work \cite{Hardy,Nevanlinna}.

Recall that the \emph{Borel transform} $B(z)$ of a formal power series $A(z) = \sum_{m \geq 0} a_m z^m$ is defined as $B(z) \deq \sum_{m \geq 0} \frac{a_m}{m!} z^m$. Then, as formal power series, we may recover $A$ from from its Borel transform $B$ by
\begin{equation} \label{A_B_rel}
A(z) \;=\; \int_0^\infty \dd t \, \ee^{-t} B(tz)\,.
\end{equation}
For $a_m$ as in \eqref{A_intro}, the formal power series $A(z)$ has a vanishing radius of convergence, but its Borel transform $B(z)$ has a positive radius of convergence, and hence defines an analytic function in a ball around the origin. If we can show that $B$ extends to an analytic function of sufficiently slow increase defined in some neighbourhood of the positive real axis, we can therefore recover the function $A(z)$ on the positive real axis from its coefficients $a_m$ via the identity \eqref{A_B_rel}. This procedure allows us to obtain control of functions $A(z)$ via the coefficients of their asymptotic expansions \eqref{A_intro}, although these expansions may have zero radius of convergence. The precise statement is given in Theorem \ref{Borel summation convergence} below, which shows that two analytic functions $A_\tau(z)$ and $A(z)$ are close provided that the coefficients of their asymptotic expansions are close. This requires bounds on the coefficients $a_{\tau,m}, a_m$ and on the remainder terms $R_{\tau,M}(z), R_M(z)$ of the asymptotic expansions of the form $\abs{a_{\tau,m}} + \abs{a_m} \leq C^m m!$ and $R_{\tau,M}(z) + R_M(z) \leq C^M M! \abs{z}^M$, which are much weaker than bounds needed to obtain a positive radius of convergence for $A$.

Thus, thanks to Borel summation, we may control the classical partition function $A(z) = \int \ee^{-z W} \, \dd \mu$ in terms of the coefficients $a_m = \frac{(-1)^m}{m!} \int W^m \, \dd \mu$. In order to identify these coefficients as the limits of the corresponding quantum coefficients, they have to be computed explicitly using Wick's theorem for the Gaussian measure $\mu$. To that end, we use the Definition of $W$ from Lemma \ref{lem:def_W} to write $a_m = \lim_{K \to \infty} \frac{(-1)^m}{m!} \int W_{[K]}^m \, \dd \mu$, and apply Wick's theorem to the resulting expression. Thus we get an integral over $2m$ integration variables, with an integrand that is a product of expressions of the form $w(x_i - x_j)$ and $G(x_i;x_j)$, where we defined the \emph{classical Green function} $G(x;y) \deq \int \dd \mu \, \bar \phi(y) \phi(x)$, which is nothing but the covariance $G = h^{-1}$ of the classical free field. A crucial observation is that, thanks to the Wick ordering in \eqref{Classical interaction}, the arguments $x_i$ and $x_j$ are always different integration variables from the list $(x_i)$, which allows us to obtain the needed estimates for the coefficients $a_m$. Without the Wick ordering in \eqref{Classical interaction}, we would get terms of the form $G(x_i;x_i)$, which yield, after integration over $x_i$, $\tr G = \tr h^{-1} = \infty$. The remainder term $R_M(z)$ of the asymptotic expansion of $A(z)$, which arises from the remainder of Taylor's theorem applied to $\ee^{-z W}$, may be estimated easily in terms of $a_m$ because $W \geq 0$.

We remark that the above analysis of the asymptotic expansion of the classical partition function $A(z) = \int \ee^{-z W} \, \dd \mu$  proceeds under much stronger assumptions on the interaction $W$ than has been considered in the field theory literature. Indeed, our main focus and the key difficulties of the proof are in the analysis of the quantum problem.

In analogy to the classical partition function $A(z) = \int \ee^{-z W} \, \dd \mu$, the (relative) quantum partition function for the modified grand canonical density operator $P_\tau^\eta$ from \eqref{reg par 1} is defined as
\begin{equation} \label{pfunct_sketch}
A_\tau(z) \;=\; \frac{\tr (P_\tau^\eta(z))}{\tr (\ee^{-H_{\tau,0}})}\,, \qquad P_{\tau}^\eta(z) \;\deq\;  \ee^{-\eta H_{\tau,0}} \ee^{-(1 - 2\eta) H_{\tau,0}  - z W_\tau} \ee^{-\eta H_{\tau,0}}\,.
\end{equation}
Since $W_\tau$ is a positive operator, it is not too hard to show that $A_\tau(z)$ is analytic for $\re z > 0$. As in the classical case, we aim to find an asymptotic expansion of the form $A_\tau(z) \;=\; \sum_{m = 0}^{M - 1} a_{\tau,m} z^m + R_{\tau,M}(z)$. The main work in the proof is
\begin{enumerate}
\item
obtain bounds of the form $\abs{a_{\tau,m}} \leq C^m m!$ and $\abs{R_{\tau,M}(z)} \leq C^M M! \abs{z}^M$ on the coefficients and the remainder term of the expansion of $A_\tau(z)$, which are strong enough to control $A_\tau(z)$ using its Borel transform;
\item
prove the convergence of the quantum coefficient $a_{\tau,m}$ to the classical coefficient $a_m$ as $\tau \to \infty$.
\end{enumerate}

Part (i) is in fact the hardest part of the proof; we now outline how these estimates are derived. The Taylor expansion of $A_\tau(z)$ is obtained from a Duhamel expansion, which arises from a repeated application of Duhamel's formula $\ee^{X + z Y} = \ee^{X} + z \int_0^1 \dd t \, \ee^{X(1 - t)} Y \ee^{t (X + z Y)}$ to the expression $\ee^{-(1 - 2\eta) H_{\tau,0}  - z W_\tau}$. This gives an explicit expression for $a_{\tau,m}$ in terms of an $m$-fold time integral $\int \dd \f t$ over an $m$-dimensional simplex (see \eqref{Explicit term a} below), whose integrand consists of the trace of a product of interaction terms $W_\tau$ and free propagators $\ee^{-(t_{i-1} - t_i) H_{\tau,0}}$ depending on the times $\f t$. We compute the trace using the quantum Wick theorem (see Appendix \ref{sec:Wick_quantum}), which says that an expression of the form
\begin{equation} \label{exp_sketch}
\frac{1}{\tr(\ee^{-H_{\tau,0}})} \tr \pB{\phi_\tau^*(x_1) \cdots \phi_\tau^*(x_m) \phi_\tau(y_1) \cdots \phi_\tau(y_m) \, \ee^{-H_{\tau,0}}}
\end{equation}
is given by a sum over all pairings of the labels $x_1, \dots, x_m, y_1, \dots, y_m$, where each pairing contributes a product over pairs of two-point functions of the form
\begin{equation*}
\frac{1}{\tr(\ee^{-H_{\tau,0}})} \tr \pB{\phi_\tau^*(x) \phi_\tau(y) \, \ee^{-H_{\tau,0}}} \;=\; G_\tau(x;y)\,,
\end{equation*}
where $G_\tau$ is the quantum Green function. It may be easily computed (see Appendix \ref{sec:Wick_quantum}) to equal $G_\tau = \frac{1}{\tau (\ee^{h_\tau/\tau} - 1)}$. Thus, for energies less than the temperature $\tau$, the quantum Green function $G_\tau$ behaves like the classical Green function $G = h^{-1}$, whereas for energies above the temperature $\tau$ it exhibits a subexponential decay.

A major hurdle in the application of the quantum Wick theorem is that the expressions obtained from the Duhamel expansion are not of the simple form \eqref{exp_sketch}, but intertwine factors of $\phi_\tau$ and $\phi_\tau^*$ with free propagators $\ee^{-(t_{i-1} - t_i) H_{\tau,0}}$. The quantum Wick theorem remains applicable in this case, except that the resulting quantum Green functions are time-dependent, and take on the form $G_{\tau,t} \deq \frac{\ee^{-t h_\tau / \tau}}{\tau (\ee^{h_\tau / \tau} - 1)}$ or $S_{\tau,t} \deq \ee^{-t h_\tau/\tau}$, depending on the structure of the pair in the pairing. We get factors of the form $G_{\tau,t}$ with $t \geq -1$ and $S_{\tau,t}$ for $t \geq 0$. Thus, these time-dependent operators are no longer uniformly Hilbert-Schmidt. In fact, the only norm under which they are uniformly bounded in time is the operator norm $\norm{\cdot}_{\fra S^\infty}$, which is far too weak for our estimates. (Note that, here and throughout this paper, \emph{time} does not refer to physical time, which is the argument of a unitary one-parameter group, but rather to imaginary physical time, which is the argument of a contraction semigroup of the form $t \mapsto \ee^{-t h_\tau/\tau}$.)

It is therefore crucial to exploit the detailed structure of the terms arising from the quantum Wick theorem and the precise time-dependence of the Green functions. We do this by introducing a graphical notation, whereby individual terms arising from the quantum Wick theorem are encoded by graphs. The value of a graph is estimated by using an interplay of $\fra S^p$-norms of the associated operators and the $L^p$-norms of their operator kernels. In order to go back and forth between these two pictures, a crucial observation that we make use of is that the operator kernels of all time-evolved Green functions $G_{\tau,t}$ and $S_{\tau,t}$ are nonnegative. This allows us to rewrite integrals of absolute values of the operator kernels, obtained after using $L^p$-estimates, as kernels of products of operators. This interplay lies at the heart of our estimates. Viewed from a slightly different angle, we work both in \emph{physical space} $x \in \Lambda$, where operators are controlled using $L^p$-norms of their kernels, and in the \emph{eigenfunction} or \emph{Fourier space} $k \in \N$, where operators are controlled using their Schatten norms. The latter space is also often referred to as \emph{frequency space} (by analogy to the case $\Lambda = \bb T^d$ and $h = - \Delta + \kappa$, where the eigenfunctions are parametrized by frequencies). 
In fact, most of our analysis takes place in physical space, in contrast to most previous work on Gibbs measures of NLS (e.g.\ \cite{B,B1,B2}) that operates in eigenfunction space. 

More concretely, we estimate a pairing by decomposing its associated graph into disjoint paths, where a path consists of a string of consecutive Green functions. Paths are connected by interaction potentials $w(x_i - x_j)$ whose two arguments may belong to different paths. Using $\norm{w}_{L^\infty} < \infty$, we may estimate the contribution of a pairing by simply dropping the interaction potentials $w$, at the cost of a multiplicative constant $\norm{w}_{L^\infty}^m$, to decouple the paths, resulting in a product of integrals of paths. Using the positivity of the remaining integral kernels, we may rewrite the contribution of each path as a trace of an operator product of Green functions. This allows us to combine the time indices $t$ of all Green functions within a path, which always sum to zero for any path, and ultimately conclude the estimate. This combination of the times within a path is essential for our estimates, and estimating the $L^p$- or $\fra S^p$-norms of the individual Green functions directly leads to bounds that are not good enough to conclude the proof.

The upper bound that we obtain for each explicit coefficient $a_{\tau,m}$ of the quantum expansion is uniform in the times $\f t$ and the parameter $\eta$. By dominated convergence, the convergence of $a_{\tau,m}$ to the classical coefficient $a_m$ may be proved using rather soft arguments, for each fixed $\f t$ in the interior of the simplex, without worrying about uniformity in $\f t$. In that case, all time-evolved Green functions are Hilbert-Schmidt (albeit with norms that may blow up as $\f t$ approaches the boundary of the simplex), and the convergence of $a_{\tau,m}$ may be relatively easily established using the convergence of $G_{\tau,t}$ to the classical Green function $G$ for all $t > -1$ and $S_{\tau,t}$ to the identity $I$ for $t > 0$. This shows that $a_{\tau,m}$, expressed as a sum over graphs, converges to an explicit expression, where the contribution of each graph is given as an integral of a monomial in factors of $w$ and $G$. It then remains to simply verify that this expression is precisely the one obtained from the classical theory using Wick's theorem, which was outlined above.

As explained above, the classical remainder term $R_M(z)$ may be trivially estimated in terms of the $M$-th explicit term $a_M \abs{z}^M$. This is not the case for the quantum problem, where the estimate of $R_{\tau,M}(z)$ is a major source of technical difficulties, and also the reason why we require $\eta > 0$ in Theorem \ref{thm:main}. The structure of the remainder term $R_{\tau,M}(z)$ is analogous to that of $a_{\tau,m}$ outlined above, except that the last free propagator $\ee^{-t_M H_{\tau,0}}$ is replaced with the full propagator $\ee^{-t_M (H_{\tau,0} + zW)}$. Hence, the quantum Wick theorem is not applicable. Moreover, estimating the remainder term in terms of an explicit term (i.e.\ setting $z = 0$ to obtain an upper bound), like in the classical case, does not work owing to the noncommutativity of the interactions $W_\tau$ and the quantum propagators. On the other hand, owing to the delicate cancellations inherent in the renormalized interaction $W_\tau$, our only means of estimating factors of $W_\tau$ effectively is the quantum Wick theorem for the quasi-free state $\rho_{\tau,0}(\cdot)$ from \eqref{rho_{tau,0}}.

Our approach is to estimate the remainder term by a careful splitting using H\"older's inequality for the Schatten spaces. To that end, we have to ensure that the total time $t_M$ carried by the full propagator $\ee^{-t_M (H_{\tau,0} + zW)}$ is bounded away from one, in order to ensure that the remaining free propagators have a total time bounded away from zero. By an appropriate choice of the splitting, we may therefore obtain a product of Schatten norms containing either (a) only the full propagator or (b) a product of free propagators and interaction potentials $W_\tau$. The former may be estimated using the Trotter-Kato product formula, using that $\re z \geq 0$ and $W_\tau$ is positive. The latter may be again estimated using Wick's theorem. A further technical complication in the use of Wick's theorem is that, in order to use Wick's theorem to estimate the Schatten norm in a factor of type (b), the total times in factors of type (b) have to be inverses of even integers. The existence of such a splitting is guaranteed by a carefully constructed splitting algorithm. We remark that it is the requirement that the full propagator have a total time separated away from one that leads to the requirement $\eta > 0$ in Theorem \ref{thm:main}.

In the one-dimensional case, where the interaction potential is not renormalized, the remainder term may be easily estimated in terms of the explicit term by an application of the Feynman-Kac formula for the full propagator. (See Proposition \ref{Remainder term bound 1D} below.) Hence, in that case we can also set $\eta = 0$. This argument requires taking the absolute value of the integral kernel of the interaction $W$, and hence destroys the delicate cancellations arising from the renormalization needed in the higher-dimensional case.

Putting everything together, we have obtained the necessary estimates to deduce from the Borel summation result in Theorem \ref{Borel summation convergence} that the relative quantum partition function $A_\tau(1)$ from \eqref{pfunct_sketch} converges to the classical partition function $A(1) = \int \ee^{-W} \, \dd \mu$.

Finally, instead of partition functions, we can analyse correlation functions by computing the expectation of observables, i.e.\ polynomials in the classical or quantum fields. This leads to a minor generalization of the approach outlined above, whose essence however remains the same. The convergence of the correlation function then follows from the convergence of observables by a simple duality argument.

\subsection{Organization of the paper}
We conclude the introduction with an outline of the remainder of the paper. In Section \ref{sec:quantum} we set up and analyse the perturbative expansion for the quantum problem. The analogous task for the classical problem is performed in Section \ref{sec:classical}, where we also identify the classical expansion as the high-temperature limit of the quantum expansion. In Section \ref{sec:one-d}, we explain the modifications required to apply our results to the one-dimensional problem without renormalization. In Section \ref{sec:counterterm} we formulate the counterterm problem precisely and solve it. Finally, in the appendices we collect various technical tools needed in our proofs. In Appendix \ref{sec:borel} we state and prove a general result about resummation of asymptotic expansions using Borel summation. In Appendix \ref{sec:Wick_quantum}, we collect some standard facts about bosonic quasi-free states. The final Appendix \ref{sec:GF_estimates} collects some basic estimates on the quantum Green function.

\section{The quantum problem} \label{sec:quantum}

Sections \ref{sec:quantum}--\ref{sec:classical} are devoted to the proof of our main result, Theorem \ref{thm:main}.

\subsection{Preparations}
We use the notation $h_\tau = \sum_{k \in \N} \lambda_{\tau,k} u_{\tau,k} u_{\tau,k}^*$ for the eigenvalues $\lambda_{\tau,k} > 0$ and eigenfunctions $u_{\tau,k} \in \fra H$ of $h_\tau$. We abbreviate $\phi_{\tau, k} \deq \phi_\tau(u_{\tau,k})$. From \eqref{CCR} we deduce that
\begin{equation}
\label{CCR k}
[\phi_{\tau, k}, \phi_{\tau,l}^*] \;=\; \frac{\delta_{kl}}{\tau} \,, \qquad [\phi_{\tau, k}, \phi_{\tau,l}] \;=\; [\phi_{\tau, k}^*, \phi_{\tau,l}^*]\;=\; 0\,.
\end{equation}
Moreover, we have the eigenfunction expansion
\begin{equation}
\label{phi_tau}
\phi_\tau(x) \;=\; \sum_{k \in \N} u_{\tau,k}(x) \phi_{\tau,k}\,,
\end{equation}
and we can write $H_{\tau,0}= \sum_{k \in \N} \lambda_{\tau,k} \, \phi_{\tau,k} \phi_{\tau,k}^*$.

Next, we define the \emph{quantum Green function}, $G_\tau$, as the one-particle correlation function of the free state \eqref{rho_{tau,0}}:
\begin{equation}
\label{quantum correlation}
\scalar{f}{G_\tau g} \;\deq\; \rho_{\tau,0}\pb{\phi^*_\tau(g) \phi_\tau(f)}\,.
\end{equation}
In particular, the quantum density $\varrho_\tau(x) = G_\tau(x;x)$ is given by the diagonal of $G_\tau$.
From Lemma \ref{Quantum Wick theorem} (i) we find
\begin{equation} \label{quantum_G}
G_\tau \;=\; \frac{1}{\tau (\ee^{h_\tau / \tau} - 1)} \;=\; \frac{1}{\tau} \frac{\ee^{- h_\tau/\tau}}{1 - \ee^{- h_\tau / \tau}}\,.
\end{equation}
Analogously, we define the \emph{classical Green function}, $G$, as the covariance of the field $\phi$ under the Gaussian measure $\mu$:
\begin{equation} \label{classical_G}
\scalar{f}{G g} \;\deq\; \int \dd \mu \, \scalar{f}{\phi} \, \scalar{\phi}{g} \;=\; \scalar{f}{h^{-1} g}\,,
\end{equation}
i.e.\ $G = h^{-1}$.

Next, we construct an observable using a self-adjoint operator $\xi$ on $\fra H^{(p)}$. More precisely, for $p \in \N$ we denote by
\begin{equation*}
\fra B_p \;\deq\; \hb{\xi \in \fra S^2(\fra H^{(p)}) \col \norm{\xi}_{\fra S^2(\fra H^{(p)})} \leq 1}
\end{equation*}
the unit ball of $S^2(\fra H^{(p)})$. We define the \emph{lift of $\xi \in \fra B_p$ to $\cal F$} through
\begin{equation}
\label{theta_tau}
\Theta_\tau(\xi) \;\deq\; \int \dd x_1 \cdots \dd x_p \, \dd y_1 \cdots \dd y_p \, \xi(x_1, \dots, x_p; y_1, \dots, y_p) \phi_\tau^*(x_1) \cdots \phi_\tau^*(x_p) \phi_\tau(y_1) \cdots \phi_\tau(y_p)\,.
\end{equation}
The main work in this section is to compute the expectation of $\Theta_\tau(\xi)$ in the state $\rho_\tau^\eta(\cdot)$, for arbitrary self-adjoint $\xi \in \fra B_p$. In addition, for the Borel summation argument we have to introduce a complex parameter in front of the interaction potential $W_\tau$. To that end, we write
\begin{equation} \label{rho_tau_frac}
\rho^\eta_{\tau}(\Theta_\tau(\xi)) \;=\; \frac{\tr \pb{\Theta_\tau(\xi) \, P_\tau^\eta}}{\tr (P_\tau^\eta)}
\;=\; \frac{\tilde \rho^\eta_{\tau,1}(\Theta_\tau(\xi))}{\tilde \rho^\eta_{\tau,1}(I)}\,,
\end{equation}
where we defined
\begin{equation} \label{def_rho_tau_tilde}
\tilde \rho^\eta_{\tau,z}(\cal A) \;\deq\; \frac{\tr \pb{\cal A \, \ee^{-\eta H_{\tau,0}} \ee^{-(1 - 2\eta) H_{\tau,0}  - z W_\tau} \ee^{-\eta H_{\tau,0}}}}{\tr (\ee^{-H_{\tau,0}})}\,.
\end{equation}
Here $z$ is a free complex parameter with nonnegative real part.

\subsection{Duhamel expansion}
\label{Duhamel expansion}

Define the function
\begin{equation} \label{def_A_tau}
A_{\tau}^{\xi}(z) \;\deq\; \tilde \rho^\eta_{\tau,z}(\Theta_\tau(\xi))\,.
\end{equation}
We now perform a Taylor expansion up to order $M \in \N$ of $A_{\tau}^{\xi}(z)$ in the parameter $z$. This is done using a Duhamel expansion. The resulting coefficients are
\begin{multline}
a_{\tau,m}^{\xi} \;\deq\; 
\tr \Bigg((-1)^m  \frac{1}{(1-2\eta)^m} \int_{\eta}^{1-\eta} \dd t_1 \int_{\eta}^{t_1} \dd t_2 \cdots \int_{\eta}^{t_{m-1}} \,\dd t_m \\
\times \Theta_\tau(\xi)\, \ee^{-(1-t_1) H_{\tau,0}} \,W_\tau \,\ee^{-(t_1-t_2)H_{\tau,0}} \,W_{\tau} \cdots \ee^{-(t_{m-1}-t_m)H_{\tau,0}} \, W_{\tau} \, \ee^{-t_m H_{\tau,0}} \Bigg) \bigg/ \tr \big(\ee^{-H_{\tau,0}}\big)\,, \label{Explicit term a}
\end{multline}
and the remainder term is
\begin{multline}
R_{\tau,M}^{\xi}(z) \;\deq\; 
\tr \Bigg((-1)^M \frac{z^M}{(1-2\eta)^{M}}  \int_{0}^{1-2\eta} \,\dd t_1 \int_{0}^{t_1} \,\dd t_2  \cdots \int_{0}^{t_{M-1}} \,\dd t_M \,
\Theta_\tau(\xi)\,\ee^{-(1-\eta-t_1) H_{\tau,0}} \,W_\tau \,\\
\times 
\ee^{-(t_1-t_2)H_{\tau,0}} \,W_{\tau} \ee^{-(t_2 - t_3) H_{\tau,0}} \cdots 
\, W_{\tau} \, \ee^{-t_M (H_{\tau,0}+\frac{z}{1-2\eta}W_{\tau})} \ee^{-\eta H_{\tau,0}}
\Bigg) \bigg/ \tr \big(\ee^{-H_{\tau,0}}\big) \,. 
\label{Remainder term R}
\end{multline}

\begin{lemma} \label{lem_Dyson_exp}
For any $M \in \mathbb{N}$ we have $A_{\tau}^{\xi}(z)=\sum_{m=0}^{M-1}a_{\tau,m}^{\xi} z^m + R_{\tau,M}^{\xi}(z)$.
\end{lemma}

\begin{proof}
By performing a Duhamel expansion of $A_{\tau}^{\xi}(z)$ in the parameter $z$ obtained by iterating the Duhamel formula $M$ times, we find that, for $m=0,1,\ldots, M-1$, the $m$-th coefficient equals
\begin{align*}
&\tr \Bigg((-1)^m \int_{0}^{1} \dd s_1 \int_{0}^{s_1} \dd s_2 \cdots \int_{0}^{s_{m-1}} \,\dd s_m 
\Theta_\tau(\xi)\,\ee^{-\eta H_{\tau,0}} \, \ee^{-(1 -s_1) H_{\tau,0}^{\eta}} \,W_\tau \,\ee^{-(s_1-s_2)H_{\tau,0}^{\eta}} \,W_{\tau}\\ 
&\qquad \qquad \cdots \ee^{-(s_{m-1}-s_m)H_{\tau,0}^{\eta}} \, W_{\tau} \, \ee^{-s_m H_{\tau,0}^{\eta}} \,\ee^{-\eta H_{\tau,0}} \Bigg) \bigg/ \tr \Big(\ee^{-H_{\tau,0}}\Big) 
\end{align*}
where $H_{\tau,0}^{\eta}\deq(1 - 2\eta) H_{\tau,0}$. We first change variables by dilation as $s'_j \deq (1-2\eta)s_j$ and then by translation as $t_j \deq s'_j+\eta$ in order to deduce \eqref{Explicit term a}. The proof of
the formula for the remainder term is analogous, except that we only perform the first change of variables.
\end{proof}

\subsection{The explicit terms I: time-evolved Green function representation} \label{sec:time_evolved_G}

We shall need time-evolved versions of the creation and annihilation operators with respect to the Hamiltonian $h_{\tau}/\tau$. 
We give a precise definition by expanding in the eigenfunctions of $h_\tau$. 
\begin{definition}
\label{time evolved operators}
Given $t \in \mathbb{R}$, we define the operator-valued distributions $\ee^{t h_\tau / \tau} \phi_\tau$ and $\ee^{t h_\tau / \tau} \phi_\tau^*$ respectively as 
\begin{equation} \label{efunction_exp}
\big(\ee^{t h_\tau / \tau} \phi_\tau\big)(x) \;\deq\; \sum_{k \in \N} \ee^{t \lambda_{\tau,k} / \tau}u_{\tau,k}(x) \phi_{\tau,k}\,, \quad
\big(\ee^{t h_\tau / \tau} \phi_\tau^*\big)(x) \;\deq\; \sum_{k \in \N} \ee^{t \lambda_{\tau,k} / \tau}\bar{u}_{\tau,k}(x) \phi_{\tau,k}^*\,.
\end{equation}
\end{definition}

These objects arise when we conjugate the creation and annihilation operators by $\ee^{t H_{\tau,0}}$.

\begin{lemma}
\label{Pullthrough formula lemma}
For $t \in \mathbb{R}$ we have
\begin{equation}
\label{Conjugation identity}
\ee^{t H_{\tau,0}} \,\phi_{\tau}^*(x) \, \ee^{-t H_{\tau,0}}\;=\; \pb{\ee^{th_{\tau}/\tau}\phi_{\tau}^*}(x)\,,
\qquad \ee^{t H_{\tau,0}} \,\phi_{\tau}(x) \, \ee^{-t H_{\tau,0}}\;=\; \pb{\ee^{-th_{\tau}/\tau}\phi_{\tau}}(x)\,.
\end{equation}
\end{lemma}

\begin{proof}
Observe first that the identities from \eqref{Conjugation identity} are equivalent by duality.
By using \eqref{phi_tau}, \eqref{H_{tau,0}}, and linearity, the first identity of \eqref{Conjugation identity} follows from
\begin{equation}
\label{Pullthrough formula}
F_\tau(t) \;\deq\; \ee^{\,t \,\phi^{*}_{\tau,k} \phi_{\tau,k}} \,\phi^{*}_{\tau,k} \,\ee^{\,-t\,\phi^{*}_{\tau,k} \phi_{\tau,k}} \;=\; \ee^{\,t/\tau} \phi^{*}_{\tau,k} \,.
\end{equation}
In order to show \eqref{Pullthrough formula}, we remark that $F_\tau(t)$ is an operator-valued distribution which depends smoothly on the parameter $t$. We compute
\begin{equation*}
F_{\tau}'(t) \;=\; \ee^{\,t \,\phi^{*}_{\tau,k} \phi_{\tau,k}} \, \Big(\phi^{*}_{\tau,k}\,\phi_{\tau,k}\,\phi^{*}_{\tau,k}-\phi^{*}_{\tau,k}\,\phi^{*}_{\tau,k}\,\phi_{\tau,k} \Big) \, \ee^{\,-t \,\phi^{*}_{\tau,k} \phi_{\tau,k}} \,
\end{equation*}
as a densely defined quadratic form.
By \eqref{CCR k}, we therefore get
\begin{equation*}
F_{\tau}'(t) \;=\; \tau^{-1} \, \ee^{\,t \,\phi^{*}_{\tau,k} \phi_{\tau,k}} \,\phi_{\tau,k}^{*} \, \ee^{\,-t \,\phi^{*}_{\tau,k} \phi_{\tau,k}}\;=\;\tau^{-1} \,F_{\tau}(t)\,.
\end{equation*}
Since $F_{\tau}(0)=\phi_{\tau,k}^{*}$, it follows that $F_{\tau}(t) =\ee^{\,t/\tau}\,\phi_{\tau,k}^{*}$, as claimed.
\end{proof}
Lemma \ref{Pullthrough formula lemma} now directly implies the following result for the conjugation of the interaction $W_\tau$ by $\ee^{t H_{\tau,0}}$.
\begin{corollary}
\label{cor:Conjugation identity}
For $t \in \mathbb{R}$ we have
\begin{align}
&\ee^{t H_{\tau,0}} \,W_\tau \,\ee^{-t H_{\tau,0}} \;=\; \frac{1}{2} \int \dd x \, \dd y \, \Big( \big(\ee^{th_{\tau}/\tau} \phi^*_\tau\big)(x) \, \big(\ee^{-th_{\tau}/\tau}\phi_\tau\big)(x) - \varrho_\tau(x)\Big) \,\notag \\
&\qquad  \qquad \qquad \quad \quad \times \,w(x - y)  \, \Big( \big(\ee^{th_{\tau}/\tau}\phi^*_\tau\big)(y) \, \big(\ee^{-th_{\tau}/\tau}\phi_\tau\big)(y) - \varrho_\tau(y)\Big)\,.
\end{align}
\end{corollary}

We denote the renormalized product of two operators $\cal A$, $\cal B$ that are each linear in $\phi_\tau, \phi_\tau^*$ by 
\begin{equation} \label{Wick ordering quadratic}
: \cal A \, \cal B : \;=\; \cal A \, \cal B - \rho_{\tau,0} ( \cal A \, \cal B ) \,.
\end{equation}
We now substitute the result of Corollary \ref{cor:Conjugation identity} into \eqref{Explicit term a} and use cyclicity of the trace to obtain
\begin{equation} \label{def_a_tau}
a_{\tau,m}^\xi \;=\; (-1)^m  \frac{1}{(1-2\eta)^m \, 2^m} \int_{\eta}^{1-\eta} \dd t_1 \int_{\eta}^{t_1} \dd t_2 \cdots \int_{\eta}^{t_{m-1}} \,\dd t_m \, f_{\tau,m}^\xi(t_1, \dots, t_m)\,,
\end{equation}
where
\begin{align}
&f_{\tau,m}^\xi(t_1, \dots, t_m) \;\deq \; \int \dd x_1 \,\ldots\,\dd x_{m+p} \, \dd y_1 \,\ldots\, \dd y_{m+p}  \notag \\
&\bigg(\prod_{i=1}^{m} \,w(x_i-y_i) \bigg) \cdot \, \xi(x_{m+1}, \dots, x_{m+p}; y_{m+1}, \dots, y_{m+p}) \notag \\ 
& \rho_{\tau,0} \Bigg(\,\prod_{i=1}^{m} \bigg( \Big[:\big(\ee^{t_ih_{\tau}/\tau} \phi^*_\tau \big)(x_i) \, \big(\ee^{-t_ih_{\tau}/\tau}\phi_\tau\big)(x_i): \Big] \,\Big[:\big(\ee^{t_ih_{\tau}/\tau} \phi^*_\tau \big)(y_i) \, \big(\ee^{-t_ih_{\tau}/\tau}\phi_\tau\big)(y_i): \Big] \, \bigg) 
\notag
\\
&\qquad  \times \prod_{i=1}^{p} \phi_\tau^*(x_{m+i})  \prod_{i=1}^{p} \phi_\tau(y_{m+i}) \Bigg)  \label{Definition of f_{tau,m}^{xi}}\,. 
\end{align}
Here we recall the definition of $\rho_{\tau,0}$ from \eqref{rho_{tau,0}}. In the above formula, and in what follows, when we write $\prod$, we always take the product of the operators in fixed order of increasing indices from left to right.

We define an abstract vertex set $\cal X$ of $4m+2p$ elements, which encode the $\phi_\tau$ and $\phi_\tau^*$ in \eqref{Definition of f_{tau,m}^{xi}}.
\begin{definition}
\label{cal X}
Let $m,p \in \N$ be given. 
\begin{enumerate}
\item
The vertex set $\cal X \equiv \cal X(m,p)$ consists of triples of the form $(i,r,\sign)$, where $i=1, \dots, m + 1$. For $i=1,\ldots,m$ we have $r=1,2$, and for $i=m+1$ we have $r=1,\ldots,p$. Finally $\sign=\pm 1$. 
In what follows, we also denote each such triple $(i,r,\sign)$ by $\alpha$. Given $\alpha=(i,r,\sign) \in \cal X$, we write its components $i,r,\sign$ as $i_{\alpha}, r_{\alpha}, \sign_{\alpha}$ respectively. 
\item On $\cal X$, we impose a linear order $\leq$ by ordering the elements of $\alpha = (i,r,\sign) \in \cal X$ in increasing order as
\begin{multline}
(1,1,+1), (1,1,-1), (1,2,+1), (1,2,-1), \dots,
(m,1,+1), (m,1,-1), (m,2,+1), (m,2,-1),
\\
(m+1,1,+1), \dots, (m+1, p,+1), (m+1,1,-1), \dots, (m+1, p, -1)\,.
\label{linear order}
\end{multline}
Moreover, given $\alpha,\beta \in \cal X$, we say that $\alpha<\beta$ if $\alpha \leq \beta$ and $\alpha \neq \beta$.
\end{enumerate}
\end{definition}

Before we proceed, let us explain the motivation for this definition. The first component $i$
indexes the operator ($w$ or $\xi$) occurring outside of $\rho_{\tau,0} $ in the integrand in \eqref{Definition of f_{tau,m}^{xi}}. Furthermore, for $i = 1, \dots, m$, the index
$r = 1,2$ tells us whether we are looking at the factor $\phi_{\tau}^\sharp(x_i)$ or $\phi_{\tau}^\sharp(y_i)$, whereas for $i=m+1$,  $r = 1, \dots, p$ indicates the choice of the factor $\phi_{\tau}^*(x_{m+r})$ or $\phi_\tau(y_{m+r})$. Finally,  $\sign = \pm1 $ is taken to be $+1$ for a $\phi_{\tau}^*$ and $-1$ for a $\phi_{\tau}$. In this way $(\cal X, \leq)$ encodes the occurrences of $\phi_\tau$ and $\phi_\tau^*$ in \eqref{Definition of f_{tau,m}^{xi}} in the appropriate order. A graphical interpretation of this encoding is given in Figure \ref{fig:graph1} below.

Next, to each vertex $\alpha=(i,r,\sign) \in \cal X$ we assign an integration label $x_\alpha = x_{i,r,\sign}$. Moreover, to each $i = 1, \dots, m$ we assign an index $t_i$, and we also use $t_{m+1} \deq 0$.
For $\alpha = (i,r,\sign)$, we interchangeably use the notations $x_{\alpha} \equiv x_{i,r,\sign}$ and $t_{\alpha} \equiv t_i$, depending on which form is more convenient.
We also abbreviate
\begin{equation} \label{x_t_labels}
\f x = (x_{\alpha})_{\alpha \in \cal X} \in \Lambda^{\cal X}\,, \qquad \f t = (t_{\alpha})_{\alpha \in \cal X} \in \R^{\cal X}\,.
\end{equation}
We always consider $(t_1, \dots, t_m)$ in the support of the integral in \eqref{def_a_tau}, in which case we have $\f t \in \fra A$, where we defined the simplex $\fra A \;\equiv\; \fra A (m)$ as
\begin{equation} \label{def_fraA}
\fra A \;\deq\; \hb{\f t \in \R^{\cal X} \col t_{i,r,\sign} = t_i \text{ with } 0 = t_{m+1} \leq \eta < t_m < t_{m - 1} < \cdots < t_2 < t_1 < 1 - \eta}\,.
\end{equation}
Note that
\begin{equation} \label{t_ordering}
\alpha < \beta \quad \Longrightarrow \quad 0 \leq t_\alpha - t_\beta < 1\,.
\end{equation}

Define a family of operator-valued distributions 
$(\cal B_{\alpha}(\f x,\f t))_\alpha$ through
\begin{align*}
\cal B_{\alpha}(\f x,\f t) \;\deq\;
\begin{cases}
\big(\ee^{t_\alpha h_{\tau}/\tau} \phi^*_\tau\big)(x_\alpha) &\mbox{if }\sign_\alpha=1\\
\big(\ee^{-t_\alpha h_{\tau}/\tau} \phi_\tau\big)(x_\alpha) &\mbox{if }\sign_\alpha=-1\,.
\end{cases}
\end{align*}

\begin{definition} \label{def_pairing}
Let $\Pi$ be a pairing of $\cal X$, i.e.\ a one-regular graph on $\cal X$. We regard its edges as ordered pairs $(\alpha, \beta)$ satisfying $\alpha < \beta$.
Let $\fra P \equiv \fra P(m,p)$ denote the set of pairings of $\cal X$ such that 
\begin{enumerate}
\item
for each $i = 1, \dots, m$ and $r = 1,2$ we have $((i,r,+1), (i,r,-1)) \notin \Pi$;
\item
for each $(\alpha, \beta) \in \Pi$ we have $\sign_{\alpha} \sign_{\beta} = -1$.
\end{enumerate}
\end{definition}

\begin{definition} \label{def:I_tau_PI}
We define the value of $\Pi \in \fra P$ through
\begin{multline} \label{def_calI}
\cal I_{\tau,\Pi}^{\xi} (\f t) \;\deq\; \int_{\Lambda^{\cal X}} \dd \f x \,\prod_{i=1}^{m} \pbb{ w(x_{i,1,1}-x_{i,2,1}) \prod_{r=1}^{2} \delta(x_{i,r,1}-x_{i,r,-1})}
\\
\times \xi(x_{m+1,1,1}, \dots, x_{m+1,p,1}; x_{m+1,1,-1}, \dots, x_{m+1,p,-1}) \prod_{(\alpha,\beta) \in \Pi} \rho_{\tau,0} \big(\cal B_{\alpha} (\f x, \f t) \, \cal B_{\beta} (\f x, \f t)
\big)\,.
\end{multline}
\end{definition}

\begin{lemma}
\label{Wick_application lemma}
For each $m,p \in \N$  we have $f_{\tau,m}^\xi(\f t) = \sum_{\Pi \in \fra P} \cal I_{\tau,\Pi}^{\xi}(\f t)$.

\end{lemma}

\begin{proof}

Analogously to the labels $\f x = (x_\alpha)_\alpha$, we assign to each vertex $\alpha=(i,r,\sign) \in \cal X$ a spectral index $k_{\alpha}=k_{i,r,\sign}$, and abbreviate $\f k = (k_\alpha)_\alpha$.

Given $k \geq 0$, we define
\begin{align*}
u^{\alpha}_{\tau,k} \;\deq\;
\begin{cases}
\bar{u}_{\tau,k} &\mbox{if }\sign_\alpha=+1\\
u_{\tau,k} &\mbox{if }\sign_\alpha=-1\,.
\end{cases}
\end{align*}
Using \eqref{efunction_exp} we write
\begin{equation}
\label{efunction_exp B}
\cal B_{\alpha} (\f x, \f t) \;=\;  \mathop{\sum}_{k_{\alpha} \in \N} \ee^{\,\sign_\alpha t_\alpha \lambda_{\tau,k_{\alpha}} / \tau} u^{\alpha}_{\tau,k_{\alpha}}(x_{\alpha}) \cal A_\alpha(\f k)
\end{equation}
where the family of operators $(\cal A_{\alpha}(\f k))_\alpha$ is defined by
\begin{equation} \label{def_cal_A}
\cal A_{\alpha}(\f k) \;\deq\;
\begin{cases}
\phi^*_{\tau,k_{\alpha}} &\mbox{if }\sign_\alpha=+1\\
\phi_{\tau,k_{\alpha}} &\mbox{if }\sign_\alpha=-1\,.
\end{cases}
\end{equation}
As before, we also use the abbreviation $\cal A_{\alpha} \equiv \cal A_{i,r,\sign}$ when $\alpha=(i,r,\sign)$.

Substituting \eqref{efunction_exp B} into \eqref{Definition of f_{tau,m}^{xi}}, it follows that 
\begin{multline}
f_{\tau,m}^\xi(\f t) \;=\; \int_{\Lambda^{\cal X}} \dd \f x \,\prod_{i=1}^{m} \pbb{ w(x_{i,1,+1}-x_{i,2,+1}) \prod_{r=1}^{2} \delta(x_{i,r,+1}-x_{i,r,-1})} 
\\
\times \xi(x_{m+1,1,+1}, \dots, x_{m+1,p,+1}; x_{m+1,1,-1}, \dots, x_{m+1,p,-1})
\\
\times \sum_{\f k} \prod_{\alpha} \ee^{\sign_\alpha t_{\alpha} \lambda_{\tau,k_\alpha} / \tau} u^{\alpha}_{\tau,k_\alpha}(x_\alpha)
\,
\rho_{\tau,0} \Bigg(\,\prod_{i=1}^{m} \prod_{r=1}^{2}  \col \qBB{\prod_{\sign = \pm 1} \cal A_{i,r,\sign}(\f k)} \col 
\prod_{\sign = \pm 1} \prod_{r = 1}^p \cal A_{m+1,r,\sign} (\f k) \,\bigg)\,
\Bigg)\,.
\label{Formula for f_{tau,m}^{xi}}
\end{multline}
In the above expression, we adopt our previous convention for $\prod$ to $\cal X$ in the sense that all the products are taken in increasing order of the vertices in $\cal X$. 
We henceforth use this convention.

Next, we claim that Lemma \ref{Quantum Wick theorem} implies
\begin{equation}
\label{Quantum Wick theorem application}
\rho_{\tau,0} \Bigg(\,\prod_{i=1}^{m} \prod_{r=1}^{2}  \col \qBB{\prod_{\sign = \pm 1} \cal A_{i,r,\sign}(\f k)} \col
\prod_{\sign = \pm 1} \prod_{r = 1}^p  \cal A_{m+1,r,\sign} (\f k) \,\bigg)\,
\Bigg) 
\;=\; 
\sum_{\Pi \in \fra P}\,\prod_{(\alpha,\beta) \in \Pi} \rho_{\tau,0} \big(\cal A_{\alpha} (\f k) \, \cal A_{\beta} (\f k) \big)\,.
\end{equation}
Indeed, we can multiply out all $2m$ factors of the form
\begin{equation*}
\col \qBB{\prod_{\sign = \pm 1} \cal A_{i,r,\sign}(\f k)} \col \;=\; \prod_{\sign = \pm 1} \cal A_{i,r,\sign}(\f k) - \rho_{\tau,0} \pBB{\prod_{\sign = \pm 1} \cal A_{i,r,\sign}(\f k)}\,,
\end{equation*}
and apply Lemma \ref{Quantum Wick theorem} (ii) to each resulting term. The result is a sum of the form \eqref{Quantum Wick theorem application} over all pairings satisfying property (i) of Definition \ref{def_pairing}. Note that it is precisely the renormalization $\col[\,\cdot\,]\col$ that gets rid of the pairings violating property (i) of Definition \ref{def_pairing}. Moreover, by Lemma \ref{Quantum Wick theorem} (i) we find that if $\Pi$ contains an edge $(\alpha, \beta)$ satisfying $\sign_{\alpha} \sign_{\beta} = 1$ then its contribution vanishes. Hence, any pairing with a nonzero contribution to \eqref{Quantum Wick theorem application} satisfies property (ii) of Definition \ref{def_pairing}. This proves \eqref{Quantum Wick theorem application}.

Substituting \eqref{Quantum Wick theorem application} into \eqref{Formula for f_{tau,m}^{xi}} and using \eqref{efunction_exp B}, we deduce the claim.
\end{proof}

Next, in order to compute the factors $\rho_{\tau,0}(\cdot)$ in \eqref{def_calI}, for $t \in \R$ we introduce the bounded operators
\begin{align}
\label{def_St}
S_{\tau,t} &\;\deq\; \ee^{-t h_\tau/\tau} \qquad (t \geq 0)\,,
\\ \label{def_Gt}
G_{\tau,t} &\;\deq\; \frac{\ee^{-t h_\tau / \tau}}{\tau (\ee^{h_\tau / \tau} - 1)} \qquad (t \geq -1)\,.
\end{align}
In particular, $G_{\tau,0} = G_\tau$.
As usual, we denote the Schwartz operator kernels of $S_{\tau,t}$ and $G_{\tau,t}$ by $S_{\tau,t}(x;y)$ and $G_{\tau,t}(x;y)$ respectively. In general, the integral kernels $S_{\tau,t}(x;y)$ and $G_{\tau,t}(x;y)$ are measures on $\Lambda^2$. 

We have the following simple but crucial positivity result for the operator kernels.

\begin{lemma}
\label{Positivity lemma}
For all $t \geq 0$ and $x,y \in \Lambda$ we have $S_{\tau,t}(x;y) = S_{\tau,t}(y;x) \geq 0$. Moreover, for all $t > -1$ and $x,y \in \Lambda$ we have $G_{\tau,t}(x;y) = G_{\tau,t}(y;x) \geq 0$.
\end{lemma}

\begin{proof}
Since $S_{\tau,t}$ and $G_{\tau,t}$ are self-adjoint, the symmetry of their kernels follows from their pointwise nonnegativity.
Using the convergent Neumann series
\begin{equation} \label{neumann_series}
G_{\tau,t} \;=\; \frac{1}{\tau} \sum_{n \geq 1} \ee^{-(t + n) h_\tau /\tau}
\end{equation}
and continuity,
we find that it suffices to prove that $S_{\tau,t}(x;y) \geq 0$ for all $t > 0$ and $x,y \in \Lambda$. We do this using the Feynman-Kac formula. For $x \in \Lambda$, denote by $W_x$ the Wiener measure on continuous paths $b \in C([0,\infty);\Lambda)$ satisfying $b(0) = x$, i.e.\ the law of the $\Lambda$-valued Brownian motion starting at $x$ with variance $\int W_x(\dd b) \, (b(t) - x)^2 = 2t$ for $t > 0$. Then we have
\begin{equation*}
S_{\tau,t}(x;y) \;=\; \int W_x(\dd b) \, \exp \pbb{-\int_0^{t/\tau} \dd s \, \qb{\kappa + v_\tau(b(s))}} \, \delta\pb{b(t/\tau) - y}\,,
\end{equation*}
which is manifestly nonnegative.
\end{proof}

\begin{lemma}
\label{Correlation functions}
Let $\alpha,\beta \in \cal X$ satisfy $\alpha < \beta$. 
\begin{enumerate}
\item
If $\sign_{\alpha} = +1$, $\sign_{\beta} = -1$, and $t_{\alpha}-t_{\beta}<1$ then
\begin{equation*}
\rho_{\tau,0} \big(\cal B_{\alpha} (\f x, \f t) \, \cal B_{\beta} (\f x, \f t)\big) \;=\; G_{\tau, -(t_{\alpha} - t_{\beta})}(x_{\alpha}; x_{\beta})\,.
\end{equation*}
\item
If $\sign_{\alpha} = -1$, $\sign_{\beta} = +1$, and $t_{\alpha}-t_{\beta} \geq 0$ then
\begin{equation*}
\rho_{\tau,0} \big(\cal B_{\alpha} (\f x, \f t) \, \cal B_{\beta} (\f x, \f t)\big) \;=\; G_{\tau, t_{\alpha} - t_{\beta}}(x_{\alpha}; x_{\beta}) + \frac{1}{\tau} S_{\tau, t_{\alpha} - t_{\beta}}(x_{\alpha}; x_{\beta})\,.
\end{equation*}
\item
In both cases (i) and (ii) we have
\begin{equation*}
\rho_{\tau,0} \big(\cal B_{\alpha} (\f x, \f t) \, \cal B_{\beta} (\f x, \f t)\big) \;\geq\; 0\,.
\end{equation*}
\end{enumerate}
\end{lemma}

\begin{proof}
We use the spectral representation to write
\begin{equation*}
S_{\tau,t}(x;y) \;=\;\sum_{k \in \N} e^{-t\lambda_{\tau,k}/\tau} \bar{u}_{\tau,k}(x) u_{\tau,k} (y)
\end{equation*}
for $t>0$, and
\begin{equation*}
G_{\tau,t}(x;y) \;=\;\sum_{k \in \N} e^{-t\lambda_{\tau,k}/\tau} \bar{u}_{\tau,k}(x) u_{\tau,k} (y) \rho_{\tau,0} \big(\phi^*_{\tau,k} \phi_{\tau,k}\big)
\end{equation*}
for $t>-1$, where we used \eqref{quantum_G}.

If $\sign_{\alpha}=+1, \sign_{\beta}=-1$ and $t_{\alpha}-t_{\beta}<1$, we compute using \eqref{efunction_exp B} and \eqref{def_cal_A}
\begin{multline*}
\rho_{\tau,0} \big(\cal B_{\alpha} (\f x, \f t) \, \cal B_{\beta} (\f x, \f t)\big) \;=\;
\sum_{k_{\alpha},k_{\beta} \in \N} e^{\,t_{\alpha} \lambda_{\tau,k_{\alpha}}/\tau-\,t_{\beta} \lambda_{\tau,k_{\beta}}/\tau} \bar{u}_{\tau,k_{\alpha}}(x_{\alpha}) u_{\tau,k_{\beta}} (x_{\beta}) \rho_{\tau,0} \big(\phi^*_{\tau,k_{\alpha}} \phi_{\tau,k_{\beta}}\big)
\\
\;=\; \sum_{k \in \N} e^{(t_{\alpha}-t_{\beta}) \lambda_{\tau,k}/\tau} \bar{u}_{\tau,k}(x_{\alpha}) u_{\tau,k} (x_{\beta})
\scalar{u_{\tau,k}}{G_{\tau} u_{\tau,k}}
\;=\; G_{\tau, -(t_{\alpha} - t_{\beta})}(x_{\alpha}; x_{\beta})\,,
\end{multline*}
as was claimed in part (i). In the second step we used \eqref{quantum correlation}, \eqref{quantum_G}, and that $G_\tau$ is diagonal in the basis $(u_{\tau,k})$.

Likewise, if $\sign_{\alpha}=-1, \sign_{\beta}=+1$ and $t_{\alpha}-t_{\beta}\geq0$, we compute using \eqref{CCR k}
\begin{align*}
\rho_{\tau,0} \big(\cal B_{\alpha} (\f x, \f t) \, \cal B_{\beta} (\f x, \f t)\big) &\;=\; \sum_{k_{\alpha},k_{\beta} \in \N} e^{\,-t_{\alpha} \lambda_{\tau,k_{\alpha}}/\tau+\,t_{\beta} \lambda_{\tau,k_{\beta}}/\tau} u_{\tau,k_{\alpha}}(x_{\alpha}) \bar{u}_{\tau,k_{\beta}} (x_{\beta}) \rho_{\tau,0} \big(\phi^*_{\tau,k_{\alpha}} \phi_{\tau,k_{\beta}}\big) \\
&\qquad + \frac{1}{\tau} \sum_{k \in \N} e^{\,-t_{\alpha} \lambda_{\tau, k}/\tau+\,t_{\beta} \lambda_{\tau, k}/\tau} u_{\tau, k}(x_{\alpha}) \bar{u}_{\tau, k} (x_{\beta}) \\
&\;=\;\sum_{k \in \N} e^{\,-(t_{\alpha}-t_{\beta}) \lambda_{\tau,k}/\tau} u_{\tau,k}(x_{\alpha}) \bar{u}_{\tau,k} (x_{\beta}) \scalar{u_{\tau,k}}{G_{\tau} u_{\tau,k}}
\\
&\qquad + \frac{1}{\tau} \sum_{k \in \N} e^{\,-(t_{\alpha}-t_{\beta}) \lambda_{\tau,k}/\tau} u_{\tau,k}(x_{\alpha}) \bar{u}_{\tau,k} (x_{\beta}) 
\\
&\;=\; \bar{G}_{\tau, t_{\alpha} - t_{\beta}}(x_{\alpha}; x_{\beta}) + \frac{1}{\tau} \bar{S}_{\tau, t_{\alpha} - t_{\beta}}(x_{\alpha}; x_{\beta})
\\
&\;=\;G_{\tau, t_{\alpha} - t_{\beta}}(x_{\alpha}; x_{\beta}) + \frac{1}{\tau} S_{\tau, t_{\alpha} - t_{\beta}}(x_{\alpha}; x_{\beta})\,.
\end{align*}
In the last equality, we used that $G_{\tau,t_{\alpha}-t_{\beta}}(x;y)$ and $S_{\tau,t_{\alpha}-t_{\beta}}(x;y)$ are real-valued kernels by Lemma \ref{Positivity lemma}. This proves part (ii).

Part (iii) now follows from Lemma \ref{Positivity lemma}.
\end{proof}

\subsection{The explicit terms II: graphical representation} \label{sec:graphs}

It is convenient to introduce a graphical representation for the vertex set $\cal X$ and the pairing $\Pi \in \fra P$. See Figure \ref{fig:graph1} for an illustration.

\begin{figure}[!ht]
\begin{center}
{\scriptsize 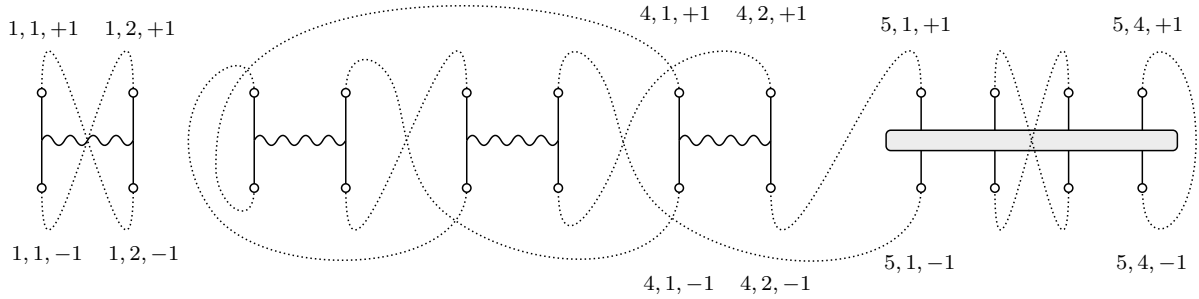}
\end{center}
\caption{A graphical depiction of a pairing $\Pi \in \fra P$. Here $m = 4$ and $p = 4$. We depict each vertex $\alpha = (i,r,\sign) \in \cal X$ with a white dot, and each edge of $\Pi$ using a dotted line. We draw vertices with $\sign = +1$ on the top and vertices with $\sign = -1$ on the bottom. We indicate the names of a few vertices next to the corresponding dots. If two vertices are connected by a vertical line, this indicates a delta function in \eqref{def_calI} that constrains their respective labels in $\f x$ to coincide, and we connect these vertical lines with a horizontal wiggly line representing a factor $w(\cdot)$ in \eqref{def_calI}. Finally, we represent the integral kernel of $\xi$ in \eqref{def_calI} with a grey box. In this picture, the leftmost $4m = 16$ vertices are ordered lexicographically according to their $(x,y)$-coordinates, and the subsequent $2p = 8$ vertices are ordered lexicographically according to their $(y,x)$-coordinates, whereby the $x$-axis is ordered from left to right and the $y$-axis from top to bottom; see \eqref{linear order}.
\label{fig:graph1}}
\end{figure}

\begin{definition} \label{def_collapsed_graph}
Fix $m,p \in \N$. To each $\Pi \in \fra P$ we assign an edge-coloured undirected multigraph\footnote{This means that the graph may have multiple edges.} $(\cal V_\Pi, \cal E_\Pi, \sigma_\Pi) \equiv (\cal V, \cal E, \sigma)$, with a colouring $\sigma \col \cal E \to {\pm 1}$, as follows.
\begin{enumerate}
\item
On $\cal X$ we introduce the equivalence relation $\alpha \sim \beta$ if and only if $i_\alpha = i_\beta \leq m$ and $r_\alpha = r_\beta$. We define the vertex set $\cal V \deq \{[\alpha] \col \alpha \in \cal X\}$ as the set of equivalence classes of $\cal X$. We use the notation ${\cal{V}} = {\cal{V}}_2 \cup {\cal{V}}_1$,
where
\begin{equation*}
{\cal V}_2 \;\deq\; \{(i,r)\col 1 \leq i \leq m, 1 \leq r \leq 2\}\,, \qquad {\cal V}_1 \;\deq\; \{(m+1,r,\pm 1)\col 1 \leq r \leq p\}\,.
\end{equation*}
\item
The set $\cal V$ carries a total order $\leq$ inherited from $\cal X$: $[\alpha] \leq [\beta]$ whenever $\alpha \leq \beta$. It is trivial to check that $\leq$ is well-defined on $\cal V$ (i.e.\ does not depend on the choice of the representatives $\alpha, \beta \in \cal X$).
\item
For a pairing $\Pi \in \fra P$, each edge $(\alpha, \beta) \in \Pi$ gives rise to an edge $e = \{[\alpha], [\beta]\}$ of $\cal E$ with $\sigma(e) \deq \sign_\beta$.
\item
We denote by $\conn(\cal E)$ the set of connected components of $\cal E$, so that $\cal E = \bigsqcup_{\cal P \in \conn(\cal E)} \cal P$. We call the connected components $\cal P$ of $\cal E$ \emph{paths}.
\end{enumerate}
\end{definition}
Since the vertex set $\cal V$ is determined by $m$ and $p$, we often refer to the multigraph $(\cal V, \cal E)$ simply as $\cal E$. An example of such a multigraph is depicted in Figure \ref{fig:graph2}. Note that the multigraph $\cal E$ may contain multiple edges, but it cannot contain loops. This absence of loops is crucial for our estimates, and is guaranteed by the condition (i) from Definition \ref{def_pairing}, which itself was a consequence of the renormalization in the interaction $W_\tau$.

\begin{figure}[!ht]
\begin{center}
{\scriptsize 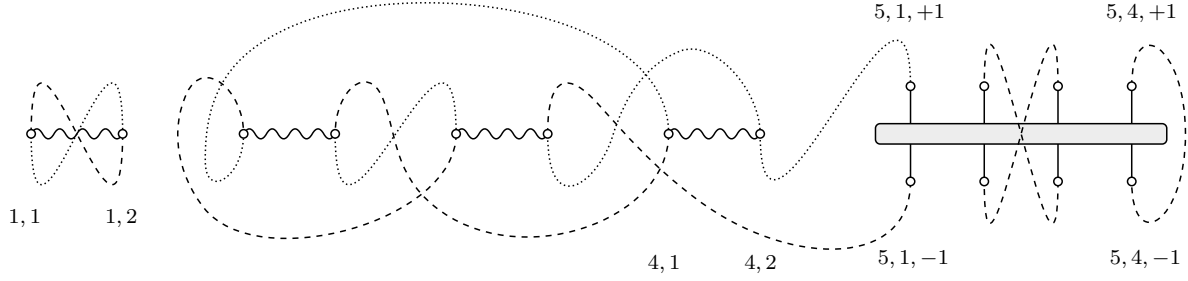}
\end{center}
\caption{A graphical depiction of the edge-coloured multigraph $\cal E$ associated with the pairing $\Pi$ from Figure \ref{fig:graph1}. We depict vertices of $\cal V$ using white dots and edges of $\cal E$ using dotted or dashed lines: edges $e$ with colour $\sigma(e) = +1$ correspond to dotted lines and 
edges $e$ with colour $\sigma(e) = -1$ to dashed lines. We indicate the names of a few vertices next to the corresponding dots. A wiggly line connecting two vertices $\va, \vb \in \cal V$ denotes a factor $w(y_{\va} - y_{\vb})$ in \eqref{I_rep}. The grey box depicts the integral kernel of $\xi$. Each line corresponding to an edge $e \in \cal E$ encodes a factor $\cal J_{\tau,e}$ in \eqref{I_rep}. The set of paths $\conn(\cal E)$ consists of two closed paths and four open paths.
\label{fig:graph2}}
\end{figure}

The proof of the following lemma is immediate from Definitions \ref{def_pairing} and \ref{def_collapsed_graph}.

\begin{lemma} \label{lem:basic_calE}
For any $\Pi \in \fra P$ the associated graph $\cal E_\Pi \equiv \cal E$ has the following properties.
\begin{enumerate}
\item Each vertex of $\cal V_2$ has degree 2.
\item Each vertex of $\cal V_1$ has degree 1.
\item There are no loops (cycles of length one).
\end{enumerate}
In particular, each $\cal P \in \conn(\cal E)$ is a path of $\cal E$ in the usual graph-theoretic sense.
\end{lemma}

For the following we fix $m,p \in \N$ and a pairing $\Pi \in \fra P$, and let $(\cal V, \cal E, \sigma)$ denote the associated graph from Definition \ref{def_collapsed_graph}.

With each $\f x=(x_{\alpha})_{\alpha \in \cal X} \in \Lambda^{\cal X}$ and $\f t = (t_{\alpha})_{\alpha \in \cal X} \in \R^{\cal X}$ we associate integration labels $\f y=(y_{\va})_{\va \in \cal V}\in \Lambda^{\cal V}$ and time labels $\f s  = (s_{\va})_{\va \in \cal V} \in \R^{\cal V}$ defined by
\begin{equation} \label{def_ys}
y_{[\alpha]} \;\deq\; x_\alpha,\qquad
s_{[\alpha]} \;\deq\; t_\alpha
\end{equation}
for any $\alpha \in \cal X$; as explained above \eqref{x_t_labels}, the definitions given above do not depend on the choice of representative $\alpha$. Note also that by \eqref{t_ordering} we have
\begin{equation} \label{s_ordered}
0 \leq s_{\va} - s_{\vb} < 1 \quad \text{for} \quad \va < \vb\,.
\end{equation}
Moreover, 
\begin{equation}
\label{s_ordered equality}
s_{\va} = s_{\vb} \quad \text{if and only if} \quad i_a = i_b\,,
\end{equation}
where we use the notation $i_{[\alpha]} \deq i_\alpha$.

\begin{definition}
\label{Open and Closed path}
Let $\cal P \in \conn(\cal E)$. We call $\cal P$ a \emph{closed path} if it only contains vertices of $\cal V_2$, and an \emph{open path} otherwise.
\end{definition}

By Lemma \ref{lem:basic_calE}, we find that any closed path $\cal P \in \conn(\cal E)$ is a closed path in $\cal V_2$ in the usual graph-theoretic sense, and any open path $\cal P \in \conn(\cal E)$ has two distinct end points in $\cal V_1$, and its remaining $\abs{\cal V(\cal P)} - 2$ vertices are in $\cal V_2$. We denote by $\cal V(\cal P) \deq \bigcup_{e \in \cal P} e \subset \cal V$ the set of vertices of the path $\cal P \in \conn(\cal E)$. Moreover, we split $\cal V(\cal P) = \cal V_2(\cal P) \sqcup \cal V_1(\cal P)$, where $\cal V_i(\cal P) \deq \cal V(\cal P) \cap \cal V_i$. Note that $\cal V_1(\cal P) = \emptyset$ whenever $\cal P$ is a closed path.

\begin{definition}
\label{J_e definition}
Let $\f y = (y_{\va})_{\va \in \cal V} \in \Lambda^{\cal V}$ and $\f s  = (s_{\va})_{\va \in \cal V} \in \R^{\cal V}$ satisfy \eqref{s_ordered}. Let $e = \{\va, \vb\} \in \cal E$ with $\va < \vb$. We define the labels $\f y_e \deq (y_a, y_b) \in \Lambda^e$ and the integral kernels
\begin{align}
\label{J_e}
\cal J_{\tau, e}(\f y_e,\f s) &\;\deq\; G_{\tau, \sigma(e) (s_{\va} - s_{\vb})}(y_{\va}; y_{\vb}) + \frac{\ind{\sigma(e) = +1}}{\tau} S_{\tau, s_{\va} - s_{\vb}}(y_{\va}; y_{\vb})\,,
\\
\hat {\cal J}_{\tau, e}(\f y_e,\f s) &\;\deq\; G_{\tau, \sigma(e) (s_{\va} - s_{\vb})}(y_{\va}; y_{\vb})\,.
\label{hat J_e} 
\end{align}
\end{definition}
Note that $\hat {\cal J}_{\tau, e}$ is always a Hilbert-Schmidt operator because $s_a - s_b < 1$. However, $\cal J_{\tau, e}$ is not Hilbert-Schmidt when $\sigma(e)=+1$ and $i_a =i_b$, in which case $s_a - s_b = 0$.

In the following we use the splitting 
\begin{equation} \label{y1_y2}
\f y \;=\; (\f y_1, \f y_2)\,
\end{equation}
where $\f y_i \deq (y_a)_{a \in \cal V_i}$ for $i = 1,2$. By a slight abuse of notation, from now on we write
\begin{equation} \label{y_1}
\xi(y_{m+1,1,+1}, \dots, y_{m+1,p,+1}; y_{m+1,1,-1}, \dots, y_{m+1,p,-1}) \;=\; \xi(\f y_1)\,.
\end{equation}

\begin{lemma} \label{lem:I_rep}
With $\f s$ defined in \eqref{def_ys}, we have
\begin{equation}
\label{I_rep}
\cal I_{\tau,\Pi}^{\xi}(\f t) \;=\; \int_{\Lambda^{\cal V}} \dd \f y \, \pBB{\prod_{i=1}^{m} w(y_{i,1}-y_{i,2})} \xi(\f y_1) \prod_{e \in \cal E} \cal J_{\tau, e}(\f y_e, \f s)\,.
\end{equation}
\end{lemma}

\begin{proof}
We define a mapping $L: {\Lambda^{\cal V}} \rightarrow {\Lambda^{\cal X}}$ by $L \f y \deq (y_{[\alpha]})_{\alpha \in \cal X}$.
We note that $L$ is a bijection mapping ${\Lambda^{\cal V}}$ onto 
\begin{equation*}
\big\{\f x \in {\Lambda^{\cal X}}\col x_{i,r,+1}=x_{i,r,-1} \mbox{ for all }1 \leq i \leq m,\,1\leq r \leq m  \big\} \,.
\end{equation*}
Moreover, using Lemma \ref{Correlation functions} and Definition \ref{J_e definition}, it follows that for $e = \{\va, \vb\}$ with $\va < \vb$
\begin{equation*}
\cal J_{\tau, e}(\f y_e, \f s) \;=\; \rho_{\tau,0} \big(\cal B_{\alpha} (L \f y, \f t) \, \cal B_{\beta} (L \f y, \f t)\big)
\end{equation*}
where $(\alpha,\beta) \in \Pi$ is chosen such that $[\alpha]=\va,\,[\beta]=\vb.$
The claim now follows by using the change of variables $\f x = L \f y$ on the right-hand side of \eqref{I_rep}, and recalling \eqref{def_calI}. 
\end{proof}

\begin{corollary} \label{Wick_application 2}
With $\f s$ defined in \eqref{def_ys}, we have
\begin{equation}
\label{Wick_application 2 bound}
\big|\cal I_{\tau,\Pi}^{\xi}(\f t)\big| \;\leq\; \norm{w}_{L^\infty}^m \int_{\Lambda^{\cal V}} \dd \f y \, \abs{\xi(\f y_1)} \prod_{e \in \cal E} \cal J_{\tau, e}(\f y_e,\f s)\,. 
\end{equation}
\end{corollary}
\begin{proof}
This follows from Lemmas \ref{lem:I_rep} and \ref{Correlation functions} (iii).
\end{proof}

\subsection{The explicit terms III: upper bound} \label{sec:upper_bound}
Recall that for the Hilbert-Schmidt norm $\norm{\cdot}_{\fra S^2(\fra H)}$ we have the basic identity
\begin{equation} \label{S2L2}
\norm{\cal A}_{\fra S^2(\fra H)} \;=\; \pbb{\int_\Lambda \dd x \, \dd y \, \abs{\cal A(x;y)}^2}^{1/2}\,.
\end{equation}
Moreover, by spectral decomposition of $\abs{\cal A} = \sqrt{\cal A^* \cal A} $, we easily find
\begin{equation}
\label{Schatten class embedding}
\|\cal A\|_{\fra S^{q_2}(\fra H)} \;\leq\; \|\cal A\|_{\fra S^{q_1}(\fra H)}
\end{equation}
for all $1 \leq q_1 \leq q_2  \leq \infty$. Both of these facts extend trivially to the $p$-particle case where $\fra H$ is replaced with $\fra H^{(p)}$.

Note also that for any $t > -1$ we have $G_{\tau,t} \in \fra S^q(\fra H)$ for all $q \in [1,\infty)$, as can be easily seen by spectral decomposition and the fact that $h_\tau^{-1} \in \fra S^2(\fra H)$. It follows that $G_{\tau,t}^q \in \fra S^1(\fra H)$ for all $1 \leq q <\infty$.

In this subsection we fix $m,p \in \N$ and a pairing $\Pi \in \fra P$, and let $(\cal V, \cal E, \sigma)$ denote the associated graph from Definition \ref{def_collapsed_graph}. The following two results are the key estimates that allow us to integrate out the variables of a single path $\cal P \in \conn(\cal E)$. The former bounds, \eqref{Closed path bound} and \eqref{Open path bound}, are used to derive an upper bound on $\big| \cal I_{\tau,\Pi}^{\xi} (\f t) \big|$; see Proposition \ref{Product of subgraphs} below. The latter bounds, \eqref{Closed path remark} and \eqref{Open path remark},  are used to prove the convergence of $\cal I_{\tau,\Pi}^{\xi} (\f t)$ for fixed $\f t \in \fra A$ in Section \ref{sec:convergence}.

\begin{lemma}
\label{Closed path}
Suppose that $\cal P \in \conn(\cal E)$ is a closed path. Then
\begin{equation}
\label{Closed path bound}
\int_{\Lambda^{\cal V(\cal P)}} \prod_{\va \in \cal V(\cal P)} \dd y_{\va} \, \prod_{e \in \cal P} \cal J_{\tau, e}(\f y_e,\f s) \;\leq\;C^{\abs{\cal V(\cal P)}}\,,
\end{equation}
and
\begin{equation}
\label{Closed path remark}
\Bigg|\int_{\Lambda^{\cal V(\cal P)}} \prod_{\va \in \cal V(\cal P)} \dd y_{\va} \Bigg(\prod_{e \in \cal P} \cal J_{\tau, e}(\f y_e,\f s)-\prod_{e \in \cal P} \hat{\cal J}_{\tau, e}(\f y_e,\f s)\Bigg)\Bigg|  \rightarrow 0 \quad \mbox{as} \quad \tau \rightarrow \infty\,.
\end{equation}
\end{lemma}

\begin{lemma}
\label{Open path}
Suppose that $\cal P \in \conn(\cal E)$ is an open path with end points $\vb_1, \vb_2 \in \cal V_1(\cal P)$.
Then
\begin{equation}
\label{Open path bound}
\Bigg\|\int_{\Lambda^{\cal V_2(\cal P)}} \prod_{\va \in \cal V_2(\cal P)} \dd y_{\va} \, \prod_{e \in \cal P} \cal J_{\tau, e} (\f y_e,\f s)\Bigg\|_{L^2_{y_{b_1}y_{b_2}}} \;\leq\; C^{\abs{\cal V_2(\cal P)}}\,,
\end{equation}
and
\begin{equation}
\label{Open path remark}
\Bigg\| \int_{\Lambda^{\cal V_2(\cal P)}} \prod_{\va \in \cal V_2(\cal P)} \dd y_{\va} \Bigg(  \prod_{e \in \cal P} \cal J_{\tau, e} (\f y_e,\f s)-  \prod_{e \in \cal P} \hat{\cal J}_{\tau, e} (\f y_e,\f s) \Bigg) \Bigg\|_{L^2_{y_{b_1}y_{b_2}}} \rightarrow 0 \quad \mbox{as} \quad \tau \rightarrow \infty\,.
\end{equation}
\end{lemma}

\begin{proof}[Proof of Lemma \ref{Closed path}]
We first prove \eqref{Closed path bound}.
Denote by $q$ the length of the path $\cal P$, and use the notation $\cal P = \{e_1, e_2, \dots, e_q\}$ for the edges of $\cal P$, whereby $e_j$ and $e_{j+1}$ are incident for $j = 1, \dots, q$. Here, and throughout the proof, the index $j$ is always understood to be modulo $q$. Denote by $\va_j$ the unique vertex in $e_{j - 1} \cap e_j$. Without loss of generality, we suppose that $\va_1 < \va_2$. Note that we always have $q \geq 2$.

Next, by construction of $\cal E$ in Definition \ref{def_collapsed_graph}, the colour $\sigma(e_j)$ of any edge $e_j \in \cal P$ is determined by $\sigma(e_1)$. Indeed,
\begin{equation} \label{sigma_determined}
\sigma(e_j) \;=\;
\begin{cases}
\sigma(e_1) &\mbox{if }\va_j<\va_{j+1}
\\
-\sigma(e_1) &\mbox{if }\va_j>\va_{j+1}\,.
\end{cases}
\end{equation}
This follows immediately from Definition \ref{def_collapsed_graph}, and is best understood using Figure \ref{fig:graph3}.

\begin{figure}[!ht]
\begin{center}
{\small 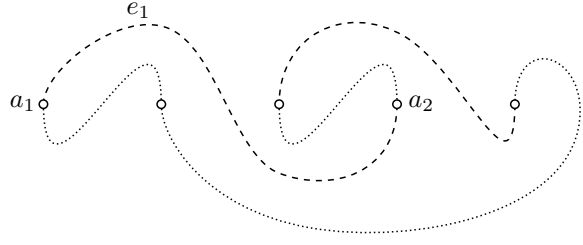}
\end{center}
\caption{A closed path $\cal P$ of length $5$, with $q = 5$. The vertices of $\cal V_2$ are ordered from left to right and depicted using white dots. Edges with colour $+1$ are depicted with dotted lines and edges with colour $-1$ with dashed lines. The edge $e_1$ is indicated, and has colour $\sigma(e_1) = -1$. Its colour determines the colours of all other edges.
\label{fig:graph3}}
\end{figure}

Next, for $j=1,\ldots,q$, we define the vertices
$\va_{j,-} \deq \min\{\va_j, \va_{j+1}\}$ and $\va_{j,+} \deq \max\{\va_j, \va_{j+1}\}$.
From \eqref{sigma_determined} we get, for $j = 1, \dots, q$,
\begin{equation}
\label{e_j label}
\sigma(e_j) (s_{\va_{j,-}}-s_{\va_{j,+}}) \;=\;\sigma(e_1) (s_{\va_j}-s_{\va_{j+1}})\,.
\end{equation}
Substituting \eqref{e_j label} into \eqref{J_e}, we get, for $j=1,\ldots,q$,
\begin{equation}
\label{J tau e_j}
\cal J_{\tau, e_j}(\f y_e,\f s) \;=\;G_{\tau, \sigma(e_1) (s_{{\va}_j} - s_{{\va}_{j+1}})}(y_{{\va}_j}; y_{{\va}_{j+1}}) + \frac{\ind{\sigma(e_j) = +1}}{\tau} S_{\tau, \sigma(e_1)(s_{{\va}_j} - s_{{\va}_{j+1}})}(y_{{\va}_j}; y_{{\va}_{j+1}})\,. 
\end{equation}
Here we used that the kernels of $G_{\tau,t}$ and $S_{\tau,t}$ are symmetric by Lemma \ref{Positivity lemma}. Note that the time arguments in \eqref{J tau e_j} are always in the corresponding domains from \eqref{def_St} and \eqref{def_Gt}, as may be easily seen using \eqref{s_ordered} and \eqref{e_j label}.
More precisely, if $\sigma(e_1)=1$, then $\sigma(e_j)=1$ implies that $\va_j <\va_{j+1}$ and hence $s_{{\va}_j} - s_{{\va}_{j+1}}\geq0$. Likewise, if $\sigma(e_1)=-1$, then $\va_j >\va_{j+1}$ and hence $s_{{\va}_j} - s_{{\va}_{j+1}}\leq0$.

We can hence rewrite the expression on the left-hand side of \eqref{Closed path bound}
as
\begin{equation}
\label{Closed path bound 2}
\tr \,\Bigg[\prod_{j=1}^{q} \bigg(G_{\tau, \sigma(e_1) (s_{{\va}_j} - s_{{\va}_{j+1}})} + \frac{\ind{\sigma(e_j) = +1}}{\tau} S_{\tau, \sigma(e_1)(s_{{\va}_j} - s_{{\va}_{j+1}})}\bigg)\Bigg]\,.
\end{equation}
Note that all operators in the above product commute, so that the order of the product is immaterial.
We let 
\begin{equation*}
J_{\cal P} \;\deq\; \{1 \leq j \leq q\col \sigma(e_j)=+1\} \,.
\end{equation*} 
By construction, $J_{\cal P} \neq \{1,\ldots,q\}$.  Indeed, by Definition \ref{def_collapsed_graph}, the smallest vertex in the path $\cal P$ is incident to an edge of colour $+1$ and to an edge of colour $-1$.

We write \eqref{Closed path bound 2} as
\begin{multline}
\sum_{I \subset J_{\cal P}} \tr \Bigg[\bigg(\prod_{j \in \{1,\ldots,q\} \setminus I} G_{\tau, \sigma(e_1) (s_{{\va}_j} - s_{{\va}_{j+1}})}\bigg) \bigg(\prod_{j \in I} \frac{1}{\tau} S_{\tau, \sigma(e_1)(s_{{\va}_j} - s_{{\va}_{j+1}})}\bigg)\Bigg]
\\
\;=\; \sum_{I \subset J_{\cal P}} \tr \Bigg[\bigg(\prod_{j \in \{1,\ldots,q\} \setminus I} G_{\tau,0}\bigg) \bigg(\prod_{j \in I} \frac{1}{\tau} S_{\tau,0}\bigg)\Bigg] \;=\; \sum_{I \subset J_{\cal P}}  \frac{1}{\tau^{|I|}} \tr \big(G_{\tau}^{q-|I|}\big)\,.
\label{Closed path bound 3}
\end{multline}
In the first equality above, we used $\sum_{j=1}^{q+1} (s_{\va_{j}}-s_{\va_{j-1}}) =s_{\va_{q+1}}-s_{\va_0}=0$, which holds since $s_{\va_{q+1}}=s_{\va_0}=0$.

By \eqref{Schatten class embedding}, we note that if $|I| \leq q- 2$ then $\tr \big(G_{\tau}^{q-|I|}\big) = \norm{G_\tau}_{\fra S^{q - \abs{I}}}^{q - \abs{I}} \leq \|G_{\tau}\|_{\fra S^2(\fra H)}^{q-|I|}$.
Consequently, \eqref{Closed path bound 2} is 
\begin{equation}
\label{Closed path bound 4}
\;\leq\; \mathop{\sum_{I \subset J_{\cal P} \col}}_{ |I| \leq q-2} \frac{1}{\tau^{|I|}} \|G_{\tau}\|_{\fra S^2(\fra H)}^{q-|I|} + \frac{q}{\tau^{q-1}} \|G_{\tau}\|_{\fra S^1(\fra H)} \;\leq\; C^{\abs{\cal V_2(\cal P)}} \, \pbb{1 +  \|G_{\tau}\|_{\fra S^2(\fra H)} + \frac{1}{\tau} \norm{G_\tau}_{\fra S^1(\fra H)}}^{\abs{\cal V(\cal P)}}\,.
\end{equation}
In the last line we used $q \geq 2$. Estimate \eqref{Closed path bound} now follows from \eqref{Closed path bound 4} by Lemma \ref{lem:conv_G_tau}.

Finally, \eqref{Closed path remark} follows by a similar argument. When estimating the right-hand side of \eqref{Closed path remark}, the only difference is that we sum over \emph{nonempty} sets $I$ in \eqref{Closed path bound 3} and \eqref{Closed path bound 4}. Recalling \eqref{Closed path bound 2}, Definition \ref{J_e definition}, and using $q \geq 2$ we thus obtain an additional factor of $\frac{1}{\tau}$ and one power of $(1 +  \|G_{\tau}\|_{\fra S^2(\fra H)} + \frac{1}{\tau} \norm{G_\tau}_{\fra S^1(\fra H)})$ less than on the right-hand side of \eqref{Closed path bound 4}, and we deduce \eqref{Closed path remark} again by using Lemma \ref{lem:conv_G_tau}.
\end{proof}

\begin{proof}[Proof of Lemma \ref{Open path}] 
We argue similarly as in the proof of Lemma \ref{Closed path}. We first prove \eqref{Open path bound}.  Let $\vb_1$ and $\vb_2$ be as in the statement of the lemma. Assume, without loss of generality, that $\vb_1<\vb_2$. By construction of $\cal E$ in Definition \ref{def_collapsed_graph}, it follows that $\sign_{\vb_1}=+1$ and  
$\sign_{\vb_2}=-1$.

Let $q \deq |\cal V_2(\cal P)|$. In the case $q = 0$ we have $\cal J_{\tau,\{\vb_1, \vb_2\}}((y_{b_1}, y_{b_2}), \f s) = G_\tau(y_{\vb_1};y_{\vb_2})$ because $\sigma(\{\vb_1, \vb_2\}) = -1$, and the claim follows immediately by Lemma \ref{lem:conv_G_tau}.

For the following we assume $q \geq 1$. We write $\cal P = \{e_1, e_2, \dots, e_q,e_{q+1}\}$ for the edges of $\cal P$, whereby $e_j$ and $e_{j+1}$ are incident for $j = 1, \dots, q$ with the further properties that $\vb_1 \in e_1, \vb_2 \in e_{q+1}$, and $\va_j \deq e_{j} \cap e_{j+1} \in \cal V_2(\cal P)$, for all $j=1, \ldots, q$. 
Furthermore, we add the conventions that $\va_{0} \deq \vb_1$ and $\va_{q+1} \deq \vb_2$.

In this case, by Definition \ref{def_collapsed_graph}, the colour $\sigma(e_j)$ of any edge $e_j \in \cal P$ is uniquely determined. Namely,
\begin{equation} \label{sigma_determined 2}
\sigma(e_j) \;=\;
\begin{cases}
-1 &\mbox{if }\va_{j-1}<\va_j
\\
+1 &\mbox{if }\va_{j-1}>\va_j\,;
\end{cases}
\end{equation}
see Figure \ref{fig:graph4}.

\begin{figure}[!ht]
\begin{center}
{\small 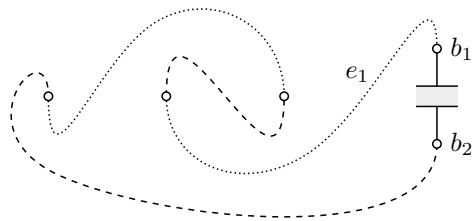}
\end{center}
\caption{An open path $\cal P$ of length $4$, with $q = 3$. The vertices of $\cal V_2$ are ordered from left to right and we draw the two vertices of $\cal V_1$ contained in $\cal P$. Edges with colour $+1$ are depicted with dotted lines and edges with colour $-1$ with dashed lines. The edge $e_1$ is indicated, and has colour $\sigma(e_1) = +1$.
\label{fig:graph4}}
\end{figure}

Taking, for $j = 1, \dots, q+1$, $\va_{j,-} \deq \min\{\va_{j-1}, \va_{j}\}$ and $\va_{j,+} \deq \max\{\va_{j-1}, \va_{j}\}$, \eqref{sigma_determined 2} implies that, for $j = 1, \dots, q+1$,
\begin{equation}
\label{e_j label 2}
\sigma(e_j) (s_{\va_{j,-}}-s_{\va_{j,+}}) \;=\; s_{\va_{j}}-s_{\va_{j-1}}\,.
\end{equation}
In particular, we substitute \eqref{e_j label 2} into \eqref{J_e} and we get, for $j=1,\ldots,q+1$,
\begin{equation}
\label{J tau e_j 2}
\cal J_{\tau, e_j}(\f y_{e_j},\f s) \;=\;G_{\tau, s_{\va_{j}}-s_{\va_{j-1}}}(y_{{\va}_{j-1}}; y_{{\va}_{j}}) + \frac{\ind{\sigma(e_j) = +1}}{\tau} S_{\tau, s_{\va_{j}}-s_{\va_{j-1}}}(y_{{\va}_{j-1}}; y_{{\va}_{j}})\,. 
\end{equation}
As for \eqref{J tau e_j}, we used the symmetry of the kernels of $G_{\tau,t}$ and $S_{\tau,t}$ from Lemma \ref{Positivity lemma}.
It immediately follows that the time arguments in \eqref{J tau e_j 2} are always in the corresponding domains from \eqref{def_St} and \eqref{def_Gt}.
From \eqref{J tau e_j 2} we deduce
\begin{equation}
\label{Open path bound 2}
\int_{\Lambda^{\cal V_2(\cal P)}} \prod_{\va \in \cal V_2(\cal P)} \dd y_{\va} \, \prod_{e \in \cal P} \cal J_{\tau, e} (\f y_e,\f s) \;=\;
\Bigg[\prod_{j=1}^{q+1} \bigg(G_{\tau,s_{\va_{j}}-s_{\va_{j-1}}} + \frac{\ind{\sigma(e_j) = +1}}{\tau} S_{\tau, s_{\va_{j}}-s_{\va_{j-1}}}\bigg) \Bigg](y_{\vb_1};y_{\vb_2})\,.
\end{equation}
We now define $J_{\cal P} \deq \{1 \leq j \leq q+1\col\sigma(e_j)=+1\}$.
We note that $\sigma(e_1)=+1$ and hence $1 \in J_{\cal P}$. Furthermore, $\sigma(e_{q+1})=-1$ and hence $q+1 \notin J_{\cal P}$. Therefore $1 \leq |J_{\cal P}|  \leq q$.

We write \eqref{Open path bound 2} as
\begin{multline}
\sum_{I \subset J_{\cal P}} \Bigg[\bigg(\prod_{j \in \{1,\ldots,q+1\} \setminus I} G_{\tau, s_{\va_{j}}-s_{\va_{j-1}}}\bigg) \bigg(\prod_{j \in I} \frac{1}{\tau} S_{\tau, s_{\va_{j}}-s_{\va_{j-1}}}\bigg)\Bigg](y_{\vb_1};y_{\vb_2})
\\
\;=\; \sum_{I \subset J_{\cal P}} \Bigg[\bigg(\prod_{j \in \{1,\ldots,q+1\} \setminus I} G_{\tau,0}\bigg) \bigg(\prod_{j \in I} \frac{1}{\tau} S_{\tau,0}\bigg)\Bigg] (y_{\vb_1};y_{\vb_2})\;=\; \sum_{I \subset J_{\cal P}}  \frac{1}{\tau^{|I|}} \big(G_{\tau}^{q+1-|I|}\big)(y_{\vb_1};y_{\vb_2})\,.
\label{Open path bound 3}
\end{multline}
In the first equality above, we used $\sum_{j=1}^{q+1}(s_{\va_{j+1}}-s_{\va_j})=s_{\va_{q+1}}-s_{\va_0}=0$, which holds since $s_{\va_{q+1}}=s_{\va_0}=0$.

Consequently,
\eqref{Open path bound 2} is
\begin{multline}
\label{Open path bound 4}
\;\leq\; \mathop{\sum_{I \subset J_{\cal P} \col}}_{|I| \leq q-1} \frac{1}{\tau^{|I|}} \|G_{\tau}\|_{\fra S^2(\fra H)}^{q-1-|I|}  \, \|G_{\tau}(y_{\vb_1};\cdot)\|_{L^2}\, \,\|G_{\tau}(y_{\vb_2};\cdot)\|_{L^2}+ \frac{q+1}{\tau^{q}}\,G_{\tau}(y_{\vb_1};y_{\vb_2})
\\
\;\leq\; C^{\abs{\cal V_2(\cal P)}}\,\pb{1 + \|G_{\tau}\|_{\fra S^2(\fra H)}}^{\abs{\cal V_2(\cal P)}} \, \Big(\|G_{\tau}(y_{\vb_1};\cdot)\|_{L^2}\, \,\|G_{\tau}(\cdot \, ; y_{\vb_2})\|_{L^2}+ G_{\tau}(y_{\vb_1};y_{\vb_2})\Big)
\,.
\end{multline}
In the first inequality we used that, for all $l \geq 2$,
\begin{multline}
\label{composition}
G_{\tau}^l (y_{\vb_1};y_{\vb_2})= \int_{\Lambda^{l-1}} \prod_{j=1}^{l-1}\dd \zeta_j \, G_{\tau}(y_{\vb_1};\zeta_1) G_{\tau}(\zeta_1;\zeta_2) \cdots G_{\tau}(\zeta_{l-2};\zeta_{l - 1}) \,G_{\tau}(\zeta_{l - 1};y_{\vb_2})
\\
\leq\; \|G_{\tau}\|_{\fra S^2(\fra H)}^{l-2}  \, \|G_{\tau}(y_{\vb_1};\cdot)\|_{L^2}\, \,\|G_{\tau}(y_{\vb_2};\cdot)\|_{L^2}\,,
\end{multline}
by the Cauchy-Schwarz inequality. Now \eqref{Open path bound} follows from \eqref{Open path bound 4} and Lemma \ref{lem:conv_G_tau}.

The proof of \eqref{Open path remark} is similar. The only difference is that we sum over \emph{nonempty} sets $I$ in \eqref{Open path bound 3} and \eqref{Open path bound 4}. We thus obtain an additional factor of $\frac{1}{\tau}$ and one power of $\pb{1 + \|G_{\tau}\|_{\fra S^2(\fra H)}}$ less than on the right-hand side of \eqref{Open path bound 4}. We deduce \eqref{Open path remark} from Lemma \ref{lem:conv_G_tau}.
\end{proof}
\begin{proposition} \label{Product of subgraphs}
For any $\Pi \in \fra P$ and $\f t \in \fra A$ we have $\big| \cal I_{\tau,\Pi}^{\xi} (\f t) \big| \leq C^{m+p} \, \|w\|_{L^{\infty}}^m$.
\end{proposition}

\begin{proof}
We rewrite \eqref{Wick_application 2 bound} as
\begin{equation}
\big|\cal I_{\tau,\Pi}^{\xi}(\f t)\big| \;\leq\; \norm{w}_{L^\infty}^m \int_{\Lambda^{\cal V_1}} \dd \f y_1 \, \absb{\xi(\f y_1)}  \,\int_{\Lambda^{\cal V_2}} \dd \f y_2
\prod_{e \in \cal E} \cal J_{\tau, e}(\f y_e,\f s)\,.
\label{Wick_application 2 bound B}
\end{equation}
Introduce the splitting $\conn(\cal E) = \conn_c(\cal E) \sqcup \conn_o(\cal E)$ into closed paths $\conn_c(\cal E)$ and open paths $\conn_o(\cal E)$. Then we have the factorization
\begin{multline}
\,\int_{\Lambda^{\cal V_2}} \dd \f y_2
\prod_{e \in \cal E} \cal J_{\tau, e}(\f y_e,\f s)\;=\; 
\\
\prod_{\cal P \in \conn_c(\cal E)} \bigg( \int_{\Lambda^{\cal V(\cal P)}} \prod_{\va \in \cal V(\cal P)} \dd y_{\va} \, \prod_{e \in \cal P} \cal J_{\tau, e}(\f y_e,\f s) \bigg) \prod_{\cal P \in \conn_o(\cal E)} \bigg(\int_{\Lambda^{\cal V_2(\cal P)}} \prod_{\va \in \cal V_2(\cal P)} \dd y_{\va} \, \prod_{e \in \cal P} \cal J_{\tau, e} (\f y_e,\f s)\bigg)\,.
\label{Wick_application 2 bound C}
\end{multline}
Substituting \eqref{Wick_application 2 bound C} into \eqref{Wick_application 2 bound B} and using Corollary \ref{Wick_application 2} and Lemmas \ref{Closed path}, it follows that
\begin{equation*}
\big| \cal I_{\tau,\Pi}^{\xi} (\f t) \big| \;\leq\;C^{m} \, \|w\|_{L^{\infty}}^m
\int_{\Lambda^{\cal V_1}}  \dd \f y_1\, \absb{\xi (\f y_1)} \,
\prod_{\cal P \in \conn_o(\cal E)} \bigg(\int_{\Lambda^{\cal V_2(\cal P)}} \prod_{\va \in \cal V_2(\cal P)} \dd y_{\va} \, \prod_{e \in \cal P} \cal J_{\tau, e} (\f y_e,\f s)\bigg)\,,
\end{equation*}
The claim now follows by applying Lemma \ref{Open path} and the Cauchy-Schwarz inequality in $y_{\vb}$ for all $\vb \in \cal V_1$, where we used that $\abs{\conn_o(\cal E)} = p$.
\end{proof}

\begin{proposition} \label{prop:est_f}
For $f_{\tau,m}^{\xi}(\f t)$ defined in \eqref{def_a_tau} we have $\big|f_{\tau,m}^{\xi}\big(\f t)| \leq (C p)^p \, \pb{C m^2  \|w\|_{L^{\infty}}}^m$.
\end{proposition}

\begin{proof}
We use Lemma \ref{Wick_application lemma}, Proposition \ref{Product of subgraphs} and $|\fra P| \leq (2m+p)!$ to deduce
\begin{equation*}
\big|f_{\tau,m}^{\xi}\big(\f t)| \;\leq\; C^{m + p} \|w\|_{L^{\infty}}^m  \,(2m+p)!\,, 
\end{equation*}
from which the claim follows.
\end{proof}

\begin{corollary} \label{cor:bound_am}
For $a_{\tau,m}^{\xi}$ defined in \eqref{Explicit term a} we have $\big|a_{\tau,m}^{\xi}\big| \leq (C p)^p \, C^m\, m! \,\|w\|_{L^{\infty}}^m$.
\end{corollary}

\begin{proof}
This follows immediately from Proposition \ref{prop:est_f} and \eqref{def_a_tau}, noting that the time integration on the right-hand side of \eqref{def_a_tau} has a volume $(1 - 2 \eta)^m / m!$.
\end{proof}

\subsection{The explicit terms IV: convergence}
\label{sec:convergence}

The following definition introduces the limit $\tau \to \infty$ of $\cal J_{\tau,e}(\f y_e, \f s)$ from Definition \ref{J_e definition}.
\begin{definition}
\label{J_e definition classical}
Let $\f y = (y_{\va})_{\va \in \cal V} \in \Lambda^{\cal V}$. We associate an integral kernel $\cal J_{e}(\f y_e) \deq G(y_{\va}; y_{\vb})$ with each edge $e = \{\va, \vb\} \in \cal E$ with $a < b$.
\end{definition}
We note that, in the above definition, the quantity $\cal J_{e}(\f y_e)$ does not depend on time. Moreover, $\cal J_{e}(\f y_e) \geq 0$ by the following lemma.

\begin{lemma}
\label{Positivity lemma 2}
We have $G(x;y) = G(y;x)\geq 0$ for all $x,y \in \Lambda$.
\end{lemma}
\begin{proof}
The claim immediately follows from the Feynman-Kac formula applied to $G = h^{-1} = \int_0^\infty \dd t \, \ee^{- ht}$, analogously to the proof of Lemma \ref{Positivity lemma}.
\end{proof}

The following definition introduces the limit $\tau \to \infty$ of $\cal I_{\tau,\Pi}^\xi(\f t)$ from Definition \ref{def:I_tau_PI}; see also \eqref{I_rep}.

\begin{definition}
\label{def_calI infinity} 
For $\Pi \in \fra P$ and $\cal E = \cal E_\Pi$ as in Definition \ref{def_collapsed_graph}, we define 
\begin{equation*}
\cal I^\xi_{\Pi} \;\deq\; \int_{\Lambda^{\cal V}} \dd \f y \, \pBB{\prod_{i=1}^{m} w(y_{i,1}-y_{i,2})} \xi(\f y_1) \prod_{e \in \cal E} \cal J_{e}(\f y_e) \,.
\end{equation*} 
\end{definition}

Recall the definition \eqref{def_fraA} of $\fra A$.

\begin{lemma}
\label{Convergence of explicit terms lemma}
For all fixed $\f t \in \fra A$ and $\Pi \in \fra P$ we have
\begin{equation*}
\cal I^\xi_{\tau,\Pi}(\f t) \rightarrow \cal I^\xi_{\Pi}\quad \mbox{as}\quad \tau \rightarrow \infty\quad \mbox{uniformly in} \quad \xi \in \fra B_p\,.
\end{equation*}
(Note that we do not prove that the convergence is uniform in $\f t \in \fra A$.)
\end{lemma}

\begin{proof}

In analogy with \eqref{I_rep} we define, for any $\Pi \in \fra P$ and $\f t \in \fra A$
\begin{equation}
\label{hat I_rep}
\hat{\cal I}_{\tau,\Pi}^{\xi}(\f t) \;\deq\; \int_{\Lambda^{\cal V}} \dd \f y \, \pBB{\prod_{i=1}^{m} w(y_{i,1}-y_{i,2})} \xi(\f y_1) \prod_{e \in \cal E} \hat{\cal J}_{\tau, e}(\f y_e, \f s)\,.
\end{equation}
Here we recall $\hat{\cal J}_{\tau, e}$ defined by \eqref{hat J_e} in Definition \ref{J_e definition}.
We observe that
\begin{equation}
\label{hat I convergence}
\cal I_{\tau,\Pi}^{\xi} (\f t)-\hat{\cal I}_{\tau,\Pi}^{\xi} (\f t) \rightarrow 0 \quad \mbox{as} \quad \tau \rightarrow \infty \quad \mbox{uniformly in} \quad \xi \in \fra B_p\,.
\end{equation}
In order to prove \eqref{hat I convergence} we first decompose $\cal I_{\tau,\Pi}^{\xi} (\f t)$ and $\hat{\cal I}_{\tau,\Pi}^{\xi} (\f t)$ as a product over paths as in \eqref{Wick_application 2 bound C}. We then use a telescoping argument for the factor corresponding to each path. We conclude by using \eqref{Closed path remark} from Lemma \ref{Closed path} and \eqref{Open path remark} from Lemma \ref{Open path} and by arguing as in the proof of Proposition \ref{Product of subgraphs}. This yields \eqref{hat I convergence}.

By \eqref{hat I convergence}, it suffices to show that
\begin{equation*}
\hat{\cal I}^\xi_{\tau,\Pi}(\f t) \rightarrow \cal I^\xi_{\Pi}\quad \mbox{as}\quad \tau \rightarrow \infty\quad \mbox{uniformly in} \quad \xi \in \fra B_p\,.
\end{equation*}
By \eqref{hat I_rep} and Definition \ref{def_calI infinity}, it follows that
\begin{equation}
\label{difference of I_Pi}
\hat{\cal I}^\xi_{\tau,\Pi}(\f t)-\cal I^\xi_{\Pi} \;=\;
\int_{\Lambda^{\cal V}} \dd \f y \, \pBB{\prod_{i=1}^{m} w(y_{i,1}-y_{i,2})} \xi(\f y_1) \Bigg[\prod_{e \in \cal E} \hat{\cal J}_{\tau,e}(\f y_e, \f s)-\prod_{e \in \cal E} \cal J_{e}(\f y_e)\Bigg] \,.
\end{equation}
By telescoping, we write
\begin{equation}
\label{telescoping}
\prod_{e \in \cal E} \hat{\cal J}_{\tau,e}(\f y_e, \f s)-\prod_{e \in \cal E} \cal J_{e}(\f y_e) \;=\; \sum_{e_0 \in \cal E}\,\Bigg[ \mathop{\prod_{e \in \cal E\col}}_{e<e_0} \hat{\cal J}_{\tau,e}(\f y_e, \f s)\,
\Big(\hat{\cal J}_{\tau,e_0}(\f y_{e_0},\f s)-\cal J_{e_0}(\f y_{e_0})\Big)\,
\mathop{\prod_{e \in \cal E\col}}_{e>e_0} \cal J_{e}(\f y_e)\Bigg]\,,
\end{equation}
where we order the elements of $\cal E$  in some arbitrary fashion.

Substituting \eqref{telescoping} into \eqref{difference of I_Pi}, it follows that
\begin{multline}
\label{difference of I_Pi absolute value}
\Big|\hat{\cal I}^\xi_{\tau,\Pi}(\f t)-\cal I^\xi_{\Pi}\Big| \leq \sum_{e_0 \in \cal E}\, \|w\|_{L^{\infty}}^m \,\int_{\Lambda^{\cal V}} \dd \f y \,\absb{\xi(\f y_1)}
\Bigg[\prod_{e <e_0} \hat{\cal J}_{\tau,e}(\f y_e, \f s)\,
\Big|\hat{\cal J}_{\tau,e_0}(\f y_{e_0},\f s)-\cal J_{e_0}(\f y_{e_0})\Big|\,
\prod_{e>e_0} \cal J_{e}(\f y_e)\Bigg]\,.
\end{multline}
Here, we used that $\hat{\cal J}_{\tau,e}(\f y_e, \f s), \cal J_{e}(\f y_e) \geq 0$ by Lemmas \ref{Positivity lemma} and \ref{Positivity lemma 2} respectively.
We denote by $\cal L_{\tau,e_0}^{\xi}(\f t)$ the summand in \eqref{difference of I_Pi absolute value} corresponding to $e_0 \in \cal E$. It suffices to show that
\begin{equation}
\label{L e_0 t}
\cal L_{\tau,e_0}^{\xi}(\f t) \rightarrow 0\quad \mbox{as}\quad \tau \rightarrow \infty\quad \mbox{uniformly in} \quad \xi \in \fra B_p\,.
\end{equation}
For the remainder of the proof, we fix $e_0 \in \cal E$ and we prove \eqref{L e_0 t}.

Let $\f y = (y_{\va})_{\va \in \cal V} \in \Lambda^{\cal V}$ and $\f s  = (s_{\va})_{\va \in \cal V} \in \R^{\cal V}$ satisfy \eqref{s_ordered}. We associate an integral kernel $\tilde{\cal J}_{\tau,e}(\f y_e,\f s)$ with each edge $e \in \cal E$ through
\begin{equation}
\label{tilde J tau e}
\tilde{\cal J}_{\tau,e}(\f y_e,\f s) \;\deq\;
\begin{cases}
\hat{\cal J}_{\tau,e}(\f y_e,\f s) &\mbox{if }e<e_0\\
\big|\hat{\cal J}_{\tau,e}(\f y_e,\f s)-\cal J_{e}(\f y_e)\big|&\mbox{if }e=e_0\\
\cal J_{e}(\f y_e) &\mbox{if }e>e_0\,. 
\end{cases}
\end{equation}
In particular
\begin{equation}
\label{L e_0 t identity}
\cal L_{\tau,e_0}^{\xi}(\f t)  \;=\;\|w\|_{L^{\infty}}^m \,\int_{\Lambda^{\cal V}} \dd \f y \,\absb{\xi(\f y_1)} \prod_{e \in \cal E} \tilde{\cal J}_{\tau,e}(\f y_e, \f s)\,.
\end{equation}
Before we proceed, we collect several useful estimates on the kernels $\tilde{\cal J}_{\tau,e}(\f y_e,\f s)$ defined in \eqref{tilde J tau e}. Here in each case, we write $e=\{\va,\vb\}$ with $\va<\vb$.
\begin{enumerate}
\item When $e<e_0$, we claim
\begin{equation}
\label{tilde J bound 1}
\|\tilde{\cal J}_{\tau,e}(\, \cdot \, ,\f s)\|_{\fra S^2(\fra H)} \;\leq\; C_{\f s}\,.
\end{equation}
We use Definition \ref{J_e definition} to write
\begin{equation}
\label{tilde J Case 1 formula}
\tilde{\cal J}_{\tau,e}(\f y_e,\f s)\;=\;G_{\tau, \sigma(e) (s_{\va} - s_{\vb})}(y_{\va}; y_{\vb})\,.
\end{equation}
We know from \eqref{s_ordered} that $s_{\va}-s_{\vb} \in [0,1)$. Hence \eqref{tilde J bound 1} follows from Lemma \ref{cor:conv_G_tau corollary}. We note that the constant $C_{\f s}$ depends on $\f s$.

\item When $e=e_0$ we claim
\begin{equation}
\label{tilde J bound 2}
\lim_{\tau \rightarrow \infty} \|\tilde{\cal J}_{\tau,e}(\, \cdot\, ,\f s)\|_{\fra S^2(\fra H)} \;=\;0\,.
\end{equation}
By Definitions \ref{J_e definition} and \ref{J_e definition classical}, it follows that 
\begin{equation}
\label{tilde J identity 2}
\tilde{\cal J}_{\tau,e}(\f y_e,\f s)\;=\; \big(G_{\tau, \sigma(e) (s_{\va} - s_{\vb})}-G\big)(y_{\va};y_{\vb})\,.
\end{equation}
We again use that $s_{\va}-s_{\vb} \in [0,1)$  and apply Lemma \ref{cor:conv_G_tau corollary}.

\item If $e>e_0$, then by Definition \ref{J_e definition classical} we have
\begin{equation}
\label{tilde J bound 3}
\|\tilde{\cal J}_{\tau,e}(\,\cdot\,,\f s)\|_{\fra S^2(\fra H)}\;=\;\|G\|_{\fra S^2(\fra H)}\;\leq\;C\,.
\end{equation}

\end{enumerate}
Arguing analogously to \eqref{Wick_application 2 bound C}, we rewrite \eqref{L e_0 t identity} as
\begin{multline}
\cal L_{\tau,e_0}^{\xi}(\f t) \;=\;
\norm{w}_{L^\infty}^m \int_{\Lambda^{\cal V_1}}  \dd \f y_1 \, \absb{\xi(\f y_1)}  
\prod_{\cal P \in \conn_c(\cal E)}\bigg( \int_{\Lambda^{\cal V(\cal P)}} \prod_{\va \in \cal V(\cal P)} \dd y_{\va} \, \prod_{e \in \cal P} \tilde{\cal J}_{\tau, e}(\f y_e,\f s) \bigg) 
\\
\times
\mathop{\prod_{\cal P \in \conn_o(\cal E)}}\bigg(\int_{\Lambda^{\cal V_2(\cal P)}} \prod_{\va \in \cal V_2(\cal P)} \dd y_{\va} \, \prod_{e \in \cal P} \tilde{\cal J}_{\tau, e} (\f y_e,\f s)\bigg)\,.
\label{L e_0 t identity 2}
\end{multline}
Let $\cal P \in \conn_c (\cal E)$. 
By the Cauchy-Schwarz inequality in $y_{\va}$ for $\va \in \cal V(\cal P)$ we obtain
\begin{equation}
\label{Closed path contribution to L e_0 t}
\int_{\Lambda^{\cal V(\cal P)}} \prod_{\va \in \cal V(\cal P)} \dd y_{\va} \, \prod_{e \in \cal P} \tilde{\cal J}_{\tau, e}(\f y_e,\f s) \;\leq\; \prod_{e \in \cal P} \|\tilde{\cal J}_{\tau,e}(\, \cdot \, ,\f s)\|_{\fra S^2(\fra H)}\,.
\end{equation}
Let $\cal P \in \conn_o (\cal E)$. 
We recall the notation from the proof of Lemma \ref{Open path}:
\begin{equation*}
 \cal P=\{e_1,\ldots,e_{q+1}\}\,,\qquad
\cal V_1(\cal P)=\{\vb_1,\vb_2\}\,, \qquad \cal V_2(\cal P)=\{\va_1,\ldots,\va_q\}\,.
\end{equation*}

By the Cauchy-Schwarz inequality in $y_{\va}$ for $\va \in \cal V_2(\cal P)$ we obtain
\begin{multline}
\label{Open path contribution to L e_0 t}
\bigg(\int_{\Lambda^{\cal V_2(\cal P)}} \prod_{\va \in \cal V_2(\cal P)} \dd y_{\va} \, \prod_{e \in \cal P} \tilde{\cal J}_{\tau, e} (\f y_e,\f s)\bigg)\;\leq\;
\\
\prod_{j=2}^{q} \|\tilde{\cal J}_{\tau,e_j}(\, \cdot\,,\f s)\|_{\fra S^2(\fra H)} \,\|\tilde{\cal J}_{\tau,e_1}((y_{b_1}, \cdot), \f s)\|_{L^2} \, \|\tilde{\cal J}_{\tau,e_1}((\cdot, y_{b_2}), \f s)\|_{L^2}\,.
\end{multline}
Substituting \eqref{Closed path contribution to L e_0 t} and \eqref{Open path contribution to L e_0 t} into \eqref{L e_0 t identity 2} and applying the Cauchy-Schwarz inequality in the $\f y_1$ variables 
we obtain
\begin{equation}
\label{L e_0 t product bound}
\cal L_{\tau,e_0}^{\xi}(\f t) \;\leq\;\norm{w}_{L^\infty}^m \,\norm{\xi}_{\fra S^2(\fra H)} \prod_{e \in \cal E} \|\tilde{\cal J}_{\tau,e}(\, \cdot\, ,\f s)\|_{\fra S^2(\fra H)} \,.
\end{equation}
Recalling \eqref{tilde J bound 1}, \eqref{tilde J bound 2}, and \eqref{tilde J bound 3}, the claim \eqref{L e_0 t} now follows from \eqref{L e_0 t product bound}.
\end{proof}

Next, for $m \in \N$ we define
\begin{equation} \label{f_m,a_m^xi classical}
a_{\infty,m}^\xi \;\deq\; \frac{(-1)^m}{m! \, 2^m} \sum_{\Pi \in \fra P} \cal I^\xi_{\Pi}\,.
\end{equation}

\begin{proposition}
\label{Convergence of the explicit terms}
For $a_{\tau,m}^{\xi}$ as defined in \eqref{Explicit term a}, we have
\begin{equation}
\label{Convergence of the explicit terms identity}
a_{\tau,m}^{\xi} \rightarrow a_{\infty,m}^\xi \quad \mbox{as}\quad \tau \rightarrow \infty\quad \mbox{uniformly in} \quad \xi \in \fra B_p\,.
\end{equation}
\end{proposition}

\begin{proof}
From Lemma \ref{Wick_application lemma}, Proposition \ref{Product of subgraphs}, and Lemma \ref{lem:conv_G_tau} it follows that $f_{\tau,m}^{\xi}(\f t)$ defined in \eqref{Definition of f_{tau,m}^{xi}} is bounded uniformly in $\tau>0$, $\xi \in \fra B_p$, and $\f t \in \fra A$. Furthermore, from Lemmas \ref{Wick_application lemma} and \ref{Convergence of explicit terms lemma}, for all $\f t \in \fra A$ we have
\begin{equation*}
f_{\tau,m}^{\xi}(\f t) \rightarrow \sum_{\Pi \in \fra P} \cal I^\xi_{\Pi} \quad \mbox{as}\quad \tau \rightarrow \infty\quad \mbox{uniformly in} \quad \xi \in \fra B_p\,.
\end{equation*}
The claim now follows by using \eqref{def_a_tau}, \eqref{f_m,a_m^xi classical}, and the dominated convergence theorem.
\end{proof}

\subsection{The remainder term} \label{sec:remainder}
In this subsection we estimate the remainder term from \eqref{Remainder term R}.

\begin{proposition} \label{prop:remainder}
For $R^\xi_{\tau,M}(z)$ defined as in \eqref{Remainder term R} and $\re z \geq 0$ we have
\begin{equation}
\abs{R^\xi_{\tau,M}(z)} \;\leq\; \pbb{\frac{Cp}{\eta^2}}^p \pbb{\frac{C \norm{w}_{L^\infty} \abs{z}}{\eta^2}}^{M} \, M!\,.
\end{equation}
\end{proposition}
The rest of this subsection is devoted to the proof of Proposition \ref{prop:remainder}. We begin by performing the change of variables $u_1=1-2\eta-t_1$ and $u_j=t_{j-1}-t_{j}$ for $2 \leq j \leq M$
in the definition \eqref{Remainder term R}, and abbreviate $\f u = (u_1, \dots, u_M)$ as well as $\abs{\f u} \deq \sum_{i = 1}^M u_i$. This yields
\begin{equation} \label{R_using_g}
R_{\tau,M}^{\xi}(z)\;=\;(-1)^M \frac{z^M}{(1-2\eta)^{M}} \int_{(0, 1 - 2 \eta)^M} \dd \f u\, \ind{\abs{\f u} < 1 - 2 \eta} \, g^\xi_{\tau,M}(z, \f u)\,,
\end{equation}
where we defined
\begin{multline}
g^\xi_{\tau,M}(z, \f u) \;\deq\; 
\tr \Big(\Theta_\tau(\xi)\,\ee^{-(\eta+u_1) H_{\tau,0}} \,W_\tau \,
\ee^{-u_2 H_{\tau,0}} \,W_{\tau} \, \ee^{-u_3 H_{\tau,0}} \cdots \\
\cdots \ee^{-u_M H_{\tau,0}} \, W_{\tau} \, \ee^{-(1 - 2 \eta - \abs{\f u}) (H_{\tau,0}+\frac{z}{1-2\eta}W_{\tau})} \ee^{-\eta H_{\tau,0}}
\Big) \Big/ \tr \big(\ee^{-H_{\tau,0}}\big) \,. 
\end{multline}
For the following we fix $\f u \in (0,1 - 2 \eta)^M$ satisfying $\abs{\f u} < 1 - 2 \eta$.

In order to estimate $g^\xi_{\tau,M}(z, \f u)$, we introduce the rescaled Schatten norm
\begin{equation*}
\norm{\cal A}_{\tilde {\fra S}^p(\cal F)} \;\deq\; \pbb{\frac{\tr \abs{\cal A}^p}{\tr (\ee^{- H_{\tau,0}})}}^{1/p} \;=\; \frac{\norm{\cal A}_{\fra S^p(\cal F)}}{\tr(\ee^{-H_{\tau,0}})^{1/p}}
\end{equation*}
for $p \in [1,\infty)$ and $\norm{\cal A}_{\tilde {\fra S}^\infty(\cal F)} \deq \norm{\cal A}_{{\fra S}^\infty(\cal F)}$ (the operator norm on $\cal F$). We note the trivial identity
\begin{equation} \label{estimated_by1}
\normb{\ee^{-t H_{\tau,0}}}_{\tilde {\fra S}^{1/t}(\cal F)} \;=\; 1
\end{equation}
for all $t > 0$.
The following result is an immediate consequence of H\"older's inequality for Schatten spaces; see \cite{Simon05}.
\begin{lemma}[H\"older's inequality]
\label{Holder's inequality in S^p}
Given $p_1,p_2 \in [1,\infty]$ and $\cal A_j \in {\fra S}^{p_j}(\cal F)$ we have
\begin{equation*}
\norm{\cal A_1 \cal A_2}_{\tilde {\fra S}^p(\cal F)} \;\leq\; \norm{\cal A_1}_{\tilde {\fra S}^{p_1}(\cal F)} \,\norm{\cal A_2}_{\tilde {\fra S}^{p_2}(\cal F)}\,,
\end{equation*}
where $\frac{1}{p}=\frac{1}{p_1}+\frac{1}{p_2}$.
\end{lemma}

By cyclicity of the trace and Lemma \ref{Holder's inequality in S^p}, we estimate
\begin{multline} \label{g_est}
\abs{g^\xi_{\tau,M}(z, \f u)} \;\leq\; \normB{\ee^{-\eta H_{\tau,0}} \Theta_\tau(\xi) \ee^{-\eta/2 \, H_{\tau,0}}}_{\tilde {\fra S}^{2/3\eta}(\cal F)} \, \normB{\ee^{-(1 - 2 \eta - \abs{\f u}) (H_{\tau,0}+\frac{z}{1-2\eta}W_{\tau})}}_{\tilde {\fra S}^{1/(1 - 2 \eta - \abs{\f u})}(\cal F)}
\\
\times \normB{\ee^{-(\eta/2+u_1) H_{\tau,0}} \,W_\tau \,
\ee^{-u_2 H_{\tau,0}} \,W_{\tau} \, \ee^{-u_3 H_{\tau,0}}
\cdots \ee^{-u_M H_{\tau,0}} \, W_{\tau}}_{\tilde {\fra S}^{1/(\eta/2 + \abs{\f u})}(\cal F)}\,.
\end{multline}
Here, and throughout the following estimates, it is crucial to work with norms inside of which the total time in the exponents add up \emph{precisely} to the inverse Schatten exponent. If this balance is broken even slightly, our estimates break down badly. This strong sensitivity to the total time in the exponents is already apparent in the simplest example $\norm{\ee^{-t H_{\tau,0}}}_{\tilde {\fra S}^{1}}$ if $t \neq 1$. Indeed, by a simple calculation analogous to the proof of Lemma \ref{Quantum Wick theorem} (i) one finds
\begin{equation*}
- \partial_t \log \tr (\ee^{-t H_{\tau,0}}) \big\vert_{t = 1} \;=\; \sum_{k \in \N} \psi(\lambda_{\tau,k} / \tau)\,, \qquad \psi(x)\;\deq\; \frac{x \ee^{-x}}{1 - \ee^{-x}}\,,
\end{equation*}
from which we deduce that $- \partial_t \log \tr (\ee^{-t H_{\tau,0}}) \vert_{t = 1} \sim \tau^{\alpha}$ with the typical behaviour $\lambda_{\tau,k} \sim k^{1/\alpha}$, where $\alpha < 2$. Hence, for $t$ close to $1$ we obtain $\norm{\ee^{-t H_{\tau,0}}}_{\tilde {\fra S}^{1}} \sim \exp(-\tau^{\alpha} (t - 1))$, which blows up rapidly for $t < 1$.

In the following three lemmas we estimate individually the three factors on the right-hand side of \eqref{g_est}.

\begin{lemma} \label{lem:gest1}
For $\f u \in (0,1 - 2 \eta)^M$ satisfying $\abs{\f u} < 1 - 2 \eta$ we have
\begin{equation} \label{g_3_est}
\normB{\ee^{-(\eta/2+u_1) H_{\tau,0}} \,W_\tau \,
\ee^{-u_2 H_{\tau,0}} \,W_{\tau} \, \ee^{-u_3 H_{\tau,0}}
\cdots \ee^{-u_M H_{\tau,0}} \, W_{\tau}}_{\tilde {\fra S}^{1/(\eta/2 + \abs{\f u})}(\cal F)} \;\leq\; \pb{C M^2 \eta^{-2} \norm{w}_{L^\infty}}^{M}\,.
\end{equation}
\end{lemma}
\begin{proof}
We would like to use the tools developed in Sections \ref{sec:time_evolved_G}--\ref{sec:upper_bound} to estimate the left-hand side of \eqref{g_3_est}. However, looking at the definition of the Schatten norms, this is only possible if the exponent $1/(\eta/2 + \abs{\f u})$ is an even integer. Since this is in general false, we need to split up the left-hand side of \eqref{g_3_est} using H\"older's inequality in such a way that every resulting piece has an exponent in $2 \N$. This splitting has to be with care to avoid overuse of H\"older's inequality, which would lead to bounds that are not affordable. Indeed, by the quantum Wick theorem we find that $\norm{\ee^{-t H_{\tau,0}} W_\tau}_{\tilde {\fra S}^{1/t}(\cal F)} \sim t^{-2}$, so that a liberal application of H\"older's inequality with very short times is a bad idea.

We begin by defining a slightly more general form of the left-hand side of \eqref{g_3_est}.
Let $0 < T \leq 1$, $1 \leq m \leq M$, and $\f s = (s_1, \dots, s_m) \in [0, T]^m$ be ordered as $s_1 \leq s_2 \leq \cdots \leq s_m$.
Define
\begin{equation*}
K_{T}(\f s) \;\deq\; \normB{\ee^{-(T - s_m) H_{\tau,0}} W_\tau \ee^{- (s_m - s_{m - 1}) H_{\tau,0}} W_\tau \cdots \ee^{-(s_2 - s_1) H_{\tau,0}} W_\tau \ee^{-s_1 H_{\tau,0}}}_{\tilde {\fra S}^{1/T}(\cal F)}\,.
\end{equation*}
Moreover, for $m = 0$ define $K_{T}(\f s) \deq \norm{\ee^{-T H_{\tau,0}}}_{\tilde {\fra S}^{1/T}(\cal F)} = 1$, where we used \eqref{estimated_by1}.

For the following we fix $1 \leq m \leq M$ as well as $T, \f s$ satisfying
\begin{equation}
T \;\geq\; \eta/2 \,, \qquad \f s \;\in\; [0, T - \eta/2]^m\,.
\end{equation}
We define the splitting time $\hat T \deq \max \hb{t \in (0,T) \col 1/t \in 2 \N}$.
We now split $T = \tilde T + \hat T$ and $m = \tilde m + \hat m$, where $\hat m \deq \sum_{i = 1}^m \ind{s_i \leq \hat T}$. We define the new families $\tilde {\f s} \in [0, \tilde T - \eta/2]^{\tilde m}$ and $\hat {\f s} \in [0,\hat T]^{\hat m}$ by setting $\hat s_i \deq s_i$ for $1 \leq i \leq \hat m$, and $\tilde s_i \deq s_{\hat m + i} - \hat T$ for $1 \leq i \leq \tilde m$.

We now apply H\"older's inequality (Lemma \ref{Holder's inequality in S^p}) in the definition of $K_{T}(\f s)$, which yields
\begin{equation} \label{K_T_split}
K_{T}(\f s) \;\leq\; K_{\tilde T} (\tilde {\f s}) K_{\hat T} (\hat {\f s})\,.
\end{equation}
See Figure \ref{fig:graph5} for an illustration of this splitting. Explicitly, we write
\begin{multline*}
K_{T}(\f s) \;=\; \normB{\ee^{-(\tilde T - \tilde s_{\tilde m}) H_{\tau,0}} W_\tau \ee^{- (\tilde s_{\tilde m} - \tilde s_{\tilde m - 1}) H_{\tau,0}} W_\tau \cdots \ee^{-(\tilde s_2 - \tilde s_1) H_{\tau,0}} W_\tau \ee^{-\tilde s_1 H_{\tau,0}} 
\\
\times 
\ee^{-(\hat T - \hat s_{\hat m}) H_{\tau,0}} W_\tau \ee^{- (\hat s_{\hat m} - \hat s_{\hat m - 1}) H_{\tau,0}} W_\tau \cdots \ee^{-(\hat s_2 - \hat s_1) H_{\tau,0}} W_\tau \ee^{-\hat s_1 H_{\tau,0}}}_{\tilde {\fra S}^{1/T}(\cal F)}
\end{multline*}
and split the norm using H\"older's inequality (Lemma \ref{Holder's inequality in S^p}) at the line break.

\begin{figure}[!ht]
\begin{center}
{\small 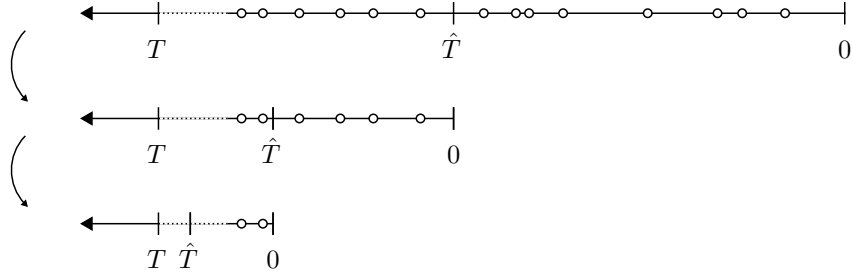}
\end{center}
\caption{An illustration of the splitting algorithm. Each expression $K_T(\f s)$ is represented graphically by a horizontal line. We draw the time axis oriented from right to left to match the ordering in the definition of $K_{T}(\f s)$. We draw a white dot, representing a factor $W_\tau$, at each $s_i$, $i = 1, \dots, m$. The interval $[0,T]$  is split into the two intervals $[0,\hat T]$ and $(\hat T, T]$, where $T = \hat T + \tilde T$. The interval $(\hat T, T]$ is the starting point of the next step of the algorithm, after a shift by $-\hat T$ which maps it to $(0,\tilde T]$. The term $K_{\tilde T}(\tilde {\f s})$ is represented graphically below its parent term $K_{T}(\f s)$, and we rename $\tilde T \to T$ and $\tilde {\f s} \to \f s$. At each step, there are $\hat m$ white dots in the first interval $[0,\hat T]$ and $\tilde m$ dots in the second interval $(\hat T, T]$. Note that, by construction, there are never dots in the interval $(T - \eta/2, T]$, which is indicated with a dotted line. The algorithm terminates when $\tilde m = 0$, i.e.\ there are no dots beyond $\hat T$.
\label{fig:graph5}}
\end{figure}

By definition of $\hat T$ we have $1/2\hat T \in \N$, so that
\begin{multline*}
K_{\hat T}(\hat {\f s})^{1/\hat T} \;=\; \frac{1}{\tr (\ee^{-H_{\tau,0}})} \tr \Biggl[ \biggl(
\pB{
\ee^{-\hat s_1 H_{\tau,0}} W_\tau \ee^{-(\hat s_2 - \hat s_1) H_{\tau,0}} \cdots W_\tau \ee^{-(\hat T - \hat s_{\hat m}) H_{\tau,0}} 
}
\\
\times \pB{\ee^{-(\hat T - \hat s_{\hat m}) H_{\tau,0}} W_\tau \cdots \ee^{-(\hat s_2 - \hat s_1) H_{\tau,0}} W_\tau \ee^{-\hat s_1 H_{\tau,0}}}\biggr)^{1/{2 \hat T}}\Biggr]\,.
\end{multline*}
As in Section \ref{sec:time_evolved_G}, using Corollary \ref{cor:Conjugation identity} we conclude that $K_{\hat T}(\hat {\f s})^{1/\hat T} = f^\emptyset_{\tau, \hat m / \hat T}(\hat {\f t})$
for some immaterial $\hat {\f t} \in [0,1)^{\hat m / \hat T}$, where $f$ was defined in \eqref{Definition of f_{tau,m}^{xi}}. Here $p = 0$, which we indicate by replacing $\xi$ with $\emptyset$. From Proposition \ref{prop:est_f} we therefore get
\begin{equation}
K_{\hat T}(\hat {\f s}) \;\leq\; \pb{C \hat m^2 \hat T^{-2} \norm{w}_{L^\infty}}^{\hat m} \;\leq\; \pb{C M^2 \eta^{-2} \norm{w}_{L^\infty}}^{\hat m}\,,
\end{equation}
where in the last step we used that $\hat T \geq T /2 \geq \eta / 4$ and $\hat m \leq m \leq M$.

Going back to \eqref{K_T_split}, we deduce the following implication:
\begin{equation} \label{splitting_iteration}
1 \leq m \leq M \,, \quad T \geq \eta/2\,, \quad \f s \in [0, T - \eta/2]^m \qquad \Longrightarrow \qquad K_T(\f s) \;\leq\; \pb{C M^2 \eta^{-2} \norm{w}_{L^\infty}}^{\hat m} K_{\tilde T}(\tilde {\f s})\,.
\end{equation}
We now iterate this procedure, starting from $m^{(0)} \deq M$, $T^{(0)} \deq \abs{\f u} + \eta/2$, and $s_k^{(0)} \deq \sum_{i = M - k + 2}^M u_i$ for $k = 1, \dots, M$. Hence, the left-hand side of \eqref{g_3_est} is equal to $K_{T^{(0)}}(\f s^{(0)})$. Note that $m^{(0)}$, $T^{(0)}$, and $\f s^{(0)}$ satisfy the assumptions of \eqref{splitting_iteration}.

We now iterate according to the following rules.
\begin{enumerate}
\item
If $m^{(i)} = 0$ then stop.
\item
If $m^{(i)} \geq 1$ then set $m^{(i + 1)} \deq \tilde m^{(i)}$, $T^{(i + 1)} \deq \tilde T^{(i)}$, and $\f s^{(i + 1)} \deq \tilde {\f s}^{(i)}$.
\end{enumerate}

Since $\hat T \geq T/2$, we deduce that $T^{(i+1)} \leq T^{(i)}/2$. Moreover, by induction we find that if $m^{(i)} \geq 1$ then we always have $T^{(i)} \geq s^{(i)}_{m^{(i)}} + \eta /2 \geq \eta/2$. We deduce that the algorithm stops after at most $\log_2 (2/\eta)$ steps. Denote by $i_*$ the time at which the algorithm stops. Then we have $\sum_{i  = 0}^{i_* - 1} \hat m^{(i)} = M$. We therefore conclude using \eqref{splitting_iteration} that
\begin{equation*}
K_{T^{(0)}}(\f s^{(0)}) \;\leq\; \pb{C M^2 \eta^{-2} \norm{w}_{L^\infty}}^{M}\,,
\end{equation*}
which is the claim.
\end{proof}

\begin{lemma} \label{lem:gest2}
For any $t \in (0,1)$ and any $z$ with $\re z \geq 0$ we have $\normb{\ee^{-t (H_{\tau,0}+z W_{\tau})}}_{\tilde {\fra S}^{1/t}(\cal F)} \leq 1$.
\end{lemma}
\begin{proof}
We use the Trotter-Kato product formula, which states that
\begin{equation} \label{trotter}
\ee^{-t (H_{\tau,0}+z W_{\tau})} \;=\; \lim_{n \to \infty} \pb{\ee^{-t H_{\tau,0}/n} \ee^{-tz W_\tau/n}}^n
\end{equation}
in $\norm{\cdot}_{\tilde {\fra S}^1(\cal F)}$ and hence also in $\norm{\cdot}_{\tilde {\fra S}^{1/t}(\cal F)}$. (See e.g.\ the presentation of \cite{NZ90} whose proof may be adapted to the case $\im z \neq 0$.)
Using Lemma \ref{Holder's inequality in S^p} we find
\begin{multline*}
\normB{\pb{\ee^{-t H_{\tau,0}/n} \ee^{-tz W_\tau/n}}^n}_{\tilde {\fra S}^{1/t}(\cal F)} \;\leq\;
\normB{\ee^{-t H_{\tau,0}/n} \ee^{-tz W_\tau/n}}_{\tilde {\fra S}^{n/t}(\cal F)}^n
\\
\leq\; \normb{\ee^{-t H_{\tau,0}/n}}_{\tilde {\fra S}^{n/t}(\cal F)}^n \, \normb{\ee^{-tz W_\tau/n}}_{\tilde {\fra S}^{\infty}(\cal F)}^n
\;\leq\; \normb{\ee^{-t H_{\tau,0}/n}}_{\tilde {\fra S}^{n/t}(\cal F)}^n \;=\; 1\,,
\end{multline*}
where in the third step we used Lemma \ref{lem_W_tau_pos} below, and in the last step we used \eqref{estimated_by1}.
The claim now follows from \eqref{trotter}.
\end{proof}

\begin{lemma} \label{lem_W_tau_pos}
The operator $W_\tau$ is positive.
\end{lemma}
\begin{proof}
Using the Fourier representation $w(x) = \int \hat w(\dd k) \, \ee^{\ii k \cdot x}$, we find
\begin{equation*}
W_\tau \;=\; \frac{1}{2} \int \hat w(\dd k)\, \pbb{\int \dd x \, \ee^{\ii k \cdot x} \pb{\phi^*_\tau(x) \phi_\tau(x) - \varrho_\tau(x)}}
\pbb{\int \dd y \, \ee^{\ii k \cdot y} \pb{\phi^*_\tau(y) \phi_\tau(y) - \varrho_\tau(y)}}^*\,,
\end{equation*}
which is manifestly a positive operator since $\hat w$ is nonnegative by assumption.
\end{proof}

\begin{lemma} \label{lem:gest3}
For $\xi \in \fra B_p$ we have $\normb{\ee^{-\eta H_{\tau,0}} \Theta_\tau(\xi) \ee^{-\eta/2 \,  H_{\tau,0}}}_{\tilde {\fra S}^{2/3\eta}(\cal F)} \leq  (C p \eta^{-2})^p$.
\end{lemma}
\begin{proof}
As in the proof of Lemma \ref{lem:gest1}, we use the quantum Wick theorem, which requires an even Schatten exponent. To that end, let $\hat \eta \deq \max \hb{s \in (0,\eta) \col 1/s \in 2 \N}$.
Then by Lemma \ref{Holder's inequality in S^p} and \eqref{estimated_by1} we have
\begin{equation} \label{xi_remaind_est0}
\normB{\ee^{-\eta H_{\tau,0}} \Theta_\tau(\xi) \ee^{-\eta/2 \,  H_{\tau,0}}}_{\tilde {\fra S}^{2/3\eta}(\cal F)}
\;\leq\; \normB{\ee^{-\hat \eta/2\, H_{\tau,0}} \Theta_\tau(\xi) \ee^{-\hat \eta/2\, H_{\tau,0}}}_{\tilde {\fra S}^{1/\hat \eta}(\cal F)}\,.
\end{equation}
Now we write
\begin{equation} \label{xi_remainder_est}
 \normB{\ee^{-\hat \eta/2\, H_{\tau,0}} \Theta_\tau(\xi) \ee^{-\hat \eta/2\, H_{\tau,0}}}_{\tilde {\fra S}^{1/\hat \eta}(\cal F)}^{1/\hat \eta} \;=\; \frac{1}{\tr (\ee^{- H_{\tau,0}})} \tr \qbb{\pB{\ee^{-H_{\tau,0}/q} \, \Theta_\tau(\xi)}^{q}}\,,
\end{equation}
where we abbreviated $q \deq 1 / \hat \eta$ and used the cyclicity of the trace. Since $\hat \eta \geq \eta / 2$, we have $q \leq 2 / \eta$.
The estimate of the right-hand side of \eqref{xi_remainder_est} is performed using the quantum Wick theorem, analogously to the argument from Sections \ref{sec:time_evolved_G}--\ref{sec:upper_bound}. Here we only explain the differences, and refer to Sections \ref{sec:time_evolved_G}--\ref{sec:upper_bound} for the full details and notations.

Analogously to Definition \ref{cal X}, we encode the creation and annihilation operators in the $q$ operators $\Theta_\tau(\xi)$ are using an abstract vertex space $\cal X \deq \{1, \dots, q\} \times \{1, \dots, p\} \times \{-1, +1\}$. We again denote the elements of $\cal X$ with triples $\alpha = (i,r,\sign)$, where $i = 1, \dots, q$ indexes the $q$ operators $\Theta_\tau(\xi)$, $r = 1, \dots, p$ the $p$ the creation or annihilation operators in the definition of $\Theta_\tau(\xi)$, and $\sign = \pm 1$ means a creation / annihilation operator. We introduce a total order on $\cal X$ that coincides with the ordering of the creation and annihilation operators: the triples $(i,r,\sign)$ are lexicographically ordered according to the string $i \sign r$, where the sets $\{1, \dots, q\}$ and $\{1, \dots, p\}$ carry the natural order and in $\{-1,+1\}$ we use the order $+1 < -1$ for the letter $\sign$. See Figure \ref{fig:graph6} for a graphical illustration of $\cal X$.

\begin{figure}[!ht]
\begin{center}
{\scriptsize 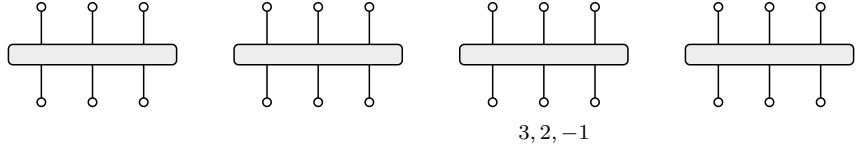}
\end{center}
\caption{
The vertex set $\cal X$ with $q = 4$ and $p = 3$, represented with white dots. In analogy to Figures \ref{fig:graph1} and \ref{fig:graph2}, we also draw grey rectangles to indicate the copy  $i = 1, \dots, q$ of $\xi$ that each vertex $(i,r,\sign)$ belongs to. For $\sign = +1$ and $\sign = -1$ we draw $(i,r,\sign)$ above and below the grey rectangle respectively. We indicate the location of the vertex $(3,2,-1)$ as an example.
\label{fig:graph6}}
\end{figure}

The set of pairings $\fra P$ is simply given by the set of bipartite pairings $\Pi$ on $\cal X$, whereby each block $(\alpha, \beta) \in \Pi$ satisfies $\sign_\alpha \sign_\beta = -1$. Then by the quantum Wick theorem (Lemma \ref{Quantum Wick theorem}) we have
\begin{equation} \label{tr_xi_Wick}
\frac{1}{\tr (\ee^{- H_{\tau,0}})} \tr \qbb{\pB{\ee^{-H_{\tau,0}/q} \, \Theta_\tau(\xi)}^{q}} \;=\; \sum_{\Pi \in \fra P} \cal I^\xi_{\tau, \Pi}\,,
\end{equation}
where we defined
\begin{equation*}
\cal I^\xi_{\tau, \Pi} \;\deq\; \int_{\Lambda^{\cal X}} \dd \f x \, (\xi^{\otimes q})(\f x) \prod_{(\alpha,\beta) \in \Pi} Q_{\alpha\beta}(x_\alpha; x_\beta)\,, \qquad Q_{\alpha\beta}(x_\alpha; x_\beta) \;\deq\; \rho_{\tau,0} \big(\cal B_{\alpha} (\f x, \f t) \, \cal B_{\beta} (\f x, \f t)
\big)\,.
\end{equation*}
Here we introduced the time labels $t_\alpha \equiv t_i \deq 1 - i/q$ for $\alpha = (i,r,\sign) \in \cal X$.

Next, we use Lemma \ref{Correlation functions} to determine the operator kernels $Q_{\alpha\beta}(x_\alpha; x_\beta)$. We consider two cases for $\alpha < \beta$.
\begin{enumerate}
\item
If $\sign_\alpha = +1$ and $\sign_\beta = -1$ then $Q_{\alpha\beta}(x_\alpha; x_\beta) = G_{\tau, -t}(x_\alpha; x_\beta)$
for some $t \leq 1 - 1/q$.
\item
If $\sign_\alpha = -1$ and $\sign_\beta = +1$ then $Q_{\alpha\beta}(x_\alpha; x_\beta) = G_{\tau,  t }(x_{\alpha}; x_{\beta}) + \frac{1}{\tau} S_{\tau,  t}(x_{\alpha}; x_{\beta})$
for some $t \geq 1/q$.
\end{enumerate}
Indeed, in case (i) we use Lemma \ref{Correlation functions} (i) and the fact that $t_i \leq 1 - 1/q$ for all $i$. In case (ii) we use that, by definition of the ordering on $\cal X$, this case can only happen if $i_\alpha < i_\beta$, which implies that $t \deq t_\alpha - t_\beta \geq 1/q$.

Either way, we deduce from \eqref{h_tau_conv} and Lemmas \ref{cor:conv_G_tau corollary} and \ref{Heat kernel estimate} that $\norm{Q_{\alpha \beta}}_{\fra S^2(\fra H)} \leq C q$.
Plugging this estimate into the definition of $\cal I^\xi_{\tau, \Pi}$ and using \eqref{S2L2} and Cauchy-Schwarz, we obtain $\abs{\cal I^\xi_{\tau, \Pi}} \leq (Cq)^{pq}$. Recalling \eqref{tr_xi_Wick} and noting that $\abs{\fra P} = (pq)!$, we conclude that
\begin{equation*}
\frac{1}{\tr (\ee^{- H_{\tau,0}})} \tr \qbb{\pB{\ee^{-H_{\tau,0}/q} \, \Theta_\tau(\xi)}^{q}} \;\leq\; (C pq^2)^{pq}\,.
\end{equation*}
Recalling \eqref{xi_remainder_est}, we find that the left-hand side of \eqref{xi_remaind_est0} is bounded by $(Cpq^2)^p$, from which the claim follows using $q \leq 2 / \eta$.
\end{proof}

Unleashing Lemmas \ref{lem:gest1}--\ref{lem:gest3} on \eqref{g_est}, we obtain the following result.

\begin{proposition} \label{prop:g_est}
For $\f u \in (0,1 - 2 \eta)^M$ satisfying $\abs{u} < 1 - 2 \eta$ and $\xi \in \fra B_p$ we have
\begin{equation*}
\abs{g^\xi_{\tau,M}(z, \f u)} \;\leq\; \pb{C M^2 \eta^{-2} \norm{w}_{L^\infty}}^{M} \, (C p \eta^{-2})^p\,.
\end{equation*}
\end{proposition}

Proposition \ref{prop:remainder} follows immediately from Proposition \ref{prop:g_est} and \eqref{R_using_g}.

We conclude this section by proving that $A_\tau^\xi$ is analytic in the right half-plane.

\begin{lemma} \label{lem:A_tau_anal}
The function $A_\tau^\xi$ is analytic in $\{z \col \re z > 0\}$.
\end{lemma}
\begin{proof}
Let $P^{(\leq n)}$ be the orthogonal projection onto the subspace $\bigoplus_{n' = 0}^n \fra H^{(n')}$ of $\cal F$. Note that $P^{(\leq n)}$ commutes with $\Theta_\tau(\xi)$, $H_{\tau,0}$, and $W_\tau$. Define
\begin{equation} \label{A xi tau n}
A_{\tau,n}^\xi(z) \;\deq\; \frac{\tr \pb{P^{(\leq n)} \,  \Theta_\tau(\xi) \, \ee^{-\eta H_{\tau,0}} \ee^{-(1 - 2\eta) H_{\tau,0}  - z W_\tau} \ee^{-\eta H_{\tau,0}}}}{\tr (\ee^{-H_{\tau,0}})}\,.
\end{equation}
Using cyclicity of the trace and H\"older's inequality from Lemma \ref{Holder's inequality in S^p}, we find for $\re z \geq 0$
\begin{align*}
\abs{A_{\tau,n}^\xi(z)} &\;\leq\; \normb{\ee^{-\eta /2 H_{\tau,0}} \Theta_\tau(\xi) \ee^{-\eta H_{\tau,0}}}_{\tilde {\fra S}^{2/3\eta}(\cal F)} \,
\normb{\ee^{-(1 - 2 \eta) H_{\tau,0} - z W_\tau}}_{\tilde {\fra S}^{1/(1 -2 \eta)}(\cal F)}\,
\normb{\ee^{-\eta/2\, H_{\tau,0}}}_{\tilde {\fra S}^{2/\eta}(\cal F)}
\\
&\;\leq\; (Cp \eta^{-2})^p\,,
\end{align*}
where in the second step we used Lemmas \ref{lem:gest2} and \ref{lem:gest3}.

Next, using that $\Theta_\tau(\xi)$ and $W_\tau$ are bounded on the range of $P^{(\leq n)}$, a simple argument using the Duhamel formula shows that $A_{\tau,n}^\xi$ is analytic for $\re z > 0$. Since $\lim_{n \to \infty} A_{\tau,n}^\xi(z) = A_\tau^\xi(z)$ and $A_{\tau,n}^\xi(z)$ is uniformly bounded for $\re z \geq 0$, we conclude that $A_\tau^\xi$ is analytic for $\re z > 0$.
\end{proof}

\section{The classical problem} \label{sec:classical}

\subsection{Construction of $W$}
In addition to the classical Green function from \eqref{classical_G}, we define its truncated version
\begin{equation}
\scalar{f}{G_{[K]} g} \;\deq\; \int \dd \mu \, \scalar{f}{\phi_{[K]}} \, \scalar{\phi_{[K]}}{g} \;=\; \scalarbb{f}{\sum_{k = 0}^K \frac{1}{\lambda_k} u_k u_k^*\, g}\,.
\end{equation}
In particular, the classical density $\varrho_{[K]} = G_{[K]}(x;x)$ is given by the diagonal of $G_{[K]}$.

Throughout this section we use the classical (i.e.\ probabilistic) Wick theorem on polynomials of $W_{[K]}$ and $\Theta(\xi)$. The resulting pairings are conveniently described by the abstract vertex set $\cal X$ from Definition \ref{cal X} (i). Note that, since the classical fields $\phi$ commute, the ordering from Definition \ref{cal X} (ii) is superfluous in this whole section. As in Section \ref{sec:time_evolved_G}, each vertex $\alpha = (i,r,\sign) \in \cal X$ encodes a field $\bar \phi$ or $\phi$, with $i = 1, \dots, m+1$ indexing the operator ($w$ for $i = 1, \dots, m$ and $\xi$ for $i = m+1$), and $\sign = +1$ and $\sign = -1$ corresponding to $\bar \phi$ and $\phi$ respectively. In this section we take over many notations from Sections \ref{sec:time_evolved_G}--\ref{sec:graphs} without further comment.

\begin{proof}[Proof of Lemma \ref{lem:def_W}]
It suffices to prove convergence of $W_{[K]}$ in $L^m(\mu)$ for a fixed $m \in 2\N$. Choose $K_1, \dots, K_m \geq K$ arbitrary sequences indexed by $K \in \N$. We claim that as $K \to \infty$ the quantity $M_{\f K} \deq \int \dd \mu \, W_{[K_1]} \cdots W_{[K_m]}$
converges to a limit that does not depend on the sequences $\f K = (K_1, \dots, K_m)$. The proof is an application of Wick's theorem. To that end, we use the vertex set $\cal X \equiv \cal X(m,0)$ from Definition \ref{cal X} (i). Recalling the definition of the set $\fra P$ of admissible pairings $\Pi$ from Definition \ref{def_pairing}, we use Wick's theorem to write $M_{\f K} = \frac{1}{2^m} \sum_{\Pi \in \fra P} \cal I_{\f K, \Pi}$, where
\begin{equation*}
\cal I_{\f K, \Pi}
\;\deq\; \int_{\Lambda^{\cal X}} \dd \f x \,\prod_{i=1}^{m} \pbb{ w(x_{i,1,1}-x_{i,2,1}) \prod_{r=1}^{2} \delta(x_{i,r,1}-x_{i,r,-1})} \,  \prod_{(\alpha,\beta) \in \Pi} G_{[K_{i_\alpha} \wedge K_{i_\beta}]}(x_\alpha; x_\beta)\,.
\end{equation*}
Here we used that $\int \dd \mu \, \phi_{K}(x) \bar \phi_{\tilde K}(y) = G_{[K \wedge \tilde K]}(x;y)$. Next, fix $\Pi \in \fra P$ and use the graph $(\cal V, \cal E)$ associated with $\Pi$ from Definition \ref{def_collapsed_graph} to write
\begin{equation} \label{I_Pi_W_constr}
\cal I_{\f K, \Pi} \;=\; \int_{\Lambda^{\cal V}} \dd \f y \, \pBB{\prod_{i=1}^{m} w(y_{i,1}-y_{i,2})} \prod_{\{\alpha, \beta\} \in \cal E} G_{[K_{i_\alpha} \wedge K_{i_\beta}]}(y_\alpha; y_\beta)\,.
\end{equation}
We denote by $\cal I_\Pi$ the expression obtained from the right-hand side of \eqref{I_Pi_W_constr} by replacing $G_{[K_{i_\alpha} \wedge K_{i_\beta}]}$ with $G$.
We claim that $\cal I_{\Pi}$ is well-defined.
For the proof, we find using Lemma \ref{Positivity lemma 2} that
\begin{equation*}
\abs{\cal I_{\Pi}} \;\leq\; \norm{w}_{L^\infty}^m \, \int_{\Lambda^{\cal V}} \dd \f y \, \prod_{\{\alpha, \beta\} \in \cal E} G(y_\alpha; y_\beta)\,.
\end{equation*}
That the right-hand side is finite follows easily from $w \in L^\infty(\Lambda)$ and the path decomposition $\cal E = \bigsqcup_{\cal P \in \conn(\cal E)} \cal P$ of $\cal E$ from Definition \ref{def_collapsed_graph}, using that $\tr G^k < \infty$ for $k \geq 2$ by the assumption $s \geq -1$ in \eqref{tr_h_assump}.

Next, we claim that $\cal I_{\f K, \Pi} \to \cal I_\Pi$ as $K \to \infty$. This follows analogously to the above argument, by telescoping on the right-hand side of
\begin{equation*}
\abs{\cal I_{\f K, \Pi} - \cal I_{\Pi}} \;\leq\; \norm{w}_{L^\infty}^m \, \int_{\Lambda^{\cal V}} \dd \f y \, \absBB{\prod_{\{\alpha, \beta\} \in \cal E} G_{[K_{i_\alpha} \wedge K_{i_\beta}]}(y_\alpha; y_\beta) - \prod_{\{\alpha, \beta\} \in \cal E} G(y_\alpha; y_\beta)}\,,
\end{equation*}
and using that $\norm{G_{[K_{i_\alpha} \wedge K_{i_\beta}]} - G}_{\fra S^2(\fra H)} \to 0$ as $K \to \infty$, by definition of $G_{[K]}$ and the assumption $s \geq -1$ in \eqref{tr_h_assump}.

We may now easily conclude the proof. For $K \to \infty$ with $\tilde K \geq K$ we have
\begin{multline*}
\lim_{K \to \infty} 2^m \, \norm{W_{[\tilde K]} - W_{[K]}}_{L^m(\mu)}^m \;=\; \lim_{K \to \infty} \sum_{l = 0}^m (-1)^l \binom{m}{l} \int \dd \mu \, W_{[\tilde K]}^{m - l} W_{[K]}^{l}
\\
=\; \lim_{K \to \infty} \sum_{l = 0}^m (-1)^l \binom{m}{l} \int \dd \mu \, W_{[K]}^{m - l} W_{[K]}^{l} \;=\; 0\,,
\end{multline*}
as claimed.
\end{proof}

\subsection{Expansion}
We proceed analogously to Section \ref{sec:quantum}, except that in the classical case the expansion and its control is much easier because we are dealing with commuting random variables.

For a self-adjoint observable $\xi \in \fra B_p$, we define the random variable
\begin{equation}
\Theta(\xi) \;\deq\; \int \dd x_1 \cdots \dd x_p \, \dd y_1 \cdots \dd y_p \, \xi(x_1, \dots, x_p; y_1, \dots, y_p) \bar \phi(x_1) \cdots \bar \phi(x_p) \phi(y_1) \cdots \phi(y_p)\,,
\end{equation}
which is the classical analogue of \eqref{theta_tau}. We write the expectation of $\Theta(\xi)$ in the state $\rho(\cdot)$ as
\begin{equation} \label{rho_frac}
\rho(\Theta(\xi)) \;=\; \frac{\tilde \rho_1(\Theta(\xi))}{\tilde \rho_1(1)}\,,
\end{equation}
where we defined
\begin{equation} \label{def_rho_z}
\tilde \rho_z(X) \;\deq\; \int X \, \ee^{- z W} \, \dd \mu
\end{equation}
for $\re z \geq 0$. Performing a Taylor expansion of the exponential $\ee^{-z W}$ up to order $M$, we get
\begin{equation} \label{def_A}
A^\xi(z) \;\deq\; \tilde \rho_z(\Theta(\xi)) \;=\; \sum_{m = 0}^{M - 1} a^\xi_m z^m + R^\xi_M(z)\,,
\end{equation}
where we defined
\begin{equation*}
a_m^\xi \;\deq\; \frac{(-1)^m}{m!} \int \Theta(\xi) \, W^m \, \dd \mu \,, \qquad R_M^\xi(z) \;\deq\; \frac{(-1)^M z^M}{M!} \int \Theta(\xi) \, W^M \, \ee^{-\tilde z W} \, \dd \mu \quad \text{for some $\tilde z \in [0,z]$}\,.
\end{equation*}

\begin{lemma} \label{lem_classical_coeff}
Recall the definition of $a_{\infty, m}^\xi$ from \eqref{f_m,a_m^xi classical}.
For any $m,p \in \N$ and $\xi \in \fra B_p$, we have 
\begin{equation}
\label{lem_classical_coeff identity}
a_m^\xi = a_{\infty, m}^\xi\,.
\end{equation}
\end{lemma}
\begin{proof}
Since $\xi \in \fra S^2(\fra H^{(p)})$ and $G \in \fra S^2(\fra H)$, we find using Wick's theorem, the Cauchy-Schwarz inequality, and \eqref{S2L2} that $\Theta(\xi) \in L^2(\mu)$.
Using Lemma \ref{lem:def_W} and H\"older's inequality, we therefore write
\begin{equation*}
a_m^\xi \;=\; \frac{(-1)^m}{m!}\lim_{K \to \infty} \int \Theta(\xi) W_{[K]}^m \, \dd \mu\,.
\end{equation*}
Using Wick's theorem, we rewrite the right-hand side as a sum over pairings. To that end, let $\cal X \equiv \cal X(m,p)$ be the vertex set from Definition \ref{cal X}, and $\fra P \equiv \fra P(m,p)$ be the set of pairings from Definition \ref{def_pairing}. Then we get from Wick's theorem that
\begin{equation*}
a_m^\xi \;=\; \frac{(-1)^m}{m! \, 2^m} \lim_{K \to \infty} \sum_{\Pi \in \fra P} \cal I_{[K], \Pi}^\xi\,,
\end{equation*}
where we defined
\begin{equation*}
\cal I_{[K], \Pi}^\xi \;\deq\; \int_{\Lambda^{\cal V}} \dd \f y \, \pBB{\prod_{i=1}^{m} w(y_{i,1}-y_{i,2})} \xi(\f y_1) \prod_{e \in \cal E} \cal J_{[K],e}(\f y_e)
\end{equation*}
and
\begin{equation*}
\cal J_{[K],e}(\f y_e) \;\deq\;
\begin{cases}
G(y_{\va}; y_{\vb}) & \text{if $e = \{a,b\}$ with $a,b \in \cal V_1$}
\\
G_{[K]}(y_{\va}; y_{\vb}) & \text{if $e = \{a,b\}$ with $a \in \cal V_2$ or $b \in \cal V_2$}\,.
\end{cases}
\end{equation*}
Here $(\cal V, \cal E)$ is the multigraph associated with $\Pi$ from Definition \ref{def_collapsed_graph}.
See also Definitions \ref{J_e definition classical} and \ref{def_calI infinity}. Here we used that
\begin{equation*}
\E \phi(x) \bar \phi_{[K]}(y) \;=\; \E \phi_{[K]}(x) \bar \phi(y) \;=\; \E \phi_{[K]}(x) \bar \phi_{[K]}(y) \;=\; G_{[K]}(x;y)\,,
\end{equation*}
as follows from the definition of $\phi_{[K]}$.

Next, by spectral decomposition, we immediately find that $\norm{G_{[K]} - G}_{\fra S^2(\fra H)} \to 0$ as $K \to \infty$. Recalling Definitions \ref{J_e definition classical} and \ref{def_calI infinity}, by a simple telescoping argument on the definition of $\cal I_{[K], \Pi}^\xi$, using \eqref{S2L2}, $w \in L^\infty(\Lambda)$, and the Cauchy-Schwarz inequality, we find  that $\cal I_{[K],\Pi}^\xi \to \cal I_{\Pi}^\xi$ as $K \to \infty$. Now the claim follows by \eqref{f_m,a_m^xi classical}.
\end{proof}

From Lemma \ref{lem_classical_coeff}, Corollary \ref{cor:bound_am}, and Proposition \ref{Convergence of the explicit terms}, we deduce the following result. (For conciseness, we obtain the upper bound from the corresponding quantum result that we have already established, Corollary \ref{cor:bound_am}, but such an upper bound could also be easily derived using Wick's theorem without passing by the quantum problem.)
\begin{corollary} \label{cor:a_m_est}
For each $m \in \N$ we have
\begin{equation}
\label{cor:a_m_est bound}
\absb{a_{m}^{\xi}} \;\leq\; (C p)^p \,(C \|w\|_{L^{\infty}})^m \, m!\,.
\end{equation}
\end{corollary}

Next, we estimate the remainder term $R_M^\xi(z)$.

\begin{lemma} \label{lem:R_est}
For each $M \in \N$ and $\re z \geq 0$ we have
\begin{equation}
\label{lem:R_est bound}
\absb{R_M^\xi(z)} \;\leq\; (Cp)^p \, \pb{C \norm{w}_{L^\infty} \abs{z}}^M M!\,.
\end{equation}
\end{lemma}

\begin{proof}
As in Lemma \ref{lem_W_tau_pos}, we easily find that $W \geq 0$. Hence, using the Cauchy-Schwarz inequality we find
\begin{equation*}
\absb{R_M^\xi(z)} \;\leq\; \frac{\abs{z}^M}{M!} \int \abs{\Theta(\xi)} \, W^M \, \dd \mu\;\leq\;  \frac{\abs{z}^M}{M!} \pbb{\int \Theta(\xi)^2 \, \dd \mu}^{1/2} \pbb{\int W^{2M} \, \dd \mu}^{1/2}\,.
\end{equation*}
Using $\xi \in \fra S^2(\fra H^{(p)})$, we find using Wick's theorem, the Cauchy-Schwarz inequality, and \eqref{S2L2}, that $\int \Theta(\xi)^2 \, \dd \mu \leq (Cp)^{2p}$. Moreover, using Corollary \ref{cor:a_m_est} we find $\int W^{2M} \, \dd \mu \leq \norm{w}_{L^\infty}^{2M} (CM)^{4M}$. The claim now follows.
\end{proof}

Finally, a simple application of Wick's theorem and the Cauchy-Schwarz inequality yields the following result.

\begin{lemma} \label{lem:A_analytic}
The function $A^\xi$ from \eqref{def_A} is analytic for $\re z > 0$.
\end{lemma}

\subsection{Proof of Theorem \ref{thm:main}} \label{sec:conclusion_proof}
We now have all the necessary ingredients to conclude the proof of Theorem \ref{thm:main}.

Using the duality $\fra S^2(\fra H^{(p)}) \cong \fra S^2(\fra H^{(p)})^*$ (see e.g.\ \cite{Simon05}), we write
\begin{equation} \label{gamma_diff_pf}
\norm{\gamma^\eta_{\tau,p} - \gamma_p}_{\fra S^2(\fra H^{(p)})} \;=\; \sup_{\xi \in \fra B_p} \absb{\tr \pb{\gamma^\eta_{\tau,p} \, \xi - \gamma_p \, \xi}} \;=\; \sup_{\xi \in \fra B_p} \absb{\rho_\tau^\eta (\Theta_\tau(\xi)) - \rho (\Theta(\xi))}\,,
\end{equation}
where we recall that $\fra B_p$ denotes the unit ball of $\fra S^2(\fra H^{(p)})$.

Let the functions $A_\tau$ and $A^\xi_\tau$ be defined as in \eqref{def_A} and \eqref{def_A_tau} respectively. We now verify the assumptions of Theorem \ref{Borel summation convergence}, with $\nu = \pb{\frac{Cp}{\eta^2}}^p$ and $\sigma = \frac{C\norm{w}_{L^\infty}}{\eta^2}$.
That $A^\xi$ and $A^\xi_\tau$ are analytic in $\cal C_R$ follows from Lemmas \ref{lem:A_analytic} and \ref{lem:A_tau_anal}. The asymptotic expansions \eqref{Asymptotic expansion tau xi} are given in \eqref{def_A} and Lemma \ref{lem_Dyson_exp}. The bound \eqref{Explicit term bound tau xi} follows from Corollaries \ref{cor:a_m_est} and \ref{cor:bound_am}. The bound \eqref{Remainder term bound tau xi} follows from Lemma \ref{lem:R_est} and Proposition \ref{prop:remainder}. Finally, the convergence \eqref{Difference of explicit terms} follows from Proposition \ref{Convergence of the explicit terms} and Lemma \ref{lem_classical_coeff}. Setting $z = 1$ and $R = 2$ in Theorem \ref{Borel summation convergence}, we therefore obtain from \eqref{A_A_tau_conv} that
\begin{equation*}
\tilde \rho^\eta_{\tau,1}(\Theta_\tau(\xi))) \;\longrightarrow\; \tilde \rho_1(\Theta(\xi)) \qquad \text{uniformly in $\xi \in \fra B_p$}\,.
\end{equation*}
Since $\tilde \rho_1(1) \neq 0$, we conclude from \eqref{rho_frac} and \eqref{rho_tau_frac} that
\begin{equation*}
\rho^\eta_{\tau}(\Theta_\tau(\xi))) \;\longrightarrow\; \rho(\Theta(\xi)) \qquad \text{uniformly in $\xi \in \fra B_p$}\,.
\end{equation*}
The claim now follows using \eqref{gamma_diff_pf}. This concludes the proof of Theorem \ref{thm:main}.

\section{The one-dimensional problem} \label{sec:one-d}

In this section we consider the case $d=1$ and prove Theorems \ref{thm:1D} and \ref{thm:1D local}. Thus, throughout this section we assume \eqref{tr_h_assump} holds for $s=0$. The main difference between this section and Sections \ref{sec:quantum}--\ref{sec:classical} is that, because $\phi \in L^2(\Lambda)$ $\mu$-almost surely, for $d = 1$ there is no need to renormalize the interaction. This simplifies the argument considerably, and the proofs of Theorems \ref{thm:1D} and \ref{thm:1D local} are easier than that of Theorem \ref{thm:main}.

In this section we take $h_\tau \equiv h$, so that the quantum Green function \eqref{quantum_G} is $G_\tau \deq \frac{1}{\tau (\ee^{h/\tau}-1)}$.
In particular, using Lemma \ref{Positivity lemma}
we have
\begin{equation}
\label{Tr G_tau 1D}
\|G_\tau\|_{\fra S^1(\fra H)} \;=\; \tr G_\tau \;=\; \sum_{k \geq 0} \frac{1}{\tau (\ee^{\lambda_k/\tau}-1)} \;\leq\;  \sum_{k \geq 0} \frac{1}{\lambda_k} \;=\; \tr h^{-1} \;<\; \infty\,.
\end{equation}

\subsection{Proof of Theorem \ref{thm:1D}, I: convergence in Hilbert-Schmidt norm}
\label{1D subsection 1}

\begin{proposition}
\label{1D proposition 1}
Under the assumptions of Theorem \ref{thm:1D} we have $\lim_{\tau \to \infty} \norm{\gamma_{\tau,p} - \gamma_p}_{\fra S^2(\fra H^{(p)})} = 0$ for all $p \in \mathbb{N}$.
\end{proposition}
This subsection is devoted to the proof of Proposition \ref{1D proposition 1}. The proof is similar to that of Theorem \ref{thm:main}, but we we do not renormalize the interaction and we set $\eta = 0$. The absence of renormalization requires us to work with a different graph structure, which allows for the presence of loops. Moreover, we are able to estimate the remainder term by using the Feynman-Kac formula, thus circumventing the more complicated estimates of Section \ref{sec:remainder} which require $\eta > 0$.
The rest of the proof proceeds analogously to that of Theorem \ref{thm:main}.

We first focus on the quantum case. We begin by setting up the Duhamel expansion analogously to Section \ref{Duhamel expansion}, with the splitting of \eqref{H_tau_1d} as $H_\tau = H_{\tau,0} + W_\tau$, where we defined
\begin{equation} \label{W_tau 1D}
H_{\tau,0} \;\deq\; \int \dd x \, \dd y \, \phi^*_\tau(x) \, h(x;y) \, \phi_\tau(y)\,, \qquad
W_\tau \;\deq\; \frac{1}{2}
\int \dd x \, \dd y \, \phi^*_\tau(x) \phi^*_\tau(y)  \, w(x - y) \, \phi_\tau(x) \phi_\tau(y)\,.
\end{equation}
Given $\xi \in \fra B_p$, we perform a Taylor expansion up to order $M \in \N$ of 
$A_{\tau}^{\xi}(z) \deq \tilde \rho_{\tau,z}(\Theta_\tau(\xi))$
in the parameter $z$ by using a Duhamel expansion. Here $\re z \geq 0$ and $\tilde \rho_{\tau,z} \equiv \tilde \rho_{\tau,z}^0$ is defined as in \eqref{def_rho_tau_tilde} with $W_\tau$ given by \eqref{W_tau 1D} and $\eta=0$. By the same argument as in Lemma \ref{lem_Dyson_exp}, the coefficients $a^\xi_{\tau,m}$ of the expansion are given by \eqref{Explicit term a} and the remainder term $R^\xi_{\tau,M}(z)$ by \eqref{Remainder term R}, where $W_\tau$ is given by \eqref{W_tau 1D} and $\eta=0$. Given $m,p \in \N$, we recall the set $\cal X \equiv \cal X(m,p)$ from Definition \ref{cal X}. 
Since $W_\tau$ is now normal-ordered, we need to modify the order on $\cal X \equiv \cal X(m,p)$.
The vertices $(i,r,\sign)$ of $\cal X$ are now ordered according to the lexicographical order of the string $i \sign r$, whereby the value of $\sign$ is ordered as $+1 < -1$.
We slightly abuse notation by also denoting this new order as $\leq$. For the remainder of this section, we adapt all of the definitions from Section \ref{sec:time_evolved_G}, but by using the above lexicographical order instead of \eqref{linear order} and setting $\eta=0$ without further comment.
Owing to the absence of renormalization, we consider a larger set of pairings than $\fra P$ from Definition \ref{def_pairing}.
\begin{definition} \label{def_pairing 1D}
We define $\fra R \equiv \fra R(m,p)$ to be the set of pairings $\Pi$ of $\cal X$ such that for each $(\alpha, \beta) \in \Pi$ we have $\sign_{\alpha} \sign_{\beta} = -1$.
\end{definition}

We now extend the Definition \ref{def_collapsed_graph} to the class of all $\Pi \in \fra R$. In particular, we assign to each $\Pi \in \fra R$ an edge-coloured multigraph $(\cal V_\Pi,\cal E_\Pi,\sigma_\Pi)=(\cal V,\cal E,\sigma)$ satisfying points (i)--(iv) of Definition \ref{def_collapsed_graph}. 
By contruction of $\fra R$, the new graph can have loops. They have to be of the form $e=\{a,a\}$ for some $a \in \cal V_2$. Here, we recall the definition of the sets $\cal V_1$ and $\cal V_2$ in point (i) of Definition \ref{def_collapsed_graph}.

Given $\Pi \in \fra R$, we define its value $\cal I_{\tau,\Pi}^{\xi}(\f t)$ at $\f t \in \fra A \equiv \fra A(m)$ as in Definition \ref{def:I_tau_PI}. Note that $\cal J_{\tau,e}(\f y_e, \f s)$ is defined as in Definition \ref{J_e definition}, with the index $\sigma(e)$ defined using the above lexicographical order on $\cal V$. Here we recall that we are taking $\eta=0$ in \eqref{def_fraA}. Note that this definition makes sense for $\Pi \in \fra{R}$. By minor modifications of the proof of Lemma \ref{Wick_application lemma}, we deduce that the coefficients $a_{\tau,m}^{\xi}$ of the Taylor expansion of $A_\tau^\xi(z)$ are given by
\begin{equation}
\label{Wick_application lemma 1D identity}
a_{\tau,m}^{\xi} \;=\; \frac{(-1)^m}{2^m} \int_{\fra A}\dd \f t\, \sum_{\Pi \in \fra R} \cal I_{\tau,\Pi}^{\xi}(\f t)\,.
\end{equation}
Note that we now sum over pairings $\Pi \in \fra R$. The additional pairings $\Pi \in \fra R \setminus \fra P$ were forbidden in Section \ref{sec:quantum} by the renormalization.

For the following we fix $m,p \in \N$ and a pairing $\Pi \in \fra R$.
We can estimate all of the integrands on the right-hand side of \eqref{Wick_application lemma 1D identity} uniformly in $\f t \in \fra A$ and $\tau \geq 1$. 
\begin{proposition} \label{Value of pairing 1D}
For any $\Pi \in \fra R$, $\f t \in \fra A$, 
we have the estimate
\begin{equation}
\label{Value of pairing 1D bound}
\big| \cal I_{\tau,\Pi}^{\xi} (\f t) \big| \;\leq\; C^{m + p} \, \|w\|_{L^{\infty}}^m  \, \big(1+\|G_{\tau}\|_{\fra S^1(\fra H)}\big)^{m+2p}\,.
\end{equation}
\end{proposition}
This result is an analogue of Proposition \ref{Product of subgraphs} in the current setting. We prove Proposition \ref{Value of pairing 1D} in several steps, following the proof of Proposition \ref{Product of subgraphs}. An important difference is that now the associated graph $(\cal V, \cal E, \sigma)$ can have loops, which requires an appropriate modification in the argument.
Moreover, we need to take into account the change of the order from \eqref{linear order} to the lexicographical order. The latter point is a minor one. It is manifested in replacing for each $1 \leq i \leq m$ every occurrence of the factor
\begin{equation*}
\rho_{\tau,0} \big(\cal B_{(i,1,-1)} (\f x, \f t) \, \cal B_{(i,2,+1)} (\f x, \f t)\big)\;=\;G_\tau(x_{(i,1,-1)};x_{(i,2,+1)})+ \frac{1}{\tau} S_{\tau,0}(x_{(i,1,-1)};x_{(i,2,+1)})
\end{equation*}
by
\begin{equation*}
\rho_{\tau,0} \big(\cal B_{(i,2,+1)} (\f x, \f t) \, \cal B_{(i,1,-1)} (\f x, \f t)\big)\;=\;G_\tau(x_{(i,2,+1)};x_{(i,1,-1)})\;=\;G_\tau(x_{(i,1,-1)};x_{(i,2,+1)})\,.
\end{equation*}
Note that this only subtracts the term $\frac{1}{\tau} S_{\tau,0}(x_{(i,1,-1)};x_{(i,2,+1)})$ and that there is no time evolution applied to the Green function in either formula.  All of the other factors of $\rho_{\tau,0} \big(\cal B_{\alpha} (\f x, \f t) \, \cal B_{\beta} (\f x, \f t)\big)$ remain the same. The analysis now proceeds as with the order given by \eqref{linear order} and we do not emphasize this point in the sequel.

With $\f y, \f s$ as in \eqref{def_ys}, $\cal J_{\tau,e}$ as in \eqref{J_e}, and $\xi(\f y_1)$ as in \eqref{y_1}, we have that 
\eqref{Wick_application 2 bound}
holds for $\Pi$, which now belongs to $\fra R$. 
Moreover, we adapt the terminology from Definition \ref{Open and Closed path} and the discussion that immediately follows to the current setting. 

Instead of Lemma \ref{Closed path}, we have the following result.

\begin{lemma}
\label{Closed path 1D}
Suppose that $\cal P \in \conn(\cal E)$ is a closed path. Then
\begin{equation}
\label{Closed path bound 1D}
\int_{\Lambda^{\cal V(\cal P)}} \prod_{\va \in \cal V(\cal P)} \dd y_{\va} \, \prod_{e \in \cal P} \cal J_{\tau, e}(\f y_e,\f s) \;\leq\; C^{\abs{\cal V(\cal P)}}\,.
\end{equation}
Moreover, with $\hat{\cal J}_{\tau,e}$ as in \eqref{hat J_e}, the convergence \eqref{Closed path remark} holds.
\end{lemma}
\begin{proof}
We first prove \eqref{Closed path bound 1D}.
We consider the cases when $\abs{\cal V(\cal P)}=1$ and when $\abs{\cal V(\cal P)} \geq 2$ separately.
In the first case,
$\cal P$ is a loop of the form $e=\{a,a\}$ for some $a \in \cal V_2$.
In particular, the expression on the left-hand side of \eqref{Closed path bound 1D} equals
$\int_{\Lambda} \dd y_{\va} \, G_\tau(y_a;y_a) = \|G_\tau\|_{\fra S^1(\fra H)}$, so that the claim holds by \eqref{Tr G_tau 1D}.
Here we used Lemma \ref{Positivity lemma} and that there is no time evolution applied to $G_\tau$.
In the second case $\cal P$ does not contain any loops so we can use the proof of Lemma \ref{Closed path} to obtain the same bound as on the right-hand side of \eqref{Closed path bound}. The convergence \eqref{Closed path remark} in this context follows from the proof of Lemma \ref{Closed path}. The only difference that we have to note is that when $\abs{\cal V(\cal P)}=1$, the claim holds trivially since $\cal J_{\tau,e}=\hat{\cal J}_{\tau,e}$ whenever $e$ is a loop. 
\end{proof}
Furthermore, if $\cal P \in \conn(\cal E)$ is an open path, then it cannot have any loops. Hence the proof of Lemma \ref{Open path} implies that \eqref{Open path bound} and \eqref{Open path remark} still hold in this setting. Putting this observation together with Lemma \ref{Closed path 1D} and \eqref{Schatten class embedding}, and arguing as in the proof of Proposition \ref{Product of subgraphs}, we deduce Proposition \ref{Value of pairing 1D}.

Combining Proposition \ref{Value of pairing 1D} and \eqref{Tr G_tau 1D}, we deduce that Corollary \ref{cor:bound_am} holds in this setting.

Next, for $\Pi \in \fra R$ and $\f t \in \fra A$, we define 
$\cal I_{\Pi}^{\xi} (\f t)$ as in Definition \ref{def_calI}.
Analogously to \eqref{f_m,a_m^xi classical}, we define
\begin{equation} \label{f_m,a_m^xi classical 1D infty}
a_{\infty,m}^\xi \;\deq\; \frac{(-1)^m}{2^m \,  m!} \sum_{\Pi \in \fra R} \cal I^\xi_{\Pi}\,.
\end{equation}
A modification of the proof of Proposition \ref{Convergence of the explicit terms} shows that \eqref{Convergence of the explicit terms identity} still holds 
with $a^{\xi}_{\tau,m}$ given by \eqref{Wick_application lemma 1D identity} and $a^{\xi}_{\infty,m}$ given by  \eqref{f_m,a_m^xi classical 1D infty}. More precisely,  \eqref{hat I convergence} still holds by \eqref{Closed path remark} and \eqref{Open path remark} applied in this setting. In the telescoping argument, we need to take into consideration that $e$ can be a loop. If this is the case, we define $\tilde {\cal J}_{\tau,e}$ as in \eqref{tilde J tau e}.

The convergence \eqref{Convergence of the explicit terms identity} then follows by using the arguments in the discussion following \eqref{telescoping} with the estimates
\begin{equation}
\label{telescoping e=loop}
\begin{cases}
\|\tilde {\cal J}_{\tau,e}(\cdot,\f s)\|_{\fra S^1} \;\leq\;C &\mbox{if }e \neq e_0\\
\lim_{\tau \rightarrow \infty} \|\tilde {\cal J}_{\tau,e}(\cdot,\f s)\|_{\fra S^1}\;=\;0 &\mbox{if }e=e_0\,.
\end{cases}
\end{equation}
Indeed, if $e\neq e_0$, then $\|\tilde {\cal J}_{\tau,e}(\cdot,\f s)\|_{\fra S^1} \leq \|G_\tau\|_{\fra S^1} + \|G\|_{\fra S^1}$ and the bound in \eqref{telescoping e=loop} holds by \eqref{Tr G_tau 1D} and $\tr h^{-1} <\infty$. For $e=e_0$, we use $\|\tilde{\cal J}_{\tau,e}(\cdot,\f s)\|_{\fra S^1} = \|G_\tau-G\|_{\fra S^1}$, which converges to zero as $\tau \rightarrow \infty$, by spectral decomposition and dominated convergence.

Next, we estimate the remainder term. Owing to the absence of renormalization, we can in fact show that it is possible to estimate the remainder term in terms of the explicit term. The key ingredient here is the Feynman-Kac formula.

\begin{proposition}
\label{Remainder term bound 1D}
Given $\f t \in \fra A(M)$ with $\eta = 0$ and $\re z \geq 0$,  we let
\begin{equation}
\label{cal R 1}
\cal R (\f t,z) \;\deq\;
\Theta_\tau(\xi)\,\ee^{-(1-t_1) H_{\tau,0}} \,W_\tau \,
\ee^{-(t_1-t_2)H_{\tau,0}} \,W_{\tau} \,\cdots 
\,\ee^{-(t_{M-1} - t_M) H_{\tau,0}} \, W_{\tau}\,\, \ee^{-t_M (H_{\tau,0}+z W_{\tau})}\,.
\end{equation}
Then we have
\begin{equation}
\label{cal R 2}
\bigg|\frac{\tr \cal R (\f t,z)}{\tr (\ee^{-H_{\tau,0}})}\bigg| \;\leq\; (Cp)^{p} 
(C M)^{2M} \, \|w\|_{L^{\infty}}^M\,.
\end{equation}
\end{proposition}

\begin{proof}
Let $P^{(n)}$ be the orthogonal projection onto the subspace $\fra H^{(n)}$ of $\cal F$.
For an operator $\cal A$ on $\cal F$, we denote $\cal A^{(n)} \deq P^{(n)} \cal A P^{(n)}$. All of the operators $\cal A$ that we consider in the following leave the subspace $\fra H^{(n)}$ invariant, in which case we have $\cal A = \sum_{n \in \N} \cal A^{(n)}$. Define the operator
\begin{equation}
\label{cal S b t}
\cal S (\f t) \;\deq\;
\Theta_\tau(\xi)\,\ee^{-(1-t_1) H_{\tau,0}} \,W_\tau \,
\ee^{-(t_1-t_2)H_{\tau,0}} \,W_{\tau} \,\ee^{-(t_2 - t_3) H_{\tau,0}} \cdots \,\ee^{-(t_{M-1} - t_M) H_{\tau,0}}
\, W_{\tau}\,,
\end{equation}
so that we may write
\begin{equation}
\label{First-quantized notation}
\tr \cal R (\f t,z) \;=\; \sum_{n \in \N} \tr \big(\cal R (\f t,z)\big)^{(n)} \;=\; 
\sum_{n \in \N} \tr \big(\cal S (\f t) \, \ee^{-t_M (H_{\tau,0}+z W_{\tau})}\big)^{(n)}\,.
\end{equation}
For any $n \in \N$ we have
\begin{equation*}
\tr \big(\cal S (\f t) \,\ee^{-t_M (H_{\tau,0}+z W_{\tau})}\big)^{(n)}
\;=\;
\int_{\Lambda^n} \dd \f x \int_{\Lambda^n} \dd \f y\, \big(\cal S (\f t) \big)^{(n)}(\f y;\f x) \, \big( \ee^{-t_M (H_{\tau,0}+z W_{\tau})} \big)^{(n)} (\f x;\f y)\,.
\end{equation*}

We now estimate each of the factors $\big( \ee^{-t_M (H_{\tau,0}+z W_{\tau})}\big)^{(n)} (\f y; \f x)$ and $\big(\cal S (\f t) \big)^{(n)}(\f x;\f y)$ separately in absolute value. 

We first estimate $\big(\ee^{-t_M (H_{\tau,0}+z W_{\tau})}\big)^{(n)} (\f y; \f x)$.
Given $\f x \in \Lambda^n$, we
denote by $W_{\f x}$ the law of the $\Lambda^n$-valued Brownian motion $\f b$ starting at $\f x$ with variance $\int W_{\f x}(\dd b) \, (\f b(t) - \f x)^2 = 2t$ for $t > 0$. By the Feynman-Kac formula, we have
\begin{equation*}
\big( \ee^{-t_M (H_{\tau,0}+z W_{\tau})}\big)^{(n)} (\f x; \f y) \;=\; \int W_{\f x}(\dd \f b) \, \ee^{-\frac{\kappa n}{\tau}}\,\ee^{-\int_0^{t_M} \dd s\, z \, \big(\frac{1}{\tau^2} \sum_{1 \leq i < j \leq n} w_{ij}(\f b(s))\big)}\,\delta(\f b(t)-\f y)\,.
\end{equation*}
Here we used the identity
$(W_\tau)^{(n)}(\f u;\f v) = \frac{1}{\tau^2} \sum_{1 \leq i < j \leq n} w(u_i-u_j) \,\prod_{l=1}^{n} \delta(u_l-v_l)$ and the notation 
$w_{ij}(\f u) \;\deq\; w(u_i-u_j)$ for $\f u=(u_1,\ldots,u_n) \in \fra H^{(n)}$ and $1 \leq i < j \leq n$.
In other words, on $\fra H^{(n)}$ the operator $W_\tau$ acts as multiplication by $\frac{1}{\tau^2}\, \sum_{1 \leq i < j \leq n} w_{ij}(\f u)$.
Since $\re z \geq 0$ and $w \geq 0$, we have
\begin{equation}
\label{Kernel estimate 1}
\big|\big( \ee^{-t_M (H_{\tau,0}+z W_{\tau})} \big)^{(n)} (\f x; \f y) \big| \;\leq\; \int W_{\f x}(\dd \f b) \, \ee^{-\frac{\kappa n}{\tau}}\,\delta(\f b(t)-\f y)\;=\;\big( \ee^{-t_M H_{\tau,0}} \big)^{(n)} (\f x; \f y) \,,
\end{equation}
where in the second step we used the Feynman-Kac formula again.

Next, we estimate $\big(\cal S (\f t)\big)^{(n)}(\f y;\f x)$.
In particular, we note that
\begin{equation}
\label{Kernel estimate 2}
\big|\big(\cal S (\f t)\big)^{(n)}(\f y;\f x)\big| \;\leq\;
\Big(\Theta_\tau(\tilde \xi)\,\ee^{-(1-t_1) H_{\tau,0}} \,W_\tau \,
\ee^{-(t_1-t_2)H_{\tau,0}} \, W_{\tau} \, \cdots \,\ee^{-(t_{M-1} - t_M) H_{\tau,0}}
\, W_{\tau}\Big)^{(n)} (\f y; \f x)\,.
\end{equation}
Here $\tilde \xi$ is defined to be the operator whose integral kernel is equal to the absolute value of the kernel of $\xi$.
We note that 
\begin{equation}
\label{tilde xi}
\tilde \xi \in \fra B_p\,.
\end{equation}
The estimate \eqref{Kernel estimate 2} follows by writing
\begin{multline*}
\big(\cal S (\f t)\big)^{(n)}(\f y;\f x) \;=\;
\Bigg(\prod_{j=1}^{2M} \int_{\Lambda^n} \dd \f v_j \Bigg)\,
\big(\Theta_\tau(\xi)\big)^{(n)} (\f y;\f v_1) \,\big(\ee^{-(1-t_1) H_{\tau,0}} \big)^{(n)}(\f v_1;\f v_2) \, \big(W_\tau\big)^{(n)}(\f v_2;\f v_3) \,
\\
\big(\ee^{-(t_1-t_2)H_{\tau,0}}\big)^{(n)}(\f v_3;\f v_4) \, \cdots \,\big(\ee^{-(t_{M-1} - t_M) H_{\tau,0}}\big)^{(n)}(\f v_{2M-1};\f v_{2M}) \, \big(W_{\tau}\big)^{(n)} (\f v_{2M}; \f x)\,
\end{multline*}
and noting that 
$\big|\big(\Theta_\tau(\xi)\big)^{(n)}\big| \leq \big(\Theta_\tau(\tilde \xi)\big)^{(n)}$, $\big(\ee^{-(t_j-t_{j-1})H_{\tau,0}}\big)^{(n)} \geq 0$, and $\big(W_{\tau}\big)^{(n)} \geq 0$
in the sense of operator kernels. Substituting \eqref{Kernel estimate 1} and \eqref{Kernel estimate 2} into \eqref{First-quantized notation}, we obtain
\begin{equation*}
|\tr \cal R (\f t,z)| \;\leq\; \tr \Big(\Theta_\tau(\tilde \xi)\,\ee^{-(1-t_1) H_{\tau,0}} \,W_\tau \,
\ee^{-(t_1-t_2)H_{\tau,0}} \, W_{\tau} \, \cdots \,\ee^{-(t_{M-1} - t_M) H_{\tau,0}}
\, W_{\tau}\,\, \ee^{-t_M H_{\tau,0}}\Big)\,.
\end{equation*}
We now recall \eqref{tilde xi} and argue as in the proof of Proposition \ref{Value of pairing 1D} with $\xi$ replaced by $\tilde{\cal \xi}$. In addition, we sum over all pairings $\Pi \in \fra R$. There are at most $(2M+p)! \leq C^{2M+p} \,p^p \,M^{2M}$ such pairings. The claim now follows.
\end{proof}

Recalling the formula \eqref{Remainder term R} for the remainder term $R^\xi_{\tau,m}$ and using Proposition \ref{Remainder term bound 1D} we deduce that the following result holds in this setting.

\begin{corollary} \label{cor:remainder 1D}
For all $M \in \N$ and $\re z \geq 0$ we have $\absb{R^\xi_{\tau,M}(z)} \leq (Cp)^p \, \pb{C \norm{w}_{L^{\infty}} \abs{z}}^M M!$.
\end{corollary}
In addition, we note that Proposition \ref{Remainder term bound 1D} and the proof of Lemma \ref{lem:A_tau_anal} imply that $A_\tau^\xi$ is analytic in $\{z \col \re z > 0\}$.

We now consider the classical case. Let $a^\xi_m$ denote the coefficients of $\tilde \rho_z (\Theta(\xi))$, where $\tilde \rho_z$ is defined for $\re z \geq 0$ as in \eqref{def_rho_z} and where $W$ is given by \eqref{def_W}  for a pointwise nonnegative $w$. Furthermore, let $R^\xi_M(z)$ denote the remainder term in the Taylor expansion up to order $M$. A simple modification of the proof of Lemma \ref{lem_classical_coeff} yields that 
\eqref{lem_classical_coeff identity} holds in this setting. In particular, adapting the proof of Corollary \ref{cor:a_m_est} implies \eqref{cor:a_m_est bound} in this setting. Finally, the proofs of Lemmas \ref{lem:R_est} and \ref{lem:A_analytic} carry over and imply that \eqref{lem:R_est bound} holds and that $A^\xi$ is analytic for $\re z>0$.

We now finish the proof of Proposition \ref{1D proposition 1} by using the duality argument given in the proof of Theorem \ref{thm:main}. Note that combining all of the results of this subsection, it follows that the assumptions of Proposition \ref{Borel summation convergence} are satisfied with 
$\nu = (Cp)^p,\, \sigma = C\norm{w}_{L^{\infty}}$.
The result of Proposition \ref{1D proposition 1} now follows as in the proof of Theorem \ref{thm:main}; see Section \ref{sec:conclusion_proof}.

\subsection{Proof of Theorem \ref{thm:1D}, II: convergence of the traces}
\label{1D subsection 2}

\begin{proposition}
\label{1D proposition 2}
Under the assumptions of Theorem \ref{thm:1D} we have $\lim_{\tau \to \infty} \tr \gamma_{\tau,p}=\tr \gamma_p$ for all $p \in \mathbb{N}$.
\end{proposition}
This subsection is devoted to the proof of Proposition \ref{1D proposition 2}. This is done in several steps, using arguments similar to those used in the proof of Proposition \ref{1D proposition 1}. The main difference is that we have to work with a different type of observable, and, consequently, with a different graph structure from the one used in Subsection \ref{1D subsection 1}.
More precisely, we fix throughout this section $p \in \N$ and we define $\xi \equiv \xi_p$ to be the identity operator on $\fra H^{(p)}$, with integral kernel
\begin{equation}
\label{xi_p}
\xi (x_1,\ldots,x_p;y_1,\ldots,y_p) \;\deq\; \prod_{j=1}^{p} \delta(x_j-y_j)\,.
\end{equation}
Note that $\xi \notin \fra B_p$. Nevertheless, we can extend the definitions \eqref{def_A_tau} and \eqref{def_A} of $A_{\tau}^{\xi}$ and $A^{\xi}$ for $\xi$ as in \eqref{xi_p}. Note that here we again take $\eta=0$.
With this definition, we have
\begin{equation}
\label{A_tau A delta}
\tr \gamma_{\tau,p} \;=\; \Theta_{\tau}^{\xi}(1)\,,\qquad \tr \gamma_p\;=\; \Theta^{\xi}(1)\,.
\end{equation}
With the same notation as in the preceding subsection, we set up the Duhamel expansion of $A_{\tau,p}^{\xi}(z)$ and $A_{\tau}^{\xi}(z)$ in $\re z \geq 0$. We first focus on the quantum case. For fixed $m \in \N$ we work again with the set 
$\cal X \equiv \cal X(m,p)$ from Definition \ref{cal X} and the set of partitions $\fra R \equiv \fra R (m,p)$ from Definition \ref{def_pairing 1D}.

We need to modify the graph structure in Definition \ref{def_collapsed_graph} for the current setting.

\begin{definition} \label{def_collapsed_graph 1D delta}
Fix $m,p \in \N$. To each $\Pi \in \fra R$ we assign an edge-coloured undirected multigraph $(\tilde{\cal V}_\Pi, \tilde{\cal E}_\Pi, \tilde \sigma_\Pi) \equiv (\tilde{\cal V},\tilde{\cal E},\tilde \sigma)$, with a colouring $\tilde \sigma \col \cal E \to {\pm 1}$, as follows.
\begin{enumerate}
\item
On $\cal X$ we introduce the equivalence relation $\alpha \sim \beta$ if and only if $i_\alpha = i_\beta$ and $r_\alpha = r_\beta$. We define the vertex set $\tilde{\cal V} \deq \{[\alpha] \col \alpha \in \cal X\}$ as the set of equivalence classes of $\cal X$. 
 We use the notation $\tilde{\cal{V}} = \tilde{\cal{V}}_2 \cup \tilde{\cal{V}}_1$,
where
$\tilde{\cal V}_2 \deq \{(i,r)\col 1 \leq i \leq m, 1 \leq r \leq 2\}$
and
$\tilde{\cal V}_1 \deq \{(m+1,r)\col 1 \leq r \leq p\}$.

\item
The set $\tilde{\cal V}$ carries a total order $\leq$ inherited from $\cal X$: $[\alpha] \leq [\beta]$ whenever $\alpha \leq \beta$. Note that here we use the lexicographic order on $\cal X$ introduced in Section \ref{1D subsection 1}.
\item
For a pairing $\Pi \in \fra R$, each edge $(\alpha, \beta) \in \Pi$ gives rise to an edge $e = \{[\alpha], [\beta]\}$ of $\tilde{\cal E}$ with $\tilde \sigma(e) \deq \sign_\beta$. 
\item
We denote by $\conn(\tilde{\cal E})$ the set of connected components of $\tilde{\cal E}$, so that $\tilde{\cal E} = \bigsqcup_{\cal P \in \conn(\tilde{\cal E})} \cal P$. We call the connected components $\cal P$ of $\tilde{\cal E}$ \emph{paths}.
\end{enumerate}
\end{definition}
Note that in point (i) of Definition \eqref{def_collapsed_graph}, we in particular view $\alpha$ and $\beta$ as equivalent when $i_\alpha=i_\beta=m+1$ and $r_\alpha=r_\beta$. This is the main difference with the graph structure from Definition \ref{def_collapsed_graph}. By a slight abuse of notation, we denote both the equivalence relation from Definition \ref{def_collapsed_graph} and from Definition \ref{def_collapsed_graph 1D delta} by $\sim$. From context, it will be clear to which equivalence relation we are referring.
From Definitions \ref{def_pairing 1D} and \ref{def_collapsed_graph 1D delta}, we deduce that each vertex of $\tilde{\cal V}$ has degree 2. Therefore $\tilde{\cal V}$ factorizes as a product of closed paths. In particular, these closed paths can be loops.
Moreover, all elements of $\conn (\tilde{\cal E})$ are closed paths.

For the following we fix $m,p \in \N$ and a pairing $\Pi \in \fra R$, and let $(\tilde{\cal V},\tilde{\cal E},\tilde{\sigma})$ denote the associated graph from Definition \ref{def_collapsed_graph 1D delta}. Given $\f t \in \fra A \equiv \fra A(m)$, we appropriately modify Definition \ref{def:I_tau_PI} and define 
\begin{equation*}
\cal I_{\tau,\Pi}^{\xi}(\f t) \;\deq\; \int_{\Lambda^{\tilde{\cal V}}} \dd \f y \, \pBB{\prod_{i=1}^{m} w(y_{i,1}-y_{i,2})} \prod_{e \in \tilde{\cal E}} \cal J_{\tau, e}(\f y_e, \f s)\,.
\end{equation*}
Note that \eqref{Wick_application lemma 1D identity} holds in this setting.

With each $\f x=(x_{\alpha})_{\alpha \in \cal X} \in \Lambda^{\cal X}$ and $\f t = (t_{\alpha})_{\alpha \in \cal X} \in \fra A$ we associate integration labels $\f y=(y_{\va})_{\va \in \tilde{\cal V}}\in \Lambda^{\tilde{\cal V}}$ and time labels $\f s  = (s_{\va})_{\va \in \tilde{\cal V}} \in \R^{\tilde{\cal V}}$ as in \eqref{def_ys}, except that we now take the vertex set $\tilde{\cal V}$ and the equivalence relation $\sim$ to be the one from Definition \ref{def_collapsed_graph 1D delta}. In addition we adapt the splitting \eqref{y1_y2} to this context, where now $\f y_i \deq (y_a)_{a \in \tilde{\cal V}_i}$.
Given $\cal P \in \conn (\tilde{\cal E})$ we denote by $\tilde{\cal V}(\cal P)$ the set of vertices of $\cal P$. Moreover, we let $\tilde{\cal V}_i(\cal P) \deq \tilde{\cal V}(\cal P) \cap \tilde{\cal V}_i$ for $i=1,2$. Given $e \in \tilde{\cal E}$, we define $\cal J_{\tau,e}$ as in Definition \ref{J_e definition}. (Here we replace each occurrence of $\cal V$ and $\cal E$ in Definition \ref{J_e definition} by $\tilde{\cal V}$ and $\tilde{\cal E}$ respectively.)

\begin{lemma}
\label{Closed path 1D delta}
Suppose that $\cal P \in \conn(\tilde{\cal E})$. Then
\begin{equation}
\label{Closed path bound 1D delta}
\int_{\Lambda^{\tilde{\cal V}(\cal P)}} \prod_{\va \in \tilde{\cal V}(\cal P)} \dd y_{\va} \, \prod_{e \in \cal P} \cal J_{\tau, e}(\f y_e,\f s) \;\leq\; C^{\abs{\tilde{\cal V}(\cal P)}}\,.
\end{equation}
\end{lemma}
\begin{proof}
We consider two cases: (i) $\tilde{\cal V}(\cal P) \subset \tilde{\cal V}_2$ and (ii) $\tilde{\cal V}_1(\cal P) \neq \emptyset$. In case (i), all of the vertices in $\cal P$ belong to $\tilde{\cal V}_2$ and the claim follows by arguing as in the proof of Lemma \ref{Closed path 1D}.

Let us therefore focus on case (ii). If $\tilde{\cal V}(\cal P)=1$, then $\cal P$ is a loop in $\tilde{\cal V}_1$ and so the left-hand side of \eqref{Closed path bound 1D delta} equals $\|G_{\tau}\|_{\fra S^1}$, which satisfies the claimed bound.

We henceforth assume  $\tilde{\cal V}(\cal P) \geq 2$.
Since $\cal P$ is a closed path, there exists $l \in \N$ and distinct elements $b_1,\ldots,b_l \in \tilde{\cal V}_1$ such that
$\cal P = \bigsqcup_{j=1}^{l} \cal P_j$,
where for each $j=1,\ldots,l$ the path $\cal P_j$ is of the form
$\cal P_j = \{e^j_1,e^j_2,\ldots,e^j_{q_j}\}$
for some $q_j \in \N$ and edges $e^j_k$ such that
$b_j \in e^j_1, b_{j+1} \in e^j_{q_j}$
and
$e^j_k \cap e^j_{k+1} \in \tilde{\cal V}_2$,
for all $k=1,\ldots,q_j-1$. Here we set $b_{l+1} \deq b_1$. (It is possible to have $q_j=1$ in which case $\cal P_j$ is the path of length $1$ joining $b_j$ and $b_{j+1}$.)

We now write the left-hand side of \eqref{Closed path bound 1D delta} as
\begin{equation}
\label{product closed path delta}
\int_{\Lambda^l}\dd y_{b_1} \cdots \dd y_{b_l}  \Bigg(\prod_{j=1}^{l}\int_{\Lambda^{\tilde{\cal V}_2(\cal P_j)}} \prod_{\va \in \tilde{\cal V}_2(\cal P_j)} \dd y_{\va} \, \prod_{e \in \cal P_j} \cal J_{\tau, e}(\f y_e,\f s)\Bigg)\,.
\end{equation}
Arguing as in \eqref{Open path bound 2}-\eqref{Open path bound 3} we can get rid of all the time evolutions in the $j$-th factor of the integrand in \eqref{product closed path delta} for all $j=1,\ldots,l$. In particular, the proof of \eqref{Open path bound 4} (in the proof of Lemma \ref{Open path}) implies that the $j$-th factor is 
\begin{equation*}
\;\leq\; C^{\abs{\tilde{\cal V}_2(\cal P_j)}} \pb{1 + \|G_{\tau}\|_{\fra S^2(\fra H)}}^{\abs{\tilde{\cal V}_2(\cal P_j)}} \, \Big(\|G_{\tau}(y_{b_j};\cdot)\|_{\fra H}\, \,\|G_{\tau}(\cdot \, ; y_{b_{j+1}})\|_{\fra H}+ G_{\tau}(y_{b_j};y_{b_{j+1}})\Big)\,.
\end{equation*}
The estimate \eqref{Closed path bound 1D delta} now follows by applying the Cauchy-Schwarz inequality in $y_{b_1},\ldots,y_{b_l}$ and using Lemma \ref{lem:conv_G_tau}.
\end{proof}

From Lemma \ref{Closed path 1D delta} and \eqref{Tr G_tau 1D} we deduce that, for any $\Pi \in \fra R$, $\f t \in \fra A$, the estimate \eqref{Value of pairing 1D bound} holds when $\xi$ is given by \eqref{xi_p}. In particular, it follows that Corollary \ref{cor:bound_am} holds in this setting.
Next, for $\Pi \in \fra R$, we define 
\begin{equation*}
\cal I^\xi_{\Pi} \;\deq\; \int_{\Lambda^{\tilde{\cal V}}} \dd \f y \, \pBB{\prod_{i=1}^{m} w(y_{i,1}-y_{i,2})} \prod_{e \in \tilde{\cal E}} \cal J_{e}(\f y_e) \,.
\end{equation*}
We define $a_{\infty,m}^\xi$ as in \eqref{f_m,a_m^xi classical 1D infty}.
With this notation, \eqref{Convergence of the explicit terms identity} still holds in this setting by using the same telescoping proof adapted to the one-dimensional setting as in Section \ref{1D subsection 1}.
For the remainder term, we note that for $\cal R(\f t,z)$ by \eqref{cal R 1} we have that \eqref{cal R 2} holds even if we take $\xi$ as in \eqref{xi_p}. The proof is a minor modification of the proof of Proposition \ref{Remainder term bound 1D}. More precisely, we note that for $\cal S(\f t)$ given by \eqref{cal S b t}, \eqref{Kernel estimate 2} holds if we take $\tilde{\xi} \deq \xi$. Namely, in this case
$\big(\Theta_{\tau}(\xi)\big)^{(n)} \geq 0$ for all $n \in \N$ in the sense of operator kernels. The proof now proceeds as in Section \ref{1D subsection 1}. 
We hence deduce that Corollary \ref{cor:remainder 1D} holds in this setting. As in the previous subsection,
we deduce that the function $A_\tau^\xi$ is analytic in $\{z \col \re z > 0\}$

We now consider the classical case. The identity \eqref{lem_classical_coeff identity} holds when $\xi$ is given by \eqref{xi_p}. Arguing as earlier, we obtain \eqref{cor:a_m_est bound} in this setting.
As in the proof of Lemma \ref{lem:R_est}, we reduce the estimate of the remainder term to that of the explicit term by the trivial estimate $\abs{R^\xi_M(z)} \leq \frac{\abs{z}^M}{M!} \int \Theta(\xi) \, W^M \, \dd \mu$.
We deduce that Lemma \ref{lem:R_est} holds in this setting. Note that now we cannot use Cauchy-Schwarz as earlier.
Moreover, the same proof allows us to deduce that the function $A^\xi$ is analytic for $\re z > 0$.

We now conclude the proof of Proposition \ref{1D proposition 2}.
Combining all of the results of this subsection, it follows that the assumptions of Proposition \ref{Borel summation convergence} are satisfied with $\nu = (Cp)^p,\, \sigma = C\norm{w}_{L^{\infty}}$. Hence, it follows that, for $\re \,z > 0$ we have $A_{\tau}^{\xi}(z) \rightarrow A^{\xi}(z)$ as $\tau \rightarrow \infty$. In particular, setting $z=1$ and recalling \eqref{A_tau A delta}, we obtain Proposition \ref{1D proposition 2}.

\subsection{Proof of Theorem \ref{thm:1D}, III: conclusion}
\label{1D subsection 3}

In order to conclude the proof of Theorem \ref{thm:1D} we need the following lemma.

\begin{lemma}
\label{1D lemma 2}
Let $p \in \N$ be fixed. Suppose that, for all $\tau>0$, $\gamma_\tau \in \fra S^1(\fra H^{(p)})$ is positive and $\gamma \in \fra S^1(\fra H^{(p)})$ is positive. Moreover, suppose that 
\begin{equation*}
\txt{(i)} \quad \lim_{\tau \to \infty} \norm{\gamma_{\tau} - \gamma}_{\fra S^2(\fra H^{(p)})} \;=\; 0\,, \qquad
\txt{(ii)} \quad  \lim_{\tau \to \infty} \tr \gamma_{\tau} \;=\; \tr \gamma\,.
\end{equation*}
Then we have $\lim_{\tau \to \infty} \norm{\gamma_{\tau} - \gamma}_{\fra S^1(\fra H^{(p)})} = 0$.
\end{lemma}
\begin{proof}
The proof is analogous to arguments in \cite[Section 2]{Simon05}; see in particular the proof of \cite[Lemma 2.20]{Simon05}.
For simplicity of notation, throughout the proof we abbreviate $\fra S^2(\fra H^{(p)})$ and $\fra S^1(\fra H^{(p)})$ as $\fra S^2$ and $\fra S^1$ respectively. All of the operators we consider act on $\fra H^{(p)}$. By scaling, we can assume without loss of generality that $\tr \gamma = 1$.

Let $\epsilon>0$ be given.
Since $\gamma \in \fra S^1$, it is possible to find an operator $P$ of finite rank such that for 
$Q \deq \mathrm{I}-P$
we have
\begin{equation}
\label{1D lemma 2 A}
\tr Q \,\gamma\, Q \;\leq \;\epsilon\,.
\end{equation}
In particular, we can choose $P$ to be the projection onto the span of a sufficiently large set of eigenvectors of $\gamma$.

Using cyclicity of the trace and $P^2=P$, we obtain 
\begin{equation}
\label{1D lemma 2 B}
\tr Q(\gamma_\tau-\gamma)Q\;=\;\big(\tr \gamma_\tau-\tr \gamma\big) -\big(\tr P \gamma_\tau-\tr P \gamma \big)\,.
\end{equation}
By assumption (ii), the first term on the right-hand side of \eqref{1D lemma 2 B} converges to zero as $\tau \rightarrow \infty$. Moreover, by the Cauchy-Schwarz inequality we have
\begin{equation*}
\big|\tr P \gamma_\tau-\tr P \gamma\big| \;\leq\; \|P\|_{\fra S^2} \|\gamma_\tau-\gamma\|_{\fra S^2}\,,
\end{equation*}
which converges to zero as $\tau \rightarrow \infty$ by assumption (i).
By \eqref{1D lemma 2 A} and \eqref{1D lemma 2 B}, we deduce
\begin{equation}
\label{1D lemma 2 C}
\tr Q \,\gamma_\tau \, Q \;\leq \;2\epsilon\,
\end{equation} 
for $\tau$ sufficiently large. We henceforth consider such $\tau$. 

Since $\gamma$ and $\gamma_\tau$ are positive, we obtain from \eqref{1D lemma 2 A} and \eqref{1D lemma 2 C}
\begin{equation}
\label{1D lemma 2 D}
\|Q \,\gamma\, Q\|_{\fra S^1}\;\leq\;\epsilon\,,\qquad \|Q \,\gamma_{\tau}\, Q\|_{\fra S^1} \;\leq\; 2\epsilon\,.
\end{equation}
We estimate
\begin{equation}
\label{1D lemma 2 E}
\|\gamma-\gamma_{\tau}\|_{\fra S^1} \;\leq\; \|Q(\gamma-\gamma_{\tau})Q\|_{\fra S^1} 
+ \|P(\gamma-\gamma_{\tau})Q\|_{\fra S^1} + \|Q(\gamma-\gamma_{\tau})P\|_{\fra S^1} 
+ \|P(\gamma-\gamma_{\tau})P\|_{\fra S^1}\,.
\end{equation}
The first term on the right-hand side of \eqref{1D lemma 2 E} is bounded by $3\epsilon$ by \eqref{1D lemma 2 D}.
By the Cauchy-Schwarz inequality, the second and third term are 
\begin{equation*}
\;\leq\; \|Q(\gamma-\gamma_{\tau})\|_{\fra S^2} \|P\|_{\fra S^2} \;\leq\;  \|\gamma-\gamma_{\tau}\|_{\fra S^2} \|P\|_{\fra S^2} \,,
\end{equation*}
which converges to zero as $\tau \rightarrow \infty$ by assumption (i).
Similarly, the fourth term is
\begin{equation*}
\;\leq\;\|P(\gamma-\gamma_{\tau})\|_{\fra S^2} \|P\|_{\fra S^2} \;\leq\;  \|\gamma-\gamma_{\tau}\|_{\fra S^2} \|P\|_{\fra S^2}\,.
\end{equation*}
In particular, for $\tau$ sufficiently large we have $\|\gamma-\gamma_{\tau}\|_{\fra S^1} \leq 4\epsilon$.
Since $\epsilon$ is arbitrary, the result follows.
\end{proof}

Theorem \ref{thm:1D} now follows from Propositions \ref{1D proposition 1} and \ref{1D proposition 2} and Lemma \ref{1D lemma 2} with $\gamma_{\tau}=\gamma_{\tau,p}$ and $\gamma=\gamma_p$.

\subsection{Proof of Theorem \ref{thm:1D local}}

The proof of Theorem \ref{thm:1D local} is similar to that of Theorem \ref{thm:1D}, but since the interaction is no longer uniformly bounded, some modifications in the proof are needed. 
We first explain several changes in the setup of the problem. 
Instead of the nonlocal $W$ from \eqref{def_W}, we consider the local $W$ from \eqref{def_W_local}.
Instead of $W_\tau$ as in \eqref{W_tau 1D}, we consider 
\begin{equation} \label{W_tau 1D local}
W_\tau \;\deq\;
\frac{1}{2} \int \dd x \, \dd y \, \phi^*_\tau(x) \phi^*_\tau(y)  \, w_\tau(x - y) \, \phi_\tau(x) \phi_\tau(y)\,.
\end{equation}
For simplicity, throughout the proof we suppose that $\alpha = \lim_{\tau \to \infty} \int w_\tau$ is equal to $1$.

Let us first consider the quantum case.
For fixed $p \in \N$ we consider the observable $\xi$ that either belongs to $\fra B_p$ or is the identity $\xi_p = I$ on $\fra H^{(p)}$ with kernel \eqref{xi_p}.
Given such a $\xi$, we again perform a Taylor expansion up to order $M \in \N$ of 
$A_{\tau}^{\xi}(z) \deq \tilde \rho_{\tau,z}(\Theta_\tau(\xi))$
in the parameter $z$ with $\re z \geq 0$. Here $\tilde \rho_{\tau,z} \equiv \tilde \rho_{\tau,z}^0$ is defined as in \eqref{def_rho_tau_tilde} with $W_\tau$ given by \eqref{W_tau 1D local} and $\eta=0$. The coefficients $a^\xi_{\tau,m}$ of the expansion are given by \eqref{Explicit term a} and the remainder term $R^\xi_{\tau,M}(z)$ by \eqref{Remainder term R}, where $W_\tau$ is given by \eqref{W_tau 1D local} and $\eta=0$. Furthermore, given $m \in \N$ and $\Pi \in \fra R\equiv \fra R(m)$ we define $\cal I_{\tau,\Pi}^{\xi}(\f t)$ at $\f t \in \fra A \equiv \fra A(m)$ by
\begin{equation}
\label{I_rep local1}
\cal I_{\tau,\Pi}^{\xi}(\f t) \;\deq\; \int_{\Lambda^{\cal V}} \dd \f y \, \pBB{\prod_{i=1}^{m} w_\tau(y_{i,1}-y_{i,2})} \xi(\f y_1) \prod_{e \in \cal E} \cal J_{\tau, e}(\f y_e, \f s)
\end{equation}
if $\xi \in \fra B_p$ and by
\begin{equation}
\label{I_rep local2}
\cal I_{\tau,\Pi}^{\xi}(\f t) \;\deq\; \int_{\Lambda^{\tilde{\cal V}}} \dd \f y \, \pBB{\prod_{i=1}^{m} w_\tau(y_{i,1}-y_{i,2})} \prod_{e \in \tilde{\cal E}} \cal J_{\tau, e}(\f y_e, \f s)
\end{equation}
if $\xi=\xi_p$.
The order $\leq$ (used in the definition of $\sigma(\cdot)$) on the sets $\cal V$ and $\tilde {\cal V}$ is the lexicographical order from Sections \ref{1D subsection 1}--\ref{1D subsection 2}.
The order $\leq$ (used in the definition of $\sigma(\cdot)$) on the sets $\cal V$ and $\tilde {\cal V}$ is the lexicographical order from Sections \ref{1D subsection 1}--\ref{1D subsection 2}.
Note that, as in Sections \ref{1D subsection 1} and \ref{1D subsection 2}, we take $\eta=0$ in the definition of $\fra A(m)$.
Arguing as in the proof of Corollary \ref{cor:bound_am}, the needed upper bounds on the explicit terms are a consequence of the following result.

Throughout the following, we use positive constants $\CP,\CA,\CB > 0$ that may only depend on $\kappa$ and the constant $C$ from \eqref{w_tau bound}.
\begin{proposition}
\label{Product of subgraphs local}
For any $\Pi \in \fra R$, $\f t \in \fra A$ and for $\cal I_{\tau,\Pi}^{\xi}(\f t)$ as defined in \eqref{I_rep local1}--\eqref{I_rep local2} we have $\big| \cal I_{\tau,\Pi}^{\xi} (\f t) \big| \leq \CP^{m+p}$.
\end{proposition}
Before proving Proposition \ref{Product of subgraphs local} we introduce some basic tools. Recalling $S_{\tau,t}$ and $G_{\tau,t}$ given by \eqref{def_St} and \eqref{def_Gt}, we define $Q_{\tau,t}$ for $t \in (-1,1)$ by
\begin{align} \label{def_Qt}
Q_{\tau,t} \;\deq\;
\begin{cases}
G_{\tau,t} + \frac{1}{\tau}S_{\tau,t} &\mbox{if }t \in (0,1)\\
G_{\tau,t} &\mbox{if }t \in (-1,0]\,.
\end{cases}
\end{align}
Note that for $t \in (-1,0)$ we can write
\begin{equation}
\label{def_QtB}
Q_{\tau,t} \;=\; \frac{\ee^{-(t+1)h/\tau}}{\tau (\ee^{h/\tau}-1)}+\frac{1}{\tau}\,\ee^{-(t+1)h/\tau} \;=\; \frac{\ee^{- \{t\}h/\tau}}{\tau (\ee^{h/\tau}-1)}+\frac{1}{\tau}\,\ee^{-\{t\}h/\tau}\,
\end{equation}
where $\{x\} \;\deq\; x -\lfloor x \rfloor \in [0,1)$ denotes the fractional part of $x$. We adopt this notation from now on. 
From \eqref{def_Qt} and \eqref{def_QtB} we note for $t \in (-1,1)$ the splitting
\begin{equation}
\label{Q_splitting}
Q_{\tau,t}=Q_{\tau,t}^{(1)}+\frac{1}{\tau}\,Q_{\tau,t}^{(2)}\,,
\end{equation} 
where
\begin{equation} \label{def_Qt_split}
Q_{\tau,t}^{(1)} \;\deq\;\frac{\ee^{- \{t\}h/\tau}}{\tau (\ee^{h/\tau}-1)} \,, \qquad Q_{\tau,t}^{(2)} \;\deq\;
\begin{cases}
\ee^{-\{t\}h/\tau} &\mbox{if }t \in (-1,1) \setminus \{0\} \\
0 &\mbox{if }t=0\,.
\end{cases}
\end{equation}
A splitting of the form \eqref{Q_splitting} valid for all $t \in (-1,1)$ is a key ingredient of our proof. From Lemma \ref{Positivity lemma} it follows that for all $t \in (-1,1)$ and $x,y \in \Lambda$ we have
\begin{equation}
\label{Q positive}
Q_{\tau,t}^{(1)}(x;y) \;=\; Q_{\tau,t}^{(1)}(y;x) \;\geq\;0\,,\qquad
Q_{\tau,t}^{(2)}(x;y) \;=\; Q_{\tau,t}^{(2)}(y;x) \;\geq\;0\,.
\end{equation}

We note the following result.
\begin{lemma}
\label{lem:kernel bounds} With $Q_{\tau,t}^{(j)}$ given by \eqref{def_Qt_split} and $Q_{\tau,t}$ given by \eqref{def_Qt} we have the following estimates.
\begin{itemize}
\item[(i)] $\|Q_{\tau,t}^{(1)}\|_{L^{\infty}(\Lambda^2)} \leq \CA$ for all $t \in (-1,1)$.
\item[(ii)] $\int \dd y \, Q_{\tau,t}^{(2)}(x;y) = \int \dd y \, Q_{\tau,t}^{(2)}(y;x) \leq 1$ for all $t \in (-1,1)$ and $x \in \Lambda$.
\item[(iii)] $\|Q_{\tau,t}^{(2)}\|_{L^{\infty}(\Lambda^2)} \leq \CA \tau$ if $\{t\} \geq 1/4$.
\item[(iv)] $Q_{\tau,t}(x;y) \geq \CB$ for all $t \in (-1,1)$ and $x,y \in \Lambda$.
\end{itemize}
\end{lemma}
\begin{proof}
Since we are working in one dimension and since $v = 0$, we change our earlier indexing convention and we index the eigenvalues of $h$ as $\lambda_k = (2 \pi k)^2 + \kappa$, where $k \in \Z$.

We first prove (i). Note that
\begin{equation*}
\|Q_{\tau,t}^{(1)}\|_{L^{\infty}(\Lambda^2)} \;=\; \Bigg\|\sum_{k \in \Z }  \frac{\ee^{-\{t\} \lambda_k/\tau}}{\tau(\ee^{\lambda_k/\tau}-1)}\,\ee^{2\pi i k (x-y)} \Bigg\|_{L^{\infty}(\Lambda^2)} \;\leq\; \sum_{k \in \Z }  \frac{\ee^{-\{t\} \lambda_k/\tau}}{\tau(\ee^{\lambda_k/\tau}-1)}
\;\leq\; \sum_{k \in \Z} \frac{1}{\lambda_k} \;<\;\infty\,.
\end{equation*}
Part (iii) is proved in the same way.

We now prove (ii). If $t=0$ the claim immediately holds since $Q^{(2)}_{\tau,0} \equiv 0$. We therefore consider $t \in (-1,1) \setminus \{0\}$. The two integrals in the claim are indeed equal by \eqref{Q positive}.
Note that
\begin{equation*}
\int \dd y \, Q_{\tau,t}^{(2)}(x;y) \;=\; \int \dd y\, \sum_{k \in \Z} \ee^{-\{t\}\lambda_k /\tau} \ee^{2\pi i k(x-y)} \;=\; \sum_{k \in \Z} \int \dd y\,  \ee^{-\{t\}\lambda_k /\tau} \ee^{2\pi i k(x-y)} \;=\, \ee^{-\{t\}\lambda_0 /\tau} \;\leq 1\,.
\end{equation*}
The interchanging of the integration in $y$ and the summation in $k$ is justified by the decay of $\ee^{-\{t\}\lambda_k /\tau}$.

Finally we prove (iv). It suffices to prove that $G_{\tau,t}(x;y) \geq \CB$. First, we observe that there exists a constant $c > 0$ such that for all $s \geq 1$ and $x,y \in \Lambda$ we have $\ee^{s \Delta}(x;y) \geq c$; this can be proved either using the Feynman-Kac formula or the Poisson summation formula. The claim (iv) then easily follows from the Neumann series representation \eqref{neumann_series} and Lemma \ref{Positivity lemma}, which yield
\begin{equation*}
G_{\tau,t}(x;y) \;\geq\; \frac{1}{\tau} \sum_{n = \tau + 1}^{2\tau + 1} \ee^{-(t + n) \kappa / \tau} \ee^{(t + n) \Delta / \tau}(x;y)\,. \qedhere
\end{equation*}
\end{proof}

We now have all the necessary ingredients to prove Proposition \ref{Product of subgraphs local}.

\begin{proof}[Proof of Proposition \ref{Product of subgraphs local}]
Let us first consider the case when $\xi \in \fra B_p$.
We decompose the multigraph $\cal E_\Pi \equiv \cal E$ into paths. We order these paths as $\cal P_1,\ldots,\cal P_k$ in an arbitrary fashion.
Using the nonnegativity of $\cal J_{\tau, e}$
we obtain the estimate
\begin{equation}
\label{Product of subgraphs local 1}
\big|\cal I_{\tau,\Pi}^{\xi}(\f t)\big| \;\leq\; \int_{\Lambda^{\cal V}} \dd \f y \, \pBB{\prod_{i=1}^{m} w_\tau(y_{i,1}-y_{i,2}) }\big|\xi(\f y_1)\big| \pBB{\prod_{j=1}^{k} \, \prod_{e \in \cal P_j} \cal J_{\tau, e}(\f y_e, \f s)}\,.
\end{equation}
Note that, due to the structure of the interaction we cannot automatically decouple the paths in \eqref{Product of subgraphs local 1} by arguing as in \eqref{Wick_application 2 bound B}-\eqref{Wick_application 2 bound C}. Indeed, arguing in this way would give us a factor of $\|w_\tau\|_{L^\infty}^m$, which we only know is bounded by $(C \tau)^m$. Such a rough upper bound is not affordable. Therefore we need to carefully distribute the factors of $w_\tau(y_{i,1}-y_{i,2})$ among the paths.

Before we proceed we introduce some notation.
Given $a=(i_a,r_a) \in \cal V_2$ we define $a^* \in \cal V_2$ by $a^* \deq (i_a,3-r_a)$. Since $r_a = 1,2$ for all $a \in \cal V_2$, we have $(a^{*})^{*}=a$. Note that we always have $a \neq a^*$.
Moreover we define for $a \in \cal V_2$ the function $\cal W_\tau^{a}$ by
\begin{align} \label{cal W}
\cal W_\tau^{a} \;\deq\;
\begin{cases}
w_\tau &\mbox{if }a \in \cal P_j,\,a^{*} \in \cal P_l \,\, \mbox{for some } 1 \leq j<l \leq k \\
1 &\mbox{if }a \in \cal P_j,\,a^{*} \in \cal P_l \,\, \mbox{for some } 1 \leq l<j \leq k \\
w_\tau &\mbox{if }a,a^* \in \cal P_j \,\, \mbox{for some } 1 \leq j \leq k \,\, \mbox{and } r_a=1 \\
1 &\mbox{if }a,a^* \in \cal P_j \,\, \mbox{for some } 1 \leq j \leq k \,\, \mbox{and } r_a=2\,.
\end{cases}
\end{align}
In particular, we can write
\begin{equation}
\label{product w tau}
\prod_{i=1}^{m} w_\tau(y_{i,1}-y_{i,2}) \;=\; \prod_{a \in \cal V_2} \cal W_\tau^{a}(y_{a}-y_{a^*})\,.
\end{equation}
Moreover, from \eqref{w_tau bound} we deduce
\begin{equation}
\label{cal W bound}
\|\cal W_{\tau}^a\|_{L^1} \;\leq\; C \,,\qquad
\|\cal W_{\tau}^a\|_{L^{\infty}} \;\leq\; C\tau\,.
\end{equation}
Given $\cal P \in \conn (\cal E)$ we define $\cal V^*(\cal P) \deq\cal \{a^{*}\col a \in \cal V_2 (\cal P) \} \setminus \cal V_2(\cal P)$.
Graphically, $\cal V^*(\cal P)$ corresponds to all of the vertices in $\cal V$ which are connected to some vertex in $\cal P$ by an interaction $w_\tau$, but which do not belong to $\cal V_2(\cal P)$.

We shall integrate out the variables $\f y_2$ by integrating successively the $\f y_2$-variables in the paths $\cal P_1, \dots, \cal P_k$.
Given $1 \leq j \leq k+1$, we define $\cal V_{2,j} \deq \cal V_2 \,\big\backslash \,\big(\bigcup_{l=1}^{j-1} \cal V_2(\cal P_l)\big)$,
which indexes the variables left after having integrated out the paths $\cal P_1, \dots, \cal P_{j-1}$.
In particular, $\cal V_{2,1} = \cal V_2$ and $\cal V_{2,k+1}= \emptyset$.
For a subset $A \subset \cal V$ we abbreviate $\f y_A \deq (y_a)_{a \in A}$ and for $r \in [1,\infty]$ we write $L^r_A$ for the space $L^r_{\f y_A}$ with the corresponding norm. We also use the splitting $\f y=(\f y_1,\f y_2)$, which, in the above notation, is shorthand for $\f y=(\f y_{\cal V_1},\f y_{\cal V_2})$.

Using \eqref{product w tau} we note that the right-hand side of \eqref{Product of subgraphs local 1} equals
\begin{multline}
\label{Product of subgraphs local 2}
\int_{\Lambda^{\cal V_1}} \dd \f y_1 \, \big|\xi(\f y_1)\big| \int_{\Lambda^{\cal V_2}} \dd \f y_2 \,\prod_{j=1}^{k} \pBB{ \prod_{e \in \cal P_j} \cal J_{\tau, e}(\f y_e, \f s) \, \prod_{a \in \cal V_2(\cal P_j)} \cal W_\tau^{a}(y_{a}-y_{a^*})}
\\
\;\leq\;
\Bigg\|\int_{\Lambda^{\cal V_2}} \dd \f y_2 \, \prod_{j=1}^{k} \pBB{\prod_{e \in \cal P_j} \cal J_{\tau, e}(\f y_e, \f s)\, \prod_{a \in \cal V_2(\cal P_j)} \cal W_\tau^{a}(y_{a}-y_{a^*})} \Bigg\|_{L^{\infty}_{\f y_1}}\,.
\end{multline}
Here we used an $L^1$-$L^{\infty}$-H\"{o}lder inequality in $\f y_1$ and $\|\xi\|_{L^1_{\f y_1}} \leq \|\xi\|_{L^2_{\f y_1}} \leq 1.$

We shall show that for all $1 \leq l \leq k$
\begin{multline}
\label{Inductive inequality 1}
\Bigg\|\int_{\Lambda^{\cal V_{2,l}}} \dd \f y_{\cal V_{2,l}} \, \prod_{j=l}^{k} \pBB{\prod_{e \in \cal P_j} \cal J_{\tau, e}(\f y_e, \f s)\, \prod_{a \in \cal V_2(\cal P_j)} \cal W_\tau^{a}(y_{a}-y_{a^*})} \Bigg\|_{L^{\infty}_{\f y_1}L^{\infty}_{\cal V_2 \setminus \cal V_{2,l}}}
\\
\;\leq\;
\CP^{\,\abs{\cal V(\cal P_l)}} \,
\Bigg\|\int_{\Lambda^{\cal V_{2,l+1}}} \dd \f y_{\cal V_{2,l+1}} \, \prod_{j=l+1}^{k} \pBB{\prod_{e \in \cal P_j} \cal J_{\tau, e}(\f y_e, \f s)\, \prod_{a \in \cal V_2(\cal P_j)} \cal W_\tau^{a}(y_{a}-y_{a^*})} \Bigg\|_{L^{\infty}_{\f y_1}L^{\infty}_{\cal V_2 \setminus \cal V_{2,l+1}}} \,.
\end{multline}
The claim of Proposition \ref{Product of subgraphs local} follows from \eqref{Product of subgraphs local 2} by iteratively applying \eqref{Inductive inequality 1}.

What remains is the proof of \eqref{Inductive inequality 1}. Fix $1 \leq l \leq k$. 
Note that 
\begin{equation}
\label{cal V 2 l}
\cal V_{2,l}\;=\;\cal V_2(\cal P_l) \sqcup \cal V_{2,l+1} \;=\;\cal V_2(\cal P_l) \sqcup \big(\cal V_{2,l+1} \cap \cal V^{*}(\cal P_l)\big) \sqcup \big(\cal V_{2,l+1} \setminus \cal V^{*}(\cal P_l)\big)\,.
\end{equation}
Using \eqref{cal V 2 l} we rewrite the left-hand side of \eqref{Inductive inequality 1} as
\begin{multline}
\label{L 1 L infty A}
\Bigg\|\int_{\Lambda^{\cal V_{2,l+1} \cap \cal V^{*}(\cal P_l)}} \dd \f y_{\cal V_{2,l+1} \cap \cal V^{*}(\cal P_l)} \,
\int_{\Lambda^{\cal V_2(\cal P_l)}} \dd \f y_{\cal V_2(\cal P_l)} \,\prod_{e \in \cal P_l} \cal J_{\tau, e}(\f y_e, \f s)\, \prod_{a \in \cal V_2(\cal P_l)} \cal W_\tau^{a}(y_{a}-y_{a^*}) 
\\
\Bigg\{\int_{\Lambda^{\cal V_{2,l+1} \setminus \cal V^{*}(\cal P_l)}} \dd \f y_{\cal V_{2,l+1} \setminus \cal V^{*}(\cal P_l)} \,
\prod_{j=l+1}^{k} \pBB{\prod_{e \in \cal P_j} \cal J_{\tau, e}(\f y_e, \f s)\, \prod_{a \in \cal V_2(\cal P_j)} \cal W_\tau^{a}(y_{a}-y_{a^*})} \Bigg\} \Bigg\|_{L^{\infty}_{\f y_1}L^{\infty}_{\cal V_2 \setminus \cal V_{2,l}}}\,.
\end{multline}
We first integrate in $\f y_{\cal V_2(\cal P_l)}$.  In doing so we note that, by construction, the 
$\f y_{\cal V_{2,l+1} \setminus \cal V^{*}(\cal P_l)}$ integral in \eqref{L 1 L infty A} does not depend on $\f y_{\cal V_2(\cal P_l)}$.
After this we use an $L^\infty$-$L^1$-H\"{o}lder inequality in $\f y_{\cal V_{2,l+1} \cap \cal V^{*}(\cal P_l)}$. Finally we take the supremum in the remaining variables. We therefore deduce that \eqref{L 1 L infty A} is
\begin{multline}
\label{L1 Linfty}
\;\leq\;
\Bigg\|\prod_{e \in \cal P_l} \cal J_{\tau, e}(\f y_e, \f s)\, \prod_{a \in \cal V_2(\cal P_l)} \cal W_\tau^{a}(y_{a}-y_{a^*})\Bigg\|_{L^{\infty}_{\f y_1}L^{\infty}_{\cal V^{*}(\cal P_l)}L^1_{\cal V_2(\cal P_l)}}
\\
\times 
\Bigg\|\int_{\Lambda^{\cal V_{2,l+1}}} \dd \f y_{\cal V_{2,l+1}} \, \prod_{j=l+1}^{k} \pBB{\prod_{e \in \cal P_j} \cal J_{\tau, e}(\f y_e, \f s)\, \prod_{a \in \cal V_2(\cal P_j)} \cal W_\tau^{a}(y_{a}-y_{a^*})} \Bigg\|_{L^{\infty}_{\f y_1}L^{\infty}_{\cal V_2 \setminus \cal V_{2,l+1}}} \,.
\end{multline}
For the first factor we used that the function which we are estimating depends only on 
$\f y_1, \f y_{\cal V_2(\cal P_l) \cup \cal V^{*}(\cal P_l)}$.
From \eqref{L1 Linfty} we note that \eqref{Inductive inequality 1} follows 
if we prove
\begin{equation}
\label{Inductive inequality 2}
\Bigg\|\prod_{e \in \cal P_l} \cal J_{\tau, e}(\f y_e, \f s)\, \prod_{a \in \cal V_2(\cal P_l)} \cal W_\tau^{a}(y_{a}-y_{a^*})\Bigg\|_{L^{\infty}_{\f y_1}L^{\infty}_{\cal V^{*}(\cal P_l)}L^1_{\cal V_2(\cal P_l)}} \;\leq\;\CP^{\,\abs{\cal V(\cal P_l)}}\,.
\end{equation}
We now prove \eqref{Inductive inequality 2} by induction on $n \deq \abs{\cal V(\cal P_l)}$.
In the sequel we use the same notation for the vertices, edges, and associated times as in the proof of Lemma \ref{Closed path} if $\cal P_l$ is a closed path and as in the proof of Lemma \ref{Open path} in $\cal P_l$ is an open path.
Suppose that $\cal P_l$ is a closed path. For $1 \leq j \leq n$, we write the time associated with the edge $e_j$ as $\zeta_j \deq \sigma(e_1) (s_{a_j}-s_{a_{j+1}})$.
The sum of the times in $\cal P_l$ is 
\begin{equation}
\label{Path condition 1}
\sum_{j=1}^{n} \zeta_j \;=\; \sum_{j=1}^{n} \sigma(e_1) (s_{a_j}-s_{a_{j+1}}) \;=\;0\,.
\end{equation}
Moreover, for all $1 \leq i \leq n$ and $1 \leq q \leq n-1$ we have 
\begin{equation}
\label{Path condition 2}
\sum_{j=i}^{i+q} \zeta_j \;=\; \sum_{j=i}^{i+q} \sigma(e_1) (s_{a_j}-s_{a_{j+1}}) \;=\; \sigma(e_1) (s_{a_i}-s_{a_{i+q+1}}) \in (-1,1)\,.
\end{equation}
In the above calculations all the indices are understood to be modulo $n$. Analogous results with obvious modifications hold for open paths $\cal P_l$ as well. We use these identities tacitly throughout the following.

Our goal is to estimate $\fra I(\cal P_l)$, where
\begin{equation*}
\fra I(\cal P) \;\deq\; \int \dd \f y_{\cal V_2(\cal P)} \, \prod_{e \in \cal P} \cal J_{\tau, e}(\f y_e, \f s)\, \prod_{a \in \cal V_2(\cal P)} \cal W_\tau^{a}(y_{a}-y_{a^*})\,,
\end{equation*}
which in general is a function of $(\f y_1, \f y_{\cal V^*(\cal P)})$.

\subsubsection*{Induction base: $n=1$ and $n=2$}
If $n=1$ then $\cal P$ is a loop at $a_1 \in \cal V_2$. In particular, $a_1 \neq a_1^*$ and $\zeta_1 = 0$, so that
\begin{equation*}
\fra I (\cal P) \;=\; \int \dd y_{a_1} \, \cal W_\tau^{a_1} (y_{a_1}-y_{a_1^*}) \, Q_{\tau,0}(y_{a_1};y_{a_1}) \;\leq\; \|Q_{\tau,0}\|_{L^{\infty}(\Lambda^2)} \, \|\cal W_\tau^{a_1}\|_{L^1(\Lambda)}\;\leq \;C\,.
\end{equation*}
Here we used Lemma \ref{lem:kernel bounds} (i) and \eqref{cal W bound}.

If $n=2$ we consider three cases.
\begin{itemize}
\item[(A)] $\cal P$ is a closed path with vertices $a_1,a_2 \in \cal V_2$ satisfying $a_2=a_1^{*}$. In this case $\zeta_{1}=\zeta_{2}=0$. Let us assume without loss of generality that $r_{a_1}=1$. Then, by \eqref{cal W} we have $\cal W_{\tau}^{a_2} = 1$. Consequently
\begin{equation*}
\fra I (\cal P) \;=\;\int \dd y_{a_1} \,\dd y_{a_2} \,\cal W_{\tau}^{a_1}(y_{a_1}-y_{a_2})\,Q_{\tau,0}(y_{a_1};y_{a_2})\,Q_{\tau,0}(y_{a_2};y_{a_1})\;\leq\; \|Q_{\tau,0}\|_{L^{\infty}(\Lambda^2)}^2 \|\cal W_{\tau}^{a_1}\|_{L^1(\Lambda)}\,.
\end{equation*}
By Lemma \ref{lem:kernel bounds} (i) and \eqref{cal W bound} this is bounded by a constant.

\item[(B)] $\cal P$ is a closed path with vertices $a_1,a_2 \in \cal V_2$, where $a_2 \neq a_1^{*}$.

In this case $\zeta_1 = \sigma(e_1) (s_{a_1}-s_{a_2}) \in (-1,1) \setminus \{0\}$, and $\zeta_2 = -\zeta_1$. It follows that 
\begin{equation*}
\fra I (\cal P)\;=\; \int \dd y_{a_1}\,\dd y_{a_2}\,\cal W_{\tau}^{a_1}(y_{a_1}-y_{a_1^{*}})\,\cal W_{\tau}^{a_2}(y_{a_2}-y_{a_2^{*}})\,Q_{\tau,\zeta_1}(y_{a_1};y_{a_2})\,Q_{\tau,-\zeta_1}(y_{a_2};y_{a_1})\,.
\end{equation*}
We now apply the decomposition \eqref{Q_splitting} to $Q_{\tau,\zeta_1}$ and $Q_{\tau,-\zeta_1}$ and write
\begin{align*}
\fra I (\cal P) &\;=\; \int \dd y_{a_1}\,\dd y_{a_2}\,\cal W_{\tau}^{a_1}(y_{a_1}-y_{a_1^{*}})\,\cal W_{\tau}^{a_2}(y_{a_2}-y_{a_2^{*}})\,Q^{(1)}_{\tau,\zeta_1}(y_{a_1};y_{a_2})\,Q^{(1)}_{\tau,-\zeta_1}(y_{a_2};y_{a_1})
\\ 
&+ \frac{1}{\tau}\,\int \dd y_{a_1}\,\dd y_{a_2}\,\cal W_{\tau}^{a_1}(y_{a_1}-y_{a_1^{*}})\,\cal W_{\tau}^{a_2}(y_{a_2}-y_{a_2^{*}})\,Q^{(2)}_{\tau,\zeta_1}(y_{a_1};y_{a_2})\,Q^{(1)}_{\tau,-\zeta_1}(y_{a_2};y_{a_1})
\\
&+\frac{1}{\tau}\,\int \dd y_{a_1}\,\dd y_{a_2}\,\cal W_{\tau}^{a_1}(y_{a_1}-y_{a_1^{*}})\,\cal W_{\tau}^{a_2}(y_{a_2}-y_{a_2^{*}})\,Q^{(1)}_{\tau,\zeta_1}(y_{a_1};y_{a_2})\,Q^{(2)}_{\tau,-\zeta_1}(y_{a_2};y_{a_1})
\\
&+\frac{1}{\tau^2}\,\int \dd y_{a_1}\,\dd y_{a_2}\, \cal W_{\tau}^{a_1}(y_{a_1}-y_{a_1^{*}})\,\cal W_{\tau}^{a_2}(y_{a_2}-y_{a_2^{*}})\,Q^{(2)}_{\tau,\zeta_1}(y_{a_1};y_{a_2})\,Q^{(2)}_{\tau,-\zeta_1}(y_{a_2};y_{a_1})\,.
\end{align*}
We need to estimate each term separately. 

The first term is
\begin{multline*}
\;\leq\; \|Q^{(1)}_{\tau,\zeta_1}\|_{L^{\infty}(\Lambda^2)} \, \|Q^{(1)}_{\tau,-\zeta_1}\|_{L^{\infty}(\Lambda^2)} \, \int \dd y_{a_1}\,\dd y_{a_2}\,\cal W_{\tau}^{a_1}(y_{a_1}-y_{a_1^{*}})\,\cal W_{\tau}^{a_2}(y_{a_2}-y_{a_2^{*}})
\\
\;=\; \|Q^{(1)}_{\tau,\zeta_1}\|_{L^{\infty}(\Lambda^2)} \, \|Q^{(1)}_{\tau,-\zeta_1}\|_{L^{\infty}(\Lambda^2)} \, \|\cal W_{\tau}^{a_1}\|_{L^1(\Lambda)} \, \|\cal W_{\tau}^{a_2}\|_{L^1(\Lambda)}\,,
\end{multline*}
which is bounded by Lemma \ref{lem:kernel bounds} (i) and \eqref{cal W bound}.

The second term is estimated, using Lemma \ref{lem:kernel bounds} (ii), as
\begin{align*}
&\leq\;
\frac{1}{\tau}\,\|\cal W_\tau^{a_2}\|_{L^{\infty}(\Lambda)} \, \|Q^{(1)}_{\tau,-\zeta_1}\|_{L^{\infty}(\Lambda^2)} \, \int \dd y_{a_1}\, \cal W_\tau^{a_1}(y_{a_1}-y_{a_1^{*}})\int \dd y_{a_2}\,Q^{(2)}_{\tau,\zeta_1}(y_{a_1};y_{a_2})
\\
&\leq\; \frac{1}{\tau}\,\|\cal W_\tau^{a_2}\|_{L^{\infty}(\Lambda)} \, \|Q^{(1)}_{\tau,-\zeta_1}\|_{L^{\infty}(\Lambda^2)} \, \int \dd y_{a_1}\, \cal W_\tau^{a_1}(y_{a_1}-y_{a_1^{*}}) 
\\
&=\;
\frac{1}{\tau}\,\|\cal W_\tau^{a_2}\|_{L^{\infty}(\Lambda)} \, \|Q^{(1)}_{\tau,-\zeta_1}\|_{L^{\infty}(\Lambda^2)} \, \|\cal W_\tau^{a_1}\|_{L^1(\Lambda)}\,,
\end{align*}
which is bounded by a constant, by Lemma \ref{lem:kernel bounds} (i) and \eqref{cal W bound}.
The third term is estimated analogously.

When estimating the fourth term we consider two cases.
\begin{itemize}
\item[(i)]$\zeta_1 \in (-3/4,0) \cup [1/4,1)$.
We estimate the fourth term using Lemma \ref{lem:kernel bounds} (ii)
\begin{multline*}
\leq\; \frac{1}{\tau^2} \, \|\cal W_{\tau}^{a_2}\|_{L^{\infty}(\Lambda)}\,\|Q_{\tau,\zeta_1}^{(2)}\|_{L^{\infty}(\Lambda^2)} \,\int \dd y_{a_1}\, \cal W_{\tau}^{a_1}(y_{a_1}-y_{a_1^{*}}) \int \dd y_{a_2} \, Q_{\tau,-\zeta_1}^{(2)}(y_{a_2};y_{a_1})
\\
\leq\; \frac{1}{\tau^2}\,\|\cal W_{\tau}^{a_1}\|_{L^{1}(\Lambda)} \, \|\cal W_{\tau}^{a_2}\|_{L^{\infty}(\Lambda)}\,\|Q_{\tau,\zeta_1}^{(2)}\|_{L^{\infty}(\Lambda^2)}\,,
\end{multline*}
which by \eqref{cal W bound} and Lemma \ref{lem:kernel bounds} (iii) is bounded by a constant.

\item[(ii)]$\zeta_1 \in (-1,-3/4] \cup (0,1/4)$.
We then estimate the fourth term using Lemma \ref{lem:kernel bounds} (ii) by
\begin{equation*}
\leq\; \frac{1}{\tau^2} \, \|\cal W_{\tau}^{a_2}\|_{L^{\infty}(\Lambda)}\,\|Q_{\tau,-\zeta_1}^{(2)}\|_{L^{\infty}(\Lambda^2)} \,\int \dd y_{a_1}\, \cal W_{\tau}^{a_1}(y_{a_1}-y_{a_1^{*}}) \int \dd y_{a_2} \, Q_{\tau,\zeta_1}^{(2)}(y_{a_1};y_{a_2})\,,
\end{equation*}
and we conclude the argument as in (i).
\end{itemize}
\item[(C)] $\cal P$ is an open path with vertices $b_1,b_2 \in \cal V_1$.
In this case $\fra I (\cal P)=Q_{\tau,0}(y_{b_1};y_{b_2})$, which is bounded by a constant 
by Lemma \ref{lem:kernel bounds} (i) .
\end{itemize}

Summarizing the base step, we have proved that $\fra I (\cal P)$ is bounded by a constant
when $n = \abs{\cal V(\cal P)} \leq 2$.

\subsubsection*{Induction step}
Suppose that $\fra I(\cal P) \leq \CP^{\abs{\cal V(\cal P)}}$
whenever $\abs{\cal V(\cal P)} \leq n-1$ with $n \geq 3$. We want to show that the bound holds if $\abs{\cal V(\cal P)}=n$. We do this by integrating a vertex in $\cal V_2(\cal P)$ and reducing the path $\cal P$ to a new path with one fewer vertex, satisfying the same properties.

We pick an arbitrary vertex $a \in \cal V_2(\cal P)$, subject to the restriction that if $a^* \in \cal V_2(\cal P)$ then $r_a = 1$. It is easy to see that such a vertex always exists.
For definiteness, we assume without loss of generality that $a=a_2$.

By construction, the only dependence on $y_a=y_{a_2}$ in the integrand defining $\fra I(\cal P)$ is 
\begin{equation*}
\cal W_\tau^{a_2}(y_{a_2}-y_{a_2^{*}}) \,Q_{\tau,\zeta_1}(y_{a_1};y_{a_2}) \, Q_{\tau,\zeta_2}(y_{a_2};y_{a_3})\,.
\end{equation*}
In what follows we prove
\begin{equation}
\label{Induction_Step}
\int \dd y_{a_2}\,\cal W_\tau^{a_2}(y_{a_2}-y_{a_2^{*}}) \,Q_{\tau,\zeta_1}(y_{a_1};y_{a_2}) \, Q_{\tau,\zeta_2}(y_{a_2};y_{a_3}) \;\leq\; \CP \, Q_{\tau,\zeta_1+\zeta_2} (y_{a_1};y_{a_3})\,.
\end{equation}
Note that \eqref{Induction_Step} implies that $\fra I (\cal P) \leq \CP \,\fra I(\cal {\hat P})$, where $\cal {\hat P}$ is the path (open or closed) obtained from $\cal P$ by deleting the vertex $a_2$ and replacing the edges $\{a_1,a_2\}$ and $\{a_2,a_3\}$ with the edge $\{a_1,a_3\}$ carrying the time $\zeta_1+\zeta_2$.
We observe that the path $\hat{\cal P}$ still satisfies the conditions \eqref{Path condition 1} and \eqref{Path condition 2} by construction.

Assuming \eqref{Induction_Step} is proved, we can complete the induction. Together with the induction base, this will conclude the proof of Proposition \ref{Product of subgraphs local}.

What remains is the proof of \eqref{Induction_Step}. Using \eqref{Q_splitting} it follows that the left-hand side of \eqref{Induction_Step} is
\begin{align*}
&\int \dd y_{a_2}\,\cal W_\tau^{a_2}(y_{a_2}-y_{a_2^{*}}) \,Q_{\tau,\zeta_1}^{(1)}(y_{a_1};y_{a_2}) \, Q_{\tau,\zeta_2}^{(1)}(y_{a_2};y_{a_3})
\\ 
&+\frac{1}{\tau}\,\int \dd y_{a_2}\,\cal W_\tau^{a_2}(y_{a_2}-y_{a_2^{*}}) \,Q_{\tau,\zeta_1}^{(2)}(y_{a_1};y_{a_2}) \, Q_{\tau,\zeta_2}^{(1)}(y_{a_2};y_{a_3})
\\
&+\frac{1}{\tau}\,\int \dd y_{a_2}\,\cal W_\tau^{a_2}(y_{a_2}-y_{a_2^{*}}) \,Q_{\tau,\zeta_1}^{(1)}(y_{a_1};y_{a_2}) \, Q_{\tau,\zeta_2}^{(2)}(y_{a_2};y_{a_3})
\\
&+\frac{1}{\tau^2}\,\int \dd y_{a_2}\,\cal W_\tau^{a_2}(y_{a_2}-y_{a_2^{*}}) \,Q_{\tau,\zeta_1}^{(2)}(y_{a_1};y_{a_2}) \, Q_{\tau,\zeta_2}^{(2)}(y_{a_2};y_{a_3})\,.
\end{align*}
We need to estimate each term separately. 

The first term is
\begin{equation*}
\;\leq\;\|Q_{\tau,\zeta_1}^{(1)}\|_{L^{\infty}(\Lambda^2)} \, \|Q_{\tau,\zeta_2}^{(1)}\|_{L^{\infty}(\Lambda^2)} \, \int \dd y_{a_2}\,\cal W_\tau^{a_2}(y_{a_2}-y_{a_2^{*}}) \;=\; \|Q_{\tau,\zeta_1}^{(1)}\|_{L^{\infty}(\Lambda^2)} \, \|Q_{\tau,\zeta_2}^{(1)}\|_{L^{\infty}(\Lambda^2)} \,\|\cal W_\tau^{a_2}\|_{L^1(\Lambda)}\,,
\end{equation*}
which is bounded by a constant, by Lemma \ref{lem:kernel bounds} (i) and \eqref{cal W bound}.

The second term is 
\begin{equation*}
\;\leq\;\frac{1}{\tau}\,\|\cal W_\tau^{a_2}\|_{L^{\infty}(\Lambda)} \, \|Q_{\tau,\zeta_2}^{(1)}\|_{L^{\infty}(\Lambda^2)} \, \int \dd y_{a_2} \,Q_{\tau,\zeta_1}^{(2)}(y_{a_1};y_{a_2})\,
\end{equation*}
which is bounded by a constant, by Lemma \ref{lem:kernel bounds} (i) and (ii) and \eqref{cal W bound}. The third term is estimated analogously.

The fourth term is
\begin{equation*}
\;\leq\; \frac{1}{\tau^2} \,\|\cal W_\tau^{a_2}\|_{L^{\infty}(\Lambda)} \,\int \dd y_{a_2}\,Q_{\tau,\zeta_1}^{(2)}(y_{a_1};y_{a_2}) \, Q_{\tau,\zeta_2}^{(2)}(y_{a_2};y_{a_3})\,
\end{equation*}
which by \eqref{cal W bound} is
\begin{equation}
\label{Induction fourth term}
\;\leq\; \frac{C}{\tau} \int \dd y_{a_2}\,Q_{\tau,\zeta_1}^{(2)}(y_{a_1};y_{a_2}) \, Q_{\tau,\zeta_2}^{(2)}(y_{a_2};y_{a_3})\,.
\end{equation}
We now estimate \eqref{Induction fourth term}. First note that, if $\zeta_1=0$ or $\zeta_2=0$, this expression equals to zero by \eqref{def_Qt_split}. We henceforth consider the case when $\zeta_1, \zeta_2 \in (-1,1) \setminus \{0\}$. In particular, we can rewrite the expression in \eqref{Induction fourth term} as
\begin{equation}
\label{Induction fourth term B}
\frac{C}{\tau} \int \dd y_{a_2}\,\ee^{-\{\zeta_1\}h/\tau}(y_{a_1};y_{a_2}) \, \ee^{-\{\zeta_2\}h/\tau}(y_{a_2};y_{a_3})\;=\; \frac{C}{\tau}  \, \ee^{-(\{\zeta_1\}+\{\zeta_2\})h/\tau}(y_{a_1};y_{a_3})\,.
\end{equation}
If $\{\zeta_1\}+\{\zeta_2\} \geq 1$, then the expression in \eqref{Induction fourth term B} is $\leq C$ for some constant $C$ by the proof of Lemma \ref{lem:kernel bounds} (iii).
Otherwise, if
$\{\zeta_1\}+\{\zeta_2\} < 1$, we have $\{\zeta_1\}+\{\zeta_2\}=\{\zeta_1+\zeta_2\} > 0$. 
Therefore, in this case, the expression in \eqref{Induction fourth term B} is
\begin{equation*}
\frac{C}{\tau}  \, \ee^{-(\{\zeta_1+\zeta_2\})h/\tau}(y_{a_1};y_{a_3}) \;=\;\frac{C}{\tau}  \,Q_{\tau,\zeta_1+\zeta_2}^{(2)}(y_{a_1};y_{a_3})\;\leq\;C  \,Q_{\tau,\zeta_1+\zeta_2}(y_{a_1};y_{a_3})\,.
\end{equation*}
Putting everything together, it follows that the expression on the left-hand side of 
\eqref{Induction_Step} is 
\begin{equation*}
\;\leq\; C+C\,Q_{\tau,\zeta_1+\zeta_2}(y_{a_1};y_{a_3})\,.
\end{equation*}
The claim \eqref{Induction_Step} now follows from Lemma \ref{lem:kernel bounds} (iv). This concludes the proof of the induction step.

We can now deduce Proposition \ref{Product of subgraphs local} when $\cal \xi \in \fra B_p$.
The proof of Proposition \ref{Product of subgraphs local} when $\xi=\xi_p$ proceeds analogously. 
The only difference is that we now work with $(\tilde{\cal V},\tilde{\cal E})$ instead of $(\cal V,\cal E)$ and we note that all of the connected components of $\tilde{\cal E}$ are closed paths. (Recall that \eqref{Product of subgraphs local 2} was estimated in terms of the $L^1$-norm of $\xi$.)
\end{proof}
From Proposition \ref{Product of subgraphs local} we deduce that Corollary \ref{cor:bound_am} holds in this setting. Next, for $\Pi \in \fra R$, we define
\begin{equation*}
\cal I^\xi_{\Pi} \;\deq\; \int_{\Lambda^{\cal V}} \dd \f y \, \pBB{\prod_{i=1}^{m} \delta(y_{i,1}-y_{i,2})} \xi(\f y_1) \prod_{e \in \cal E} \cal J_{e}(\f y_e)
\end{equation*} 
if $\xi \in \fra B_p$ and
\begin{equation*}
\cal I^\xi_{\Pi} \;\deq\; \int_{\Lambda^{\tilde{\cal V}}} \dd \f y \, \pBB{\prod_{i=1}^{m} \delta(y_{i,1}-y_{i,2})} \prod_{e \in \tilde{\cal E}} \cal J_{e}(\f y_e) \,.
\end{equation*} 
if $\xi = \xi_p$.
In other words, we replace $w$ by $\delta$ in Definition \ref{def_calI infinity}. With $\cal I^\xi_{\Pi}$ given in this way we define $a_{\infty,m}^\xi$ by \eqref{f_m,a_m^xi classical 1D infty}.
We want to show that \eqref{Convergence of the explicit terms identity} still holds in this setting. The following analogue of Lemma \ref{Convergence of explicit terms lemma} holds in this setting.

\begin{lemma}
\label{Convergence of explicit terms delta}
For all fixed $\f t \in \fra A$ and $\Pi \in \fra R$ we have
\begin{equation*}
\cal I^\xi_{\tau,\Pi}(\f t) \rightarrow \cal I^\xi_{\Pi}\quad \mbox{as}\quad \tau \rightarrow \infty\quad \mbox{uniformly in} \quad \xi \in \fra B_p \cup \{\xi_p\}\,.\end{equation*}
\end{lemma}
\begin{proof} 
Since we are working in one dimension, it follows from the definition of $\cal J_{\tau, e}, \cal J_{e}$ and from a telescoping argument that
\begin{equation}
\label{convergence in C}
\lim_{\tau \rightarrow \infty}\Big\|\prod_{e \in \hat{\cal E}} \cal J_{\tau, e}(\f y_e, \f s)-\prod_{e \in \hat{\cal E}} \cal J_{e}(\f y_e)\Big\|_{L^\infty_{\f y}} \;=\;0\,,
\end{equation}
where $\hat {\cal E}$ stands for $\cal E$ if $x \in \fra B_p$ or $\tilde {\cal E}$ if $\xi = \xi_p$.
Note that this convergence is not uniform in $\f s$. In the telescoping step of the proof of \eqref{convergence in C} we use $\frac{1}{\tau} \|S_{\tau,t}\|_{L^\infty(\Lambda^2)} \rightarrow 0$ as $\tau \rightarrow \infty$ as well as the estimate 
\begin{equation*}
\bigg|\frac{\ee^{t\lambda_k/\tau}}{\tau(\ee^{\lambda_k/\tau} - 1)}\bigg| \;\leq\;
C \begin{cases}
\frac{1}{\lambda_k}&\mbox{if }\lambda_k \leq \tau \\
\frac{1}{\tau}\ee^{-(1-t)\lambda_k/\tau}&\mbox{if }\lambda_k > \tau\,,
\end{cases}
\end{equation*}
for all $t \leq 1$. The latter allows us to apply the dominated convergence theorem and deduce that $\|\cal J_{\tau, e}(\f y_e, \f s)-\cal J_{e}(\f y_e)\|_{L^\infty_{\f y}} \rightarrow 0$ as $\tau \rightarrow \infty$.

We consider first the case when $\xi \in \cal B_p$. We compute
\begin{multline}
\label{I difference local}
\cal I_{\tau,\Pi}^{\xi}(\f t)-\cal I_{\Pi}^{\xi}\;=\;\int_{\Lambda^{\cal V}} \dd \f y \, \pBB{\prod_{i=1}^{m} w_\tau(y_{i,1}-y_{i,2})} \xi(\f y_1) \Bigg(\prod_{e \in \cal E} \cal J_{\tau, e}(\f y_e, \f s)-\prod_{e \in \cal E} \cal J_{e}(\f y_e)\Bigg)\\
+
\int_{\Lambda^{\cal V_2}} \dd \f y_2 \, \Bigg[\int_{\Lambda^{\cal V_1}} \dd \f y_1\, \prod_{e \in \cal E} \cal J_{e}(\f y_e) \, \xi(\f y_1)\Bigg] \pBB{\prod_{i=1}^{m} w_\tau(y_{i,1}-y_{i,2})-\prod_{i=1}^{m} \delta(y_{i,1}-y_{i,2})}\,.
\end{multline}
The first line on the right-hand side of \eqref{I difference local} is bounded in absolute value by
\begin{multline*}
\Big\|\prod_{e \in \cal E} \cal J_{\tau, e}(\f y_e, \f s)-\prod_{e \in \cal E} \cal J_{e}(\f y_e)\Big\|_{L^\infty_{\f y}} \,\int_{\Lambda^{\cal V}} \dd \f y \,\pBB{\prod_{i=1}^{m} w_\tau(y_{i,1}-y_{i,2})} \,|\xi(\f y_1)|
\\
\leq\; \Big\|\prod_{e \in \cal E} \cal J_{\tau, e}(\f y_e, \f s)-\prod_{e \in \cal E} \cal J_{e}(\f y_e)\Big\|_{L^\infty_{\f y}} \norm{\xi}_{L^1_{\f y_1}}\,,
\end{multline*}
where we used that $\int w_\tau = 1$ and that the arguments $y_{i,1},y_{i,2}$ of each $w_\tau$ are in $\f y_2$.
This converges to zero as $\tau \rightarrow \infty$ by \eqref{I difference local}, uniformly in $\xi \in \fra B_p$.
Furthermore, we note that by \eqref{I difference local}, $\xi \in L^2_{\f y_1}$ and the dominated convergence theorem we find that the expression in square brackets on the second line of \eqref{I difference local} is a bounded continuous function of $\f y_2$. Therefore, the second line on the right-hand side of \eqref{I difference local} also converges to zero as $\tau \rightarrow \infty$ by the assumption that $w_\tau$ converges weakly to the delta function at $0$. The claim when $\xi \in \fra B_p$ follows, and the argument for the case $\xi = \xi_p$ is analogous.
\end{proof}
Since $w_\tau \geq 0$ pointwise, we can apply the Feynman-Kac formula as in the proof of Proposition \ref{Remainder term bound 1D} and reduce the estimates on the remainder term to those on the explicit terms, as we did in Sections \ref{1D subsection 1} and \ref{1D subsection 2}. In particular, from Proposition \ref{Product of subgraphs local} we obtain that for all $M \in \N$ and $\re z \geq 0$ we have
\begin{equation*}
\notag
\absb{R^\xi_{\tau,M}(z)} \;\leq\; (C\CP\, p)^p \, \pb{C \CP \abs{z}}^M M!\,,
\end{equation*}
for the constant $\CP$ from Proposition \ref{Product of subgraphs local}. Moreover, for fixed $\tau>0$, the function $w_\tau \in L^{\infty}(\Lambda)$ and hence the function $A_\tau^{\xi}$ is analytic in $\{z\col\re z > 0\}$ by the same arguments as in Sections \ref{1D subsection 1} and \ref{1D subsection 2}. 

We now consider the classical case. Let $a^\xi_m$ denote the Taylor coefficients and $R_M^{\xi}(z)$ the remainder term of $A^{\xi}(z) \deq \tilde \rho_z (\Theta(\xi))$, where $\tilde \rho_z$ is defined for $\re z \geq 0$ as in \eqref{def_rho_z}, except that now $W$ is given by \eqref{def_W_local}. The identity $a_m^\xi = a_{\infty, m}^\xi\,$ follows from Wick's theorem. 
Proposition \ref{Product of subgraphs local} and Lemma \ref{Convergence of explicit terms delta} imply that for all $m \in \N$ we have
\begin{equation}
\label{am bound local}
\absb{a^\xi_{m}} \;\leq\; (C\CP\, p)^p \, \pb{C \CP}^m m!\,,
\end{equation}
for the constant $\CP$ from Proposition \ref{Product of subgraphs local}. 
If $\xi \in \cal B_p$ then the proof of Lemma \ref{lem:A_analytic} carries over this setting and shows that the $R_M^{\xi}(z)$ satisfies 
\begin{equation}
\label{RM bound local}
\absb{R^\xi_{M}(z)} \;\leq\; (C\CP\, p)^p \, \pb{C \CP \abs{z}}^M M!\,,
\end{equation}
and that the function $A^\xi$ is analytic for $\re z>0$. If $\xi=\xi_p$ we argue as in Section \ref{1D subsection 2} and we reduce the estimate of the remainder term to that of the explicit term and so for $\xi=\xi_p$ we obtain \eqref{RM bound local} from \eqref{am bound local}.  Furthermore, we deduce that the function $A^\xi$ is analytic for $\re z > 0$.

Putting everything together we deduce Theorem \ref{thm:1D local} by using the same duality arguments as in the proof of Theorem \ref{thm:1D} in Section \ref{1D subsection 3}. This concludes the proof of Theorem \ref{thm:1D local}.

\section{The counterterm problem} \label{sec:counterterm}

In this section we formulate the counterterm problem from \eqref{counter_intro} precisely, and solve it. We fix the spatial domain $\Lambda=\mathbb{R}^d$, with $d=2,3$. 
Throughout this section, we use the convention
\begin{equation*}
\hat{f}(p)\;\deq\;\int \dd x\, f(x)\,\ee^{-\ii p\cdot x}
\end{equation*}
for the Fourier transform, in order to avoid factors of $2\pi$ in the calculations below. Note that this is a slightly different convention from the one used earlier.

For a nonnegative function $u \col \Lambda \to [0,\infty)$ we use the abbreviations
\begin{equation*}
G_\tau^u \;\deq\; \frac{1}{\tau (\ee^{(-\Delta + u)/\tau} - 1)}\,, \qquad \varrho_\tau^u(x) \;\deq\; G_\tau^u(x;x)\,.
\end{equation*}
Thus, in the notation \eqref{quantum_G} and \eqref{def_varrho_tau} we have $G_\tau = G_\tau^{\kappa + v_\tau}$ and $\varrho_\tau = \varrho_\tau^{\kappa + v_\tau}$ (recall the remark after \eqref{quantum correlation}). For given $\kappa > 0$ we now make the choice
\begin{equation} \label{def_rho_bar}
\bar \varrho_\tau \;\equiv\; \bar \varrho_\tau^\kappa \;\deq\; \varrho_\tau^\kappa(x;x)\,,
\end{equation}
which is independent of $x$ by translation invariance of the operator $-\Delta + \kappa$.

With these notations, we may write the counterterm problem \eqref{counter_intro} as 
\begin{equation}\label{eq:fix} 
v_{\tau} \;=\; V + w* (\varrho^{\kappa+v_\tau}_\tau - \bar{\varrho}^\kappa_\tau)\,,
\end{equation}
where the dependence of the right-hand side on $v_\tau$ is now made explicit.

The main result of this section, Theorem \ref{thm:counter} below, states that, for a general class of external potentials $V$, for sufficiently large $\kappa > 0$ and for all $\tau > 0$, the counterterm problem \eqref{eq:fix} has a unique solution $v_\tau$. Since the solution $v_\tau$ of (\ref{eq:fix}) depends on $\tau$, we also show that $v_\tau$ approaches a limit $v$ as $\tau \to\infty$. This one-body potential $v$ is the one-body potential in the classical one-particle Hamiltonian \eqref{def_h} that yields the correct rigorous construction of the formal measure \eqref{def_classical_H} with external potential $V$.

We start by imposing some conditions on the external potential $V$ in the original Hamiltonian (\ref{eq:standard}). We require that $V \in L^\infty_\txt{loc} (\Lambda)$ satisfies $V(x) \geq 0$ for all $x \in \Lambda$, that $h_V \deq -\Delta + \kappa + V$ satisfies $\| h_V^{-1} \|_{\fra S^2} < \infty$, and that there exists a constant $C > 0$ such that  
\begin{equation}\label{eq:assV} V(x+y) \;\leq\; C V(x) V(y)
\end{equation}
for all $x,y \in \Lambda$.
Moreover, for some of our results, we also need the additional assumption $V \in C^1 (\Lambda)$, with $\| \nabla V / V \|_{L^\infty} < \infty$; see the statement of Theorem \ref{thm:counter} below.

Note that by the condition $\| h_V^{-1} \|_{\fra S^2} < \infty$ we have $V \not \equiv 0$. Hence, (\ref{eq:assV}) implies that $V (x) > 0$ for all $x \in \Lambda$. (In fact, since we have the freedom of choosing the constant $\kappa$, we can always assume that $V(x) \geq c$ for some constant $c$.) Moreover, (\ref{eq:assV}) implies that $V$ grows at most exponentially at infinity. In particular, the following result holds.

\begin{lemma} 
\label{lem:Vexp}
For $V>0$ satisfying \eqref{eq:assV} there exists a constant $C_0 \geq 0$ such that
\begin{equation}\label{eq:Vexp} V(x) \;\leq\; e^{C_0 (\abs{x} + 1)}\, 
\end{equation}
for all $x \in \Lambda$.
\end{lemma}
\begin{proof}
We let $f(x) \deq \log V(x)$ and note that (\ref{eq:assV}) implies $f(x+y) - f(x) - f(y) \leq C$
for all $x,y \in \Lambda$. Hence $f(2x) \leq C + 2 f(x)$ for all $x \in \Lambda$. Iterating this inequality, we find
\begin{equation}\label{eq:f-func} f(x) \;\leq\; \sum_{j=0}^{n-1} 2^j C + 2^{n} f(x/2^{n}) \;\leq\; 2^n C + 2^n f(x/2^n)  \,.
\end{equation}
For a given $x \in \Lambda$ with $\abs{x} > 1$, we find a unique $n \in \N \setminus \{ 0 \}$ such that $2^{n-1} < \abs{x} \leq 2^n$. {F}rom (\ref{eq:f-func}) we conclude that $f(x) \leq  2 \abs{x} C + 2 \abs{x} \sup_{\abs{y} \leq 1} f(y) \leq C_0 \abs{x}$
for all $x \in \Lambda$ with $\abs{x} > 1$. Here we defined $C_0 \deq \max\{2C + 2\sup_{\abs{y} \leq 1} \log V(y),0\} \geq 0$. Since trivially $f(x) \leq C_0$ for all $x \in \Lambda$ with $\abs{x} \leq 1$, we conclude that $f(x) \leq C_0 (\abs{x} + 1)$ for all $x \in \Lambda$, which implies (\ref{eq:Vexp}). 
\end{proof}

To prove the existence and uniqueness of $v_{\tau}$ satisfying (\ref{eq:fix}), we regard (\ref{eq:fix}) as a fixed point equation on an appropriate metric space. This allows us to apply Banach's fixed point theorem. We fix $V \in L^\infty_\txt{loc} (\Lambda)$ with $V \geq 0$ satisfying (\ref{eq:assV}). For $f \in L^\infty_\txt{loc} (\Lambda)$, we set
\[ \| f \|_V \;\deq\; \sup_{x \in \Lambda} | f(x) / V(x)|\,. \] 
Then we define the Banach space $B \deq \{ f \in L^\infty_\txt{loc} (\Lambda) \col \| f \|_V < \infty \}$.
For $\rad > 0$, we denote by 
\[ B_\rad (V) \;\deq\; \{ f \in B \col \| f - V \|_V \leq \rad \} \]
the closed ball of radius $\rad$ around $V$ in the space $B$. This is a complete metric space with respect to the metric inherited from $B$.

The main result of this section is that the \emph{counterterm problem}
\begin{equation} \label{counterterm_problem}
u \;=\; V + w * (\varrho^{\kappa+u}_\tau - \bar{\varrho}^\kappa_\tau) \,, \qquad u \in B_\rad(V)\,,
\end{equation}
has a unique solution $u$, which we call $v_\tau$, and that this solution is continuous in $\tau$ in an appropriate topology.

\begin{theorem}\label{thm:counter}
Let $V\in L^\infty_\txt{loc} (\Lambda)$ satisfy $V(x) \geq 0$ for all $x \in \Lambda$, $[ -\Delta + \kappa + V ]^{-1} \in \fra S^2$, and (\ref{eq:assV}). Let $w \in L^\infty (\Lambda)$ be an even function such that 
\begin{equation}\label{eq:assw} 
\int |w(y)| \pb{1 + V^2 (y)} \, \dd y \;<\; \infty \,.
\end{equation}

Then for every $\rad \in (0,1)$ there exists $\kappa_0 \equiv \kappa_0(\rad) > 0$ with the following properties. For all $\kappa > \kappa_0$ and for all $\tau > 0$ there exists a unique solution $u \eqd v_\tau$ of the counterterm problem \eqref{counterterm_problem}. Furthermore, for all $\kappa > \kappa_0$ there exists $v \equiv v(\kappa) \in B_\rad (V)$ such that $\lim_{\tau \to \infty} \| v_\tau - v \|_V = 0$. Under the additional assumption $V \in C^1 (\Lambda)$ with $\| \nabla V \|_V < \infty$, we also have  
\begin{equation}\label{eq:HS-conv} 
\lim_{\tau \to \infty} \left\| [-\Delta + \kappa + v_\tau ]^{-1} - [-\Delta + \kappa + v]^{-1} \right\|_{\fra S^2} \;=\; 0 \,. \end{equation}
\end{theorem}

\begin{remark}
\label{counterterm remark}
The condition $u \in B_\rad (V)$, for a $\rad \in (0,1)$, implies that 
\[ (1-\rad) V(x) \;\leq\; u (x) \;\leq\; (1+\rad) V(x)\,. \]
In particular, $u (x) \geq 0$ for all $x \in \Lambda$ and the Hamiltonian $h_u = -\Delta + \kappa + u$ is such that
\begin{equation}\label{eq:hhV} (1-\rad) h_V \;\leq\; h_u \;\leq\; (1+\rad) h_V \, . \end{equation}
This also implies that $\| h_u^{-1} \|_{\fra S^2} \leq (1-\rad)^{-1} \| h_V^{-1} \|_{\fra S^2} < \infty$ for all $u \in B_\rad (V)$.
\end{remark}

In order to prove Theorem \ref{thm:counter}, we need some basic properties of the translation invariant Gibbs states associated with $h_0 = -\Delta + \kappa$. Similar properties for the Gibbs state associated with $h_v = h_0 +v$, for a potential $v \in B_\rad (V)$ follow from the Feynman-Kac representation of $e^{-\alpha h_v}$. 

\begin{lemma}\label{lm:Lk}
\begin{itemize}
\item[(a)] For every $\alpha \in [0,2)$ there exists a constant $C_\alpha > 0$ such that
\[ \left[ \frac{1}{\tau^2} \frac{e^{\alpha (-\Delta +\kappa) / \tau}}{(e^{(-\Delta + \kappa)/ \tau}-1)^2} \right] (x;x) \;\leq\; C_\alpha \kappa^{-(2-d/2)} \,.\]
The constant $C_\alpha$ diverges as $\alpha$ approaches $2$. Because of translation invariance, the left-hand side is actually independent of $x$. 
\item[(b)] There exists a constant $C > 0$ such that 
\[ 0 \;\leq\; \left[ \frac{1}{\tau} \frac{e^{\alpha (-\Delta + \kappa) / \tau}}{e^{(-\Delta + \kappa)/\tau} - 1} \right] (x;y) \;\leq\; C \kappa^{d/2-1} \, e^{-\sqrt{\kappa} |x-y|/4} \]
for all $x,y \in \Lambda$ with $|x-y| > 1$, all $\kappa \geq 1$ and all  $0 \leq \alpha < 1$. The constant $C$ can be chosen independently of $\alpha$, for $\alpha \in [0,1)$. In the sense of distributions, the bound continues to hold for $\alpha = 1$, with the convention that the Dirac $\delta$-function satisfies $\delta (x-y) = 0$, if $|x-y| \geq 1$. 
\end{itemize}
\end{lemma}

\begin{proof}
(a) In Fourier space, we find 
\begin{equation}\label{eq:lm1} \left[ \frac{1}{\tau^2} \frac{e^{\alpha (-\Delta + \kappa)/\tau}}{(e^{(-\Delta + \kappa) / \tau} - 1)^2} \right] (x;x) \;=\; \int \dd q \, \frac{1}{\tau^2} \frac{e^{\alpha \frac{\abs{q}^2 + \kappa}{\tau}}}{\big( e^{\frac{\abs{q}^2 + \kappa}{\tau}} - 1\big)^2}\,. \end{equation}
To estimate the integral we observe, first of all, that the function $x \to e^{\alpha x}/(e^x-1)^2$ is monotonically decreasing in $x$, for $x \geq 0$. In fact
\[\frac{e^{\alpha x}}{(e^x - 1)^2} \;=\; \frac{e^{(\alpha-2)x}}{(1-e^{-x})^2}
\;=\; \sum_{n \geq 1} n e^{-(n+1-\alpha) x}\,, \]
which is a sum of positive and monotonically decreasing functions (since $\alpha < 2$ by assumption). We divide the integral in (\ref{eq:lm1}) in two parts. For $|q|^2 \leq \kappa$, we bound $\abs{q}^2 + \kappa > \kappa$. For $|q|^2 \geq \kappa$ we use $\abs{q}^2 + \kappa > \abs{q}^2$. We find
\begin{equation}\label{eq:lm2} 
\begin{split} 
\left[ \frac{1}{\tau^2} \frac{e^{\alpha (-\Delta + \kappa)/\tau}}{(e^{(-\Delta + \kappa)/ \tau} - 1)^2} \right] (x;x) &\;\leq\; \frac{C \kappa^{d/2}}{\tau^2} \frac{e^{\alpha \kappa/\tau}}{(e^{\kappa/\tau}-1)^2} + \frac{1}{\tau^{2}} \int_{|q| > \sqrt{\kappa}}  \frac{e^{\alpha \abs{q}^2/\tau}}{(e^{\abs{q}^2/\tau} -1)^2} \, \dd q \\ &\;\leq\; \frac{C}{\kappa^{2-d/2}} \sup_{x \geq 0}  \frac{x^2 e^{\alpha x}}{(e^x-1)^2} + \frac{1}{\tau^{2-d/2}} \int_{|q| > \sqrt{\kappa /\tau}}  \frac{e^{\alpha \abs{q}^2}}{(e^{\abs{q}^2} -1)^2} \, \dd q \,.
\end{split} \end{equation}

If $\kappa \leq  4\tau$, we have 
\[ \begin{split} \frac{1}{\tau^{2-d/2}}  \int_{|q| > \sqrt{\kappa/\tau}} \frac{e^{\alpha \abs{q}^2}}{(e^{\abs{q}^2} - 1)^2} \, \dd q & \;\leq\; \frac{C}{\tau^{2-d/2}} \int_{\sqrt{\kappa/\tau} \leq |q| \leq 2} \frac{1}{|q|^4} \, \dd q + \frac{1}{\tau^{2-d/2}} \int_{|q| > 2} \frac{e^{\alpha \abs{q}^2}}{(e^{\abs{q}^2} - 1)^2} \, \dd q \\& \;\leq\; C \left[ \frac{1}{\kappa^{2-d/2}} + \frac{1}{\tau^{d/2-2}} \right] \;\leq\; \frac{C}{\kappa^{d/2-2}} \end{split}  \]
for a constant $C$ depending on $\alpha$ (being finite, for all fixed $\alpha < 2$). If $\kappa > 4\tau$, we obtain 
\[ \begin{split} \frac{1}{\tau^{d/2-2}} \int_{|q| > \sqrt{\kappa/\tau}} \frac{e^{\alpha \abs{q}^2}}{(e^{\abs{q}^2} - 1)^2} \, \dd q  \;\leq\; \frac{C}{\tau^{d/2-2}} \int_{|q| > \sqrt{\kappa/\tau}} e^{-(2-\alpha) \abs{q}^2} \, \dd q \;\leq\; \frac{C}{\tau^{2-d/2}} e^{-c\kappa /\tau} \;\leq\; \frac{C}{\kappa^{2-d/2}} \,.\end{split} \]
Also here, $C$ depends on $\alpha$. {F}rom (\ref{eq:lm2}), we conclude that 
\[ \left[ \frac{1}{\tau^2} \frac{e^{\alpha (-\Delta + \kappa)/\tau}}{(e^{(-\Delta + \kappa)/\tau} - 1)^2} \right] (x;x)  \;\leq\; \frac{C}{\kappa^{2-d/2}}\,, \]
as claimed.

(b) We have 
\[ \frac{1}{\tau} \frac{e^{\alpha (-\Delta + \kappa) / \tau}}{e^{ (-\Delta + \kappa) /\tau} - 1} \;=\; \frac{1}{\tau} e^{(\alpha-1) (-\Delta + \kappa) /\tau} \sum_{n \geq 0} e^{-n (-\Delta + \kappa) /\tau} \;=\; \frac{1}{\tau} \sum_{n \geq 0} e^{-(n+1-\alpha) (-\Delta + \kappa) / \tau}\,.  \]
Since $\exp (- (n+1-\alpha) (p^2 + \kappa)/\tau )$ has a positive Fourier transform, for all $n \in \N$, $\alpha \in [0,1)$, we conclude that 
\[ \left[ \frac{1}{\tau} \frac{e^{\alpha (-\Delta + \kappa) / \tau}}{e^{(-\Delta + \kappa) /\tau} - 1} \right] (x;y) \;\geq\; 0 \]
for all $x,y \in \Lambda$ (the kernel depends only on $x-y$). To show the upper bound, we compute
\begin{equation}\label{eq:lm3} 
\begin{split} \left[ \frac{1}{\tau} \frac{e^{\alpha (-\Delta + \kappa) / \tau}}{e^{ (-\Delta + \kappa) /\tau} - 1} \right] (x;y) & \;=\; \frac{1}{\tau} \sum_{n \geq 0}  \int \dd q \, e^{-(n+1-\alpha) (\abs{q}^2 + \kappa)/\tau} e^{-iq \cdot (x-y)} \\ &\;=\; \pi^{d/2} \tau^{d/2-1} \sum_{n\geq 0} \frac{1}{(n+1-\alpha)^{d/2}} e^{-(n+1-\alpha) \kappa/\tau} e^{-\frac{\abs{x-y}^2 \tau}{4(n+1-\alpha)}} 
\\ &\;= \;  \pi^{d/2} \tau^{d/2-1} \sum_{\txt{I}} \frac{1}{(n+1-\alpha)^{d/2}} e^{-(n+1-\alpha) \kappa/\tau} e^{-\frac{\abs{x-y}^2 \tau}{4(n+1-\alpha)}}  \\ &\qquad + \pi^{d/2} \tau^{d/2-1} \sum_{\txt{II}} \frac{1}{(n+1-\alpha)^{d/2}} e^{-(n+1-\alpha) \kappa/\tau} e^{-\frac{\abs{x-y}^2 \tau}{4(n+1-\alpha)}}\,,\end{split} 
\end{equation} 
where $\sum_{\txt{I}}$ denotes the sum over all $n\in \N$ with $2(n+1-\alpha) > \tau |x-y| / \sqrt{\kappa}$ and $\sum_{\txt{II}}$ indicates the sum over all $n\in \N$ with $2(n+1-\alpha) \leq \tau |x-y|/\sqrt{\kappa}$. In the first sum, we neglect the factor $\exp (-\abs{x-y}^2 \tau / 4 (n+1-\alpha))$ and we estimate (comparing the sum of a monotone sequence with the corresponding integral) 
\begin{equation}\label{eq:sumI} \begin{split} \pi^{d/2} \tau^{d/2-1} \sum_{\txt{I}} \frac{1}{(n+1-\alpha)^{d/2}} &e^{-(n+1-\alpha) \kappa/\tau} e^{-\frac{\abs{x-y}^2 \tau}{4(n+1-\alpha)}}  \\ &\;\leq\; C \tau^{d/2-1} e^{-\sqrt{\kappa} |x-y|/4}   \int_{|z| > \tau |x-y|/\sqrt{\kappa}}  \frac{e^{-z \kappa / \tau}}{z^{d/2}} \, \dd z \\
&\;\leq\; C \kappa^{d/2-1} e^{-\sqrt{\kappa} |x-y| /4} \int_{|z| > \sqrt{\kappa} |x-y|} \frac{e^{-z}}{z^{d/2}} \, \dd z \\
&\;\leq\; C \kappa^{d/2-1} e^{-\sqrt{\kappa} |x-y| /4} \end{split} 
\end{equation}
for all $\kappa > 1$ and $|x-y| > 1$. In the second sum on the right-hand side of (\ref{eq:lm3}), we neglect the factor $\exp (-(n+1-\alpha) \kappa/\tau)$. Furthermore, we introduce a new index $\ell \in \N$ and we sum first over all $n \in \N$ such that 
\begin{equation}\label{eq:index-l} \frac{\tau |x-y|}{(\ell+1) \sqrt{\kappa}} \;<\; 2(n+1-\alpha)  \;\leq\; \frac{\tau |x-y|}{\ell \sqrt{\kappa}} \end{equation}
and then over all $\ell \in \N \setminus \{ 0 \}$. Since there are at most 
\[ \frac{|x-y| \tau}{2 \ell \sqrt{\kappa}} - \frac{|x-y| \tau}{2 (\ell+1) \sqrt{\kappa}} \;=\; \frac{|x-y| \tau}{2 \ell (\ell+1)\sqrt{\kappa}} \;\leq\; \frac{|x-y| \tau}{2 \ell^2 \sqrt{\kappa} } \]
indices $n \in \N$ satisfying (\ref{eq:index-l}), we conclude that
\begin{multline*}
\pi^{d/2} \tau^{d/2-1} \sum_{\txt{II}}  \frac{1}{(n+1-\alpha)^{d/2}} e^{-(n+1-\alpha) \kappa/\tau} e^{-\frac{\abs{x-y}^2 \tau}{4(n+1-\alpha)}}
\\
\leq\; C \kappa^{d/4-1/2} |x-y|^{1-d/2} \sum_{\ell \geq 1} \frac{e^{-\sqrt{\kappa} |x-y| \ell/2}}{\ell^{2-d/2}} \;\leq\; C \kappa^{d/2-1} e^{-\sqrt{\kappa}|x-y|/4}
\end{multline*}
for all $x,y \in \Lambda$ with $|x-y| > 1$ and all $\kappa \geq 1$.
Together with (\ref{eq:lm3}) and (\ref{eq:sumI}), the last equation implies that there exists a constant $C > 0$ such that 
\begin{equation}\label{eq:lm4} \left[ \frac{1}{\tau} \frac{e^{\alpha (-\Delta + \kappa) / \tau}}{e^{ (-\Delta + \kappa)/\tau} - 1} \right] (x;y)  \;\leq\; C \kappa^{d/2-1} e^{-\sqrt{\kappa} |x-y|/4} 
\end{equation}
for all $x,y \in \Lambda$ with $|x-y| >1$ and all $\kappa \geq  1$. 

Notice that the same computation holds for $\alpha = 1$. In this case, however, the term with $n=0$ on the first line in the right-hand side of (\ref{eq:lm3}) has to be handled differently; it gives precisely the contribution $\tau^{-1} \delta (x-y)$ which, in the sense of distributions, vanishes for $|x-y| > 1$. All other terms in the sum over $n$ can be handled as above.
\end{proof}

We are now ready to proceed with the proof of Theorem \ref{thm:counter}.

\begin{proof}[Proof of Theorem \ref{thm:counter}] 
For $\tau > 0$ and $u \in B$ we define 
\[ \Phi_\tau (u) \;\deq\; V + w * (\varrho^{\kappa+u}_\tau - \bar{\varrho}^\kappa_\tau)\,. \]
We claim, first of all, that for fixed $\rad \in (0,1)$ there exists $\kappa_0$ such that $\Phi_\tau (u) \in B_\rad (V)$ for all $u \in B_\rad (V)$, $\tau > 0$ and $\kappa > \kappa_0$. In fact 
\begin{align}
\varrho_\tau^{\kappa + u}(y) - \bar \varrho_\tau^{\kappa} &\;=\; G^{\kappa + u}_\tau(y;y) - G^{\kappa}_\tau(y;y)
\notag
\\
&\;=\; \frac{1}{\tau} \pbb{\frac{1}{\ee^{(-\Delta + \kappa + u)/\tau} - 1} - \frac{1}{\ee^{(-\Delta + \kappa)/\tau} - 1}}(y;y)
\notag
\\
&\;=\; \frac{1}{\tau} \pbb{\frac{1}{\ee^{(-\Delta + \kappa + u)/\tau} - 1} \pb{\ee^{(-\Delta + \kappa)/\tau} - \ee^{(-\Delta + \kappa + u)/\tau}} \frac{1}{\ee^{(-\Delta + \kappa)/\tau} - 1}}(y;y)
\notag
\\
&\;=\; - \frac{1}{\tau} \int_0^1 \dd t \, \pbb{\frac{\ee^{t (-\Delta + \kappa + u)/\tau}}{\ee^{(-\Delta +\kappa + u)/\tau} - 1} \, \frac{u}{\tau}\, \frac{\ee^{(1-t) (-\Delta + \kappa)/\tau}}{\ee^{(-\Delta + \kappa)/\tau} - 1}}(y;y)
\notag
\\
&\;=\; - \int_0^1 \dd t \int \dd z \, \pbb{\frac{1}{\tau}\frac{\ee^{t (-\Delta + \kappa + u)/\tau}}{\ee^{(-\Delta + \kappa  + u)/\tau} - 1}}(y;z) \, u(z) \, \pbb{\frac{1}{\tau}\frac{\ee^{(1-t) (-\Delta + \kappa) /\tau}}{\ee^{ (-\Delta + \kappa) /\tau} - 1}}(z;y)\,.
\label{Fixed point 1}
\end{align}
In the third equality, we used the resolvent identity and in the fourth equality a Duhamel expansion.

Now we split the integral in two parts. In the region with $|y-z| \leq 1$, we estimate \[ 0\; \leq \; u(z) \;\leq\; \| u \|_V V (z) \;\leq\; C \| u \|_V V(y)\,. \]
In the last estimate we used $V \in L^\infty_\txt{loc} (\Lambda)$ and the assumption (\ref{eq:assV}) to deduce
\[ V(z) \;=\; V(y + (z-y)) \;\leq\; C V(y) \sup_{\abs{x} \leq 1} V(x) \;\leq\; C V(y)\,. \]
We observe that 
\[ \frac{e^{t (-\Delta + \kappa + u)/\tau}}{e^{(-\Delta + \kappa +u)/\tau} - 1} \;=\; e^{(t-1) (-\Delta + \kappa +u)/\tau} \frac{1}{1- e^{-(-\Delta + \kappa + u)/\tau}} \;=\; \sum_{n \geq 0} e^{-(n+1-t) (-\Delta + \kappa + u)/\tau}\,. \]
For all $t < 1$, we have $n+1-t > 0$. Using the Feynman-Kac representation for the heat kernel $\exp (-(n+1-t) ( -\Delta +\kappa + u)/\tau)$ and $u (x) \geq 0$ for all $x \in \Lambda$, we find that 
\begin{equation}\label{eq:kac}  0 \;\leq\; \left( \frac{1}{\tau} \frac{e^{t(-\Delta + \kappa + u)/\tau}}{e^{(-\Delta + \kappa + u)/\tau} - 1} \right) (y;z)  \;\leq\;  \left( \frac{1}{\tau} \frac{e^{t (-\Delta + \kappa) /\tau}}{e^{(-\Delta + \kappa)/\tau} - 1} \right) (y;z) \, . \end{equation}
Hence, by (\ref{eq:kac}), $u \geq 0$, and (\ref{eq:assV}) we have
\begin{equation} \label{Fixed point 2A}
\begin{split} 0 \;\leq\; \int_0^1 \dd t \int_{|y-z| \leq 1}  \dd z \, &\pbb{\frac{1}{\tau}\frac{\ee^{t (-\Delta + \kappa + u)/\tau}}{\ee^{(-\Delta + \kappa + u)/\tau} - 1}}(y;z) \, u(z) \, \pbb{\frac{1}{\tau}\frac{\ee^{(1-t)(-\Delta + \kappa)/\tau}}{\ee^{(-\Delta + \kappa)/\tau} - 1}}(z;y) \\ &\hspace{4cm} \;\leq\; C \| u \|_V V(y) \left[ \frac{1}{\tau^2} \frac{e^{(-\Delta + \kappa) / \tau}}{(e^{(-\Delta + \kappa)/\tau} - 1)^2} \right] (y;y) \,,\end{split} 
\end{equation}
which by part (a) of Lemma \ref{lm:Lk} is
\begin{equation}
\label{Fixed point 2}
\;\leq\;C  \kappa^{d/2-2} \| u \|_V V(y) \,.
\end{equation}
Next, we consider the integral in the region $|y-z| > 1$. Using part (b) of Lemma \ref{lm:Lk}, (\ref{eq:kac}), and $u \geq 0$ we obtain
\begin{equation}
\label{Fixed point 3}
\begin{split} 0 \;\leq\; \int_0^1 \dd t \int_{|y-z| > 1}  \dd z \, &\pbb{\frac{1}{\tau} \frac{\ee^{t (-\Delta + \kappa + u)/\tau}}{\ee^{(-\Delta + \kappa + u)/\tau} - 1}}(y;z) \, u(z) \, \pbb{\frac{1}{\tau}\frac{\ee^{(1-t)(-\Delta + \kappa)/\tau}}{\ee^{(-\Delta + \kappa)/\tau} - 1}}(z;y) \\  &\;\leq\; C \kappa^{d-2} \| u \|_V \int_{|y-z| > 1} e^{-\sqrt{\kappa} |y-z|/2} V(z) \, \dd z \\ &\;\leq\; C \kappa^{d-2} \| u \|_V  V(y) \int_{|x|>1} e^{-\sqrt{\kappa} |x|/2} V(x) \, \dd x \\ & \;\leq\; C \kappa^{d/2 -2} \| u \|_V V(y)  \end{split}  
\end{equation}
for $\kappa$ large enough. In the second inequality, we used (\ref{eq:assV}) and the change of variables $x =y-z$. In the third inequality, we used Lemma \ref{lem:Vexp} and integrated over the region $|x|>1$.
We conclude from \eqref{Fixed point 1}, \eqref{Fixed point 2A}, \eqref{Fixed point 2}, and \eqref{Fixed point 3} that
\begin{equation}\label{eq:rhodiff} 
0 \;\leq\; -(\varrho_\tau^{\kappa +u} (y) - \bar{\varrho}_\tau^{\kappa}) \;\leq\; C  \kappa^{d/2-2}  \|u \|_V  V(y) \end{equation}
for all $y \in \Lambda$ and all $\kappa >0$ large enough, depending on the constant $C_0$ in the exponential bound (\ref{eq:Vexp}) for $V$. Using  (\ref{eq:assV}), we find 
\begin{multline*}
\left| w * (\varrho_\tau^{\kappa +u} - \bar{\varrho}_\tau^{\kappa}) (x)\right|  \;\leq\; C \kappa^{d/2-2}  \| u \|_V \int |w(y)|  V(x-y) \, \dd y \\ 
\;\leq\; C \kappa^{d/2-2} \| u \|_V V(x) \int |w(y)| V(y) \, \dd y \;\leq\; C \kappa^{d/2-2}  \| u \|_V V(x) \,,
\end{multline*}
where we used \eqref{eq:assw}.
We therefore obtain
\begin{equation}\label{eq:wstar} \left\| w* (\varrho_\tau^{\kappa+u} - \bar{\varrho}_\tau^\kappa) \right\|_V \;\leq\; C \kappa^{d/2-2} \| u \|_V  \;\leq\; C \kappa^{d/2-2} \| u - V \|_V + C \kappa^{d/2-2}  \end{equation}
because $\| V \|_V  = 1$. Hence, since $d/2-2<0$, it follows that for every $\rad \in (0,1)$, there exists $\kappa_0 \equiv \kappa_0(\rad) > 0$ such that $\Phi_\tau : B_\rad (V) \to B_\rad (V)$ for all $\kappa > \kappa_0$ and all $\tau > 0$.

Furthermore, if $u_1, u_2 \in B_\rad (V)$, we find $\Phi_\tau (u_1) - \Phi_\tau (u_2) = w * (\varrho_\tau^{\kappa+u_1} - \varrho_\tau^{\kappa + u_2})$.
Now we have
\[ \begin{split}  
\varrho_\tau^{\kappa+u_1} (y) &- \varrho_\tau^{\kappa + u_2} (y) \\ &\;=\; \frac{1}{\tau} \frac{1}{e^{(-\Delta + \kappa + u_1)/\tau} -1}  \left[ e^{(-\Delta + \kappa + u_1)/\tau} - e^{(-\Delta + \kappa + u_2)/\tau} \right] 
 \frac{1}{e^{(-\Delta + \kappa + u_2)/\tau} -1} (y;y) \\  
&\;=\; \frac{1}{\tau} \int_0^1 \dd t \, \frac{e^{t(-\Delta + \kappa +u_1)/\tau}}{e^{(-\Delta + \kappa + u_1)/\tau} - 1} \frac{u_1 - u_2}{\tau} \frac{e^{(1-t)(-\Delta + \kappa +u_2)/\tau}}{e^{(-\Delta + \kappa + u_2)/\tau} - 1} (y;y) \\
&\;=\; \int_0^1 \dd t \int \dd z \, \left[ \frac{1}{\tau} \frac{e^{t(-\Delta + \kappa +u_1)/\tau}}{e^{(-\Delta + \kappa + u_1)/\tau} - 1} \right] (y;z)  (u_1 (z)- u_2 (z)) \left[ \frac{1}{\tau} \frac{e^{(1-t)(-\Delta + \kappa +u_2)/\tau}}{e^{(-\Delta + \kappa + u_2)/\tau} - 1} \right] (z;y) \,.\end{split} \]
Similarly as before, dividing the integral in two parts, we obtain 
\[ |\varrho_\tau^{\kappa + u_1} (y) - \varrho_\tau^{\kappa + u_2} (y) | \;\leq\; C \kappa^{d/2-2} \| u_1 - u_2 \|_V V(y) \,, \]
which implies that 
\[ \left\| \Phi_\tau (u_1) - \Phi_\tau (u_2) \right\|_V \;\leq\; C \kappa^{d/2-2} \| u_1 - u_2 \|_V\,. \]
Hence, for $\kappa$ large enough and for all $\tau > 0$, $\Phi_\tau \col B_\rad (V) \to B_\rad (V)$ is a contraction on $B_\rad (V)$. Since $B_\rad (V)$ is complete, there exists a unique fixed point $v_\tau \in B_\rad (V)$ such that (\ref{eq:fix}) holds.. 

Next, we show that the sequence of fixed points $v_\tau$ has a limit in $B_\rad (V)$, as $\tau \to \infty$. For $\tau > 0$,  let $h_\tau = -\Delta +\kappa +v_\tau$. We find
\begin{equation}\label{eq:v-v} v_\tau- v_{\tau'} \;=\; \Phi_\tau (v_{\tau}) - \Phi_{\tau'} (v_{\tau'}) \;=\; w * \left( (\varrho^{\kappa + v_\tau}_\tau - \bar{\varrho}_\tau^\kappa) - (\varrho^{\kappa + v_{\tau'}}_{\tau'} - \bar{\varrho}_{\tau'}^\kappa \right) \,.\end{equation}
We have
\begin{equation}\label{eq:AB} \begin{split} & \mspace{-20mu} \left[ \varrho_\tau^{\kappa + v_\tau} (y) - \bar{\varrho}_{\tau} \right] - \left[ \varrho_{\tau'}^{\kappa + v_{\tau'}} (y) - \bar{\varrho}_{\tau'} \right] \\ &\;= \; \left[ \frac{1}{\tau} \frac{1}{e^{h_\tau/\tau}-1} - \frac{1}{\tau} \frac{1}{e^{(-\Delta + \kappa)/\tau}-1} \right] (y;y) - \left[ \frac{1}{\tau'} \frac{1}{e^{h_\tau/\tau'}-1} - \frac{1}{\tau'} \frac{1}{e^{(-\Delta + \kappa)/\tau'}-1} \right] (y;y) \\ &\qquad + \left[ \frac{1}{\tau'} \frac{1}{e^{h_\tau/\tau'} -1} - \frac{1}{\tau'} \frac{1}{e^{h_{\tau'}/\tau'} - 1} \right] (y;y) \\ &\;\eqd\; A (y) + B (y)\,, \end{split} \end{equation}
where $A(y)$ and $B(y)$ denote the contributions of the first and second lines respectively. We start with the term $B$. We obtain, by using the resolvent identity and a Duhamel expansion as in \eqref{Fixed point 1},
\begin{equation} 
\label{eq:rhodiffB}
\begin{split} 
&\mspace{-20mu} \Big| \Big[ \frac{1}{\tau'} \frac{1}{e^{h_\tau /\tau'} -1} - \frac{1}{\tau'} \frac{1}{e^{h_{\tau'}/\tau'}- 1} \Big] (y;y) \Big| \\ &\;= \;  \left| \int_0^1 \dd t \int \dd z \, \left[ \frac{1}{\tau'} \frac{e^{t h_{\tau} / \tau'}}{e^{h_\tau / \tau'} -1} \right] (y;z)  (v_\tau - v_{\tau'}) (z) \left[ \frac{1}{\tau'} \frac{e^{(1-t) h_{\tau'} / \tau'}}{e^{h_{\tau'} / \tau'} -1 } \right] (z;y) \right| \\
&\;\leq \;  \| v_\tau - v_{\tau'} \|_V \, \int_0^1 \dd t \int \dd z \, \left[  \frac{1}{\tau'} \frac{e^{t h_{\tau} / \tau'}}{e^{h_\tau /\tau'} -1} \right]  (y;z)  V (z) \left[  \frac{1}{\tau'} \frac{e^{(1-t) h_{\tau'} / \tau'}}{e^{h_{\tau'} / \tau'} -1 } \right] (z;y) 
\\ &\;\leq \; C \kappa^{d/2-2} \,  \| v_\tau - v_{\tau'} \|_V \, V(y) \,.
\end{split} 
\end{equation}

The last bound can be proved similarly to (\ref{eq:rhodiff}), integrating first over the region $|y-z| \leq 1$ using part (a) of Lemma \ref{lm:Lk} and then over $|y-z| > 1$ by using
\[ \left[ \frac{1}{\tau'} \frac{e^{\alpha h_{\tau'} / \tau'}}{e^{h_{\tau'} / \tau'} -1 } \right] (y;z) \;\leq\; \left[ \frac{1}{\tau'} \frac{e^{\alpha (-\Delta +\kappa) / \tau'}}{e^{(-\Delta + \kappa)/ \tau'} -1 } \right] (y;z) \;\leq\; C \kappa^{d/2-1} e^{-\sqrt{\kappa} |y-z|/4} \]
for all $|y-z| > 1$, uniformly in $\alpha \in [0,1)$, which follows from Part (b) of Lemma \ref{lm:Lk} (the same estimate holds if we replace $h_{\tau'}$ with $h_\tau$). 
We conclude that 
\[ \left| \int w(x-y) \, \Big[ \frac{1}{\tau'} \frac{1}{e^{h_\tau /\tau'} -1} - \frac{1}{\tau'} \frac{1}{e^{h_{\tau'}/\tau'}- 1} \Big] (y;y) \, \dd y \right| \;\leq\; C \kappa^{d/2-2} \| v_\tau - v_\tau' \|_V \, V(x)\,. \]
Here we again used \eqref{eq:assV} and \eqref{eq:assw}. The estimate \eqref{eq:rhodiffB} now follows.
From (\ref{eq:v-v}), (\ref{eq:AB}), and \eqref{eq:rhodiffB} we find 
\begin{equation}\label{eq:v-v2} (1- C \kappa^{d/2-2}) \| v_\tau - v_{\tau'} \|_V \;\leq\; \| w * A \|_V \,.
\end{equation}
Next, we consider the contribution of the term $A$, defined in (\ref{eq:AB}). We claim that $\| w * A \|_V \to 0$ as $\tau,\tau' \to \infty$. By (\ref{eq:v-v2}), this also implies that $\| v_\tau - v_{\tau'} \|_V \to 0$, as $\tau,\tau' \to \infty$, if $\kappa > 0$ is large enough. Writing
\[ A(y) \;=\; - \int_0^1 \dd t \, \left[ \frac{1}{\tau} \frac{e^{th_\tau /\tau}}{e^{h_\tau / \tau} -1} \frac{v_\tau}{\tau} \frac{e^{(1-t) (-\Delta + \kappa)/\tau}}{e^{(-\Delta + \kappa)/\tau}-1} - \frac{1}{\tau'} \frac{e^{th_\tau /\tau'}}{e^{h_\tau / \tau'} -1} \frac{v_\tau}{\tau'} \frac{e^{(1-t) (-\Delta + \kappa)/\tau'}}{e^{(-\Delta + \kappa)/\tau'}-1}\right] (y;y)\,,\]
we remark that 
\[ \begin{split} &\| w * A \|_V \\ &\leq\; \int_0^1 \dd t \, \left\| \int \dd y \, w(\cdot-y) \left[ 
\frac{1}{\tau} \frac{e^{th_\tau/\tau}}{e^{h_\tau /\tau}-1} \frac{v_\tau}{\tau} \frac{e^{(1-t)(-\Delta + \kappa) / \tau}}{e^{(-\Delta + \kappa)/\tau} - 1}  -  \frac{1}{\tau'} \frac{e^{t h_\tau / \tau'}}{e^{h_\tau /\tau'}-1} \frac{v_\tau}{\tau'} \frac{e^{(1-t)(-\Delta + \kappa) / \tau'}}{e^{(-\Delta + \kappa)/\tau'} - 1} \right] (y;y) \right\|_V \,.\end{split}\]
Proceeding similarly as in the derivation of (\ref{eq:wstar}), the integrand of $t$-integral is
\begin{equation*} 
\begin{split} 
& 
\;\leq\; \left\| \int \dd y \, w(\cdot-y) \left[ 
\frac{1}{\tau} \frac{e^{th_\tau/\tau}}{e^{h_\tau /\tau}-1} \frac{v_\tau}{\tau} \frac{e^{(1-t)(-\Delta + \kappa) / \tau}}{e^{(-\Delta + \kappa)/\tau} - 1} \right] (y;y) \right\|_V \\ & \qquad \qquad + \left\| \int \dd y \, w(\cdot-y) \left[\frac{1}{\tau'} \frac{e^{th_\tau/\tau'}}{e^{h_\tau /\tau'}-1} \frac{v_\tau}{\tau'} \frac{e^{(1-t)(-\Delta + \kappa) / \tau'}}{e^{(-\Delta + \kappa)/\tau'} - 1} \right] (y;y) \right\|_V \;\leq\; 2C \kappa^{d/2-2} \| v \|_V \,.
\end{split} 
\end{equation*}
Since the right-hand side is clearly integrable over $t \in [0,1]$, the convergence $\| w*A \|_V \to 0$ follows from the dominated convergence theorem if we can prove that 
\begin{equation}\label{eq:t-fix} 
\left\| \int \dd y \, w(\cdot-y) \left[ 
\frac{1}{\tau} \frac{e^{th_\tau/\tau}}{e^{h_\tau /\tau}-1} \frac{v_\tau}{\tau} \frac{e^{(1-t)(-\Delta + \kappa) / \tau}}{e^{(-\Delta + \kappa)/\tau} - 1}  -  \frac{1}{\tau'} \frac{e^{th_\tau/\tau'}}{e^{h_\tau /\tau'}-1} \frac{v_\tau}{\tau'} \frac{e^{(1-t)(-\Delta + \kappa) / \tau'}}{e^{(-\Delta + \kappa)/\tau'} - 1} \right] (y;y) \right\|_V \;\to\; 0 \end{equation}
as $\tau,\tau' \to \infty$, for every fixed $t \in (0,1)$. To this end, we decompose
\begin{equation}\label{eq:I+II} \begin{split} &\left[ \frac{1}{\tau} \frac{e^{th_\tau / \tau}}{e^{h_\tau/\tau} - 1} \frac{v_\tau}{\tau} \frac{e^{(1-t) (-\Delta + \kappa) /\tau}}{e^{(-\Delta + \kappa) / \tau}-1} - \frac{1}{\tau'} \frac{e^{th_\tau / \tau'}}{e^{h_\tau/\tau'} - 1} \frac{v_\tau}{\tau'} \frac{e^{(1-t) (-\Delta + \kappa) /\tau'}}{e^{(-\Delta + \kappa) / \tau'}-1} \right](y;y) \\ &\qquad \;=\; \left[ \frac{1}{\tau} \frac{e^{th_\tau / \tau}}{e^{h_\tau/\tau} - 1} - \frac{1}{\tau'} \frac{e^{th_\tau/\tau'}}{e^{h_\tau/\tau'}-1} \right] \frac{v_\tau}{\tau} \frac{e^{(1-t) (-\Delta + \kappa)/\tau}}{e^{(-\Delta+\kappa)/\tau}-1}(y;y) \\ &\qquad \qquad +
\frac{1}{\tau'} \frac{e^{th_\tau/\tau'}}{e^{h_\tau/\tau'}-1} v_\tau \left[ \frac{1}{\tau} \frac{e^{(1-t)(-\Delta + \kappa) / \tau}}{e^{(-\Delta + \kappa) / \tau}-1} - \frac{1}{\tau'} \frac{e^{(1-t) (-\Delta + \kappa) / \tau'}}{e^{(-\Delta + \kappa) / \tau'}-1} \right](y;y) \,. \end{split} \end{equation}
The absolute value of the first term is
\begin{equation} \label{Counterterm A}
\begin{split} 
& \;=\; \left| \int \dd z \, \left[  \frac{1}{\tau} \frac{e^{th_\tau / \tau}}{e^{h_\tau/\tau} - 1} - \frac{1}{\tau'} \frac{e^{th_\tau/\tau'}}{e^{h_\tau/\tau'}-1} \right] (y;z) v_\tau (z) \left[ \frac{1}{\tau} \frac{e^{(1-t) (-\Delta + \kappa) / \tau}}{e^{(-\Delta + \kappa) / \tau} - 1} \right] (z;y) \right| \\  
&\;\leq\; \alpha \int \dd z \, \left| \left[ \frac{1}{\tau} \frac{e^{th_\tau / \tau}}{e^{h_\tau/\tau} - 1} - \frac{1}{\tau'} \frac{e^{th_\tau/\tau'}}{e^{h_\tau/\tau'}-1} \right] (y;z) \right|^2  + \alpha^{-1} \int \dd z \, v^2_\tau (z) \left| \left[ \frac{1}{\tau} \frac{e^{(1-t) (-\Delta + \kappa) / \tau}}{e^{(-\Delta + \kappa) / \tau} - 1} \right] (z;y) \right|^2 \end{split} 
\end{equation}
for an arbitrary $\alpha > 0$. On the one hand, using part (a) of Lemma \ref{lm:Lk} and arguing as in the proof of \eqref{Fixed point 2}, we have
\begin{equation} \label{Counterterm B}
\begin{split} 
\int_{|y-z| \leq 1} \dd z \, v^2_\tau (z) \left| \left[ \frac{1}{\tau} \frac{e^{(1-t) (-\Delta + \kappa) / \tau}}{e^{(-\Delta + \kappa) / \tau} - 1} \right] (z;y) \right|^2 &\;\leq\; C \| v_\tau \|_V^2 V(y)^2 \left[ \frac{1}{\tau} \frac{e^{2(1-t) (-\Delta + \kappa) / \tau}}{(e^{(-\Delta + \kappa) / \tau} - 1)^2} \right] (y;y)  \\ &\;\leq\; C \kappa^{d/2-2} \| v_\tau \|_V^2 V(y)^2 \end{split}
\end{equation}
for a constant $C > 0$ depending on $t$, finite for all fixed $0 \leq t < 1$. On the other hand, by using part (b) of Lemma \ref{lm:Lk} and arguing as in the proof of \eqref{Fixed point 3}, we find 
\begin{equation} \label{Counterterm C}
\begin{split} 
\int_{|y-z| > 1} \dd z \, v^2_\tau (z) \left| \left[ \frac{1}{\tau} \frac{e^{(1-t) (-\Delta + \kappa) / \tau}}{e^{(-\Delta + \kappa) / \tau} - 1} \right] (z;y) \right|^2 &\;\leq\; C \kappa^{d-2} \| v_\tau \|_V^2 \int V^2(z) e^{-\sqrt{\kappa} |y-z|/2} \, \dd z \\ &\;\leq\; C \kappa^{d/2-2} \| v_\tau \|_V^2  \, V(y)^2 \end{split} 
\end{equation}
for $\kappa$ large enough. By \eqref{Counterterm A}, \eqref{Counterterm B}, and \eqref{Counterterm C},  we conclude that
\[ \begin{split}  &\left| \int \dd y \, w(x-y) \, \left[ \frac{1}{\tau} \frac{e^{th_\tau / \tau}}{e^{h_\tau/\tau} - 1} - \frac{1}{\tau'} \frac{e^{th_\tau/\tau'}}{e^{h_\tau/\tau'}-1} \right] \frac{v_\tau}{\tau} \frac{e^{(1-t) (-\Delta + \kappa)/\tau}}{e^{(-\Delta+\kappa)/\tau}-1}  (y;y) \right| \\ &\qquad \;\leq\; C \alpha^{-1} \kappa^{d/2-2} \| v_\tau \|^2_V  \int |w(x-y)| V^2 (y) \, \dd y + \alpha \| w \|_{L^\infty} \left\|  \frac{1}{\tau} \frac{e^{th_\tau / \tau}}{e^{h_\tau/\tau} - 1} - \frac{1}{\tau'} \frac{e^{th_\tau/\tau'}}{e^{h_\tau/\tau'}-1}  \right\|_{\fra S^2}^2\,.
\end{split} \]
With the assumptions (\ref{eq:assV}) and (\ref{eq:assw}) and with the optimal choice
\begin{equation*}
\alpha \;=\;\frac{V(x) \,\kappa^{d/4-1}}{\big\|  \frac{1}{\tau} \frac{e^{th_\tau / \tau}}{e^{h_\tau/\tau} - 1} - \frac{1}{\tau'} \frac{e^{th_\tau/\tau'}}{e^{h_\tau/\tau'}-1}  \big\|_{\fra S^2}} \,,
\end{equation*}
we obtain the bound
\[ \begin{split}  &\left\| \int \dd y \, w(x-y) \, \left[ \frac{1}{\tau} \frac{e^{th_\tau / \tau}}{e^{h_\tau/\tau} - 1} - \frac{1}{\tau'} \frac{e^{th_\tau/\tau'}}{e^{h_\tau/\tau'}-1} \right] \frac{v_\tau}{\tau} \frac{e^{(1-t) (-\Delta + \kappa)/\tau}}{e^{(-\Delta+\kappa)/\tau}-1}  (y;y) \right\|_V \\ &\hspace{3.5cm} \;\leq\; C \kappa^{d/4-1} \| v_\tau \|_V \, \left\|  \frac{1}{\tau} \frac{e^{th_\tau / \tau}}{e^{h_\tau /\tau} - 1} - \frac{1}{\tau'} \frac{e^{th_\tau/\tau'}}{e^{h_\tau/\tau'}-1}  \right\|_{\fra S^2}\,. \end{split} 
\]
Since
\begin{equation} \label{Counterterm D}
\left\|   \frac{1}{\tau} \frac{e^{th_\tau / \tau}}{e^{h_\tau/\tau} - 1} - \frac{1}{\tau'} \frac{e^{th_\tau/\tau'}}{e^{h_\tau/\tau'}-1}  \right\|_{\fra S^2}  \;\leq\; \left\|   \frac{1}{\tau} \frac{e^{th_\tau / \tau}}{e^{h/\tau} - 1} - \frac{1}{h_\tau} \right\|_{\fra S^2} + \left\|   \frac{1}{\tau'} \frac{e^{th_\tau / \tau'}}{e^{h_\tau/\tau'} - 1} - \frac{1}{h_\tau} \right\|_{\fra S^2} \,,
\end{equation}
it follows from Lemma \ref{cor:conv_G_tau corollary2} that 
\begin{equation}\label{eq:I-last}  \left\| \int \dd y \, w(\cdot -y) \, \left[ \frac{1}{\tau} \frac{e^{th_\tau / \tau}}{e^{h/\tau} - 1} - \frac{1}{\tau'} \frac{e^{th_\tau/\tau'}}{e^{h_\tau/\tau'}-1} \right] \frac{v_\tau}{\tau} \frac{e^{(1-t) (-\Delta + \kappa)/\tau}}{e^{(-\Delta+\kappa)/\tau}-1}  (y;y) \right\|_V  \;\to\; 0 \end{equation}
as $\tau, \tau' \to \infty$. The application of Lemma \ref{cor:conv_G_tau corollary2} is justified by the assumption that $v_\tau \in B_\rad(V)$ and Remark \ref{counterterm remark}. Note that, for the second term on the right-hand side of \eqref{Counterterm D}, we first let $\tau' \rightarrow \infty$ and we apply Lemma \ref{cor:conv_G_tau corollary2} with the parameter $\tau'$ and the operator $h_\tau$ which is constant in the parameter $\tau'$.
Next, we consider the second term on the right-hand side of
(\ref{eq:I+II}), which is
\[ \begin{split} 
&\;=\;\int \dd z \, \left[ \frac{1}{\tau'} \frac{e^{th_\tau / \tau'}}{e^{h_\tau / \tau'} - 1} \right] (y;z) v_\tau (z) \left[ \frac{1}{\tau} \frac{e^{(1-t)(-\Delta + \kappa) / \tau}}{e^{(-\Delta + \kappa) / \tau}-1} - \frac{1}{\tau'} \frac{e^{(1-t) (-\Delta + \kappa) / \tau'}}{e^{(-\Delta + \kappa) / \tau'}-1} \right] (z-y) \\ &\;\leq\;  
\alpha \int \dd z \, v^2_\tau (z) \, \left| \frac{1}{\tau'} \frac{e^{th_\tau / \tau'}}{e^{h_\tau / \tau'} - 1} (y;z) \right|^2  
+ \alpha^{-1} \int \dd z \, \left| \left[ \frac{1}{\tau} \frac{e^{(1-t)(-\Delta + \kappa) / \tau}}{e^{(-\Delta + \kappa) / \tau}-1} - \frac{1}{\tau'} \frac{e^{(1-t) (-\Delta + \kappa) / \tau'}}{e^{(-\Delta + \kappa) / \tau'}-1} \right]  (z) \right|^2 
\end{split} \]
for every $\alpha > 0$. Proceeding as above (considering separately the regions with $|y-z| \leq 1$ and $|y-z| > 1$, and using the results of Lemma \ref{lm:Lk}), this quantity is
\[ \begin{split} 
\;\leq\;  C \alpha \kappa^{d/2-2} \| v_\tau \|_V^2 V(y)^2 + C \alpha^{-1} \, \left[ \frac{1}{\tau} \frac{e^{(1-t)(-\Delta + \kappa) / \tau}}{e^{(-\Delta + \kappa) / \tau}-1} - \frac{1}{\tau'} \frac{e^{(1-t) (-\Delta + \kappa) / \tau'}}{e^{(-\Delta + \kappa) / \tau'}-1} \right]^2 (0;0) \,. \end{split} \]
With the optimal choice
\begin{equation*}
\alpha \;=\; \frac{\Big| \Big[ \frac{1}{\tau} \frac{e^{(1-t)(-\Delta + \kappa) / \tau}}{e^{(-\Delta + \kappa) / \tau}-1} - \frac{1}{\tau'} \frac{e^{(1-t) (-\Delta + \kappa) / \tau'}}{e^{(-\Delta + \kappa) / \tau'}-1} \Big]^2 (0;0) \Big|^{1/2}\,\kappa^{-d/4+1}}{\|v_\tau\|_V \,V(y)}
\,,
\end{equation*}
 we find 
\[ \begin{split} 
\frac{1}{\tau'} \frac{e^{th/\tau'}}{e^{h/\tau'}-1} v &\left[ \frac{1}{\tau} \frac{e^{(1-t)(-\Delta + \kappa) / \tau}}{e^{(-\Delta + \kappa) / \tau}-1} - \frac{1}{\tau'} \frac{e^{(1-t) (-\Delta + \kappa) / \tau'}}{e^{(-\Delta + \kappa) / \tau'}-1} \right] (y;y) \\ & \;\leq\; C \kappa^{d/4-1} \|v_\tau \|_V V(y) \left| \left[ \frac{1}{\tau} \frac{e^{(1-t)(-\Delta + \kappa) / \tau}}{e^{(-\Delta + \kappa) / \tau}-1} - \frac{1}{\tau'} \frac{e^{(1-t) (-\Delta + \kappa) / \tau'}}{e^{(-\Delta + \kappa) / \tau'}-1} \right]^2 (0;0) \right|^{1/2} \,.\end{split} \]
We conclude that 
\begin{equation}\label{eq:II-last} \begin{split} &\left\| \int \dd y \, w(\cdot-y) \frac{1}{\tau'} \frac{e^{th/\tau'}}{e^{h/\tau'}-1} v \left[ \frac{1}{\tau} \frac{e^{(1-t)(-\Delta + \kappa) / \tau}}{e^{(-\Delta + \kappa) / \tau}-1} - \frac{1}{\tau'} \frac{e^{(1-t) (-\Delta + \kappa) / \tau'}}{e^{(-\Delta + \kappa) / \tau'}-1} \right] (y;y) \right\|_V \\ &\hspace{5cm} \;\leq\; C \kappa^{d/2-1} \| v_\tau \|_V \left| \left[ \frac{1}{\tau} \frac{e^{(1-t)(-\Delta + \kappa) / \tau}}{e^{(-\Delta + \kappa) / \tau}-1} - \frac{1}{\tau'} \frac{e^{(1-t) (-\Delta + \kappa) / \tau'}}{e^{(-\Delta + \kappa) / \tau'}-1} \right]^2 (0;0) \right|^{1/2}\,. \end{split} \end{equation}
To prove that the right-hand side converges to zero, as $\tau,\tau' \to\infty$, we show that
\begin{equation}\label{eq:dom-p}\lim_{\tau \to \infty} \left[ \frac{1}{\tau} \frac{e^{(1-t) (-\Delta + \kappa) / \tau}}{e^{(-\Delta + \kappa) /\tau} - 1} - \frac{1}{(-\Delta + \kappa)} \right]^2 (0;0) \;=\; \lim_{\tau \to \infty} \int \dd p \, \left[ \frac{1}{\tau} \frac{e^{(1-t) (p^2 + \kappa)/\tau}}{e^{(p^2 + \kappa)/\tau} -1} - \frac{1}{p^2 + \kappa} \right]^2 \;=\; 0 \,. \end{equation}
Since 
\[\left|\frac{1}{\tau} \frac{e^{t (p^2+\kappa)/\tau}}{e^{(p^2+\kappa)/\tau}-1} - \frac{1}{p^2 + \kappa} \right| \;\leq\; \frac{C}{p^2 + \kappa} \]
for an appropriate constant $C > 0$, since $(p^2 + \kappa)^{-2}$ is integrable (in dimensions $d \leq 3$), and since 
\[  \lim_{\tau \to \infty} \left[ \frac{1}{\tau} \frac{e^{t (p^2+\kappa)/\tau}}{e^{(p^2+\kappa)/\tau}-1} - \frac{1}{p^2 + \kappa} \right] \;=\; 0 \]
pointwise, for all $p \in \Lambda$, the dominated convergence theorem implies (\ref{eq:dom-p}). {F}rom (\ref{eq:II-last}), we obtain that 
\[ \left\| \int \dd y \, w(\cdot-y) \frac{1}{\tau'} \frac{e^{th/\tau'}}{e^{h/\tau'}-1} v \left[ \frac{1}{\tau} \frac{e^{(1-t)(-\Delta + \kappa) / \tau}}{e^{(-\Delta + \kappa) / \tau}-1} - \frac{1}{\tau'} \frac{e^{(1-t) (-\Delta + \kappa) / \tau'}}{e^{(-\Delta + \kappa) / \tau'}-1} \right] (y;y) \right\|_V \to 0 \]
as $\tau,\tau' \to \infty$. Together with (\ref{eq:I-last}) and (\ref{eq:I+II}), we conclude that (\ref{eq:t-fix}) holds, for all fixed $t\in (0,1)$. This shows that $\| v_\tau - v_{\tau'}\|_V \to 0$, for $\tau,\tau' \to \infty$. Hence $v_\tau$ is a Cauchy sequence in the complete metric space $B_\rad (V)$. As a consequence, there exists $v \in B_\rad (V)$ with $v_\tau \to v$, as $\tau \to \infty$. 

Finally, we prove (\ref{eq:HS-conv}). As before, we use the notation $h_\tau = -\Delta + \kappa + v_\tau$. Moreover, we set $h = -\Delta + \kappa + v$. From the resolvent identity we get
\begin{multline*}
\| h_\tau^{-1} - h^{-1} \|^2_{\fra S^2} \;=\; \| h_{\tau}^{-1} (v_\tau - v) h^{-1} \|^2_{\fra S^2} \;=\; \tr \, h_{\tau}^{-1} (v_\tau - v) h^{-2} (v_\tau - v) h_{\tau}^{-1} \\ \leq\; \| h_{\tau}^{-1} V \|_{\fra S^\infty} \| V h^{-1} \|_{\fra S^\infty} \, \| h^{-1} \|_{\fra S^2} \, \| h_{\tau}^{-1} \|_{\fra S^2}\, \| v_\tau - v \|^2_V \,.
\end{multline*}
By Remark \ref{counterterm remark}, it follows that the above expression is 
\begin{equation}
\label{eq:HS-conv1}
\;\leq\; C\, \| h_{\tau}^{-1} V \|_{\fra S^\infty} \| V h^{-1} \|_{\fra S^\infty} \, \| h^{-1} \|_{\fra S^2}^2 \,\| v_\tau - v \|^2_V \,. 
\end{equation}

Since $\| v_\tau - v \|_V \to 0$, as $\tau \to \infty$, to prove that $\|h_\tau^{-1} - h^{-1} \|_{\fra S^2} \to 0$ it is enough to show that the operator norms $\| h_\tau^{-1} V \|_{\fra S^\infty}$, $\| V h^{-1} \|_{\fra S^\infty}$ are bounded, uniformly in $\tau$. To this end, we observe that, for all $\psi \in L^2 (\Lambda)$,  
\begin{equation}\label{eq:h-1V} \begin{split} 
\| h_\tau^{-1} V \psi \|^2 &\;= \; \left\langle V \psi , h_\tau^{-2} V \psi \right\rangle \\ &\;= \; \left\langle V^{1/2} \psi , \left\{ h_\tau^{-1} V^{1/2} + \left[ V^{1/2} , h_\tau^{-1} \right] \right\} \left\{ V^{1/2} h_\tau^{-1} + \left[ h_\tau^{-1} , V^{1/2} \right] \right\} V^{1/2} \psi \right\rangle \\  & \;\leq \; 2 \left\langle V^{1/2} \psi , h_\tau^{-1} V h_\tau^{-1} V^{1/2} \psi \right\rangle + 2 \left\langle V^{1/2} \psi, h_\tau^{-1} \left[ -\Delta , V^{1/2} \right] h_\tau^{-2} \left[ V^{1/2} , -\Delta \right] h_\tau^{-1} V^{1/2} \psi \right\rangle  \\ &\;\leq \; 2 \left\langle V^{1/2} \psi , h_\tau^{-1} V h_\tau^{-1} V^{1/2} \psi \right\rangle  + 2 \left\langle V^{1/2} \psi, h_\tau^{-1} \nabla \cdot \frac{\nabla V}{\sqrt{V}} \, h_\tau^{-2} \, \frac{\nabla V}{\sqrt{V}} \cdot \nabla h_\tau^{-1} V^{1/2} \psi \right\rangle \\ &\qquad + 2 \left\langle V^{1/2} \psi, h^{-1}_\tau \frac{\nabla V}{\sqrt{V}} \cdot \nabla \, h_\tau^{-2} \, \nabla \cdot \frac{\nabla V}{\sqrt{V}} h_\tau^{-1} V^{1/2} \psi \right\rangle \,, \end{split} \end{equation}
where we used the Cauchy-Schwarz inequality in the form $\scalar{x}{(A + B) (A^* + B^*) x} \leq 2 \scalar{x}{AA^* x} + 2 \scalar{x}{BB^* x}$, as well as the identities $[A,B^{-1}] = B^{-1} [B,A] B^{-1}$ and $[-\Delta, V^{1/2}] = - \nabla \cdot \frac{\nabla V}{\sqrt{V}} - \frac{\nabla V}{\sqrt{V}} \cdot \nabla$. We observe that 
\[ \left\| h_\tau^{-1/2}  V^{1/2} \varphi \right\|^2 \;=\; \left\langle V^{1/2} \varphi , h_\tau^{-1} \, V^{1/2} \varphi \right\rangle  \;\leq\; \| V/v_\tau \|_{L^\infty} \| \varphi \|^2 \;\leq\; \frac{1}{1-\rad} \| \varphi \|^2\,, \]
since $v_\tau (x) \geq (1-\rad) V(x)$ for all $\tau > 0$ (because $v_\tau \in B_\rad (V)$). This implies that 
\[ \| h_\tau^{-1/2} V^{1/2} \|_{\fra S^\infty} \;=\; \| V^{1/2} h_\tau^{-1/2}  \|_{\fra S^\infty} \;\leq\; \frac{1}{1-\rad} \,.\]
Similarly, 
\[ \left\| h_\tau^{-1/2} \frac{\nabla V}{\sqrt{V}} \right\|_{\fra S^\infty} \;=\; \left\| \frac{\nabla V}{\sqrt{V}} h_\tau^{-1/2} \right\|_{\fra S^\infty} \;\leq\; \frac{\| \nabla V \|_V}{1-\rad}\,. \] 
Moreover, $\| \nabla \, h_\tau^{-1/2} \|_{\fra S^\infty} = \| h_\tau^{-1/2} \nabla \|_{\fra S^\infty} \leq 1$.

{F}rom (\ref{eq:h-1V}), we find
\[ \begin{split} \| h_\tau^{-1} V \psi \|^2 &\;\leq \; \left[ 2 \| h_\tau^{-1/2} V^{1/2} \|_{\fra S^{\infty}}^4 + 4 \| h_\tau^{-1/2} V^{1/2} \|^2_{\fra S^\infty} \| \nabla h_\tau^{-1/2} \|_{\fra S^\infty}^2 \left\| \frac{\nabla V}{\sqrt{V}} h_\tau^{-1/2} \right\|_{\fra S^\infty}^2 \| h_\tau^{-1} \|_{\fra S^\infty} \right] \| \psi \|^2 \\ 
&\;\leq \;  \frac{4}{(1-\rad)^4} (1+ \kappa^{d/2-2} \| \nabla V \|_V^2) \| \psi \|^2 \,. \end{split} \]
This proves that $\| h_\tau^{-1} V \|_{\fra S^\infty} < \infty$, uniformly in $\tau > 0$. Analogously, we find $\| h^{-1} V \|_{\fra S^\infty} < \infty$. Hence, (\ref{eq:HS-conv1}) implies $\lim_{\tau \to \infty} \| h_\tau^{-1} - h^{-1} \|_{\fra S^2} = 0$.
\end{proof}

\appendix

\section{Borel summation} \label{sec:borel}

The following result is the key tool that allows us to deduce the convergence of analytic functions from the convergence of the coefficients of their asymptotic expansions. For its statement, for $R \geq 1$ we introduce the open ball
\begin{equation} \label{def_CR}
\mathcal{C}_R \;\deq\; \{z \in \C \col \re z^{-1}>R^{-1}\}\,.
\end{equation}

\begin{theorem}
\label{Borel summation convergence}
Let $(A^{\xi})_{\xi}$ and $(A_{\tau}^{\xi})_{\xi,\tau}$ be families of functions that are analytic in $\mathcal{C}_R$. The first family is indexed by a parameter $\xi$ in some arbitrary set and the second one is, in addition, indexed by $\tau>0$.
Suppose that, for all $M \in \mathbb{N}$, $A^{\xi}$ and $A_{\tau}^{\xi}$ are given by the asymptotic expansions
\begin{equation}
\label{Asymptotic expansion tau xi}
A^{\xi}(z)\;=\;\sum_{m=0}^{M-1}a_{m}^{\xi} z^m + R_{M}^{\xi}(z) \quad \mbox{and} \quad  A_{\tau}^{\xi}(z)\;=\;\sum_{m=0}^{M-1}a_{\tau,m}^{\xi} z^m + R_{\tau,M}^{\xi}(z)\quad \mbox{in}\quad \mathcal{C}_R \,,
\end{equation}
where the explicit terms satisfy
\begin{equation}
\label{Explicit term bound tau xi}
\sup_{\xi} \big|a_{m}^{\xi}\big| +
\sup_{\tau,\xi} \big|a_{\tau,m}^{\xi}\big| \;\leq\; \nu  \sigma^m  m!
\end{equation}
and the remainder terms satisfy
\begin{equation}
\label{Remainder term bound tau xi}
\sup_{\xi} \Big|R_{M}^{\xi}(z)\Big|+
\sup_{\tau,\xi} \Big|R_{\tau,M}^{\xi}(z)\Big| \;\leq\; \nu  \sigma^M  M!  |z|^M\quad \mbox{for all}\quad z \in \mathcal{C}_R
\end{equation}
for some constants $\nu>0$ and $\sigma \geq 1$, which are independent of $m$ and $M$.

Moreover, suppose that the differences of the explicit terms satisfy
\begin{equation}
\label{Difference of explicit terms}
\sup_{\xi} \big|a_{\tau,m}^{\xi}-a_m^{\xi}\big| \rightarrow 0  \quad \mbox{as}\quad \tau \rightarrow \infty \,.
\end{equation}
Then
\begin{equation} \label{A_A_tau_conv}
\notag
\sup_{\xi} \big|A_{\tau}^{\xi} - A^{\xi}\big| \rightarrow 0 \quad \mbox{pointwise in}\quad \mathcal{C}_R \quad \mbox{as}\quad \tau \rightarrow \infty \,.
\end{equation}
\end{theorem}

The proof of Theorem \ref{Borel summation convergence} is based on Borel summation techniques inspired by the following result of Sokal from \cite{Sokal}.

\begin{theorem}[after \cite{Sokal}]
\label{Borel summation}
Given $R \geq 1$, let $A$ be an analytic function on the ball \eqref{def_CR}.
Suppose that for all $M \in \mathbb{N}$ the function $A$ is given by the asymptotic expansion
\begin{equation}
\label{Asymptotic expansion}
A(z) \;=\; \sum_{m=0}^{M-1}a_m z^m + R_M(z)\quad\mbox{in}\quad\mathcal{C}_R 
\end{equation}
where the explicit terms satisfy
\begin{equation}
\label{Explicit term bound}
|a_m| \;\leq\;  \nu \sigma^m m!
\end{equation}
and the remainder term satisfies
\begin{equation}
\label{Remainder term bound}
\big|R_M(z)\big| \;\leq\; \nu  \sigma^M  M!  |z|^M\quad\mbox{for all}\quad z \in \mathcal{C}_R
\end{equation}
for some constants $\nu>0$ and $\sigma \geq 1$, which are independent of $m$ and $M$.

Then, the following statements hold.
\begin{enumerate}
\item The series
\begin{equation}
\label{B(t)}
B(t) \;\deq\; \sum_{m=0}^{\infty} a_m t^m /m!
\end{equation}
converges absolutely for all $|t| < \frac{1}{\sigma}$ and it has an analytic extension to the region 
\begin{equation}
\label{Region S}
S_{\sigma} \;\deq\; \bigg\{t\col \dist (t,\mathbb{R}^{+}) <\frac{1}{\sigma}\bigg\}\,.
\end{equation}
\item In the region $\bar{S}_{2\sigma}=\Big\{t\col \dist (t,\mathbb{R}^{+}) \leq \frac{1}{2\sigma}\Big\}$, the function $B(t)$ satisfies the bound
\begin{equation}
\label{Bound on B}
\big|B(t)\big| \;\leq\; C_0  R  \ee^{|t|/R}  \nu  \sigma
\end{equation}
for some universal constant $C_0>0$.
\item We have
\begin{equation}
\label{Formula for f}
A(z) \;=\; \frac{1}{z}  \int_0^{+\infty} \ee^{-t/z}  B(t)\,\dd t \quad \mbox{for all} \quad z \in \mathcal{C}_R \,.
\end{equation}
\end{enumerate}
\end{theorem}

The proof of Theorem \ref{Borel summation} was outlined in \cite{Sokal} which, in turn, is based on the methods used in the proof of \cite[Theorem 136]{Hardy} and it rediscovers and extends the work of \cite{Nevanlinna}. The proof in \cite{Hardy} is based on the work
\cite{Watson}, in which the function $A$ is assumed to be analytic in a larger domain. For completeness and with applications to the proof of Theorem \ref{Borel summation convergence} below in mind, we give a full proof of Theorem \ref{Borel summation}. Theorem \ref{Borel summation convergence} itself is proved at the end of this appendix, using tools from the proof of Theorem \ref{Borel summation}.

\begin{proof}[Proof of Theorem \ref{Borel summation}]
The fact that the series $B(t)$, defined in \eqref{B(t)} converges absolutely for all $|t| < \frac{1}{\sigma}$ follows immediately from the assumption \eqref{Explicit term bound}.
We now explain how to extend $B$ analytically to the whole region $S_{\sigma}$.
Following \cite{Sokal}, we define for  $k \in \N,r \in (0,R), t \geq 0$ the quantities
\begin{equation}
\label{b_k(t)}
b_k(t)\;\deq\;a_k+\frac{1}{2 \pi \mathrm{i}}  \mathop{\int}_{\re \zeta=r^{-1}} \ee^{t \zeta}  \zeta^{k-1}  \bigg(A(\zeta^{-1})-\sum_{m=0}^{k}a_m \zeta^{-k}\bigg) \, \dd \zeta\,.
\end{equation}
Note that the above integrals are absolutely convergent for $t \geq 0$ and independent of $r$ for $0<r<R$.

In order to deduce the absolute convergence, we note that by \eqref{Asymptotic expansion} and \eqref{Remainder term bound}, for $\re \zeta=r^{-1}$ we have
\begin{equation}
\notag
\bigg|A(\zeta^{-1})-\sum_{m=0}^{k}a_m \zeta^{-m}\bigg| \;\leq\; \nu \sigma^{k+1}  (k+1)!  |\zeta|^{-(k+1)}\,.
\end{equation}
Hence
\begin{multline*}
\int_{\re \zeta=r^{-1}}   \bigg|\ee^{t \zeta}  \zeta^{k-1}  \Big(A(\zeta^{-1})-\sum_{m=0}^{k}a_m \zeta^{-m}\Big) \bigg| \,|\dd\zeta| 
\;\leq\; \mathop{\int}_{\re \zeta=r^{-1}} \ee^{t/r}  |\zeta|^{k-1}  \nu  \sigma^{k+1}  (k+1)!  |\zeta|^{-(k+1)} \,|\dd\zeta|
\\
=\;\ee^{t/r}  \nu  \sigma^{k+1}  (k+1)!  \mathop{\int}_{\re \zeta=r^{-1}} |\zeta|^{-2} \,|\dd \zeta|
\;=\; \nu \pi  r  \ee^{t/r}  \sigma^{k+1}  (k+1)!
\end{multline*}
In the last equality, we used the identity
\begin{equation}
\label{Integral 1/|zeta|^2}
\mathop{\int}_{\re \zeta=r^{-1}} |\zeta|^{-2} \,|\dd \zeta| \;=\;\pi r\,.
\end{equation}
Absolute convergence now follows.

In order to see that each $b_k(t)$ is independent of $r$ for $0<r<R$, we note that the function
\begin{equation}
\notag
\sign(\zeta)\;\deq\;\ee^{t\zeta}  \zeta^{k-1}  \Big(A(\zeta^{-1})-\sum_{m=0}^{k}a_m  \zeta^{-m}\Big)
\end{equation}
is analytic in the region $\re \zeta>0$,
and that for $\zeta=r^{-1}$, it satisfies the bound
\begin{equation}
\notag
|\sign(\zeta)| \;\leq\; \ee^{t/r}  \nu  \sigma^{k+1}  (k+1)!  \frac{1}{|\zeta|^2}\,. 
\end{equation} 
The claim that $b_m(t)$ is independent of $r$ follows from Cauchy's theorem and the fact that, for fixed $0<r_2<r_1<R$
\begin{equation}
\notag
\int_{1/r_1}^{1/r_2} \frac{1}{|s\pm L \mathrm{i}|^2} \,\dd s  \;\leq\; \Big(\frac{1}{r_2}-\frac{1}{r_1}\Big)  \frac{1}{L^2} \;\rightarrow\; 0 \quad\mbox{as}\quad L \rightarrow \infty \,.
\end{equation}
In particular, we can take $r$ to be arbitrarily close to $R$.

In the following we use the identity
\begin{equation}
\label{Borel summation star}
\frac{1}{2 \pi \mathrm{i}} \mathop{\int}_{\re \zeta=r^{-1}} \ee^{t \zeta}  \zeta^{-(m+1)} \,\dd \zeta \;=\; \frac{t^m}{m!}\quad\mbox{for all}\quad m \in \N \,,
\end{equation}
which follows easily from the residue theorem.
Recalling \eqref{b_k(t)}, we write, for $M \in \mathbb{N}$
\begin{align*}
b_0(t) &\;=\; a_0+\frac{1}{2 \pi \mathrm{i}}  \mathop{\int}_{\re \zeta=r^{-1}} \ee^{t\zeta}  \zeta^{-1}  \Big(A(\zeta^{-1})-a_0\Big) \,\dd \zeta
\\
&\;=\; a_0+\frac{1}{2 \pi \mathrm{i}}  \mathop{\int}_{\re \zeta=r^{-1}} \ee^{t \zeta}  \zeta^{-1}  \bigg(A(\zeta^{-1})-\sum_{m=0}^{M-1}a_m\zeta^{-m}/m! \bigg) \,\dd \zeta
\\
&\qquad
+\frac{1}{2 \pi \mathrm{i}}  \mathop{\int}_{\re \zeta=r^{-1}} \ee^{t \zeta}  \zeta^{-1}  \bigg(\sum_{m=1}^{M-1}a_m\zeta^{-m}/m!\bigg)  \,\dd \zeta\,,
\end{align*}
which by \eqref{Asymptotic expansion} and \eqref{Borel summation star} is equal to
\begin{equation}
\label{seven}
\bigg(a_0+\sum_{m=1}^{M-1}a_m t^m/m!\bigg)+\frac{1}{2 \pi \mathrm{i}}  \mathop{\int}_{\re \zeta=r^{-1}} \ee^{t \zeta}  \zeta^{-1}  R_M(\zeta^{-1})  \,\dd \zeta \,.
\end{equation}
Let $t>0$. We choose $M>\frac{t}{R}$ in \eqref{seven}, and we let 
\begin{equation}
\label{choice of r}
r\;\deq\;\frac{t}{M} \,. 
\end{equation}
For parameters chosen in this way, we need to estimate the remainder term in \eqref{seven}. We use \eqref{Remainder term bound} in order to deduce that
\begin{multline}
\label{7B}
\bigg|\frac{1}{2 \pi \mathrm{i}}  \mathop{\int}_{\re \zeta=r^{-1}} \ee^{t \zeta}  \zeta^{-1}  R_M(\zeta^{-1})  \,\dd \zeta\bigg| \;\leq\; \frac{1}{2\pi}  \mathop{\int}_{\re \zeta=r^{-1}} \ee^{t/r}  \nu  \sigma^M  M!  |\zeta|^{-(M+1)} \,|\dd \zeta| 
\\
\leq\; \frac{1}{2\pi}  \ee^{t/r}  \nu  \sigma^M  M!   r^{M-1}  \mathop{\int}_{\re \zeta=r^{-1}} \frac{1}{|\zeta|^2} \,\dd \zeta \;=\; \frac{1}{2}  \ee^{t/r}  \nu  \sigma^M  M!   r^M\,,
\end{multline}
where we used \eqref{Integral 1/|zeta|^2}.
In particular, by \eqref{choice of r} and by Stirling's formula, the right-hand side of \eqref{7B} is
\begin{equation*}
\;\leq\; C  e^M  \nu  \sigma^M  \sqrt{M}  \Big(\frac{M}{e}\Big)^M  \Big(\frac{t}{M}\Big)^M 
\;\leq\;  C  \nu  \sqrt{M}  \big(\sigma t)^M\,,
\end{equation*}
which converges to zero as $M \rightarrow \infty$ provided that $t< \frac{1}{\sigma}$.
In particular, we note that 
\begin{equation}
\label{b_0(t) identity}
b_0(t)\;=\;\sum_{m=0}^{\infty} a_m t^m/m!\;=\;B(t)
\end{equation}
for all $0 < t < \frac{1}{\sigma}$.

In order to show that the function $B$ has the desired analytic continuation property, we relate the functions $b_k$ defined in \eqref{b_k(t)} with $b_0$. Moreover, we show upper bounds on $|b_k(t)|$. In order to address the first point, we show that, for all $k \in \N$ and  $t \geq 0$
\begin{equation}
\label{kth derivative of b_0}
b_0^{(k)}(t)\;=\;b_k(t) \,.
\end{equation}
In particular, using \eqref{b_0(t) identity} in \eqref{kth derivative of b_0}, we deduce that, for all $k\in \N$
\begin{equation}
\label{b_k(0)}
b_k(0)\;=\;a_k \,.
\end{equation}
We show the identity \eqref{kth derivative of b_0} inductively. The claim when $k=0$ holds by definition. 

We compute
\begin{align*}
b_0'(t) &\;=\;\frac{1}{2 \pi \mathrm{i}}  \mathop{\int}_{\re \zeta=r^{-1}} \ee^{t \zeta}  \Big(A(\zeta^{-1})-a_0\Big) \,d\zeta
\\
&\;=\;a_1  \frac{1}{2 \pi \mathrm{i}}  \mathop{\int}_{\re \zeta=r^{-1}} \ee^{t\zeta}  \zeta^{-1} \,d\zeta +\frac{1}{2 \pi \mathrm{i}}  \mathop{\int}_{\re \zeta=r^{-1}} \ee^{t\zeta}  \Big(A(\zeta^{-1})-a_0-a_1\zeta^{-1}\Big) \,\dd \zeta
\\
&\;=\; a_1+\frac{1}{2 \pi \mathrm{i}}  \mathop{\int}_{\zeta=r^{-1}} \ee^{t\zeta}  \Big(A(\zeta^{-1})-a_0-a_1\zeta^{-1}\Big) \, \dd \zeta \;=\; b_1(t) \,,
\end{align*}
where in the last step we used \eqref{Borel summation star}.
We know from the discussion immediately following \eqref{b_k(t)} that the obtained integral is absolutely convergent, which justified the differentiation under the integral. We iterate this procedure in order to obtain the identity \eqref{kth derivative of b_0}.

In order to obtain a good bound on $|b_k(t)|$, we first estimate
\begin{equation}
\label{b_k(t) bound first step}
\big|b_k(t) - a_k\big| \;=\; \bigg|\frac{1}{2 \pi \mathrm{i}}  \mathop{\int}_{\re \zeta=r^{-1}} \ee^{t \zeta}  \zeta^{k-1}  R_{k+1}(\zeta^{-1}) \, \dd\zeta\bigg| \,.
\end{equation}
By \eqref{Remainder term bound}, the expression in \eqref{b_k(t) bound first step} is bounded by
\begin{multline*}
\frac{1}{2\pi}  \mathop{\int}_{\re \zeta=r^{-1}} \ee^{t/r}  |\zeta|^{k-1}  \nu  \sigma^{k+1}  (k+1)!  |\zeta|^{-(k+1)} \,|\dd \zeta|
\\
=\;\frac{1}{2\pi}  \ee^{t/r}  \nu  \sigma^{k+1}  (k+1)!  \mathop{\int}_{\re \zeta=r^{-1}}  \frac{1}{|\zeta|^2} \, \dd \zeta=\frac{1}{2}  r  \ee^{t/r}  \nu  \sigma^{k+1}  (k+1)! \,,
\end{multline*}
where we again used \eqref{Integral 1/|zeta|^2}.
In particular, we can let $r \rightarrow R$ and deduce that
\begin{equation}
\notag
\big|b_k(t)-a_k\big| \;\leq\; R  \ee^{t/R}  \nu  \sigma^{k+1}  (k+1)! \,.
\end{equation}
Using \eqref{Explicit term bound} and $R,\sigma \geq 1$, it follows that
\begin{equation}
\label{Bound on b_k(t)}
\big|b_k(t)\big| \;\leq\; 2 R  \ee^{t/R}  \nu  \sigma^{k+1}  (k+1)! 
\end{equation}
for all $t \geq 0$.

For fixed $t_0 \geq 0$, define the function
\begin{equation}
\label{B^{t_0}(t)}
B^{t_0}(t)\;\deq\;\sum_{k=0}^{\infty} b_k(t_0)  (t-t_0)^k/k! \,.
\end{equation}
By using \eqref{Bound on b_k(t)}, it follows that the series in \eqref{B^{t_0}(t)} is absolutely convergent for $t \in \mathbb{C}$ with $|t-t_0| < \frac{1}{\sigma}$\,. Hence, it defines an analytic function in this domain. 
From \eqref{b_k(0)}, it follows that $B^0 = B$.
Suppose that $0 \leq t_0 <t_1$ are such that $t_1-t_0<1/\sigma$. Then, using \eqref{kth derivative of b_0}, \eqref{Bound on b_k(t)}, and Taylor's theorem, it follows that for all $t \in (t_0,t_1)$
\begin{equation}
\notag
\sum_{k=0}^{\infty} b_k(t_0)  (t-t_0)^k /k! \;=\; \sum_{k=0}^{\infty} b_k(t_1)  (t-t_1)^k /k! \;=\; b_0(t) \,.
\end{equation}
In particular, the functions $B^{t_0}$ and $B^{t_1}$ agree on the interval $(t_0,t_1)$. By analyticity, it follows that they agree on the whole intersection of their domains of definition, i.e.\ on the intersection of the balls of radius $1/\sigma$ centred at $t_0$ and at $t_1$ respectively.
Consequently, we can extend $B=B^0$ to the whole strip $S_{\sigma}$ by letting
\begin{equation}
\label{definition of B}
B(t)\;\deq\;B^{t_j}(t)
\end{equation}
for some $t_j>0$ with $|t-t_j|<1/\sigma$. This is a well-defined analytic function in $S_{\sigma}$. Part (i) now follows.

We now prove part (ii), i.e.\ the bound on $\big|B(t)\big|$ given by \eqref{Bound on B}.

Let $t \in \bar{S}_{2\sigma}$ be given. We can find $t_0>0$ such that $|t-t_0| \leq 1/(2\sigma)$\,.
By the construction in \eqref{definition of B}, we know that $B(t)=B^{t_0}(t)$. We substitute \eqref{Bound on b_k(t)} into \eqref{B^{t_0}(t)} in order to deduce that
\begin{equation}
\notag
\big|B(t)\big| \;\leq\; 2R  \ee^{t_0/R}  \nu  \sum_{k=0}^{\infty} (k+1)  \sigma  \Big(\frac{1}{2}\Big)^k \;=\; C  R  \ee^{t_0/R}  \nu  \sigma
\end{equation}
for some universal constant $C>0$. Moreover,
\begin{equation}
\notag
\ee^{t_0/R}\;=\;\ee^{|t_0|/R} \;\leq\; \ee^{(|t|+|t-t_0|)/R}\;=\; \ee^{|t|/R}  \ee^{|t-t_0|/R} \;\leq\; \ee^{|t|/R}  \ee^{1/(2\sigma R)} \;\leq\; \ee^{|t|/R}  \ee^{1/2} \,.
\end{equation}
Here we used that $R, \sigma \geq 1$.
Consequently, $\big|B(t)\big| \leq C_0  R  \ee^{|t|/R}  \nu  \sigma$
for some universal constant $C_0$. Hence, \eqref{Bound on B} follows.

We now prove part (iii), i.e.\ \eqref{Formula for f}, which allows us to write $f$ in terms of $B$. In the proof of part (i), we saw that, for all $t \geq 0$,
\begin{equation}
\notag
B(t)\;=\;b_0(t)\;=\; a_0+\frac{1}{2 \pi \mathrm{i}}  \mathop{\int}_{\re \zeta=r^{-1}} \ee^{t\zeta}  \frac{A(\zeta^{-1})-a_0}{\zeta} \, \dd \zeta \,.
\end{equation}
Here $r \in (0,R)$ is arbitrary.
We then compute, for $z \in \mathcal{C}_R$, the right-hand side of \eqref{Formula for f}
\begin{align}
\frac{1}{z}  \int_0^{+\infty} \ee^{-t/z}  B(t) \, \dd t 
\label{RHS of formula for f}
\;=\;
a_0+ \frac{1}{z}  \frac{1}{2 \pi \mathrm{i}}  \int_0^{+\infty} \mathop{\int}_{\re \zeta=r^{-1}}  \ee^{-t/z+t \zeta}  \frac{A(\zeta^{-1})-a_0}{\zeta} \, \dd \zeta\,\dd t \,.
\end{align}
By \eqref{Remainder term bound}, we know that $\absb{\frac{A(\zeta^{-1})-a_0}{\zeta}} \leq C/|\zeta|^2$
for some universal constant $C>0$; this is integrable over $\re \zeta =r^{-1}$ for all $r \in (0,R)$.
Furthermore, choose $r>0$ in the integral in \eqref{RHS of formula for f} such that $\re z^{-1}>r^{-1}>R^{-1}$. This is possible to do since $z \in \mathcal{C}_R$. In particular, for all $t>0$ we have $\re \pb{-\frac{t}{z}+t \zeta} < 0$.
Hence, the integral in \eqref{RHS of formula for f} is absolutely convergent. Therefore, we can interchange the orders of integration and deduce that
\begin{multline*}
\frac{1}{z}  \int_0^{+\infty} \ee^{-t/z}  B(t) \,\dd t \;=\;
a_0 +\frac{1}{z}  \frac{1}{2 \pi \mathrm{i}}  \mathop{\int}_{\re \zeta =r^{-1}} \int_0^{+\infty}  \ee^{-t/z+t \zeta}  \frac{A(\zeta^{-1})-a_0}{\zeta} \, \dd t\, \dd \zeta
\\
=\; a_0+\frac{1}{z}  \frac{1}{2 \pi \mathrm{i}}  \mathop{\int}_{\re \zeta=r^{-1}} \frac{1}{\frac{1}{z}-\zeta}  \frac{A(\zeta^{-1})-a_0}{\zeta} \, \dd \zeta \;=\; a_0-\frac{1}{z}  \frac{1}{2 \pi \mathrm{i}}  \mathop{\int}_{\re \zeta=r^{-1}}  \frac{A(\zeta^{-1})-a_0}{(\zeta-\frac{1}{z})  \zeta} \, \dd \zeta\,.
\end{multline*}
Note that
\begin{equation}
\notag
f(\zeta)\;\deq\;\frac{A(\zeta^{-1})-a_0}{(\zeta-\frac{1}{z})  \zeta}
\end{equation}
is a meromorphic function in $\re \zeta>R^{-1}$.  We now evaluate the above integral
by using the residue theorem. For fixed $L>0$, we let $\Phi_L$ be the  rectangular path in $\mathbb{C}$ whose top and bottom sides are $\im=L$ and $\im=-L$ respectively. The left side of $\Phi_L$ is $\re=r^{-1}$ and the right side is $\re=r^{-1}+L$. We orient $\Phi_L$ clockwise. For $L$ large enough, $\frac{1}{z}$ lies within $\Phi_L$. We henceforth assume that this is the case.

The only pole of $f$ inside $\Phi_L$ is at $\frac{1}{z}$, which is a simple pole.
We know that 
$ |A(\zeta^{-1})-a_0| \leq C/|\zeta|$ for some universal constant $C>0$. Consequently, for $|\zeta| \gg \frac{1}{|z|}$, it follows that $|f(\zeta)| \leq \frac{C}{|\zeta|^3}$
for some universal constant $C>0$.
Therefore, if we take $L \gg \frac{1}{|z|}$, then the top, right, and bottom contributions are $\r O\big(\frac{L}{L^3}\big) = \r O\big(\frac{1}{L^2}\big)$.
Hence, putting everything together, it follows from the residue theorem that

\begin{multline*}
\frac{1}{z}  \int_0^{+\infty} \ee^{-t/z}  B(t) \, \dd t\;=\;a_0- \frac{1}{z}   (-1)  \Res \bigg(\frac{A(\zeta^{-1})-a_0}{(\zeta-\frac{1}{z})  \zeta};\frac{1}{z}\bigg)
\\
\;=\;a_0+\frac{1}{z}  \frac{A(z)-a_0}{1/z}\;=\;a_0+(A(z)-a_0)\;=\;A(z) \,.
\end{multline*}
The factor of $-1$ in the above calculation came from the fact that $\Phi_L$ is oriented clockwise.
\end{proof}

\begin{proof}[Proof of Theorem \ref{Borel summation convergence}]
By translation, we can assume, without loss of generality, that $A^{\xi} \equiv 0$, i.e.\ $a_m^{\xi}=0$ for all $m$ and $\xi$.
Define
\begin{equation}
\label{Sum representation tau xi 0}
B_{\tau}^{\xi}(t)\;\deq\;\sum_{m=0}^{\infty} a_{\tau,m}^{\xi} t^m /m! \,.
\end{equation}
By using the bounds \eqref{Explicit term bound tau xi} and \eqref{Remainder term bound tau xi} and by the proof of Theorem \ref{Borel summation}, it follows that $B_{\tau}^{\xi}$ converges absolutely as a series for all $|t| < \frac{1}{\sigma}$ and that it has an analytic extension to the region $S_{\sigma}$. Moreover, the bound 
\begin{equation}
\label{Bound on B tau xi}
\sup_{\tau, \xi} \big|B_{\tau}^{\xi}(t)\big| \;\leq\; C_0  R  \ee^{|t|/R}  \nu  \sigma
\end{equation}
holds in $\bar{S}_{2\sigma}$ for some universal constant $C_0>0$. Finally, 
for all $\tau,\xi$,
\begin{equation}
\label{Formula for A tau xi}
A_{\tau}^{\xi}(z) \;=\; \frac{1}{z}  \int_0^{+\infty} \ee^{-t/z}  B_{\tau}^{\xi}(t) \, \dd t\quad \mbox{for all}\quad z \in \mathcal{C}_R\,.
\end{equation}
Given a non-negative integer $j$, let $t_j\deq\frac{j}{4\sigma}$ and let $S_j\deq\overline{B(t_j,\frac{1}{2\sigma})}$ be the closed ball of radius $\frac{1}{2\sigma}$ centred at $t_j$. Again, using the proof of Theorem \ref{Borel summation}, it follows that, on $S_j$ one can write
\begin{equation}
\label{Sum representation tau xi j}
B_{\tau}^{\xi}(t) \;=\; \sum_{k=0}^{\infty} b_{\tau,k}^{\xi}(t_j)  (t-t_j)^k/k!
\end{equation}
where the coefficients $b_{\tau,k}^{\xi}(t_j)$ satisfy
\begin{equation}
\label{Bound on the coefficients tau xi}
\sup_{\tau,\xi} \big|b_{\tau,k}^{\xi}(t_j)\big| \;\leq\; C_j  \sigma^k  k! 
\end{equation}
for some $C_j>0$ which depends on $j,R,\sigma,\nu$. Namely, for the latter bound, we are using \eqref{Bound on b_k(t)} applied in this context.

We now show that for all $j$
\begin{equation}
\label{Delta_j}
\sup_{\zeta \in S_j} \sup_{\xi} \big| B_{\tau}^{\xi}(\zeta)\big| \rightarrow 0 \quad \mbox{as}\quad  \tau \rightarrow \infty \,.
\end{equation}
We show \eqref{Delta_j} by induction.

For the base case $j=0$, we consider $t \in S_0$, which means that $|t| \leq \frac{1}{2\sigma}$. In particular, we can use \eqref{Sum representation tau xi 0} and deduce that
\begin{equation}
\notag
\sup_{\xi} \big| B_{\tau}^{\xi}(t)\big| \;\leq\; \sum_{m=0}^{\infty} \Big(\sup_{\xi}\big|a_{\tau,m}^{\xi}\big|\Big)  |t|^m /m! 
\;\leq\; \sum_{m=0}^{\infty} \Big(\sup_{\xi}\big|a_{\tau,m}^{\xi}\big|\Big)  \frac{1}{(2\sigma)^m}  \frac{1}{m!}\,. 
\end{equation}
The latter expression converges to zero as $\tau \rightarrow \infty$ by using \eqref{Explicit term bound tau xi}, \eqref{Difference of explicit terms} and the dominated convergence theorem. This proves the base case.

We now prove the induction step. Namely, we suppose that \eqref{Delta_j} holds for some $j$. We now show that it holds for $j+1$. Let $\Gamma_{j+1}$ denote the circle centred at $t_{j+1}$ of radius $\frac{1}{10\sigma}$. We note that then $\Gamma_{j+1}$ is contained in the ball $S_j$. We give $\Gamma_{j+1}$ the positive orientation. By Cauchy's integral formula,
\begin{equation}
\notag
b_{\tau,k}^{\xi}(t_{j+1}) \;=\; \frac{k!}{2\pi i}  \mathop{\oint}_{\Gamma_{j+1}} \frac{B_{\tau}^{\xi}(\zeta)}{(\zeta-t_{j+1})^{k+1}} \, \dd \zeta \,.
\end{equation}
In particular,
\begin{multline} \label{(**)_j}
\sup_{\xi} \big|b_{\tau,k}^{\xi}(t_{j+1})\big| \;\leq\; \frac{k!}{2\pi}  \mathop{\oint}_{\Gamma_{j+1}} \frac{\sup_{\xi} \big|B_{\tau}^{\xi}(\zeta)\big|}{|\zeta-t_{j+1}|^{k+1}} \, |\dd \zeta| 
\;\leq\; \frac{k!}{2\pi}  \frac{2\pi  \big(\frac{1}{10 \sigma}\big)}{\big(\frac{1}{10 \sigma}\big)^{k+1}}   \sup_{\zeta \in \Gamma_{j+1}} \sup_{\xi} \big|B_{\tau}^{\xi}(\zeta)\big|
\\
\leq\; 10^k  \sigma^k  k! \sup_{\zeta \in S_j} \sup_{\xi} \big|B_{\tau}^{\xi}(\zeta)\big| \rightarrow 0 \quad \mbox{as}\quad \tau \rightarrow \infty \,.
\end{multline}
In order to deduce the last step, we used the induction assumption.
Furthermore, by \eqref{Sum representation tau xi j} applied on $S_{j+1}$, we note that, for $t \in S_{j+1}$
\begin{equation}
\notag
\sup_{\xi} \big|B_{\tau}^{\xi}(t)\big| \;\leq\; \sum_{k=0}^{\infty} \Big(\sup_{\xi}\big|b_{\tau,k}^{\xi}(t_{j+1})\big|\Big)  \frac{1}{(2\sigma)^k}  \frac{1}{k!} \,.
\end{equation}
The quantity on the right-hand side converges to zero as $\tau \rightarrow \infty$ by using \eqref{Bound on the coefficients tau xi}, \eqref{(**)_j}, and the dominated convergence theorem.
In particular, we deduce that
\begin{equation}
\notag
\sup_{\zeta \in S_{j+1}} \sup_{\xi} \big| B_{\tau}^{\xi}(\zeta)\big| \rightarrow 0 \quad \mbox{as}\quad  \tau \rightarrow \infty \,.
\end{equation}
This finishes the induction.

From \eqref{Delta_j}, we deduce that
\begin{equation}
\label{Pointwise convergence xi tau}
\big|B_{\tau}^{\xi}\big| \rightarrow 0 \quad \mbox{pointwise on}\quad \bar{S}_{2\sigma} \quad \mbox{as}\quad \tau \rightarrow \infty \,.
\end{equation}
Now, \eqref{Formula for A tau xi} implies that, for all $z \in \mathcal{C}_R$
\begin{equation}
\notag
\sup_{\xi} \big|A_{\tau}^{\xi}(z)\big| \;\leq\; \frac{1}{|z|}  \int_{0}^{+\infty} \ee^{-t \re z^{-1}}  \Big(\sup_{\xi} \big|B_{\tau}^{\xi}(t)\big|\Big) \, \dd t \,,
\end{equation}
which converges to zero as $\tau \rightarrow \infty$ by \eqref{Bound on B tau xi}, \eqref{Pointwise convergence xi tau}, and the dominated convergence theorem.
\end{proof}

\section{The quantum Wick theorem} \label{sec:Wick_quantum}

In this appendix we review some standard facts about bosonic quasi-free states. Throughout this appendix, let $h > 0$ be a positive self-adjoint operator on $\fra H$ satisfying $\tr \ee^{-h} < \infty$. We write $h=\sum_{k \geq 0} \lambda_k u_k u_k^*$ and abbreviate $b_k \deq b(u_k)$.
We define the quasi-free state
\begin{equation*}
\rho_0^h(\cal A) \;\deq\; \frac{\tr \p{\cal A \, \ee^{-\sum_k \lambda_k b^*_k b_k}}}{\tr(\ee^{-\sum_k \lambda_k b^*_k b_k})}\,.
\end{equation*}

\begin{lemma}[Quantum Wick theorem]
\label{Quantum Wick theorem}
With the above notations, the following holds.
\begin{enumerate}
\item
We have
\begin{equation}
\label{QWT_i_1}
\scalar{f}{G^h \, g}\;\deq\; \rho_{0}^h \pb{b^*(g) \, b(f)} \;=\; \scalarbb{f}{\frac{1}{\ee^{h} - 1}\, g}
\end{equation}
and
\begin{equation}
\label{QWT_i_2}
\rho_{0}^h\pb{b(f) \, b(g)} \;=\; \rho_{0}^h\pb{b^*(f) \, b^*(g)} \;=\; 0
\end{equation}
for all $f,g \in \fra H$.
\item
Let $\cal A_1, \dots, \cal A_n$ be operators of the form $\cal A_i = b(f_i)$ or $\cal A_i = b^*(f_i)$, where $f_1, \dots, f_n \in \fra H$. Then we have
\begin{equation}
\label{QWT_ii}
\rho_{0}^h (\cal A_1 \cdots \cal A_n) \;=\; \sum_{\Pi} \prod_{(i,j) \in \Pi} \rho_{0}^h(\cal A_i \cal A_j)\,,
\end{equation}
where the sum ranges over all pairings of $\{1, \dots, n\}$, and we label the edges of $\Pi$ using ordered pairs $(i,j)$ with $i < j$.
\end{enumerate}
\end{lemma}

Before we proceed with the proof of Lemma \ref{Quantum Wick theorem}, we record three auxiliary results. First, using \eqref{CCR_b}, we get
\begin{equation}
\label{quotient rule}
\rho_0^h(\cal A) \;=\; \frac{\tr \p{\cal A \, \ee^{-\lambda_k b^*_k b_k}}}{\tr(\ee^{-\lambda_k b^*_k b_k})}\,,
\end{equation}
whenever $\cal A$ is a polynomial in the variables $b_k$ and $b_k^*$.
Second, we have the following factorization property.
\begin{lemma}
\label{factorization property}
Let $\cal A_1, \dots, \cal A_n$ be operators of the form $\cal A_i = b_{k_i}$ or $\cal A_i = b_{k_i}^*$ for $k_1, \dots, k_n \in \N$. 
Then we have
\begin{equation*}
\rho_{0}^h (\cal A_1 \cdots \cal A_n) \;=\;\prod_{k \geq 0} 
\,\rho_{0}^h \Bigg(\prod_{i \col k_i=k}\cal{A}_i\Bigg)\,,
\end{equation*}
where the ordering of the operators in the product is always increasing in $i$.
\end{lemma}

\begin{proof}
This is immediate from the definition of $\rho^h_0(\cdot)$ and \eqref{CCR_b}.
\end{proof}

\begin{lemma}[Gauge invariance]
\label{Quantum Wick Theorem Lemma}
Let $\cal A_1, \dots, \cal A_n$ be as in Lemma \ref{factorization property}.
Given $k \in \N$, we define
\begin{equation*}
\cal{N}_k^-\;\deq\; \Big|\big\{1 \leq i \leq n\col \cal A_i=b_{k}\big\}\Big|\,,\qquad \cal{N}_k^+\;\deq\; \Big|\big\{1 \leq i \leq n\col \cal A_i=b_{k}^*\big\}\Big|\,.
\end{equation*}
Then we have $\rho_{0}^h (\cal A_1 \cdots \cal A_n) = 0$
unless $\cal{N}_k^-=\cal{N}_k^+$ for all $k \in \N$.
\end{lemma}

\begin{proof}[Proof of Lemma \ref{Quantum Wick Theorem Lemma}]
By Lemma \ref{factorization property}, 
the claim follows if we show that, for fixed $k \in \N$ we have
\begin{equation*}
\rho_{0}^h \Bigg(\prod_{i \col k_i=k}\cal{A}_i\Bigg) \;=\; \tr \Bigg(\prod_{i \col k_i=k}\cal{A}_i \,\ee^{- \lambda_k b_k^* b_k}\Bigg) \bigg/ \tr\pb{\ee^{- \lambda_k b_k^* b_k}}\;=\;0
\end{equation*}
unless $\cal{N}_k^-=\cal{N}_k^+$. This, in turn, is a consequence of the identity
\begin{equation}
\label{Wick_prod}
\rho_{0}^h \Bigg(\prod_{i \col k_i=k}\cal{A}_i\Bigg) \;=\; \ee^{(\cal{N}^+_k-\cal{N}^-_k) \lambda_{k}} \rho_{0}^h \Bigg(\prod_{i \col k_i=k}\cal{A}_i\Bigg)\,,
\end{equation}
which follows by a repeated application of $\ee^{t b_k^* b_k} \, b_k^* \, \ee^{-t b_k^* b_k} = \ee^{t} \, b_k^*$, which follows from \eqref{Pullthrough formula} by setting $\tau = 1$.
\end{proof}

\begin{proof}[Proof of Lemma \ref{Quantum Wick theorem}]

(i) The claim \eqref{QWT_i_2} follows immediately by Lemma \ref{Quantum Wick Theorem Lemma}. We now show the claim \eqref{QWT_i_1}. By linearity and Lemma \ref{Quantum Wick Theorem Lemma} it suffices to prove
\begin{equation}
\label{Wick_A}
\rho_{0}^h(b^*_k\,b_k) \;=\;\frac{1}{e^{\lambda_k}-1}\,.
\end{equation}

We prove \eqref{Wick_A} by using the occupation state basis $(\psi_{\f m})_{\f m \in \N^\N}$, an orthonormal basis of $\cal F$ defined by $\psi_{\f m} \deq \prod_{l \in \N} \frac{(b^*_l)^{m_l}}{\sqrt{m_l!}} \, \Omega$,
where $\Omega=  (1,0,0,\dots) \in \cal F$ is the vacuum state and $\f m = (m_l)_{l \in \N}$ ranges over $\N^{\N}$. Using $b_k \Omega = 0$ and \eqref{CCR_b} we easily find
\begin{equation}
\begin{cases}
\label{Occupation State Basis b_k}
b_k \,\psi_{\f m} = \sqrt{m_k}\, \psi_{\tilde{\f m}} & \text{for } \tilde{m}_l=m_l-\delta_{kl}\\
b_k^*\; \psi_{\f m} = \sqrt{m_k+1}\, \psi_{\tilde{\f m}} & \text{for } \tilde{m}_l=m_l+\delta_{kl}\,.
\end{cases}
\end{equation}
Here we use the convention that $\psi_{\f m} = 0$ if an entry of $\f m$ is negative.
In particular, we have 
\begin{equation}
\label{eigenstate}
b_k^*\, b_k \, \psi_{\f m} \;=\; m_k \,\psi_{\f m}\,.
\end{equation}
Substituting \eqref{eigenstate} into \eqref{quotient rule}, we obtain

\begin{equation*}
\rho_{0}^h(b^*_k\, b_k)
\;=\; \frac{\sum_{m_k} m_k\, \ee^{-\lambda_k m_k}}{\sum_{m_k} \ee^{-\lambda_k  m_k}}\;=\;\frac{1}{e^{\lambda_k}-1}\,,
\end{equation*}
as was claimed.

(ii) We prove this part of the Lemma by first considering several special cases and by then deducing the general result.
In what follows, we denote for fixed $k \in \mathbb{N}$
\begin{equation}
\label{c_tau}
c(k) \;\deq\; \rho_{0}^h(b^*_k\,b_k)\;=\;\frac{1}{e^{\lambda_k}-1} \,.
\end{equation}
Moreover, we abbreviate $\cal A \deq\cal A_1 \cdots \cal A_n$.
\begin{enumerate}
\item[(1)] Suppose that $n=2r$ and, for a fixed $k \in \N$, $\cal A_i=b_k^*$ for $1 \leq i \leq r$ and $\cal A_i=b_k$ for $r+1 \leq i \leq n$. In other words, the operators $\cal A_1, \dots, \cal A_n$ are \emph{normal-ordered}: all creation operators are to the left of all annihilation operators. In this case we have
\begin{equation}
\label{Wick Case 1}
\rho_{0}^h (\cal A) \;=\;r! \,\big(c(k)\big)^r\,.
\end{equation}
We note that, by Lemma \ref{Quantum Wick Theorem Lemma} and \eqref{c_tau}, the expression on the right-hand side of \eqref{Wick Case 1} indeed equals to $\sum_{\Pi} \prod_{(i,j) \in \Pi} \rho_{0}^h(\cal A_i \cal A_j)$ in this case. We prove \eqref{Wick Case 1} by again using the occupation state basis. Namely, using \eqref{quotient rule} and \eqref{Occupation State Basis b_k}, we have
\begin{equation*}
\rho_{0}^h \Big((b_k^*)^r \,(b_k)^r \Big) \;=\;\frac{\sum_{m_k} m_k (m_k-1) \cdots (m_k-r+1)\,\ee^{-\lambda_k m_k}}{\sum_{m_k} \ee^{-\lambda_k m_k}}\;=\; r! \, \bigg(\frac{1}{e^{\lambda_k}-1}\bigg)^r\,,
\end{equation*}
as was claimed.
\item[(2)] Suppose that, for a fixed $k \in \N$ and for all $1 \leq i \leq n$, $\cal A_i=b^*_k$ or $\cal A_i=b_k$. 

Now the operators are no longer normal-ordered.
By Lemma \ref{Quantum Wick Theorem Lemma}, it follows that
\begin{equation*}
\rho_{0}(\cal A)\;=\;\sum_{\Pi} \prod_{(i,j) \in \Pi} \rho_{0}^h(\cal A_i \cal A_j)\;=\;0
\end{equation*}
unless $n=2r$ is even and that there are exactly $r$ values of $i$ for which $\cal A_i=b_k$. 
Let us now assume that this is the case.

If $\cal B_1, \ldots, \cal B_p$ are operators each of which is $b_k$ or $b^*_k$, then we write 
\begin{equation*}
\col \prod_{j=1}^{p} \cal B_j \col \;=\; (b^*_k)^{\cal N^+} \, (b_k)^{p-\cal N^+}
\end{equation*}
for the normal-ordering of $\prod_{j=1}^{p} \cal B_j$, where $\cal N^+$ is the number of $1 \leq j \leq p$ such that $\cal B_j = b^*_k$. (This is a slight abuse of notation since, in \eqref{Wick ordering quadratic}, we used $\col \cdot \col$ to denote the renormalized product of two operators that are linear in $b_k, b_k^*$. Since we never use these two operations simultaneously, it will be clear from context to which one we are referring.)

Given the sequence $\f {\cal A} =(\cal A_1, \ldots, \cal A_n)$ as in the assumptions, we say that the pair $(i,j)$ is an \emph{inversion} if $1 \leq i <j \leq n$ and $\cal A_i=b_k, \,\cal A_j=b^*_k$. We denote the set of all inversions by $\fra I \equiv  \fra I (\f {\cal A})$.
Furthermore, for $0 \leq s \leq r$, we say that $\cal M= \{(i_1,j_1),\ldots,(i_s,j_s)\}$ is an \emph{$s$-matching} if
\begin{enumerate}
\item[(i)] $i_1,\ldots,i_s,j_1,\ldots,j_s$ are distinct elements of $\{1,\ldots,n \}$.
\item[(ii)] $(i_q,j_q) \in \fra I$ for all $q=1,\ldots,s$. 
\end{enumerate}
Finally, given an $s$-matching $\cal M$, we write
\begin{equation*}
\cal A^{\cal M} \;\deq\; \cal A_1 \cdots  \hat{\cal A}_{i_1} \cdots \hat{\cal A}_{j_s} \cdots \cal A_n\,,
\end{equation*}
where we have omitted the factors $\cal A_{i_1}, \ldots, \cal A_{i_s}, \cal A_{j_1},\ldots,\cal A_{j_s}$ from the product and we have kept the original ordering which was increasing in $i$.
We denote the set of all $s$-matchings by $\bb M_s \equiv \bb M_s (\f {\cal A})$.

We then have
\begin{equation}
\label{normal ordered sum}
\cal A \;=\; \sum_{s=0}^{r} \sum_{\cal M \in \bb M_s} \col \cal A^{\cal M}\col\,.
\end{equation}
This claim is proved by induction on the number of inversions.

Namely, in the base case, where there are zero inversions, the operator $\cal A$ is already normal-ordered and the identity \eqref{normal ordered sum} immediately follows. For the induction step, we assume that the claim holds when there are $q \geq 0$ inversions. Suppose that $\f {\cal A} = (\cal A_1,\ldots,\cal A_n)$ has $q+1$ inversions. Then, we can find $1 \leq i \leq n$ such that $(i,i+1)$ is an inversion. We choose $i$ to be minimal. We note that by \eqref{CCR_b} we have
\begin{equation}
\label{sum induction step}
\cal A \;=\; \cal A_1 \cdots \cal A_{i-1} \cal A_{i+1} \cal A_i  \cal A_{i+2} \cdots \cal A_n + \cal A_1 \cdots \cal A_{i-1}  \cal A_{i+2} \cdots \cal A_n \,.
\end{equation}
Furthermore, we note that, in the sequences 
\begin{equation*}
\f {\cal A}_1\;=\;(\cal A_i ,\ldots, \cal A_{i-1},\cal A_{i+1},\cal A_i, \cal A_{i+2},\ldots,\cal A_n),\quad \f {\cal A}_2\;=\;(\cal A_i ,\ldots, \cal A_{i-1}, \cal A_{i+2},\ldots,\cal A_n)
\end{equation*}
there are exactly $q$ inversions. Therefore, we can apply the induction hypotheses to them. 
We can now deduce \eqref{normal ordered sum}. The contribution from $\f {\cal A}_1$ corresponds to the sum over all $\cal M$ with $(i,i+1) \notin \cal M$ and the contribution from $\f {\cal A}_2$ corresponds to the sum over all $\cal M$ with $(i,i+1) \in \cal M$.

Using \eqref{CCR_b}, we can write the right-hand side of \eqref{QWT_ii} as
$\sum_{\Pi} \prod_{(i,j) \in \Pi} \big(c(k)+\f 1_{\fra I} (i,j) \big)$.
Here $\f 1_{\fra I} (i,j)=1$ if $(i,j) \in \fra I$ and it is $0$ otherwise.
Multiplying everything out, recalling the result of case (1), and using \eqref{normal ordered sum} we obtain the claim in this case.

\item[(3)] Suppose that $\cal A_1, \dots, \cal A_n$ are operators of the form $\cal A_i = b_{k_i}$ or $\cal A_i = b_{k_i}^*$ for $k_1, \dots, k_n \in \N$. 

The claim in this case follows from case (2) and from Lemma \ref{factorization property}.

\item[(4)] The general case follows from case (3) and linearity. \qedhere
\end{enumerate}
\end{proof}

\section{Estimates on the quantum Green function} \label{sec:GF_estimates}

\begin{lemma} \label{lem:conv_G_tau}
Suppose that $h,h_\tau > 0$ satisfy $h^{-1}, h_\tau^{-1} \in \fra S^2(\fra H)$ and $\lim_{\tau \to \infty} \norm{h_\tau^{-1} - h^{-1}}_{\fra S^2(\fra H)} = 0$. Then
\begin{equation} \label{conv_h_tau}
\lim_{\tau \to \infty} \, \normbb{\frac{1}{\tau (\ee^{h_\tau / \tau} - 1)} - h^{-1}_\tau}_{\fra S^2(\fra H)} \;=\; 0\,,
\end{equation}
and
\begin{equation} \label{conv_tr_htau}
\lim_{\tau \to \infty} \frac{1}{\tau} \tr \frac{1}{\tau (\ee^{h_\tau/\tau} - 1)} \;=\; 0\,.
\end{equation}
\end{lemma}
\begin{proof}
We begin with \eqref{conv_h_tau}.
Note first that
\begin{multline} \label{convG2_tau}
\tr (h_\tau^{-2} - h^{-2}) \;\leq\; \norm{h_\tau^{-1} (h_\tau^{-1} - h^{-1})}_{\fra S^1(\fra H)} + \norm{(h_\tau^{-1} - h^{-1}) h^{-1}}_{\fra S^1(\fra H)}
\\
\leq\; \pb{\norm{h_\tau^{-1}}_{\fra S^2(\fra H)} + \norm{h^{-1}}_{\fra S^2(\fra H)}} \norm{h_\tau^{-1} - h^{-1}}_{\fra S^2(\fra H)} \to 0\,.
\end{multline}
Let $(\lambda_k)$ and $(\lambda_{\tau,k})$ denote the eigenvalues of $h$ and $h_\tau$ respectively. We order the eigenvalues in a nondecreasing fashion. Then we have that
\begin{equation} \label{conv_lambda_k1}
\lim_{\tau \to \infty} \lambda_{\tau,k} \;=\; \lambda_k
\end{equation}
for all $k$.
Here we used the infinite-dimensional version of the Hoffman-Wielandt inequality \cite[Theorem 2]{BhatiaElsner} applied to the operators $h_\tau^{-1}$ and $h^{-1}$. (Note that, since the eigenvalues are ordered in a nondecreasing fashion, there is no permutation of the indices. Moreover, even though the spectrum of $h^{-1}$ has a limit point at zero, we always consider the convergence in \eqref{conv_lambda_k1} for a fixed, finite $k$. Therefore, the additional subtlety arising from working with extended enumerations of the eigenvalues as in \cite{BhatiaElsner} is not present.)

Next, let $\epsilon > 0$ and pick $K = K(\epsilon)$ such that $\sum_{k > K} \lambda_k^{-2} \leq \epsilon$.
From \eqref{convG2_tau} and \eqref{conv_lambda_k1} we deduce that
\begin{equation} \label{conv_lambda_k2}
\lim_{\tau \to \infty} \sum_{k > K} \frac{1}{\lambda_{\tau,k}^2} \;\leq\; \epsilon\,.
\end{equation}
Using $\tau (\ee^{\lambda/\tau} - 1) \geq \lambda$ for all $\tau,\lambda > 0$ we get the estimate
\begin{align}
\label{intermediate step}
\normbb{\frac{1}{\tau (\ee^{h_\tau / \tau} - 1)} - h_\tau^{-1}}_{\fra S^2(\fra H)}^2 &\;=\; \sum_{k \in \N} \pbb{\frac{1}{\tau (\ee^{\lambda_{\tau,k}/\tau} - 1)} - \frac{1}{\lambda_{\tau,k}}}^2
\\
&\;\leq\; \sum_{k  = 0}^K \pbb{\frac{1}{\tau (\ee^{\lambda_{\tau,k}/\tau} - 1)} - \frac{1}{\lambda_{\tau,k}}}^2 + \sum_{k > K} \frac{1}{\lambda_{\tau,k}^2}\,,
\end{align}
from which we deduce, together with \eqref{conv_lambda_k1} and \eqref{conv_lambda_k2} that
\begin{equation*}
\limsup_{\tau \to \infty} \normbb{\frac{1}{\tau (\ee^{h_\tau / \tau} - 1)} - h_\tau^{-1}}_{\fra S^2(\fra H)}^2 \;\leq\; \epsilon\,.
\end{equation*}
Since $\epsilon > 0$ was arbitrary, \eqref{conv_h_tau} follows.

Next, we prove \eqref{conv_tr_htau}. We write
\begin{multline} \label{est_tr_htau}
\frac{1}{\tau} \tr \frac{1}{\tau (\ee^{h_\tau/\tau} - 1)} \;=\; \frac{1}{\tau^2} \sum_k  \frac{\ind{\lambda_{\tau,k} \leq \tau}}{\ee^{\lambda_{\tau,k}/\tau} - 1} + \frac{1}{\tau^2} \sum_k  \frac{\ind{\lambda_{\tau,k} > \tau}}{\ee^{\lambda_{\tau,k}/\tau} - 1}
\\
\leq\; \frac{C}{\tau} \sum_k  \frac{\ind{\lambda_{\tau,k} \leq \tau}}{\lambda_{\tau,k}} + C \sum_k  \frac{\ind{\lambda_{\tau,k} > \tau}}{\lambda_{\tau,k}^2}\,.
\end{multline}
Using \eqref{conv_lambda_k1}, $\lim_{k \to \infty} \lambda_k = \infty$, and \eqref{conv_lambda_k2}, we easily find that the second term on the right-hand side of \eqref{est_tr_htau} vanishes in the limit $\tau \to \infty$.

In order to estimate the first term on the right-hand side of \eqref{est_tr_htau}, we introduce the probability measure
\begin{equation}
\label{nu_tau}
\nu_\tau \;\deq\; \frac{1}{Z_\tau} \sum_k \frac{1}{\lambda_{\tau,k}^2} \, \delta_{\lambda_{\tau,k}}\,,
\end{equation}
where $Z_\tau > 0$ is a normalization constant. Hence, the first term of \eqref{est_tr_htau} is equal to
\begin{equation*}
\frac{C Z_\tau}{\tau} \int \nu_\tau(\dd x) \, x \, \ind{x \leq \tau}\,.
\end{equation*}
By \eqref{convG2_tau}, we know that $Z_\tau$ converges to some $Z > 0$ as $\tau \to \infty$, so that it suffices to prove that
\begin{equation*}
\lim_{\tau \to \infty} \frac{1}{\tau} \int \nu_\tau(\dd x) \, x \, \ind{x \leq \tau} \;=\; 0\,.
\end{equation*}
We argue by contradiction and suppose that there exists an $\epsilon > 0$ and a sequence $\tau_i \to \infty$ such that
\begin{equation} \label{contr_nu}
\frac{1}{\tau_i} \int \nu_{\tau_i}(\dd x) \, x \, \ind{x \leq \tau_i} \;\geq\; \epsilon
\end{equation}
for all $i$. By the convergence $Z_\tau \to Z$, \eqref{conv_lambda_k1}, $\lim_{k \to \infty} \lambda_k = \infty$, and \eqref{conv_lambda_k2}, we find that there exists a $x_* > 0$ such that $\nu_{\tau_i}([x_*,\infty)) \leq \epsilon/2$ for all $i$. Suppose without loss of generality that $\tau_i > x_*$ for all $i$. Then we get
\begin{align*}
\frac{1}{\tau_i} \int \nu_{\tau_i}(\dd x) \, x \, \ind{x \leq \tau_i} &\;=\; \frac{1}{\tau_i} \int \nu_{\tau_i}(\dd x) \, x \, \ind{x \leq x_*} +
\frac{1}{\tau_i} \int \nu_{\tau_i}(\dd x) \, x \, \ind{x_* < x \leq \tau_i}
\\
&\;\leq\; \frac{x_*}{\tau_i} + \frac{\epsilon}{2} \;<\; \epsilon
\end{align*}
for large enough $i$. This is the desired contradiction to \eqref{contr_nu}. This concludes the proof of \eqref{conv_tr_htau}.
\end{proof}

\begin{lemma} \label{cor:conv_G_tau corollary}
With the notations and assumptions as in Lemma \ref{lem:conv_G_tau}, we have that uniformly in $t \in (-1,1)$ 
\begin{equation} 
\label{conv_h_zeta_tau}
\lim_{\tau \to \infty} \, (1 + t) \normbb{\frac{\ee^{-t h_\tau / \tau}}{\tau (\ee^{h_\tau / \tau} - 1)} - h^{-1}_\tau}_{\fra S^2(\fra H)} \;=\; 0\,.
\end{equation}
\end{lemma}

\begin{proof}
By \eqref{conv_h_tau}, the claim \eqref{conv_h_zeta_tau} is equivalent to
\begin{equation} 
\label{conv_h_zeta_tau B}
\lim_{\tau \to \infty} (1 + t)  \normbb{\frac{\ee^{-t h_\tau / \tau}}{\tau (\ee^{h_\tau / \tau} - 1)} - \frac{1}{\tau (\ee^{h_\tau / \tau} - 1)}}_{\fra S^2(\fra H)} \;=\; 0
\end{equation}
uniformly in $t \in (-1,1)$.
We write
\begin{equation} \label{HS_t_est}
(1 + t)^2\normbb{\frac{\ee^{-t h_\tau / \tau}}{\tau (\ee^{h_\tau / \tau} - 1)}-\frac{1}{\tau (\ee^{h_\tau / \tau} - 1)}}_{\fra S^2(\fra H)}^2 \;=\; \sum_{k} (1 + t)^2 \pbb{\frac{\ee^{-t \lambda_{\tau,k}/\tau}-1}{\tau (\ee^{\lambda_{\tau,k}/\tau} - 1)}}^2\,.
\end{equation}
and split the sum into the parts $\lambda_{\tau,k} \leq \tau$ and $\lambda_{\tau,k} > \tau$. The first piece is estimated by
\begin{equation*}
    \sum_{k} \ind{\lambda_{\tau,k} \leq \tau} (1 + t)^2 \pbb{\frac{C \abs{t} \lambda_{\tau,k} / \tau}{\lambda_{\tau,k}}}^2 \;\leq\; \frac{C}{\tau^2} \sum_{k} \ind{\lambda_{\tau,k} \leq \tau} \;\leq\; \frac{C}{\tau} \sum_{k} \frac{\ind{\lambda_{\tau,k} \leq \tau}}{\lambda_{\tau,k}}\,,
\end{equation*}
which goes to zero as $\tau \to \infty$, as shown after \eqref{est_tr_htau}.

The second piece of \eqref{HS_t_est} is estimated as
\begin{multline*}
\sum_{k} \ind{\lambda_{\tau,k} > \tau} (1 + t)^2 \pbb{\frac{1 + \ee^{- t \lambda_{\tau,k}/\tau}}{\tau \ee^{\lambda_{\tau,k}/\tau}}}^2 \;\leq\;
C \sum_{k} \frac{\ind{\lambda_{\tau,k} > \tau}}{\tau^2} (1 + t)^2 \pb{\ee^{-2\lambda_{\tau,k}/\tau} + \ee^{-2(1 + t)\lambda_{\tau,k}/\tau}}
\\
\leq\; C \sum_{k} \ind{\lambda_{\tau,k} > \tau} (1 + t)^2 \pbb{\frac{1}{\lambda_{\tau,k}^2} + \frac{1}{(1 + t)^2\lambda_{\tau,k}^2}}
\;\leq\; C \sum_{k} \frac{\ind{\lambda_{\tau,k} > \tau}}{\lambda_{\tau,k}^2}\,,
\end{multline*}
which goes to zero as $\tau \to \infty$, as shown after \eqref{est_tr_htau}. This concludes the proof.
\end{proof}
An analogous result to Lemma \ref{cor:conv_G_tau corollary} holds under slightly different assumptions.

\begin{lemma} \label{cor:conv_G_tau corollary2}
Suppose that $h_{\tau}>0$ is such that $h_{\tau}^{-1} \in \fra S^2(\fra H)$. Furthermore, suppose that there exists $\rad \in (0,1)$ and a sequence $(\lambda_k)_k$ of positive real numbers such that 
\begin{equation}
\label{lambda comparison}
(1-\rad) \lambda_k \;\leq\; \lambda_{\tau,k} \; \leq \; (1+\rad) \lambda_{k}
\end{equation}
for all $\tau,k$, where $\lambda_{\tau,k}$ denote the eigenvalues of $h_{\tau}$.
Then \eqref{conv_h_zeta_tau} holds uniformly in $t \in (-1,1)$. 
\end{lemma}
\begin{proof}
By \eqref{lambda comparison} and $h_{\tau}^{-1} \in \fra S^2(\fra H)$ we deduce that 
\begin{equation}
\label{lambda k comparison}
\sum_k \frac{1}{\lambda_k^2} \;<\;\infty\,.
\end{equation}
Moreover, since by \eqref{lambda comparison}
\begin{equation*}
\bigg|\frac{1}{\tau(\ee^{\lambda_{\tau,k}/\tau}-1)}-\frac{1}{\lambda_{\tau,k}}\bigg| \;\leq\; \frac{C}{\lambda_{\tau,k}} \;\leq\; \frac{C}{\lambda_k}\,,
\end{equation*}
the dominated convergence theorem implies \eqref{lem:conv_G_tau}. 
Therefore, we reduce the claim to proving \eqref{conv_h_zeta_tau B} uniformly in $t \in (-1,1)$. Arguing as after \eqref{HS_t_est}, we need to show that $\txt{I}_\tau \deq \frac{1}{\tau} \sum_{k} \frac{\ind{\lambda_{\tau,k} \leq \tau}}{\lambda_{\tau,k}}$ and $\txt{II}_\tau \deq \sum_{k} \frac{\ind{\lambda_{\tau,k} > \tau}}{\lambda_{\tau,k}^2}$ converge to zero as $\tau \rightarrow \infty$. By \eqref{lambda comparison}, we have
\begin{equation*}
\txt{I}_\tau \;\leq\; \frac{1}{(1-\rad)}\,\frac{1}{\tau} \,\sum_{k} \frac{\ind{\lambda_k \leq (1+\rad) \tau}}{\lambda_k}\,,
\end{equation*}
which converges to zero as $\tau \rightarrow \infty$ by considering the probability measure 
$\nu \deq \frac{1}{Z} \sum_k \frac{1}{\lambda_k^2} \delta_{\lambda_k}$ for an appropriate normalization constant $Z>0$ and using the argument following \eqref{est_tr_htau}.
Furthermore,
\begin{equation*}
\txt{II}_\tau \;\leq\; \frac{1}{(1-\rad)^2} \frac{1}{\tau} \sum_k \frac{1}{\lambda_k^2} \indbb{\lambda_k > \frac{\tau}{(1+\rad)}}\,, 
\end{equation*}
which converges to zero as $\tau \rightarrow \infty$ by \eqref{lambda k comparison}.
\end{proof}

\begin{lemma} With the notations and assumptions as in Lemma \ref{lem:conv_G_tau}, we have uniformly in $t \in (0,1)$ that
\label{Heat kernel estimate}
\begin{equation*}
\lim_{\tau \rightarrow \infty} \frac{t}{\tau} \,\big\|\ee^{-t h_\tau/\tau}\big\|_{\fra S^2(\fra H)} \;=\;0\,.
\end{equation*}
\end{lemma}

\begin{proof}
We compute
\begin{multline*}
\frac{t^2}{\tau^2}\, \big\|\ee^{-t h_\tau/\tau}\big\|_{\fra S^2(\fra H)}^2 \;=\;\frac{t^2}{\tau^2}\, \sum_{k \in \N} \ee^{-2t \lambda_{\tau,k}/\tau}
\;\leq\; \sum_{k} \frac{t^2}{\tau^2} \,\frac{C \tau}{t \,\lambda_{\tau,k}}\ind{\lambda_{\tau,k} \leq \tau} +
\sum_{k} \frac{t^2}{\tau^2} \,\frac{C \tau^2}{t^2 \,\lambda_{\tau,k}^2}\ind{\lambda_{\tau,k} > \tau} 
\\
\leq\; \frac{C}{\tau} \sum_k  \frac{\ind{\lambda_{\tau,k} \leq \tau}}{\lambda_{\tau,k}} + C \sum_k  \frac{\ind{\lambda_{\tau,k} > \tau}}{\lambda_{\tau,k}^2}\,,
\end{multline*}
which tends to zero as $\tau \to \infty$, as shown after \eqref{est_tr_htau}.
\end{proof}

\end{document}

%% file: graph1.pdf_tex
\begingroup%
  \makeatletter%
  \providecommand\color[2][]{%
    \errmessage{(Inkscape) Color is used for the text in Inkscape, but the package 'color.sty' is not loaded}%
    \renewcommand\color[2][]{}%
  }%
  \providecommand\transparent[1]{%
    \errmessage{(Inkscape) Transparency is used (non-zero) for the text in Inkscape, but the package 'transparent.sty' is not loaded}%
    \renewcommand\transparent[1]{}%
  }%
  \providecommand\rotatebox[2]{#2}%
  \ifx\svgwidth\undefined%
    \setlength{\unitlength}{447.5731457bp}%
    \ifx\svgscale\undefined%
      \relax%
    \else%
      \setlength{\unitlength}{\unitlength * \real{\svgscale}}%
    \fi%
  \else%
    \setlength{\unitlength}{\svgwidth}%
  \fi%
  \global\let\svgwidth\undefined%
  \global\let\svgscale\undefined%
  \makeatother%
  \begin{picture}(1,0.24220397)%
    \put(0,0){\includegraphics[width=\unitlength,page=1]{graph1.pdf}}%
    \put(-0.03775483,0.01926034){\color[rgb]{0,0,0}\makebox(0,0)[lb]{\smash{}}}%
    \put(-0.00111277,0.21676995){\color[rgb]{0,0,0}\makebox(0,0)[lb]{\smash{$1,1,+1$}}}%
    \put(0.08021472,0.21676995){\color[rgb]{0,0,0}\makebox(0,0)[lb]{\smash{$1,2,+1$}}}%
    \put(0.00208154,0.02730372){\color[rgb]{0,0,0}\makebox(0,0)[lb]{\smash{$1,1,-1$}}}%
    \put(0.08340903,0.02730372){\color[rgb]{0,0,0}\makebox(0,0)[lb]{\smash{$1,2,-1$}}}%
    \put(0.53089017,0.22943668){\color[rgb]{0,0,0}\makebox(0,0)[lb]{\smash{$4,1,+1$}}}%
    \put(0.61221771,0.22943668){\color[rgb]{0,0,0}\makebox(0,0)[lb]{\smash{$4,2,+1$}}}%
    \put(0.53408448,0.00197026){\color[rgb]{0,0,0}\makebox(0,0)[lb]{\smash{$4,1,-1$}}}%
    \put(0.61541202,0.00197026){\color[rgb]{0,0,0}\makebox(0,0)[lb]{\smash{$4,2,-1$}}}%
    \put(0.73355797,0.22034479){\color[rgb]{0,0,0}\makebox(0,0)[lb]{\smash{$5,1,+1$}}}%
    \put(0.92888634,0.22034479){\color[rgb]{0,0,0}\makebox(0,0)[lb]{\smash{$5,4,+1$}}}%
    \put(0.73675223,0.02015405){\color[rgb]{0,0,0}\makebox(0,0)[lb]{\smash{$5,1,-1$}}}%
    \put(0.93208065,0.02015405){\color[rgb]{0,0,0}\makebox(0,0)[lb]{\smash{$5,4,-1$}}}%
  \end{picture}%
\endgroup%

%% file: graph2.pdf_tex
\begingroup%
  \makeatletter%
  \providecommand\color[2][]{%
    \errmessage{(Inkscape) Color is used for the text in Inkscape, but the package 'color.sty' is not loaded}%
    \renewcommand\color[2][]{}%
  }%
  \providecommand\transparent[1]{%
    \errmessage{(Inkscape) Transparency is used (non-zero) for the text in Inkscape, but the package 'transparent.sty' is not loaded}%
    \renewcommand\transparent[1]{}%
  }%
  \providecommand\rotatebox[2]{#2}%
  \ifx\svgwidth\undefined%
    \setlength{\unitlength}{443.34347754bp}%
    \ifx\svgscale\undefined%
      \relax%
    \else%
      \setlength{\unitlength}{\unitlength * \real{\svgscale}}%
    \fi%
  \else%
    \setlength{\unitlength}{\svgwidth}%
  \fi%
  \global\let\svgwidth\undefined%
  \global\let\svgscale\undefined%
  \makeatother%
  \begin{picture}(1,0.22869251)%
    \put(0,0){\includegraphics[width=\unitlength,page=1]{graph2.pdf}}%
    \put(-0.04765545,0.00665652){\color[rgb]{0,0,0}\makebox(0,0)[lb]{\smash{}}}%
    \put(-0.00112339,0.04035178){\color[rgb]{0,0,0}\makebox(0,0)[lb]{\smash{$1,1$}}}%
    \put(0.08098,0.04035178){\color[rgb]{0,0,0}\makebox(0,0)[lb]{\smash{$1,2$}}}%
    \put(0.54234885,0.00198906){\color[rgb]{0,0,0}\makebox(0,0)[lb]{\smash{$4,1$}}}%
    \put(0.62445229,0.00198906){\color[rgb]{0,0,0}\makebox(0,0)[lb]{\smash{$4,2$}}}%
    \put(0.73462491,0.2150728){\color[rgb]{0,0,0}\makebox(0,0)[lb]{\smash{$5,1,+1$}}}%
    \put(0.93181678,0.2150728){\color[rgb]{0,0,0}\makebox(0,0)[lb]{\smash{$5,4,+1$}}}%
    \put(0.73784964,0.00575428){\color[rgb]{0,0,0}\makebox(0,0)[lb]{\smash{$5,1,-1$}}}%
    \put(0.93504157,0.00575428){\color[rgb]{0,0,0}\makebox(0,0)[lb]{\smash{$5,4,-1$}}}%
  \end{picture}%
\endgroup%

%% file: graph3.pdf_tex
\begingroup%
  \makeatletter%
  \providecommand\color[2][]{%
    \errmessage{(Inkscape) Color is used for the text in Inkscape, but the package 'color.sty' is not loaded}%
    \renewcommand\color[2][]{}%
  }%
  \providecommand\transparent[1]{%
    \errmessage{(Inkscape) Transparency is used (non-zero) for the text in Inkscape, but the package 'transparent.sty' is not loaded}%
    \renewcommand\transparent[1]{}%
  }%
  \providecommand\rotatebox[2]{#2}%
  \ifx\svgwidth\undefined%
    \setlength{\unitlength}{214.67248345bp}%
    \ifx\svgscale\undefined%
      \relax%
    \else%
      \setlength{\unitlength}{\unitlength * \real{\svgscale}}%
    \fi%
  \else%
    \setlength{\unitlength}{\svgwidth}%
  \fi%
  \global\let\svgwidth\undefined%
  \global\let\svgscale\undefined%
  \makeatother%
  \begin{picture}(1,0.40646223)%
    \put(0,0){\includegraphics[width=\unitlength,page=1]{graph3.pdf}}%
    \put(0.20263944,0.38522712){\color[rgb]{0,0,0}\makebox(0,0)[lb]{\smash{$e_1$}}}%
    \put(-0.00232003,0.2214511){\color[rgb]{0,0,0}\makebox(0,0)[lb]{\smash{$a_1$}}}%
    \put(0.69730149,0.2214511){\color[rgb]{0,0,0}\makebox(0,0)[lb]{\smash{$a_2$}}}%
  \end{picture}%
\endgroup%

%% file: graph4.pdf_tex
\begingroup%
  \makeatletter%
  \providecommand\color[2][]{%
    \errmessage{(Inkscape) Color is used for the text in Inkscape, but the package 'color.sty' is not loaded}%
    \renewcommand\color[2][]{}%
  }%
  \providecommand\transparent[1]{%
    \errmessage{(Inkscape) Transparency is used (non-zero) for the text in Inkscape, but the package 'transparent.sty' is not loaded}%
    \renewcommand\transparent[1]{}%
  }%
  \providecommand\rotatebox[2]{#2}%
  \ifx\svgwidth\undefined%
    \setlength{\unitlength}{183.59099947bp}%
    \ifx\svgscale\undefined%
      \relax%
    \else%
      \setlength{\unitlength}{\unitlength * \real{\svgscale}}%
    \fi%
  \else%
    \setlength{\unitlength}{\svgwidth}%
  \fi%
  \global\let\svgwidth\undefined%
  \global\let\svgscale\undefined%
  \makeatother%
  \begin{picture}(1,0.4301044)%
    \put(0,0){\includegraphics[width=\unitlength,page=1]{graph4.pdf}}%
    \put(0.68751979,0.28583713){\color[rgb]{0,0,0}\makebox(0,0)[lb]{\smash{$e_1$}}}%
    \put(0.90321664,0.33159101){\color[rgb]{0,0,0}\makebox(0,0)[lb]{\smash{$b_1$}}}%
    \put(0.90321664,0.13768172){\color[rgb]{0,0,0}\makebox(0,0)[lb]{\smash{$b_2$}}}%
  \end{picture}%
\endgroup%

%% file: graph5.pdf_tex
\begingroup%
  \makeatletter%
  \providecommand\color[2][]{%
    \errmessage{(Inkscape) Color is used for the text in Inkscape, but the package 'color.sty' is not loaded}%
    \renewcommand\color[2][]{}%
  }%
  \providecommand\transparent[1]{%
    \errmessage{(Inkscape) Transparency is used (non-zero) for the text in Inkscape, but the package 'transparent.sty' is not loaded}%
    \renewcommand\transparent[1]{}%
  }%
  \providecommand\rotatebox[2]{#2}%
  \ifx\svgwidth\undefined%
    \setlength{\unitlength}{322.80094595bp}%
    \ifx\svgscale\undefined%
      \relax%
    \else%
      \setlength{\unitlength}{\unitlength * \real{\svgscale}}%
    \fi%
  \else%
    \setlength{\unitlength}{\svgwidth}%
  \fi%
  \global\let\svgwidth\undefined%
  \global\let\svgscale\undefined%
  \makeatother%
  \begin{picture}(1,0.31369016)%
    \put(0,0){\includegraphics[width=\unitlength,page=1]{graph5.pdf}}%
    \put(0.96604731,0.24861112){\color[rgb]{0,0,0}\makebox(0,0)[lb]{\smash{$0$}}}%
    \put(0.50343516,0.24861112){\color[rgb]{0,0,0}\makebox(0,0)[lb]{\smash{$\hat T$}}}%
    \put(0.15920821,0.24861112){\color[rgb]{0,0,0}\makebox(0,0)[lb]{\smash{$T$}}}%
    \put(0.15920821,0.12567147){\color[rgb]{0,0,0}\makebox(0,0)[lb]{\smash{$T$}}}%
    \put(0.15920821,0.00273183){\color[rgb]{0,0,0}\makebox(0,0)[lb]{\smash{$T$}}}%
    \put(0.50941416,0.12567147){\color[rgb]{0,0,0}\makebox(0,0)[lb]{\smash{$0$}}}%
    \put(0.29866061,0.00273183){\color[rgb]{0,0,0}\makebox(0,0)[lb]{\smash{$0$}}}%
    \put(0,0){\includegraphics[width=\unitlength,page=2]{graph5.pdf}}%
    \put(0.29268138,0.12567147){\color[rgb]{0,0,0}\makebox(0,0)[lb]{\smash{$\hat T$}}}%
    \put(0.19608593,0.00273183){\color[rgb]{0,0,0}\makebox(0,0)[lb]{\smash{$\hat T$}}}%
    \put(0,0){\includegraphics[width=\unitlength,page=3]{graph5.pdf}}%
  \end{picture}%
\endgroup%

%% file: graph6.pdf_tex
\begingroup%
  \makeatletter%
  \providecommand\color[2][]{%
    \errmessage{(Inkscape) Color is used for the text in Inkscape, but the package 'color.sty' is not loaded}%
    \renewcommand\color[2][]{}%
  }%
  \providecommand\transparent[1]{%
    \errmessage{(Inkscape) Transparency is used (non-zero) for the text in Inkscape, but the package 'transparent.sty' is not loaded}%
    \renewcommand\transparent[1]{}%
  }%
  \providecommand\rotatebox[2]{#2}%
  \ifx\svgwidth\undefined%
    \setlength{\unitlength}{319.06740518bp}%
    \ifx\svgscale\undefined%
      \relax%
    \else%
      \setlength{\unitlength}{\unitlength * \real{\svgscale}}%
    \fi%
  \else%
    \setlength{\unitlength}{\svgwidth}%
  \fi%
  \global\let\svgwidth\undefined%
  \global\let\svgscale\undefined%
  \makeatother%
  \begin{picture}(1,0.16419768)%
    \put(0,0){\includegraphics[width=\unitlength,page=1]{graph6.pdf}}%
    \put(0.60340959,0.00276379){\color[rgb]{0,0,0}\makebox(0,0)[lb]{\smash{$3,2,-1$}}}%
  \end{picture}%
\endgroup%

%% file: gibbs.bbl
\begin{thebibliography}{References}







\bibitem{AFP} Z.~{Ammari}, M.~{Falconi} and B.~{Pawilowski}. {On the rate of convergence for the mean-field approximation of many-body quantum dynamics}. 
\emph{Comm. Math. Sci.} \textbf{14} (2016) 1417--1442. 



\bibitem{AN} Z.~{Ammari} and F.~{Nier}. {Mean-field limit for bosons and propagation of Wigner measures}. \emph{J. Math. Phys.} \textbf{50} (2009), 042107.

\bibitem{Bach} V. Bach. Ionization energies of bosonic Coulomb systems. \emph{Lett. Math. Phys.} {\bf 21} (1991), 139--149.



\bibitem{BKS}
G.~{Ben Arous}, K.~{Kirkpatrick} and B.~{Schlein}. {A central limit
  theorem in many-body quantum dynamics}. \emph{Comm. Math. Phys.} \textbf{321} (2013), 371--417.


\bibitem{BL} R. Benguria and E. H. Lieb. Proof of the Stability of Highly Negative Ions in the
Absence of the Pauli Principle. \emph{Phys. Rev. Lett.} {\bf 50} (1983), 1771--1774.

\bibitem{BhatiaElsner} R. Bhatia, L. Elsner. The Hoffman-Wielandt inequality in infinite dimensions. \emph{Proc. Ind. Acad. Sci.} \textbf{104} (1994), no. 3, 483--494. 

\bibitem{BCS} C. Boccato, S. Cenatiempo, B. Schlein. Quantum many-body fluctuations around nonlinear Schr\"odinger dynamics. 
\emph{Ann. Inst. H. Poincar\'{e} Anal. Non Lin\'{e}aire} \textbf{18} (2017), 113--191. 

\bibitem{B} J. Bourgain. Periodic nonlinear Schr\"{o}dinger equation and invariant measures. \emph{Comm. Math. Phys.} {\bf 166} (1994), 1--26.


\bibitem{Bourgain_ZS} J. Bourgain. On the Cauchy problem and invariant measure problem for the periodic Zakharov system. \emph{Duke Math. J.} {\bf 76} (1994), 175--202.



\bibitem{B1} J. Bourgain. Invariant measures for the 2D-defocusing nonlinear Schr\"odinger equation. \emph{Comm. Math. Phys.}  {\bf 176} (1996), 421-445. 

\bibitem{B2} J. Bourgain. Invariant measures for the Gross-Pitaevskii equation. \emph{J. Math. Pures 
Appl.}  {\bf 76} (1997), 649--702.


\bibitem{Bourgain1998} J.~{Bourgain}. {Refinements of Strichartz's Inequality and applications to 2D-NLS with critical nonlinearity.} \emph{Int. Math. Research Notices.} {\bf 5} (1998), 253--283.


\bibitem{B3} J. Bourgain. Invariant measures for {NLS} in infinite volume. \emph{Comm. Math. Phys.}  {\bf 210} (2000), 605--620. 

\bibitem{BourgainBulut} J.~{Bourgain}, A.~{Bulut}.  Gibbs measure evolution in radial nonlinear wave and Schr\"{o}dinger equations on the ball. \emph{C. R. Math. Acad. Sci. Paris} {\bf 350} (2012), 571--575.
\bibitem{BourgainBulut2} J. ~{Bourgain}, A.~{Bulut}. 
Almost sure global well posedness for the radial nonlinear Schr\"{o}dinger equation on the unit ball II: the 3D case. \emph{J. Eur. Math. Soc. (JEMS).} {\bf16} (2014), 1289--1325.
\bibitem{BourgainBulut4} J.~{Bourgain}, A. ~{Bulut}. Almost sure global well posedness for the radial nonlinear Schr\"{o}dinger equation on the unit ball I: the 2D case.  \emph{Ann. Inst. H. Poincar\'{e} Anal. Non Lin\'{e}aire.} {\bf31} (2014), 1267--1288. 

\bibitem{BrydgesSlade} D. Brydges, G. Slade. Statistical mechanics of the $2D$ focusing nonlinear Schr\"{o}dinger equation. \emph{Comm. Math. Phys.} {\bf 182} (1996), 485--504.

\bibitem{BSS} S.~{Buchholz}, C.~{Saffirio} and B.~{Schlein}. {Multivariate central limit theorem in quantum dynamics}. \emph{J. Stat. Phys.} \textbf{154} (2014), 113--152.


\bibitem{BurqThomannTzvetkov} N.~{Burq}, L.~{Thomann} and N.~{Tzvetkov}. Long time dynamics for the one dimensional non linear Schr\"{o}dinger equation. \emph{Ann. Inst. Fourier (Grenoble)} {\bf 63} (2013), 2137--2198. 



\bibitem{BT} N. Burq and N. Tzvetkov. Random data Cauchy theory for supercritical wave
equations. I. Local theory. \emph{Invent. Math.} {\bf 173} (2008), 449--475.



\bibitem{BT2} N.~{Burq} and N.~{Tzvetkov}. Random data Cauchy theory for supercritical wave equations II. A global existence result. \emph{Invent. Math.} {\bf173} (2008), 477-496.


\bibitem{Cacciafesta_deSuzzoni1}
F. ~{Cacciafesta}, A.-S. ~{de Suzzoni}.
Invariant measure for the Schr\"{o}dinger equation on the real line. \emph{J. Funct. Anal.} {\bf 269} (2015), 271--324.

\bibitem{Cacciafesta_deSuzzoni2}
F. ~{Cacciafesta}, A.-S. ~{de Suzzoni}.
On Gibbs measure and weak flow for the cubic NLS with non-localised initial data. Preprint arXiv:1507.03820.


\bibitem{CLS}
L.~{Chen}, J.~{Oon Lee} and B.~{Schlein}. {Rate of convergence towards
  {H}artree dynamics}. \emph{J. Stat. Phys.} \textbf{144} (2011), no.~4, 872--903.

\bibitem{CP} T.~{Chen} and N.~{Pavlovi\'{c}}. {The quintic NLS as the mean-field limit of a boson gas with three-body interactions}. \emph{J. Funct. Anal.} \textbf{260} (2011), 959--997.

\bibitem{XC} X.~{Chen}. {Second order corrections to mean-field evolution for weakly interacting
bosons in the case of three-body interactions}. \emph{Arch. Rational Mech. Anal.} \textbf{203} (2012), 455--497.

\bibitem{CH} X.~{Chen} and J.~{Holmer}. {Focusing quantum many-body dynamics: the rigorous derivation of the $1D$ focusing cubic nonlinear Schr\"odinger equation}. 
\emph{Arch. Rational Mech. Anal.} \textbf{221} (2016), 631--676.

\bibitem{CH2} X.~{Chen} and J.~{Holmer}. {Focusing quantum many-body dynamics II: the rigorous derivation of the $1D$ focusing cubic nonlinear Schr\"odinger equation from 3D}. 
\emph{Anal.  PDE} \textbf{10--3} (2017), 589--633.



\bibitem{CO} J. Colliander and T. Oh. Almost sure well-posedness of the cubic nonlinear
Schr\"odinger equation below $\ell^2 (T)$. \emph{Duke Math. J.} {\bf 161} (2012), 367--414.

\bibitem{deS}
A.-S. de Suzzoni. Invariant measure for the cubic wave equation on the unit ball of
$\R^3$. \emph{Dyn. Partial Differ. Equ.} \textbf{8} (2011), 127--147.


\bibitem{Deng} Y.~{Deng}. Two-dimensional nonlinear Schr\"{o}dinger equation with random initial data. \emph{Anal. PDE.} {\bf 5} (2012), 913--960.


\bibitem{DengTzvetkovVisciglia} Y.~{Deng}, N.~{Tzvetkov} and N.~{Visciglia}. Invariant Measures and Long Time Behaviour for the Benjamin-Ono Equation III. \emph{Comm. Math. Phys.} {\bf 339} (2015), 815--857.


\bibitem{Dolbeault_Felmer_Loss_Paturel} J. Dolbeault, P. Felmer, M. Loss, E. Paturel. Lieb-Thirring type inequalities and Gagliardo-Nirenberg inequalities for systems. \emph{J. Funct. Ana.} \textbf{238} (2006), no. 1, 193--220.


\bibitem{E} P. Ehrenfest. Bemerkung \"uber die angen\"aherte G\"ultigkeit der klassichen Machanik
innerhalb der Quanatenmechanik. \emph{Z. Physik} {\bf 45} (1927), 455--457.

\bibitem{EESY} A.~{Elgart}, L.~{Erd\H{o}s}, B.~{Schlein}, H.-T.~{Yau}. Gross-Pitaevskii Equation as the Mean Field Limit of Weakly Coupled Bosons. \emph{Arch. Rational Mech. Anal.} \textbf{179} (2006), 265--283. 

\bibitem{ES} A.~{Elgart}, B.~{Schlein}.  Mean-field dynamics for boson
stars. \emph{Comm. Pure Applied Math.} \textbf{60} (2007), no. 4, 500--545.

\bibitem{ESY1} L.~{Erd\H{o}s}, B.~{Schlein}, H.-T.~{Yau}. Derivation of the Gross-Pitaevskii hierarchy for the dynamics of Bose-Einstein condensate. \emph{Comm. Pure Appl. Math.} \textbf{59} (2006), no. 12, 1659--1741.
\bibitem{ESY2} L.~{Erd\H{o}s}, B.~{Schlein}, H.-T.~{Yau}. Derivation of the cubic non-linear Schr\"{o}dinger equation from quantum dynamics of many-body systems, \emph{Invent. Math.} \textbf{167} (2007), no. 3, 515--614.
\bibitem{ESY3} L.~{Erd\H{o}s}, B.~{Schlein}, H.-T.~{Yau}. Rigorous derivation of the Gross-Pitaevskii equation, \emph{Phys. Rev. Lett.} \textbf{98} (2007), no.4, 040404.
\bibitem{ESY4} L.~{Erd\H{o}s}, B.~{Schlein}, H.-T.~{Yau}. Rigorous derivation of the Gross-Pitaevskii equation with a large interaction potential. \emph{J. Amer. Math. Soc.} \textbf{22} (2009), no. 4, 1099--1156
\bibitem{ESY5} L.~{Erd\H{o}s}, B.~{Schlein}, H.-T.~{Yau}. Derivation of the Gross-Pitaevskii equation for the dynamics of Bose-Einstein condensate. \emph{Ann. of Math. (2).} \textbf{172} (2010), no. 1, 291--370.

\bibitem{EY}
L.~{Erd{\H{o}}s} and H.-T.~{Yau}. {Derivation of the
  nonlinear {S}chr\"{o}dinger equation from a many-body {C}oulomb system}. \emph{Adv.
  Theor. Math. Phys.} \textbf{5} (2001), no.~6, 1169--1205.

\bibitem{FSV} M. Fannes, H. Spohn, and A. Verbeure. Equilibrium states for mean field models. \emph{J.
Math. Phys.} {\bf 21} (1980), 355--358.

\bibitem{FKP} J.~{Fr\"ohlich}, A.~{Knowles} and A.~{Pizzo}. {Atomism and quantization}. \emph{J.\ Phys.\ A: Math.\ Theor.} \textbf{40} (2007), 3033--3045.

\bibitem{FrKnScSo_2017} J.~{Fr\"{o}hlich}, A.~{Knowles}, B.~{Schlein}, V.~{Sohinger}.
{A microscopic derivation of time-dependent correlation functions of the $1D$ cubic nonlinear Schr\"{o}dinger equation}.  Preprint arXiv: 1703.04465.

\bibitem{FKS}
J.~{Fr\"{o}hlich}, A.~{Knowles} and S.~{Schwarz}. {On the  mean-field limit of bosons with {C}oulomb two-body interaction}. \emph{Comm. Math. Phys.} \textbf{288} (2009), no.~3, 1023--1059.

\bibitem{GLV} G. Genovese, R. Luc{\`{a}} and D. Valeri. Gibbs measures associated to the integrals of motion of the periodic {dNLS}. \emph{Sel. Math. New Ser.} (2016), DOI 10.1007/s000029-016-0225-2, (online version).

\bibitem{Glimm_Jaffe} J. ~{Glimm} and A ~{Jaffe}. {Quantum Physics. A Functional Integral Point of View}. \emph{Springer-Verlag}, Second edition, 1987.

\bibitem{GV}
J.~Ginibre and G.~Velo. {The classical field limit of scattering theory
  for nonrelativistic many-boson systems. {I} and {II}}. \emph{Comm. Math. Phys.} \textbf{66}
  (1979), no.~1, 37--76, and \textbf{68} (1979), no.~1, 45--68.

\bibitem{GS}  P. Grech and R. Seiringer. The excitation spectrum for weakly interacting bosons in
a trap. \emph{Comm. Math. Phys.} {\bf 322} (2013),  559--591.

\bibitem{GM} M. Grillakis, M. Machedon. Beyond mean field: on the role of pair excitations in the evolution of condensates. Preprint arXiv:1509.07371.

\bibitem{GMM1}
M.~Grillakis, M.~Machedon and D.~Margetis. {Second-order corrections to 
mean-field evolution of weakly interacting bosons. {I}.} \emph{Comm. Math. Phys.} \textbf{294} 
(2010), no.~1, 273--301.

\bibitem{GMM2}
M.~Grillakis, M.~Machedon and D.~Margetis. {Second-order corrections to mean-field evolution of weakly interacting bosons. {II}.} \emph{Adv. Math.} \textbf{228} (2011), no.~3, 1788--1815.

\bibitem{Hardy} G. H. Hardy. Divergent Series. \emph{Oxford at the Clarendon Press}, 1949. 

\bibitem{H}
K.~{Hepp}. {The classical limit for quantum mechanical correlation
  functions}. \emph{Comm. Math. Phys.} \textbf{35} (1974), 265--277.
  
\bibitem{HS} S.~{Herr} and V. ~{Sohinger}. The Gross-Pitaevskii hierarchy on general rectangular tori. \emph{Arch. Rational Mech Anal.} \textbf{220} (2016), no. 3, 1119--1158. 

\bibitem{KSS}
K.~{Kirkpatrick}, B.~{Schlein} and G.~{Staffilani}. {Derivation of the two dimensional nonlinear
{S}chr\"odinger equation from many-body quantum dynamics}. \emph{Amer. J. Math.} \textbf{133} (2011), no.~1, 91-130. 

\bibitem{K} M. K.-H. Kiessling. The Hartree limit of Born's ensemble for the ground state of a bosonic atom or ion. \emph{J. Math. Phys.} {\bf 53} (2012), 095223.

\bibitem{KP}
A.~{Knowles} and P.~{Pickl}. {Mean-field dynamics: singular potentials
  and rate of convergence}. \emph{Comm. Math. Phys.} \textbf{298} (2010), no.~1,
  101--138.

\bibitem{LRS} J. Lebowitz, H. Rose, E. Speer. Statistical mechanics of the nonlinear Schr\"odinger equation. \emph{J. Stat. Phys.} {\bf 50} (1988), 657--687.
 
\bibitem{Lewin_Nam_Rougerie} M. Lewin, P. T. Nam, N. Rougerie. Derivation of nonlinear Gibbs measures from many-body quantum mechanics. \emph{J. \'{E}c. Polytech. Math} \textbf{2} (2015), 65--115.
 
\bibitem{LNR16} M. Lewin, P. T. Nam, N. Rougerie. Bose gases at positive temperature and non-linear Gibbs measures. \emph{Proceedings of the 18th ICMP}, Santiago de Chile, July 2015.

\bibitem{LNR0} M. Lewin, P.T. Nam, N. Rougerie. Derivation of Hartree's theory for generic mean-field Bose systems. \emph{Adv. Math.} {\bf 254} (2014), 570-621.

\bibitem{Lewin_Nam_Rougerie2} M. Lewin, P. T. Nam, N. Rougerie. Gibbs measures based on 1D (an)harmonic oscillators as mean-field limits. Preprint arXiv: 1703.09422.

\bibitem{LNS} M.~{Lewin}, P.~T.~{Nam} and B.~{Schlein}. {Fluctuations around Hartree states in the mean-field regime}. \emph{Amer. J. Math} \textbf{137} (2015), 1613--1650. 

\bibitem{LNSS}  M.~{Lewin}, P.~T.~{Nam}, S. Serfaty, J.P. Solovej. Bogoliubov spectrum of interacting Bose gases. \emph{Comm. Pure Appl. Math.} {\bf 68} (2012), 413-471.

\bibitem{LS} E.H.~Lieb, R. Seiringer.
Proof of {B}ose-{E}instein condensation for dilute trapped gases.
\textit{Phys. Rev. Lett.} \textbf{88} (2002), 170409-1-4.

\bibitem{LSY} E.H.~Lieb, R. Seiringer, J. Yngvason. Bosons in a trap:
a rigorous derivation of the {G}ross-{P}itaevskii energy functional.
\textit{Phys. Rev A} \textbf{61} (2000), 043602.

\bibitem{LY} E. H. Lieb and H.-T. Yau. The Chandrasekhar theory of stellar collapse as the limit of quantum mechanics. \emph{Commun. Math. Phys.} {\bf 112} (1987), 147--174.

\bibitem{NN} P.T. Nam, M. Napiorkowski. Bogoliubov correction to the mean-field dynamics of interacting bosons. Preprint arXiv:1509.04631. To appear in \emph{Adv. Theor. Math. Phys.} 

\bibitem{NZ90} H. Neidhart and V. A. Zagrebnov. The Trotter-Kato product formula for Gibbs semigroups. \emph{Comm. Math. Phys.} \textbf{131} (1990), 333--346.

\bibitem{Nevanlinna} F. Nevanlinna, Zur Theorie der asymptotischen Potenzreihen, \emph{Ann. Acad. Sci. Fen. Ser. A} \textbf{12}, no. 3, 1918--1919.

\bibitem{NORBS} A.~{Nahmod}, T.~{Oh}, L.~{Rey-Bellet}, G.~{Staffilani}. 
Invariant weighted Wiener measures and almost sure global well-posedness for the periodic derivative NLS. \emph{J. Eur. Math. Soc.} {\bf 14} (2012), 1275--1330.

\bibitem{NRBSS} A.~{Nahmod}. L.~{Rey-Bellet}, S.~{Sheffield}, G.~{Staffilani}. Absolute continuity of Brownian bridges under certain gauge transformations. \emph{Math. Res. Lett.} {\bf 18} (2011), 875--887.

\bibitem{OQ}
T. Oh and J. Quastel. On invariant Gibbs measures conditioned on mass and
momentum. \emph{J. Math. Soc. Japan} \textbf{65} (2013), 13--35.

\bibitem{RW} G. A. Raggio and R. F. Werner. Quantum statistical mechanics of general mean field
systems. \emph{Helv. Phys. Acta} {\bf 62} (1989), 980--1003.

\bibitem{RS1} M. Reed and B. Simon. Methods of modern mathematical physics, I: functional analysis. \emph{Academic Press}, 1980.

\bibitem{RS}
I.~{Rodnianski} and B.~{Schlein}. {Quantum fluctuations and rate of
  convergence towards mean-field dynamics}. \emph{Comm. Math. Phys.} \textbf{291}
  (2009), no.~1, 31--61.

\bibitem{Schr1} E. ~{Schr\"{o}dinger}.   {Der stetige \"{U}bergang von der Mikro-zur Makromechanik}. \emph{Die Naturwissenschaften.} {\bf 14.} Jahrgang, Heft 28, S. 664--666 (1926).

\bibitem{Schr2} E. ~{Schr\"{o}dinger}.  {Gesammelte Abhandlungen}, Band {\bf 3}, ``Beitr\"{a}ge zur Quantentheorie'', p.137.

\bibitem{Se} R. Seiringer. The excitation spectrum for weakly interacting bosons. \emph{Commun. Math.
Phys.} {\bf 306} (2011), 565--578.

\bibitem{Simon74} B.~{Simon}. The $P(\Phi)_2$ Euclidean (Quantum) Field Theory. \emph{Princeton Univ. Press}, 1974.

\bibitem{Simon05} B. Simon. Trace ideals and their applications. \emph{Amer. Math. Soc.}, Second edition, 2005.

\bibitem{S2} V.~{Sohinger}. A rigorous derivation of the defocusing nonlinear Schr\"{o}dinger equation on $\mathbb{T}^3$ from the dynamics of many-body quantum systems. \emph{Ann. Inst. H. Poincar\'{e} Anal. Non Lin\'{e}aire.} \textbf{32} (2015), no. 6, 1337--1365.

\bibitem{Sokal} A.D. Sokal. An improvement of Watson's theorem on Borel summability. \emph{J. Math. Phys.} \textbf{21} (1980), no. 2, 261--263.

\bibitem{S} J. P. Solovej. Asymptotics for bosonic atoms. \emph{Lett. Math. Phys.} {\bf 20} (1990), 165--172.

\bibitem{Spohn}
H.~{Spohn}. {Kinetic equations from {H}amiltonian dynamics: {M}arkovian
  limits}. \emph{Rev. Modern Phys.} \textbf{52} (1980), no.~3, 569--615.

\bibitem{TTz}
L. Thomann and N. Tzvetkov. Gibbs measure for the periodic derivative nonlinear {S}chr\"odinger equation. \emph{Nonlinearity} \textbf{23} (2010), 2771.

\bibitem{Tz1}
N.~{Tzvetkov}.
Invariant measures for the Nonlinear Schr\"{o}dinger equation on the disc.
\emph{Dynamics of PDE.} {\bf 2} (2006) 111--160.

\bibitem{Tz}
N. Tzvetkov. Invariant measures for the defocusing nonlinear Schr\"odinger equation. \emph{Ann. Inst. Fourier (Grenoble)} \textbf{58} (2008), 2543--2604.

\bibitem{Watson} G.N. Watson. A Theory of Asymptotic Series. \emph{Philos. Trans. Soc. London, Ser. A} \textbf{211} (1912).

\bibitem{Zhidkov} P.E. Zhidkov. An invariant measure for the nonlinear Schr\"{o}dinger equation. (Russian) \emph{Dokl. Akad. Nauk SSSR} {\bf 317} (1991) 543--546; translation in Soviet Math. Dokl. {\bf 43}, 431--434.


\end{thebibliography}
